\documentclass[12pt]{article}
\usepackage{amsmath}
\usepackage{graphicx}
\usepackage{enumerate}
\usepackage{natbib}
\usepackage{url} % not crucial - just used below for the URL 

\usepackage{bibunits}  % For separate bibliographies
\usepackage{amsmath}
\usepackage{graphicx,psfrag,epsf}
\usepackage{enumerate}
\usepackage{natbib}
\usepackage{url} % not crucial - just used below for the URL 
\usepackage{amsthm}
\usepackage{booktabs}
\usepackage{framed}  
\usepackage{caption}
\usepackage{pgfplots}
\usepgfplotslibrary{dateplot}
\usetikzlibrary{snakes}
\usepackage{rotating}
\usepackage{longtable}
\usepackage{float}
\usepackage{amsfonts}
\usepackage{multirow, booktabs}
\usepackage{cleveref}
\usepackage{subcaption}
\usepackage[ruled,vlined]{algorithm2e}
\usepackage[T1]{fontenc}
\usepackage[utf8]{inputenc}
\usepackage{authblk}
\usepackage[multiple]{footmisc}
\usepackage{blindtext,titlefoot}
\usepackage{sectsty}
\usepackage{xcolor}
\usepackage{tikz}
\usetikzlibrary{positioning, shapes.geometric,calc}
\usepackage[nodisplayskipstretch]{setspace}
\usepackage{tabularx}
\usepackage{ragged2e} % added for better adjust of cells' content
\newcolumntype{L}{>{\RaggedRight}X} % for cells with left aligned content
\usepackage{lipsum} % just for dummy text

\usepackage[ruled,vlined]{algorithm2e}
\usepackage{amsmath, amssymb}
\usepackage{amsfonts, multirow, epsfig, subfig}
\usepackage{graphicx, pdflscape, verbatim, enumerate, colortbl, setspace}
\usepackage{setspace, color,bm}
\usepackage[normalem]{ulem}
\usepackage{cite}
\usepackage{multirow}
\usepackage{booktabs,array}
\usepackage{url}
\usepackage{bbm}
\usepackage{mathtools, nccmath}
\newtheorem{lemma}{Lemma}

\newtheorem{theorem}{Theorem}
\newtheorem{example}{Example}
\newtheorem{remark}{Remark}

\newtheorem{assumption}{Assumption}

\makeatletter
\newcommand*{\rom}[1]{\expandafter\@slowromancap\romannumeral #1@}
\makeatother
\let\oldfootnote\footnote
\renewcommand{\footnote}{\unskip\oldfootnote}% 

\begin{document}
\begin{bibunit}[apalike]  % Use apalike or agsm for appendix

\def\spacingset#1{\renewcommand{\baselinestretch}%
{#1}\small\normalsize} \spacingset{1}

%%%%%%%%%%%%%%%%%%%%%%%%%%%%%%%%%%%%%%%%%%%%%%%%%%%%%%%%%%%%%%%%%%%%%%%%%%%%%%
%TCIMACRO{\TeXButton{Section}{\sectionfont{\bfseries\large\sffamily}}}%
%BeginExpansion
\sectionfont{\bfseries\large\sffamily}%
%EndExpansion
%
\newcommand*\emptycirc[1][1ex]{\tikz\draw (0,0) circle (#1);} 
\newcommand*\halfcirc[1][1ex]{%
  \begin{tikzpicture}
  \draw[fill] (0,0)-- (90:#1) arc (90:270:#1) -- cycle ;
  \draw (0,0) circle (#1);
  \end{tikzpicture}}
\newcommand*\fullcirc[1][1ex]{\tikz\fill (0,0) circle (#1);} 

\let\oldthebibliography\thebibliography
\let\endoldthebibliography\endthebibliography
\renewenvironment{thebibliography}[1]{
  \begin{oldthebibliography}{#1}
    \setlength{\itemsep}{0em}
    \setlength{\parskip}{0em}
}
{
  \end{oldthebibliography}
}

%TCIMACRO{\TeXButton{Subsection}{\subsectionfont{\bfseries\sffamily\normalsize
%}}}%
%BeginExpansion
\subsectionfont{\bfseries\sffamily\normalsize}%
%EndExpansion
%
%\bibliographystyle{natbib}

\def\spacingset#1{\renewcommand{\baselinestretch}%
{#1}\small\normalsize} \spacingset{1}

%\begin{center}
    %\Large \bf Design-Based Causal Inference with Missing Outcomes: Missingness Mechanisms, Imputation-Assisted Randomization Tests, and Covariate Adjustment
%\end{center}

\begin{center}
    \Large \bf Design-Based Causal Inference with Missing Outcomes: Missingness Mechanisms, Imputation-Assisted Randomization Tests, and Covariate Adjustment
\end{center}

%Analyzing Randomized Experiments Subject to Outcome Misclassification via Integer Programming

\begin{center}
  \large  $\text{Siyu Heng}^{*, 1}$, $\text{Jiawei Zhang}^{*, 2, 3}$ and $\text{Yang Feng}^{\dagger, 1}$
\end{center}

\begin{center}
   \textit{$^{1}$Department of Biostatistics, New York University, New York, NY}
\end{center}
\begin{center}
  \textit{$^{2}$Courant Institute of Mathematical Sciences, New York University, New York, NY}
\end{center}
\begin{center}
   \textit{$^{3}$Data Science Institute, The University of Chicago, Chicago, IL}
\end{center}

\let \thefootnote\relax\footnotetext{$^{*}$Siyu Heng (siyuheng@nyu.edu) and Jiawei Zhang (jz4721@nyu.edu) contributed equally to this work.}

 \let \thefootnote\relax\footnotetext{$^{\dagger}$Contact: Yang Feng (yang.feng@nyu.edu), Department of Biostatistics, New York University, New York, NY. }

\begin{abstract}

Design-based causal inference, also known as randomization-based or finite-population causal inference, is one of the most widely used causal inference frameworks, largely due to the merit that its validity can be guaranteed by study design (e.g., randomized experiments) and does not require assuming specific outcome-generating distributions or super-population models. Despite its advantages, design-based causal inference can still suffer from other issues, among which outcome missingness is a prevalent and significant challenge. This work systematically studies the outcome missingness problem in design-based causal inference. First, we propose a general and flexible outcome missingness mechanism that can facilitate finite-population-exact randomization tests of no treatment effect. Second, under this general missingness mechanism, we propose a general framework called ``imputation and re-imputation" for conducting randomization tests in design-based causal inference with missing outcomes. We prove that our framework can still ensure finite-population-exact type-I error rate control even when the imputation model was misspecified or when unobserved covariates or interference exist in the missingness mechanism. Third, we extend our framework to conduct covariate adjustment in randomization tests and construct finite-population-valid confidence regions with missing outcomes. Our framework is evaluated via extensive simulation studies and applied to a large-scale randomized experiment. 

\end{abstract}

\noindent%

%\spacingset{1.9} % DON'T change the spacing!
\spacingset{1.75} % Previously used spacing

\section{Introduction}\label{sec: introduction}

Design-based causal inference, also known as randomization-based or finite-population causal inference, is one of the most fundamental causal inference frameworks and has become increasingly popular in the past two decades \citep{rosenbaum2002observational, imbens2015causal, choi2017estimation, li2017general, basse2019randomization, luo2021leveraging, zhang2023randomization, cohen2024no, zhao2024adjust}. Compared with model-based (i.e., super-population or sampling-based) causal inference, design-based causal inference has two characteristics. First, the only probability distribution needed in design-based causal inference is the randomization of treatment assignments, which can be guaranteed by design, and no distributional assumptions on potential outcomes are required. Second, design-based causal inference focuses explicitly on the study subjects by conditioning on their potential outcomes and makes causal inference more relevant for the study subjects in hand. See \citet{athey2017econometrics} and \citet{li2023randomization} for detailed discussions.

Although design-based causal inference is immune to outcome model/distribution misspecification, it may still suffer from other data-related issues, among which missingness in outcomes is a common one. However, contrary to the rapidly growing literature on missing outcomes issues in model-based (super-population) causal inference, existing literature on outcome missingness in design-based (finite-population) causal inference is scarce, largely due to the challenge that design-based causal inference does not impose any outcome models and, therefore, cannot directly incorporate existing model-based approaches to missing outcomes. Most existing approaches to design-based causal inference with missing outcomes fall into one of the following two categories: 1) \textbf{Class-One Approaches} (in which covariate information is ignored when handling missing outcomes): These methods involve non-informative imputation techniques (e.g., median, mean, or worst-case imputation) or complete case analysis \citep{edgington2007randomization, kennes2012choice, rosenberger2019randomization, heussen2023randomization}; or 2) \textbf{Class-Two Approaches} (in which covariate information is used): These methods integrate model-based multiple imputation or weighting approaches into design-based inference \citep{ivanova2022randomization}.

However, both class-one and class-two approaches have significant limitations. In class-one approaches, all the fruitful covariate information will be discarded when handling missing outcomes, which can substantially decrease statistical efficiency. For class-two approaches, previous simulation studies have shown that the corresponding type-I error rates are often either inflated or conservative, even under some simple settings \citep{ivanova2022randomization}. This is because class-two approaches adopt an empirical model-based strategy for incorporating covariate information and fail to stick to the randomization design. To our knowledge, in design-based causal inference, there is no existing framework for handling outcome missingness that can both incorporate covariate information and rigorously follow the original randomization design.

Our work fills this gap by systematically investigating the outcome missingness problem in design-based causal inference. Specifically, the main contributions of our work include:
\begin{itemize}
    \item We use the potential outcomes framework to propose a general outcome missingness mechanism that can facilitate finite-population-exact randomization tests of no treatment effect in design-based causal inference. We show that the proposed outcome missingness mechanism is flexible enough to cover many existing outcome missingness mechanisms as special cases.
    
    \item Built upon the above missingness mechanism, we propose a general framework called ``imputation and re-imputation" for constructing randomization tests of no treatment effect with outcome missingness. Our framework can incorporate any existing missing outcomes imputation algorithms (from simple linear models to advanced machine learning algorithms) to leverage covariate information during imputation and meanwhile guarantee finite-population-exact type-I error rate control, even when the imputation algorithm or model was misspecified or when unobserved covariates or interference exists in the missingness mechanism. This is also the first design-based causal testing framework that can both 1) incorporate covariate information when handling missing outcomes and 2) rigorously follow the randomization design.

    \item We extend our framework to allow conducting covariate adjustment in finite-population-exact randomization tests with missing outcomes, seeking to further increase power. Some other recent studies investigate covariate adjustment with outcome missingness from the model-based (super-population) perspective (e.g., \citealp{zhao2023covariate, wang2024handling}). Parallelly, our work presents the first rigorous strategy for covariate adjustment with outcome missingness in the design-based (finite-population) setting, which can ensure finite-population-exact type-I error rate control \textit{without} the need for correctly specifying either the outcome-generating or the missingness model.
    
    \item We show how to construct a design-based (finite-population) confidence region with missing outcomes by inverting the proposed imputation-assisted randomization tests.
\end{itemize}
The finite-population-exact type-I error rate control and gains in power under our framework are examined via simulation studies. As a data application, we apply our framework to a large-scale cluster randomized experiment called the Work, Family, and Health Study (\citealp{WFHS2018}). We have also developed open-source and user-friendly \textsf{Python} and \textsf{R} packages \texttt{iArt} (\textbf{i}mputation-\textbf{A}ssisted \textbf{r}andomization \textbf{t}ests) for implementation of our framework (URL: \url{https://iart.readthedocs.io/en/latest/}).

\section{Brief Review of Design-Based Causal Inference With Complete Outcome Data}

We first review the classic framework for design-based causal inference with complete outcome data (i.e., no missing outcomes) \citep{rosenbaum2002observational, imbens2015causal}. Suppose there are $I\geq 1$ strata (or blocks, study centers, or matched sets), and $n_{i}\geq 2$ subjects in stratum $i$. In total, there are $N=\sum_{i=1}^{I}n_{i}$ subjects. For a given subject $j$ in stratum $i$, let $Z_{ij}\in \{0, 1\}$ denote the treatment indicator ($Z_{ij}=1$ if receiving treatment and $Z_{ij}=0$ if receiving control), $\mathbf{x}_{ij}\in \mathbbm{R}^{p_{x}}$ the observed covariates (i.e., pre-treatment characteristics measured before the treatment), and $\mathbf{u}_{ij}\in \mathbbm{R}^{p_{u}}$ the unobserved covariates. Suppose there are $K$ outcomes of interest ($K\geq 1$), and let $Y_{ijk}$ denote the $k$-th outcome of subject $j$ in stratum $i$ (we can omit $k$ in the single outcome case). Let $\mathbf{Z}=(Z_{11}, \dots, Z_{In_{I}})\in \{0,1\}^{N}$ denote the treatment indicator vector, $\mathbf{X}= \{\mathbf{x}_{ij}: i\in[I], j \in [n_{i}]\}$ all the observed covariate information measured before the treatment assignments (henceforth, we let $[a]=\{1, \dots, a\}$ for any positive integer $a$), $\mathbf{U}=\{\mathbf{u}_{ij}: i \in [I], j \in [n_{i}]\}$ all the unobserved covariates, $\mathbf{Y}_{k}=(Y_{11k}, \dots, Y_{In_{I}k})$ all the $k$-th outcomes, and $\mathbf{Y}=(\mathbf{Y}_{1},\dots, \mathbf{Y}_{K})$ all the $K$ outcomes. Under the potential outcomes framework (\citealp{neyman1923application, rubin1974estimating}), we let $Y_{ijk}(1)$ and $Y_{ijk}(0)$ denote the potential $k$-th outcome of subject $j$ in stratum $i$ under treatment and that under control, respectively. Then, we have $Y_{ijk}=Z_{ij}Y_{ijk}(1)+(1-Z_{ij})Y_{ijk}(0)$. In design-based causal inference, one of the most widely considered null hypotheses is Fisher's sharp null of no effect, of which the form (for the $k$-th outcome) is $H_{0, k}: Y_{ijk}(1)=Y_{ijk}(0)$ for all $i, j$. When there are multiple outcomes of interest (i.e., $K\geq 2$), the overall sharp null is $H_{0}: \bigcap_{k=1}^{K}H_{0, k}$. The hypothesis testing part of our work focuses on Fisher's sharp null and its extensions in the single and multiple outcomes cases, which is often regarded as a first step in a causal analysis and widely considered in both methodology research and applied research \citep{rosenbaum2002observational, rosenbaum2020design, imbens2015causal, li2017general}. 

In design-based causal inference,  all the potential outcomes $\mathcal{Y}=\{(Y_{ijk}(0), Y_{ijk}(1)): i \in [I], j \in [n_{i}], k \in [K]\}$ are fixed values, with the only probability distribution that enters into statistical analysis being the randomization of treatment assignments (i.e., the probability distribution of $\mathbf{Z}$ prespecified by design) (\citealp{rosenbaum2002observational, imbens2015causal, li2017general}). For example, in a stratified randomized experiment, researchers randomly assign the treatment to $m_{i}$ subjects among the $n_{i}$ total subjects in stratum $i$, where each $m_{i}$ is a prespecified number. Specifically, let $\mathcal{Z}=\{ \mathbf{z}=(z_{11},\dots, z_{In_{I}})\in \{0, 1\}^{N}:\sum_{j=1}^{n_{i}}z_{ij}=m_{i}, i \in [I]\}$ denote all possible treatment assignments, then we have 
\begin{equation}\label{eqn: randomization assumption}
    P(\mathbf{Z}=\mathbf{z})=\prod_{i=1}^{I}{n_{i}\choose m_{i}}^{-1} \quad \text{for all $\mathbf{z}\in \mathcal{Z}$}. 
\end{equation}
The study design (\ref{eqn: randomization assumption}) (i.e., stratified randomization design) reduces to a completely randomized experiment when $I=1$ and reduces to a paired randomized experiment when $n_{i}=2$ and $m_{i}=1$ for all $i$. In a Bernoulli randomized experiment, the probability of receiving the treatment for each study subject independently and identically follows $\text{Bernoulli}(1/2)$. In a cluster randomized experiment, the unit of treatment assignments is cluster instead of individual. See \citet{imbens2015causal} and \citet{rosenbaum2020design} for detailed introductions to various randomization-based study designs. After fixing a study design (e.g., complete randomization, stratified randomization, Bernoulli randomization, or cluster randomization, among many others), we let $\Omega$ denote the collection of all possible treatment assignments and $\mathcal{P}$ the probability distribution of $\mathbf{Z}$ over $\Omega$, induced by the study design (i.e., we have $\mathbf{Z} \sim \mathcal{P}$). For example, in a stratified randomized experiment, the $\Omega$ is the aforementioned $\mathcal{Z}$, and the $\mathcal{P}$ is the randomization distribution (\ref{eqn: randomization assumption}). In the single outcome case, given the observed value $t$ of some test statistic $T(\mathbf{Z}, \mathbf{Y})$ (e.g., the Wilcoxon rank sum test, the permutational $t$-test, the Kolmogorov-Smirnov test, among many others), we can calculate the finite-population-exact (permutational) $p$-value under $H_{0}$:
\begin{equation}\label{eqn: exact p-value}
   P(T(\mathbf{Z}, \mathbf{Y}) \geq t\mid H_{0})=\sum_{\mathbf{z}\in \Omega}\mathbbm{1}\{T(\mathbf{Z}=\mathbf{z}, \mathbf{Y})\geq t\}\times P(\mathbf{Z}=\mathbf{z}), \text{ where $\mathbf{Z} \sim \mathcal{P}$}. 
\end{equation}
When the sample size $N$ is large, the exact $p$-value (\ref{eqn: exact p-value}) can be approximated via the Monte Carlo method with arbitrary precision (\citealp{rosenbaum2002observational, imbens2015causal}). In the multiple outcomes case, for testing the overall null $H_{0}: \bigcap_{k=1}^{K}H_{0, k}$, researchers can either directly use a test statistic $T(\mathbf{Z}, \mathbf{Y}_{1},\dots, \mathbf{Y}_{K})$ that combines all the $K$ outcomes (e.g., some linear combination of $K$ statistics calculated separately from the $K$ individual outcomes \citep{rosenbaum2016using}), or combine the $K$ $p$-values based on (\ref{eqn: exact p-value}) for each individual null $H_{0,k}$ using the Bonferroni correction or the Holm-Bonferroni method. In addition to testing the null effect, researchers may also want to construct confidence regions for causal parameters. For example, consider the following general class of treatment effect models: 
\begin{equation}\label{eqn: effect model}
    Y_{ijk}(1)=f_{k}(Y_{ijk}(0), \Vec{\beta}_{k}, \mathbf{x}_{ij}), \text{ for } i \in [I], j \in [n_{i}], k \in [K], 
\end{equation}
where each $f_{k}$ is a prespecified map from $Y_{ijk}(0)$ to $Y_{ijk}(1)$ that involve causal parameter(s) $\Vec{\beta}_{k}$ and may also involve the observed covariates $\mathbf{x}_{ij}$. For example, when $f_{k}=Y_{ijk}(0)+\beta_{k}$, the treatment effect model (\ref{eqn: effect model}) is the additive constant effect model. When $f_{k}=\beta_{k}Y_{ijk}(0)$, the model (\ref{eqn: effect model}) corresponds to the multiplicative effect model. When there are interaction terms between $\Vec{\beta}_{k}$ and $\mathbf{x}_{ij}$, the model (\ref{eqn: effect model}) allows for heterogeneous treatment effects. Researchers can obtain a design-based (randomization-based) confidence region for the causal parameter(s) $(\Vec{\beta}_{1},\dots, \Vec{\beta}_{K})$ via inverting the randomization tests for testing Fisher's sharp null with transformed outcomes $Y_{ijk, \Vec{\beta}_{k}}=Z_{ij}Y_{ijk}+(1-Z_{ij})f_{k}(Y_{ijk}, \Vec{\beta}_{k}, \mathbf{x}_{ij})$ (because each $Y_{ijk, \Vec{\beta}_{k}}$ will be invariant under any $\mathbf{Z}$ if model (\ref{eqn: effect model}) holds with the prespecified causal parameters $\Vec{\beta}_{k}$); see \citet{rosenbaum2002observational} for details.

\section{Constructing Finite-Population-Exact and Powerful Randomization Tests with Missing Outcomes}

\subsection{Clarifying the Outcome Missingness Mechanism for Design-Based Testing}

In many practical studies, among the $K\geq 1$ outcomes, one or multiple of them have missingness. Let $M_{ijk}$ denote the observed outcome missingness indicator for the $k$-th outcome of subject $j$ in stratum $i$ (hereafter, we call it subject $ij$ for short): $M_{ijk}=1$ if the $k$-th outcome of subject $ij$ was missing, and $M_{ijk}=0$ otherwise. That is, if we let $Y_{ijk}^{*}$ denote the realized value (possibly with missingness) of the $k$-th outcome (henceforth called ``realized $k$-th outcome") of subject $ij$, then we have $Y_{ijk}^{*}=Y_{ijk}$ if $M_{ijk}=0$ (recall that each $Y_{ijk}$ denotes the true, possibly unobserved, outcome) and $Y_{ijk}^{*}=\text{``Missing"}$ if $M_{ijk}=1$. We let $\mathbf{M}_{k}=(M_{11k}, \dots, M_{In_{I}k})$, $\mathbf{M}=(\mathbf{M}_{1}, \dots, \mathbf{M}_{K})$, $\mathbf{Y}^{*}_{k}=(Y_{11k}^{*}, \dots, Y_{In_{I}k}^{*})$, and $\mathbf{Y}^{*}=(\textbf{Y}_{1}^{*}, \dots, \textbf{Y}_{K}^{*})$. We still let $\mathbf{Y}=(\mathbf{Y}_{1}, \dots, \mathbf{Y}_{K})$ (where $\mathbf{Y}_{k}=(Y_{11k}, \dots, Y_{In_{I}k})$) denote the true, possibly unobserved, outcomes. Since each missingness status $M_{ijk}$ is a post-treatment variable, following the potential outcomes framework \citep{neyman1923application, rubin1974estimating, frangakis1999addressing}, we let $M_{ijk}(\mathbf{z})$ denote the potential missingness status of the $k$-th outcome of subject $ij$ under treatment assignments $\mathbf{z}\in \Omega$. Then, the observed missingness status $M_{ijk}=M_{ijk}(\mathbf{z})$ if $\mathbf{Z}=\mathbf{z}$. Note that here we allow arbitrary interference in the missingness mechanism, i.e., we allow $M_{ijk}$ to depend on $Z_{i^{\prime}j^{\prime}}$ for $i\neq i^{\prime}$ and/or $j\neq j^{\prime}$. We then let $Y_{ijk}^{*}(\mathbf{z})$ denote the potential $k$-th realized outcome of subject $ij$, where $Y_{ijk}^{*}(\mathbf{z})=Y_{ijk}(z_{ij})$ if $M_{ijk}(\mathbf{z})=0$ and $Y_{ijk}^{*}(\mathbf{z})=\text{``Missing"}$ if $M_{ijk}(\mathbf{z})=1$. Then, the observed realized outcome $Y_{ijk}^{*}=Y_{ijk}^{*}(\mathbf{z})$ if $\mathbf{Z}=\mathbf{z}$. In design-based causal inference, all the potential post-treatment variables are fixed values (including the potential true outcomes $\mathcal{Y}=\{(Y_{ijk}(1), Y_{ijk}(0)): i \in [I], j \in [n_{i}], k \in [K]\}$, the potential outcome missingness indicators $\mathcal{M}=\{M_{ijk}(\mathbf{z}): i \in [I], j \in [n_{i}], k \in [K], \mathbf{z}\in\Omega \}$, and the potential realized outcomes $\mathcal{Y}^{*}=\{Y_{ijk}^{*}(\mathbf{z}): i \in [I], j \in [n_{i}], k \in [K], \mathbf{z}\in\Omega \}$), and the only source of randomness is the randomization of treatment assignments $\mathbf{Z}$ \citep{rosenbaum2002observational}. 

%In addition to missingness in outcomes, missingness in observed covariates may also exist in practice. We let $\mathbf{X}^{*}$ denote the realized values (possibly with missingness) of the observed covariates among the study subjects and still let $\mathbf{X}$ denote the corresponding true values of the covariates among the study subjects. As will be discussed in detail in Remark~\ref{rem: missing covariates}, a nice property of our proposed framework is that it does not require any assumptions on the missingness mechanisms of observed covariates $\mathbf{X}$. 

In practice, missingness can occur not only in outcomes but also in observed covariates. We let $\mathbf{X}^{*}$ denote the realized values (possibly with missingness) of the observed covariates for the study subjects and still let $\mathbf{X}$ denote the corresponding true covariates values for the study subjects. As will be discussed in detail in Remark~\ref{rem: missing covariates}, a key advantage of our proposed framework is that it does not require any assumptions on the missingness mechanisms of the observed covariates $\mathbf{X}$.

%\begin{remark}\label{rem: missing covariates}
%\emph{\textcolor{red}{The imputation and re-imputation framework (i.e., Algorithms \ref{alg: part} and \ref{alg: part with covariate adjustment}) also allows the observed covariates $\mathbf{X}^{*}$ to have missingness (assuming these observed covariates are measured before the treatment \citep{zhao2024adjust}) because Theorems \ref{thm: hypothesis testing} and \ref{thm: hypothesis testing with covariate adjustment} still hold when both missingness in observed covariates and unobserved covariates exist (i.e., validity of the imputation and re-imputation framework still holds).}} %Therefore, in Algorithms~\ref{alg: part} and \ref{alg: part with covariate adjustment}, we use $\mathbf{X}^{*}$ to denote the possibly incomplete values of the covariates $\mathbf{X}$ (i.e., potential covariates missingness in $\mathbf{X}^{*}$) to distinguish it with the complete values of the observed covariates $\mathbf{X}$. }} 
%\end{remark}

In this work, we consider the following general class of randomization designs:
\begin{assumption}[A General Class of Randomization Designs]\label{assump: randomization design}
   The distribution of treatment assignments $\mathbf{Z}$ is fixed before observing the outcome data and known to us by design. 
\end{assumption}
Assumption~\ref{assump: randomization design} includes many commonly used randomization designs, such as complete randomization, stratified/blocked randomization, paired randomization, Bernoulli randomization, and cluster randomization. Basically, as long as a randomization design is not adaptive or sequential (in which the study population or the distribution of $\mathbf{Z}$ is not fixed before observing part of the outcome data), it naturally satisfies Assumption~\ref{assump: randomization design}.

We now delineate the assumption concerning the outcome missingness mechanism. There are many works clarifying the assumptions concerning the outcome missingness mechanisms required for model-based causal inference with missing outcomes under various settings. In contrast, parallel literature in design-based causal inference is highly scarce. To our knowledge, the only existing work that explicitly listed the assumptions concerning outcome missingness mechanism for randomization-based inference is \citet{heussen2023randomization}. Our discussions on outcome missingness mechanisms broadly differ from the related discussions in \citet{heussen2023randomization} in three aspects. First, our discussions on outcome missingness mechanisms are built on the potential outcomes framework \citep{neyman1923application, rubin1974estimating, frangakis1999addressing}, which is the common language in design-based causal inference. Second, contrary to \citet{heussen2023randomization}, our work does not impose any randomness on missingness and strictly follows the design-based causal inference framework. Third, our discussions explicitly clarify the minimal (weakest) assumption concerning (i) the outcome missingness mechanism required for conducting finite-population-exact design-based hypothesis testing (see Assumption~\ref{assump: conditional indep} stated below), which is typically more flexible than the assumptions considered in \citet{heussen2023randomization} (see Remarks 1--5 below), and (ii) the missingness mechanism required for constructing finite-population-valid design-based confidence regions with missing outcomes (see Assumption~\ref{assump: conditional indep with more restrictions} in Section~\ref{sec: CI with missing outcomes}), respectively.

\begin{assumption}[A General Outcome Missingness Mechanism for Design-Based Hypothesis Testing]\label{assump: conditional indep}
   Conditional on the finite-population dataset in hand, there exists an unknown map $\eta: \mathbb{R}^{N \times (p_{x}+p_{u}+K)} \rightarrow \{0,1\}^{N \times K}$ such that $(\mathbf{M}_{1},\dots, \mathbf{M}_{K})=\eta(\mathbf{X}, \mathbf{U}, \mathbf{Y}_{1}, \dots, \mathbf{Y}_{K})$. 
\end{assumption}

 Assumption~\ref{assump: conditional indep} claims that, conditional on the observed covariates $\mathbf{X}$ and unobserved covariates $\mathbf{U}$ (including unobserved pre-treatment error terms), the treatment assignments $\mathbf{Z}$ can only affect the missingness status $(\mathbf{M}_{1},\dots, \mathbf{M}_{K})$ through their effects on the post-treatment outcomes $(\mathbf{Y}_{1}, \dots, \mathbf{Y}_{K})$. Therefore, another version of Assumption~\ref{assump: conditional indep} is: the outcome missingness $\mathbf{M}$ is unaffected by the treatment assignments $\mathbf{Z}$ \textit{as long as the true outcomes $(\mathbf{Y}_{1}, \dots, \mathbf{Y}_{K})$ are unaffected by $\mathbf{Z}$} (i.e., as long as Fisher's sharp null $H_{0}$ holds). 
 
This assumption is more likely to hold when the randomized experiment is double-blinded and the outcomes are self-reported by study participants. In such cases, participants have no knowledge about the treatment assignments $\mathbf{Z}$ but know the true values of their outcomes. These outcome variables may, in turn, influence whether the study participants choose to respond to outcome-related survey questions. However, Assumption~\ref{assump: conditional indep} is violated if the treatment assignments $\mathbf{Z}$ directly affect outcome missingness $\mathbf{M}$ beyond their effects on the outcomes $\mathbf{Y}$. For example, in an unblinded experiment, the study participants are aware of their treatment status $Z_{ij}$, which may influence their willingness to remain in the study or respond to follow-up survey questions.

We here make five important remarks on Assumption~\ref{assump: conditional indep}. 
\begin{remark}\emph{
    The missingness mechanism described in Assumption~\ref{assump: conditional indep} is a missing-not-at-random missingness mechanism as it allows the missingness status to depend on unobserved covariates and/or true outcome values. Also, because missingness may depend on the true outcome itself (self-censoring), Assumption~\ref{assump: conditional indep} cannot be falsified based on the observed data without additional identification conditions \citep{d2010new, wang2014instrumental}.}
\end{remark}
\begin{remark}\emph{
    Assumption~\ref{assump: conditional indep} allows interference between study subjects -- one subject's covariates and outcomes may affect another subject's outcome missingness status. }
\end{remark}
\begin{remark}\emph{
    In the multiple outcomes case, Assumption~\ref{assump: conditional indep} also allows interference between different outcome variables -- one outcome variable can affect the missingness of another outcome variable. Such a scenario often occurs when the multiple outcomes are measured sequentially in randomized trials \citep{ivanova2022randomization}. }
\end{remark}
\begin{remark}\label{rem: self-censoring in Assumption 2}
\emph{
    Assumption~\ref{assump: conditional indep} allows each missingness status $M_{ijk}$ to depend on $Y_{ijk}$, which is also called ``self-censoring" in the literature \citep{d2010new, wang2014instrumental}. Self-censoring in outcomes is common in practice, especially for self-reported outcomes.}
\end{remark}
\begin{remark}\emph{
    The map $\eta$ in Assumption~\ref{assump: conditional indep} can be any unknown map.}
\end{remark}

In summary, the outcome missingness mechanism considered in Assumption~\ref{assump: conditional indep} allows much flexibility as it allows dependence on observed and unobserved covariates, interference between study subjects and different outcome variables, self-censoring, and arbitrary form of these dependencies (i.e., the map $\eta$ can be an arbitrary unknown map). See Figures 1 and 2 for an illustration of the proposed outcome missingness mechanism in Assumption~\ref{assump: conditional indep}. Meanwhile, researchers should also be aware that Assumption~\ref{assump: conditional indep} will be violated if the treatment assignments $\mathbf{Z}$ can directly affect outcome missingness $\mathbf{M}$, beyond their effect through the path $\mathbf{Z} \rightarrow \mathbf{Y} \rightarrow \mathbf{M}$ (see Remark~\ref{rem: counterexample to Assumption 3} in Appendix D for the detailed causal diagram). This violation is particularly likely in unblinded randomized experiments. 

%Actually, as will be discussed in Remark~\ref{rem: weakest assumption} in Section~\ref{subsec: PART}, Assumption~\ref{assump: conditional indep} is almost the minimal (weakest) assumption (concerning outcome missingness) required for conducting a finite-population-exact design-based test of no treatment effect $H_{0}$. 

\begin{figure}[h]
    \centering
    \begin{minipage}[t]{.5\textwidth}
        \centering
        \begin{tikzpicture}
            \node (XU) at (0,0) {$X, U$};
            \node (Y) [below = 1.2cm of XU] {$Y$};
            \node (M) [right = of Y] {$M$};
            \node (Z) [left = of Y] {$Z$};
    
            \path[->] (XU) edge (Y);
            \path[->] (XU) edge (M);
            \path[->] (Z) edge (Y);
            \path[->] (Y) edge (M);
            \draw[->] (XU) -- (Z);
    
            % Adding cross in red for XU to Z
            \path (XU) -- (Z) coordinate[pos=0.5] (midpoint1);
            \draw[ thick] ($(midpoint1) + (-0.08,0) +(-0.15,-0.15)$) -- ($(midpoint1)+ (-0.08,0)+ (0.15,0.15)$);
            \draw[ thick] ($(midpoint1)+ (-0.08,0)+(0.15,-0.15)$) -- ($(midpoint1)+ (-0.08,0)+(-0.15,0.15)$);
            \node (annotation1) [left = 0.8cm of midpoint1, align=center] {\scriptsize Removed by \\ \scriptsize Randomization};
            \draw[->] (annotation1) -- ($(midpoint1) + (-0.08,0)$);
        \end{tikzpicture}
        \caption{The single outcome case of \\ Assumptions \ref{assump: randomization design} and \ref{assump: conditional indep}.}
    \end{minipage}%
    \begin{minipage}[t]{.5\textwidth}
        \centering
        \begin{tikzpicture}
            \node (XU) at (0,0) {$X, U$};
            \node (Y1) [below =0.5cm of XU] {$Y_{1}$};
            \node (Y2) [below = 0.2cm of Y1] {$Y_{2}$};
            \node (M1) [right =1cm of Y1] {$M_{1}$};
            \node (M2) [right =1cm of Y2] {$M_{2}$};
            \node (Z) [left = 1cm of Y1] {$Z$};

            \path[->] (XU) edge (Y1);
            \path[->] (XU) edge (M1);
            \path[->] (XU) edge[bend right=45] (Y2);
            \path[->] (XU) edge[bend left=75] (M2);
            \path[->] (Z) edge (Y1);
            \path[->] (Z) edge (Y2);
            \path[->] (Y1) edge (M1);
            \path[->] (Y1) edge (M2);
            \path[->] (Y2) edge (M1);
            \path[->] (Y2) edge (M2);
            \draw[->] (XU) -- (Z);

            % Adding cross in red for XU to Z
            \path (XU) -- (Z) coordinate[pos=0.5] (midpoint2);
            \draw[thick] ($(midpoint2)+(-0.15,-0.15)$) -- ($(midpoint2)+(0.15,0.15)$);
            \draw[ thick] ($(midpoint2)+(0.15,-0.15)$) -- ($(midpoint2)+(-0.15,0.15)$);
            \node (annotation2) [left = 0.8cm of midpoint2, align=center] {\scriptsize Removed by\\ \scriptsize Randomization};
            \draw[->] (annotation2) -- (midpoint2);
        \end{tikzpicture}
         \caption{The multiple outcomes case of Assumptions \ref{assump: randomization design} and \ref{assump: conditional indep}.}
    \end{minipage}
\end{figure}

\subsection{Imputation-Assisted Randomization Tests for Design-Based Hypothesis Testing With Missing Outcomes: An ``Imputation and Re-Imputation" Framework}\label{subsec: PART}

As mentioned in Section~\ref{sec: introduction}, most existing studies on design-based hypothesis testing with missing outcomes rely on either non-informative imputation approaches that ignore covariate information (\textbf{Class-One Approaches}) or empirical model-based imputation approaches that do not ensure precise control of the type-I error rate, even in relatively simple settings (\textbf{Class-Two Approaches}). To derive design-based tests with both finite-population-exact type-I error rate control and improved power, we propose a general class of imputation-assisted randomization tests implemented via an ``imputation and re-imputation" framework described in Algorithm~\ref{alg: part}.

The main idea of the imputation and re-imputation framework described in Algorithm~\ref{alg: part} is to embed a missing outcome imputation procedure into each permutation run of a randomization test. Therefore, the imputation and re-imputation framework offers a way to conduct a randomization test with missing outcomes that can fully incorporate covariate information for imputation and, meanwhile, faithfully follow the original randomization design. 1) Compared with a randomization test incorporated with some empirical multiple imputation or weighting method (i.e., the aforementioned Class-Two Approaches), Algorithm~\ref{alg: part} can achieve finite-population-exact type-I error rate control under the general missingness mechanism described in Assumption~\ref{assump: conditional indep}, as will be shown in Theorem~\ref{thm: hypothesis testing} and simulation studies in Section~\ref{subsec: simulation studies}. 2) Compared with a randomization test based on non-informative imputation such as median/mean imputation (i.e., the aforementioned Class-One Approaches), Algorithm~\ref{alg: part} can improve power by making full use of the covariates and outcome information, as will also be shown in Section~\ref{subsec: simulation studies}. 

\begin{algorithm} %\SetAlgoNoLine
\SetAlgoLined
\caption{The ``imputation and re-imputation" framework.} \label{alg: part}
\vspace*{0.12 cm}
\KwIn{A prespecified number of re-imputation runs $L$ (e.g., $L=10{,}000$), some chosen imputation algorithm $\mathcal{G}$, some chosen test statistics $T$ ($T$ can be any test statistic; see Remark~\ref{rem: test statistics}), the randomization design $\mathcal{P}$ (see Assumption~\ref{assump: randomization design}), the observed treatment indicators $\mathbf{Z}$, the observed covariates $\mathbf{X}^{*}$ with possible missingness, and the observed realized outcomes with missingness $\mathbf{Y}^{*}=(\mathbf{Y}_{1}^{*}, \dots, \mathbf{Y}_{K}^{*})$ for the $K$ outcomes ($K\geq 1$). Note that $\mathbf{Y}^{*}$ contains all the observed missingness status information $\mathbf{M}$. }
\vspace*{0.12 cm}
\begin{enumerate}
    \item Use $(\mathbf{Z}, \mathbf{X}^{*}, \mathbf{Y}^{*})$ and the chosen imputation algorithm $\mathcal{G}$ ($\mathcal{G}$ can be any imputation algorithm; see Remark~\ref{rem: G algorithm}) to obtain the full imputed outcomes $\widehat{\mathbf{Y}}$: 
    \begin{equation*}
 \mathcal{G}: (\mathbf{Z}, \mathbf{X}^{*}, \mathbf{Y}^{*})\mapsto \widehat{\mathbf{Y}}=(\widehat{\mathbf{Y}}_{1}, \dots, \widehat{\mathbf{Y}}_{K}). \quad \text{(The Imputation Step)}
    \end{equation*}
    We then calculate $t=T(\mathbf{Z}, \widehat{\mathbf{Y}})$.
    \item For each $l=1, \dots, L$:
\begin{enumerate}
        \item Randomly generate $\mathbf{Z}^{(l)}$ according to the randomization design $\mathcal{P}$.
    \item Obtain the full imputed outcomes $\widehat{\mathbf{Y}}^{(l)}$ based on $(\mathbf{Z}^{(l)}, \mathbf{X}^{*}, \mathbf{Y}^{*})$ and algorithm $\mathcal{G}$:
    \begin{equation*}
  \mathcal{G}: (\mathbf{Z}^{(l)}, \mathbf{X}^{*}, \mathbf{Y}^{*})\mapsto \widehat{\mathbf{Y}}^{(l)}=(\widehat{\mathbf{Y}}_{1}^{(l)}, \dots, \widehat{\mathbf{Y}}_{K}^{(l)}). \quad \text{(The Re-Imputation Step)}
    \end{equation*}
    \item Calculate $T^{(l)}=T(\mathbf{Z}^{(l)}, \widehat{\mathbf{Y}}^{(l)})$ for the $l$-th re-imputation run.
    \end{enumerate} 
\end{enumerate}
\textbf{Output:} The approximate finite-population-exact $p$-value $\widehat{p}=\frac{1}{L}\sum_{l=1}^{L}\mathbbm{1}\{T^{(l)}\geq t\}$.
\end{algorithm}

We here use a simple example in the single outcome case to illustrate the necessity of the re-imputation step in Algorithm~\ref{alg: part}. Consider a completely randomized experiment with four subjects (indexed by IDs 1--4), among which there are two treated subjects and two control subjects. Therefore, there are six possible treatment assignments: $\mathbf{Z}^{(1)}=(1,0,1,0)$, $\mathbf{Z}^{(2)}=(1,0,0,1)$, $\mathbf{Z}^{(3)}=(0,1,1,0)$, $\mathbf{Z}^{(4)}=(0,1,0,1)$, $\mathbf{Z}^{(5)}=(1,1,0,0)$, and $\mathbf{Z}^{(6)}=(0,0,1,1)$. Suppose the observed treatment assignment $\mathbf{Z}=\mathbf{Z}^{(1)}=(1,0,1,0)$ and the corresponding observed realized outcomes $\mathbf{Y}^{*}=(1, 0, \text{NA}, \text{NA})$, where ``NA" means ``Missing." We consider two imputation-assisted randomization tests: one based on the imputation and re-imputation approach and one based on a one-shot imputation procedure (i.e., without re-imputation). For both approaches, we consider permutational $t$-test statistic (i.e., the sum of outcome values among the treated subjects).

We first calculate the exact $p$-value under sharp null $H_{0}$ using the imputation and re-imputation approach. We consider a natural outcome imputation algorithm $\mathcal{G}$ defined as the following: for each treated (or control) subject whose outcome is missing, we impute its missing outcome using the mean value of all the non-missing outcomes among the treated (or control) subjects; if all the treated (or control) subjects' outcomes are missing, we impute $0.5$ for all the missing outcomes of the treated (or control) subjects. This imputation algorithm is perhaps the most natural one without covariate information. For each treatment assignment $\mathbf{Z}^{(s)}$ $(s=1,\dots, 6)$, based on the observed realized outcomes $\mathbf{Y}^{*}=(1, 0, \text{NA}, \text{NA})$ (fixed under sharp null $H_{0}$ and Assumption~\ref{assump: conditional indep} for different treatment assignments) and imputation algorithm $\mathcal{G}$, the corresponding imputed outcomes $\widehat{\mathbf{Y}}^{(s)}$ can be obtained and the corresponding imputation-assisted permutational $t$-test statistic $T^{(s)}=\mathbf{Z}^{(s)}(\widehat{\mathbf{Y}}^{(s)})^{T}$ can be calculated, which are summarized in Figure~\ref{fig: illustrating re-imput}. We take Datasets 1, 3, and 5 as illustrative examples to show the detailed imputation process. For Dataset 1 in Figure~\ref{fig: illustrating re-imput} (the observed dataset), based on the imputation algorithm $\mathcal{G}$ defined above and the training dataset $(\mathbf{Z}^{(1)}, \mathbf{Y}^{*})$, the imputed outcome for a treated subject will be $1$ (because the only treated subject with available outcome data has outcome $1$) and that for a control subject will be $0$ (because the only control subject with available outcome data has outcome $0$). Instead, for Dataset 3 in Figure~\ref{fig: illustrating re-imput} (a permuted dataset), based on the imputation algorithm $\mathcal{G}$ and the training dataset $(\mathbf{Z}^{(3)}, \mathbf{Y}^{*})$, the imputed outcome for a treated subject will be $0$ (because the only treated subject with available outcome data has outcome $0$ in the permuted dataset $(\mathbf{Z}^{(3)}, \mathbf{Y}^{*})$) and that for a control subject will be $1$ (because the only control subject with available outcome data has outcome $1$ in the permuted dataset). For Dataset 5 in Figure~\ref{fig: illustrating re-imput} (a permuted dataset), the imputation algorithm $\mathcal{G}$, when applied to the training dataset $(\mathbf{Z}^{(5)}, \mathbf{Y}^{*})$, assigns an imputed outcome of $0.5$ for a control subject. This is because there are no control subjects with available outcome data in the permuted dataset $(\mathbf{Z}^{(5)}, \mathbf{Y}^{*})$. Therefore, we assign the imputed outcome of $0.5$, representing the mean/median of the possible outcomes of $0$ and $1$, to control subjects with missing outcomes. Meanwhile, for Dataset 5, both treated subjects have available outcome data, and their ``imputed" outcomes are naturally set as the values of their available outcomes. See Figure~\ref{fig: illustrating re-imput} for the imputed results for all six datasets. Therefore, the one-sided $p$-value reported by the imputation and re-imputation approach given the observed dataset (Dataset 1 in Figure~\ref{fig: illustrating re-imput}) is $P(T\geq T^{(1)}\mid H_{0})=|\{\mathbf{Z}^{(s)}: T^{(s)}=\mathbf{Z}^{(s)}(\widehat{\mathbf{Y}}^{(s)})^{T}\geq 2, s=1,\dots, 6\}|/6=1/3$.
% The example table
\newcolumntype{C}[1]{>{\centering\arraybackslash}p{#1}}
\begin{figure}[t]
% First row of tables
\begin{tabular}
{|C{0.57cm}|C{0.57cm}|C{0.57cm}|C{0.57cm}|C{0.57cm}|}
\hline
\multicolumn{5}{|c|}{Dataset 1 (Observed)} \\
\hline
ID & \(\mathbf{Z}^{(1)}\) & \(\mathbf{Y}^{*}\) & \(\widehat{\mathbf{Y}}^{(1)}\) & \(\widehat{\mathbf{Y}}_{*}^{(1)}\) \\
\hline
1 & $1$ & $1$ & $1$ & $1$ \\
2 & $0$ & $0$ & $0$ & $0$ \\
3 & $1$ & NA & $1$ & $1$ \\
4 & $0$ & NA & $0$ & $0$ \\
\hline
\multicolumn{5}{|c|}{\( T^{(1)}= 2; T^{(1)}_{*} = 2 \)} \\
\hline
\end{tabular}
\quad % Space between tables
% Second row of tables
\begin{tabular}{|C{0.57cm}|C{0.57cm}|C{0.57cm}|C{0.57cm}|C{0.57cm}|}
\hline
\multicolumn{5}{|c|}{Dataset 2 (Permuted)} \\
\hline
ID & \(\mathbf{Z}^{(2)}\) & \(\mathbf{Y}^{*}\) & \(\widehat{\mathbf{Y}}^{(2)}\) & \(\widehat{\mathbf{Y}}_{*}^{(2)}\) \\
\hline
1 & $1$ & $1$ & $1$ & $1$ \\
2 & $0$ & $0$ & $0$ & $0$ \\
3 & $0$ & NA & $0$ & $0$ \\
4 & $1$ & NA & $1$ & $1$ \\
\hline
\multicolumn{5}{|c|}{\( T^{(2)} = 2; T^{(2)}_{*} = 2 \)} \\
\hline
\end{tabular}
\quad % Space between tables
\begin{tabular}{|C{0.57cm}|C{0.57cm}|C{0.57cm}|C{0.57cm}|C{0.57cm}|}
\hline
\multicolumn{5}{|c|}{Dataset 3 (Permuted)} \\
\hline
ID & \(\mathbf{Z}^{(3)}\) & \(\mathbf{Y}^{*}\) & \(\widehat{\mathbf{Y}}^{(3)}\) & \(\widehat{\mathbf{Y}}_{*}^{(3)}\) \\
\hline
1 & $0$ & $1$ & $1$ & $1$ \\
2 & $1$ & $0$ & $0$ & $0$ \\
3 & $1$ & NA & $0$ & $1$ \\
4 & $0$ & NA & $1$ & $0$ \\
\hline
\multicolumn{5}{|c|}{\( T^{(3)} = 0; T^{(3)}_{*}= 1 \)} \\
\hline
\end{tabular}

\vspace{0.5cm} % Space between rows of tables

\begin{tabular}{|C{0.57cm}|C{0.57cm}|C{0.57cm}|C{0.57cm}|C{0.57cm}|}
\hline
\multicolumn{5}{|c|}{Dataset 4 (Permuted)} \\
\hline
ID & \(\mathbf{Z}^{(4)}\) & \(\mathbf{Y}^{*}\) & \(\widehat{\mathbf{Y}}^{(4)}\) & \(\widehat{\mathbf{Y}}_{*}^{(4)}\) \\
\hline
1 & $0$ & $1$ & $1$ & $1$ \\
2 & $1$ & $0$ & $0$ & $0$ \\
3 & $0$ & NA & $1$ & $0$ \\
4 & $1$ & NA & $0$ & $1$ \\
\hline
\multicolumn{5}{|c|}{\( T^{(4)} = 0; T^{(4)}_{*} = 1 \)} \\
\hline
\end{tabular}
\quad % Space between tables
\begin{tabular}{|C{0.57cm}|C{0.57cm}|C{0.57cm}|C{0.57cm}|C{0.57cm}|}
\hline
\multicolumn{5}{|c|}{Dataset 5 (Permuted)} \\
\hline
ID & \(\mathbf{Z}^{(5)}\) & \(\mathbf{Y}^{*}\) & \(\widehat{\mathbf{Y}}^{(5)}\) & \(\widehat{\mathbf{Y}}_{*}^{(5)}\) \\
\hline
1 & $1$ & $1$ & $1$ & $1$ \\
2 & $1$ & $0$ & $0$ & $0$ \\
3 & $0$ & NA & $0.5$ & $0$ \\
4 & $0$ & NA & $0.5$ & $0$ \\
\hline
\multicolumn{5}{|c|}{\( T^{(5)} = 1; T^{(5)}_{*} = 1 \)} \\
\hline
\end{tabular}
\quad
\begin{tabular}{|C{0.57cm}|C{0.57cm}|C{0.57cm}|C{0.57cm}|C{0.57cm}|}
\hline
\multicolumn{5}{|c|}{Dataset 6 (Permuted)} \\
\hline
ID & \(\mathbf{Z}^{(6)}\) & \(\mathbf{Y}^{*}\) & \(\widehat{\mathbf{Y}}^{(6)}\) & \(\widehat{\mathbf{Y}}_{*}^{(6)}\) \\
\hline 
1 & $0$ & $1$ & $1$ & $1$ \\
2 & $0$ & $0$ & $0$ & $0$ \\
3 & $1$ & NA & $0.5$ & $1$ \\
4 & $1$ & NA & $0.5$ & $1$ \\
\hline
\multicolumn{5}{|c|}{\( T^{(6)} = 1; T^{(6)}_{*} = 2 \)} \\
\hline
\end{tabular}
\caption{An illustrative example of the necessity of re-imputation.}
\label{fig: illustrating re-imput}
\end{figure}

We then calculate the exact $p$-value under sharp null $H_{0}$ using a one-shot imputation procedure (i.e., without re-imputation). Specifically, to facilitate a fair comparison, we still use the aforementioned algorithm $\mathcal{G}$ and a training dataset (e.g., that obtained from sample splitting or some external dataset) that is the same as that used in the imputation step of the imputation and re-imputation approach (i.e., the observed Dataset 1 in Figure~\ref{fig: illustrating re-imput}). This will give us an ``oracle" imputation model $\mathcal{G}_{*}$: we impute value 1 for missing outcomes among the treated subjects and value 0 for missing outcomes among the control subjects. Therefore, for Dataset 1 (the observed dataset), the imputed outcomes given by the one-shot imputation approach are the same as those given by the imputation and re-imputation approach, i.e., we have $\mathcal{G}_{*}(\mathbf{Z}^{(1)}, \mathbf{Y}^{*})=\mathcal{G}(\mathbf{Z}^{(1)}, \mathbf{Y}^{*})=(1, 0, 1, 0)$. However, an intrinsic difference between the imputation and re-imputation approach using algorithm $\mathcal{G}$ and a one-shot imputation approach using model $\mathcal{G}_{*}$ is that the former one will re-train the imputation model using $\mathcal{G}$ for each of the six permutations of $\mathbf{Z}$ (i.e., the re-imputation step). In contrast, the latter one will not re-train the imputation model and will stick to $\mathcal{G}_{*}$ (i.e., imputing the missing outcome as value $1$ for treated subjects with missing outcomes and value $0$ for control subjects with missing outcomes) in each permutation. We can then calculate the imputed outcomes reported by this one-shot imputation approach (denoted as $\widehat{\mathbf{Y}}_{*}^{(s)}$) based on each permuted dataset $(\mathbf{Z}^{(s)}, \mathbf{Y}^{*})$, as well as the corresponding permutational $t$-test statistic $T^{(s)}_{*}=\mathbf{Z}^{(s)}(\widehat{\mathbf{Y}}^{(s)}_{*})^{T}$; see Figure~\ref{fig: illustrating re-imput} for details. Therefore, the one-sided (greater than) exact $p$-value reported by the one-shot imputation approach is $P(T_{*}\geq T^{(1)}_{*}\mid H_{0})=|\{\mathbf{Z}^{(s)}: T^{(s)}_{*}=\mathbf{Z}^{(s)}(\widehat{\mathbf{Y}}^{(s)}_{*})^{T}\geq 2, s=1,\dots, 6\}|/6=1/2$, which is much larger than the $p$-value 1/3 reported by the imputation and re-imputation approach. 

The above example is intended to illustrate the necessity of re-imputation compared with the one-shot imputation approach. It is not to show the superior performance of the imputation and re-imputation approach over classic randomization tests based on non-informative imputation (e.g., median or mean imputation). If we use median/mean imputation in the above example, the imputed outcomes will always be $(1, 0, 0.5, 0.5)$ for all six possible treatment assignments $\mathbf{Z}^{(s)}$ $(s=1,\dots, 6)$, and the corresponding $p$-value is $1/3$, which is the same as the $p$-value reported by the imputation and re-imputation approach. This is because there is no covariate information or imbalance of missingness proportion among the treated versus control groups in this simple example. In Section~\ref{subsec: simulation studies}, we will conduct comprehensive simulation studies to show that, when covariate information is available, compared with the commonly used non-informative imputation approach in design-based causal inference, the imputation and re-imputation approach can increase statistical power by making better use of the covariate information.

Here are some additional remarks on Algorithm~\ref{alg: part}.
\begin{remark}\emph{
    For the full imputed outcomes $\widehat{\mathbf{Y}}=(\widehat{\mathbf{Y}}_{1},\dots, \widehat{\mathbf{Y}}_{K})$, in which $\widehat{\mathbf{Y}}_{k}=(\widehat{Y}_{11k},\dots, \widehat{Y}_{In_{I}k})$, we have each $\widehat{Y}_{ijk}=Y_{ijk}^{*}=Y_{ijk}$ if missing indicator $M_{ijk}=0$ and $\widehat{Y}_{ijk}$ equals the outcome value imputed by some pre-specified imputation algorithm $\mathcal{G}$ if $M_{ijk}=1$. }
\end{remark}

\begin{remark}\label{rem: G algorithm}\emph{
    As will be shown in Theorem~\ref{thm: hypothesis testing}, the finite-sample validity of the $p$-value reported by Algorithm~\ref{alg: part} does \textit{not} require the outcome imputation algorithm/model $\mathcal{G}$ to be correctly specified. That is, the algorithm $\mathcal{G}$ for outcome imputation can be chosen from any existing imputation algorithms based on practical researchers' preference, such as $k$-nearest neighbors imputation, hot-deck imputation \citep{andridge2010review}, or chained equations imputation based on linear regressions (ordinary or regularized) or flexible machine learning methods (e.g., boosting or deep neural networks) \citep{pedregosa2011scikit}. The choice of the outcome imputation algorithm/model $\mathcal{G}$ embedded in Algorithm~\ref{alg: part} does \textit{not} affect the finite-sample validity of the reported $p$-value but may affect power. }
\end{remark}

\begin{remark}\label{rem: test statistics}\emph{
   The test statistic $T=T(\mathbf{Z}, \widehat{\mathbf{Y}})$ can be any randomization-based test statistic based on $\mathbf{Z}$ and $\widehat{\mathbf{Y}}$. In the multiple outcomes case (when $K\geq 2$), we can either 1) choose $T(\mathbf{Z}, \widehat{\mathbf{Y}})$ to be some overall test statistic that uses the information of all the $K$ outcomes (\citealp{rosenbaum2016using}), or 2) combine the $K$ individual $p$-values through the Bonferroni correction or the Holm-Bonferroni method. See Remark~\ref{rem: test statistics with additional details} in Appendix D for more details.}
\end{remark}

\begin{remark}\label{rem: missingness as test}
    \emph{The test statistic $T_{M}(\mathbf{Z}, \mathbf{M})=\sum_{i=1}^{I}\sum_{j=1}^{n_{i}}\{Z_{ij}(\sum_{k=1}^{K}M_{ijk})\}$ (i.e., the number of missing outcomes among the treated subjects), which has been used in previous studies to test if the treatment would affect missingness, is in principle also valid for testing the null $H_{0}$ under Assumption~\ref{assump: conditional indep}. However, in many settings, the $T_{M}$ is not practically useful for testing the null $H_{0}$ because it has no statistical power when the treatment has an effect on the true outcomes but no effect on outcome missingness status (i.e., when Assumption~\ref{assump: conditional indep with more restrictions} in Section~\ref{sec: CI with missing outcomes} holds, which is an important special case of Assumption~\ref{assump: conditional indep}). In contrast, randomization tests constructed by Algorithm~\ref{alg: part} do not have this deficiency.
    }
\end{remark}

The following Theorem~\ref{thm: hypothesis testing} shows that, under Assumptions \ref{assump: randomization design} and \ref{assump: conditional indep},  as the number of re-imputation runs $L$ increases, the $p$-value $\widehat{p}$ reported by Algorithm~\ref{alg: part} converges almost surely to the true $p$-value under $H_{0}$ with rate $O_{p}(1/\sqrt{L})$. This result does not depend on the sample size $N$ and does \textit{not} require correct specification of the missingness model.
\begin{theorem}\label{thm: hypothesis testing}
   Consider the imputation and re-imputation framework in Algorithm~\ref{alg: part} paired with any chosen imputation algorithm $\mathcal{G}: (\mathbf{Z}, \mathbf{X}^{*}, \mathbf{Y}^{*})\mapsto \widehat{\mathbf{Y}}=(\widehat{\mathbf{Y}}_{1}, \dots, \widehat{\mathbf{Y}}_{K})$ and any chosen test statistic $T=T(\mathbf{Z}, \widehat{\mathbf{Y}})$. Let $t$ be the observed value of $T$ based on the observed data $(\mathbf{Z}, \mathbf{X}^{*}, \mathbf{Y}^{*})$ and imputation algorithm $\mathcal{G}$. Let $p=P(T\geq t\mid H_{0})$ denote the true finite-population-exact $p$-value under $H_{0}$. For the approximate $p$-value $\widehat{p}$ reported by Algorithm~\ref{alg: part}, under Assumptions \ref{assump: randomization design} and \ref{assump: conditional indep}, we have $\widehat{p}\xrightarrow{a.s.} p$ as the number of re-imputation runs $L\rightarrow \infty$. Moreover, for any $\epsilon>0$ and for all $L$, we have $P(| \widehat{p}-p|\geq \epsilon)\leq 2\exp(-2L\epsilon^{2})$.

\end{theorem}

The detailed proofs of all theorems in this paper can be found in Appendix A in the online supplementary materials. Here are three additional remarks we would like to offer.
\begin{remark}\label{rem: weakest assumption}\emph{
    As shown in Appendix A.1, the key point of finite-population-exact type-I error rate control for testing $H_{0}$ with missing outcomes is that, \textit{under the null $H_{0}$},  the observed realized outcomes $\mathbf{Y}^{*}$ (including the observed missingness status $\mathbf{M}$) are invariant under different treatment assignments $\mathbf{Z}$. This also agrees with the principle of testing $H_{0}$ with complete outcome data using classic randomization tests \citep{rosenbaum2002observational, imbens2015causal}. Instead, if the missingness indicators $\mathbf{M}$ also depend on some post-treatment variables other than $(\mathbf{X}, \mathbf{U}, \mathbf{Y})$ (i.e., if Assumption~\ref{assump: conditional indep} was violated), then $\mathbf{M}$ can vary with different treatment assignments $\mathbf{Z}$ \textit{even under the null $H_{0}$}, which makes it impossible to construct a finite-population-exact randomization test for $H_{0}$ without additional strong assumptions. }
\end{remark}

\begin{remark}\label{rem: missing covariates}
    \emph{
    Theorem~\ref{thm: hypothesis testing} (as well as Theorem~\ref{thm: hypothesis testing with covariate adjustment} in Section~\ref{sec: covariate adjustment}) remains valid when missingness occurs in observed covariates. The key rationale behind this property is that the covariates $\mathbf{X}^{*}$, which may contain missing values, are measured before the treatment assignments. As a result, the missingness patterns in $\mathbf{X}^{*}$ remain unchanged across different $\mathbf{Z}$. Similarly,  under Fisher's sharp null $H_{0}$ and Assumption~\ref{assump: conditional indep}, the observed realized outcomes $\mathbf{Y}^{*}$, which may also have missing values, are invariant under different $\mathbf{Z}$. These properties ensure that the only source of randomness in the imputation-assisted randomization test described in Algorithm~\ref{alg: part} comes from the randomization of treatment assignments $\mathbf{Z}$. Therefore, by drawing Monte Carlo simulations from the randomization distribution of $\mathbf{Z}$, the re-imputation step in Algorithm~\ref{alg: part} (as well as Algorithm~\ref{alg: part with covariate adjustment} in Section~\ref{sec: covariate adjustment}) can asymptotically fully characterize the distribution of the test statistic under $H_{0}$. See proofs in Appendices A.1 and A.2 and the simulation studies in Appendix B.3 for more details. }
\end{remark}

\begin{remark}
   \emph{Note that whenever we use the term ``finite-population-exact" throughout this paper (as referenced in Algorithm~\ref{alg: part} and Theorem~\ref{thm: hypothesis testing}), it denotes that a certain object (e.g., $p$-value) is exact with respect to the finite population in the study dataset which could be of any sample size. It is important to emphasize that the asymptotics discussed in Theorems 1--3 of this paper are for increasing the number of re-imputation runs, $L$ (which can be sufficiently large given adequate computational resources), and do not necessitate an increase in the sample size, $N$. This is because our methods and results hold for any $N$.}
\end{remark}

\subsection{Simulation Studies}\label{subsec: simulation studies}

We conduct simulation studies to investigate the type-I error rate and power of imputation-assisted randomization tests constructed by the imputation and re-imputation framework described in Algorithm~\ref{alg: part}. We compare these tests with both the randomization tests based on non-informative imputation approaches such as median imputation (referred to as the \textbf{Class-One Approaches} in Section~\ref{sec: introduction}) and the randomization tests on based model-based imputation approaches (referred to as the \textbf{Class-Two Approaches} in Section~\ref{sec: introduction}).

We consider a stratified randomized experiment with $I$ strata. In each stratum, there are $10$ subjects, among which we randomly assign $5$ subjects to receive the treatment and the remaining $5$ subjects to receive the control. We generate a five-dimensional observed covariates vector $(x_{ij1}, \dots, x_{ij5})$ for each subject $j$ in stratum $i$ using the following data generating process: $(x_{ij1}, x_{ij2}) \overset{\text{i.i.d.}}{\sim} \mathcal{N}\left(\begin{pmatrix} \frac{1}{2} \\ -\frac{1}{3} \end{pmatrix}, \begin{pmatrix} 1 & \frac{1}{2} \\ \frac{1}{2} & 1 \end{pmatrix}\right)$, $(x_{ij3}, x_{ij4}) \overset{\text{i.i.d.}}{\sim} \text{Laplace}\left(\begin{pmatrix} 0 \\ \frac{1}{\sqrt{3}} \end{pmatrix}, \begin{pmatrix} 1 & \frac{1}{\sqrt{2}} \\ \frac{1}{\sqrt{2}} & 1 \end{pmatrix}\right)$, and $x_{ij5} \overset{\text{i.i.d.}}{\sim} \text{Bernoulli}(1/3)$.  We also generate an aggregate unobserved covariate $u_{ij}$ (including the unobserved error term in the outcome missingness generating process) for each subject $ij$ according to $u_{ij} \overset{\text{i.i.d.}}{\sim} N(0, 0.2)$. In the outcome-generating process, we include a stratum-level random effect $\alpha_{i}$ and an individual-level random effect $\epsilon_{ij}$, where $\alpha_{i} \overset{\text{i.i.d.}}{\sim} N(0, 0.1)$ and $\epsilon_{ij} \overset{\text{i.i.d.}}{\sim} N(0, 0.2)$. Therefore, the total variance of unobserved terms is $\text{var}(u_{ij})+\text{var}(\alpha_{i})+\text{var}(\epsilon_{ij})=0.5$ (note that we will normalize the coefficients of the terms involving observed covariates to make the variance contributed by observed terms and that contributed by unobserved terms comparable). We consider the following four data-generating models consisting of a model for generating the true outcome $Y_{ij}$ and a model for generating the missingness indicator $M_{ij}$, with increasing complexity:
\begin{itemize}
    \item  Model 1 (Constant treatment effect; linear model for the true outcome; linear selection model for the missingness status; without interference in the missingness mechanism): $Y_{ij} = \beta Z_{ij} + \frac{1}{\sqrt{5}} \sum \limits_{p=1}^{5} x_{ijp} + u_{ij} + \alpha_{i}+\epsilon_{ij}$ and $M_{ij} = \mathbbm{1} \left\{ \frac{1}{\sqrt{5}} \sum \limits_{p=1}^{5} p x_{ijp} + Y_{ij} + u_{ij} > \lambda \right\}$.
        
    \item Model 2 (Constant treatment effect; non-linear model for the true outcome; non-linear selection model for the missingness status; without interference in the missingness mechanism): $Y_{ij} = \beta Z_{ij} + \frac{1}{\sqrt{5}} \sum \limits_{p=1}^{5} x_{ijp}+ \frac{1}{5} \sum \limits_{p=1}^{5} \sum \limits_{p'=1}^{5} x_{ijp} \sigma(1 - x_{ijp^{\prime}}) + u_{ij} + \alpha_{i} + \epsilon_{ij}$ and $M_{ij} = \mathbbm{1} \left\{ \frac{1}{\sqrt{5}} \sum \limits_{p=1}^{5} p x_{ijp} + \frac{1}{\sqrt{5}} \sum \limits_{p=1}^{5} p \cos(x_{ijp}) + 10 \sigma(Y_{ij}) + u_{ij} > \lambda \right\}$.

    \item Model 3 (Heterogeneous treatment effects; non-linear model for the true outcome; non-linear selection model for the missingness status; without interference in the missingness mechanism): $Y_{ij} = \beta Z_{ij} \left(1 + x_{ij1} + \frac{1}{\sqrt{5}} \sum_{p=1}^{5} |x_{ijp}| \right) + \frac{1}{\sqrt{5}} \sum_{p=1}^{5} x_{ijp} + \frac{1}{5} \sum_{p=1}^{5} \sum_{p'=1}^{5} x_{ijp} \sigma(1 - x_{ijp^{\prime}}) + u_{ij} + \alpha_{i}+\epsilon_{ij}$, and the model for $M_{ij}$ is the same as in Model 2.

    \item Model 4 (Heterogeneous treatment effects; non-linear model for the true outcome; non-linear selection model for the missingness status; with interference in the missingness mechanism): the model for $Y_{ij}$ is the same as in Model 3 and\\ $M_{ij} = \mathbbm{1} \left\{ \frac{1}{\sqrt{5}} \sum \limits_{p=1}^{5} p x_{ijp} + \frac{1}{\sqrt{5}} \sum \limits_{p=1}^{5} p \cos(x_{ijp}) + 10 \sigma(Y_{ij}) + u_{ij} + \frac{1}{3}\sum \limits_{j'=1}^{10}x_{ij'1} +  \frac{1}{3}\sum \limits_{j'=1}^{10}Y_{ij'} > \lambda \right\}$.
 \end{itemize}   
In Models 1--4, each $\lambda$ is a tuning parameter for controlling the outcome missingness rate (we set $\lambda$ such that the outcome missingness rate is $50\%$ for each simulated dataset), and the $\sigma$ function is defined as $\sigma(x)=\exp(x)/(1+\exp(x))$. For each model, we consider two scenarios: a small sample size scenario in which we set the total sample size $N=50$ (corresponding to $I=5$) and a large sample size scenario in which we set $N=1000$ (corresponding to $I=100$). In each model and each simulation scenario, we implement the following four methods for design-based hypothesis testing with missing outcomes: 
\begin{itemize}
    \item Method 1 (Median Imputation): Classic randomization test based on median imputation for missing outcomes. 
    
    \item Method 2 (Algo 1 -- Linear): Imputation-assisted randomization test based on the imputation and re-imputation framework described in Algorithm~\ref{alg: part}, in which we set the embedded outcome imputation algorithm $\mathcal{G}$ to be chained equations imputation based on linear regression (more specifically, Bayesian ridge regression). 
    
    \item Method 3 (Algo 1 -- Boosting): Imputation-assisted randomization test based on the imputation and re-imputation framework described in Algorithm~\ref{alg: part}, in which we set the embedded outcome imputation algorithm $\mathcal{G}$ to be chained equations imputation based on a boosting algorithm in machine learning. Specifically, in the small sample size scenario (when $N=50$), the boosting algorithm we choose is the well-known XGBoost algorithm \citep{chen2016xgboost}, which is a gradient boosting decision tree algorithm that has relatively robust performance even under a small or moderate sample size. In the large sample size scenario (when $N=1000$), the boosting algorithm we choose is the widely used LightGBM algorithm \citep{ke2017lightgbm}, which is a newer gradient boosting decision tree algorithm that has less computational cost than XGBoost and is particularly suitable for implementing more re-imputation runs (i.e., larger $L$) in Algorithm~\ref{alg: part} with large datasets. 
    
    \item Method 4 (Oracle): Classic randomization test with the complete outcome data (with the oracle/true outcomes). 
\end{itemize}

When studying the type-I error rate control, we also compare our methods with two model-based imputation methods (one based on a linear imputation model and one based on boosting) for randomization tests with missing outcomes. The key difference between these model-based methods and our proposed framework is that these model-based methods do \textit{not} have the re-imputation step; instead, they keep the imputed outcomes fixed when permuting the treatment labels in a randomization test (\citealp{ivanova2022randomization}). As will be shown in Figure~\ref{fig: single outcome simulations}, these model-based methods fail to ensure valid type-I error rate control under the given simulation settings. Therefore, to ensure fair comparisons, we include these model-based methods only in type-I error rate evaluations and exclude them from the power simulations. Also, for all the above methods, we use an adjusted Wilcoxon rank sum test statistic $T_{\text{adj}}(\mathbf{Z}, \widehat{\mathbf{Y}})$ (one-sided), in which $T_{\text{adj}}(\mathbf{Z}, \widehat{\mathbf{Y}})=\sum_{i=1}^{I}\sum_{j=1}^{n_{i}}Z_{ij}\text{rank}_{\text{adj}}(\widehat{Y}_{ij})$, where we set $\text{rank}_{\text{adj}}(\widehat{Y}_{ij})=\sum_{i^{\prime}j^{\prime}: M_{i^{\prime}j^{\prime}}=0}\mathbbm{1}\{Y_{ij}\geq Y_{i^{\prime}j^{\prime}}\}$ if $M_{ij}=0$ and $\text{rank}_{\text{adj}}(\widehat{Y}_{ij})=\sum_{i^{\prime}j^{\prime}: M_{i^{\prime}j^{\prime}}=1}\mathbbm{1}\{Y_{ij}\geq Y_{i^{\prime}j^{\prime}}\}$ if $M_{ij}=1$. That is, $T_{\text{adj}}(\mathbf{Z}, \widehat{\mathbf{Y}})$ considers the ranks among the non-missing outcomes and those among the imputed values for the missing outcomes separately. Therefore, the statistical scores (i.e., adjusted ranks) contributed by the subjects with non-missing outcome data are unchanged under Methods 1--4, and the differences in the statistical power of Methods 1--4 can only be due to their different performances in imputing missing outcomes. This facilitates a fair and transparent comparison.

The simulated type-I error rate (obtained via setting the effect size $\beta=0$ in the outcome-generating model, with level $\alpha=0.05$ and 10,000 simulated datasets) and the simulated power (corresponding to $\beta> 0$ in the outcome generation, with level $\alpha=0.05$ and 2000 simulated datasets) are reported in Figure~\ref{fig: single outcome simulations}. The simulation results in Figure~\ref{fig: single outcome simulations} deliver the following messages. First, the imputation and re-imputation framework ensures finite-population-exact type-I error rate control (for the $\beta=0$ case) with either linear imputation models (Method 2) or a flexible machine learning imputation algorithm (Method 3) across all Models 1--4. This confirms the theoretical guarantee of finite-population-exact type-I error rate control of the imputation and re-imputation framework, even when the imputation algorithm/model is misspecified or when unobserved covariates or interference exists in the missingness mechanism, as proved in Theorem~\ref{thm: hypothesis testing}. In contrast, model-based imputation methods fail to maintain valid type-I error rate control, exhibiting type-I error rates significantly exceeding the nominal 0.05 significance level under the considered simulation settings. Second, the imputation and re-imputation framework can increase power by using an imputation algorithm (either based on a simple linear model or a flexible machine learning algorithm) to incorporate the covariate information into a randomization test with imputed outcomes. Third, when the sample size is small (e.g., $N=50$), the imputation and re-imputation approach based on a linear model (i.e., Method 1) is typically more powerful than that based on a boosting algorithm under Models 1--4. Fourth, as the sample size increases (e.g., $N=1000$), and as the outcome and missingness models evolve from simpler (Model 1) to more complex forms (Models 3 and 4), shifting from linear to nonlinear models and from constant to heterogeneous treatment effects, the incorporation of a flexible machine learning algorithm such as boosting into the imputation and re-imputation framework becomes increasingly advantageous. 

\begin{figure}[H]
      \centering
        \captionsetup[subfigure]{skip=2pt} % Adjust this value as needed
        \begin{minipage}{\textwidth}
          \centering
          \includegraphics[width=\textwidth, trim=0 10 0 0, clip]{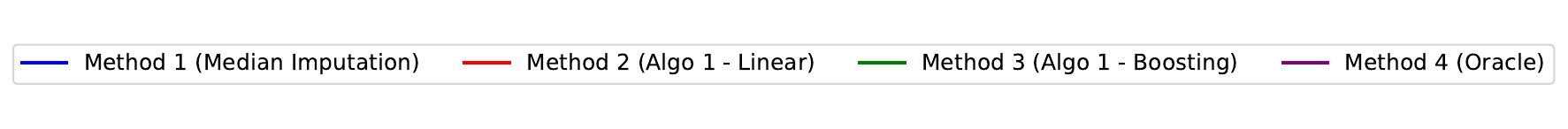}
        \end{minipage}
        
        \vspace*{-8pt}
        
        \begin{minipage}{\textwidth}
          \centering
          \makebox[\textwidth]{
            \includegraphics[width=1.2\textwidth]{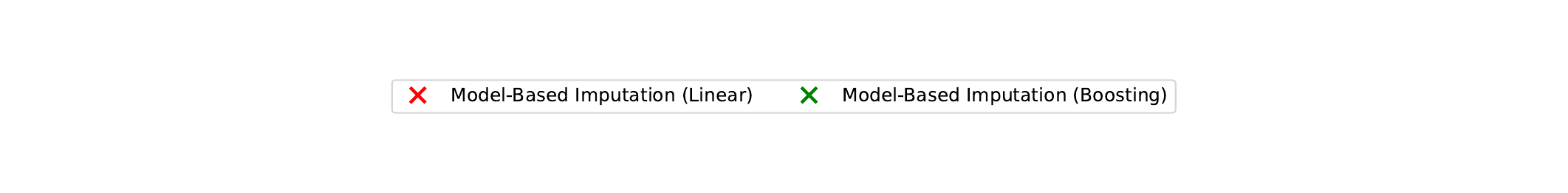}
          }
        \end{minipage}
        \vspace*{-20pt}
          
      \begin{subfigure}[b]{0.4\textwidth}
        \includegraphics[width=\textwidth]{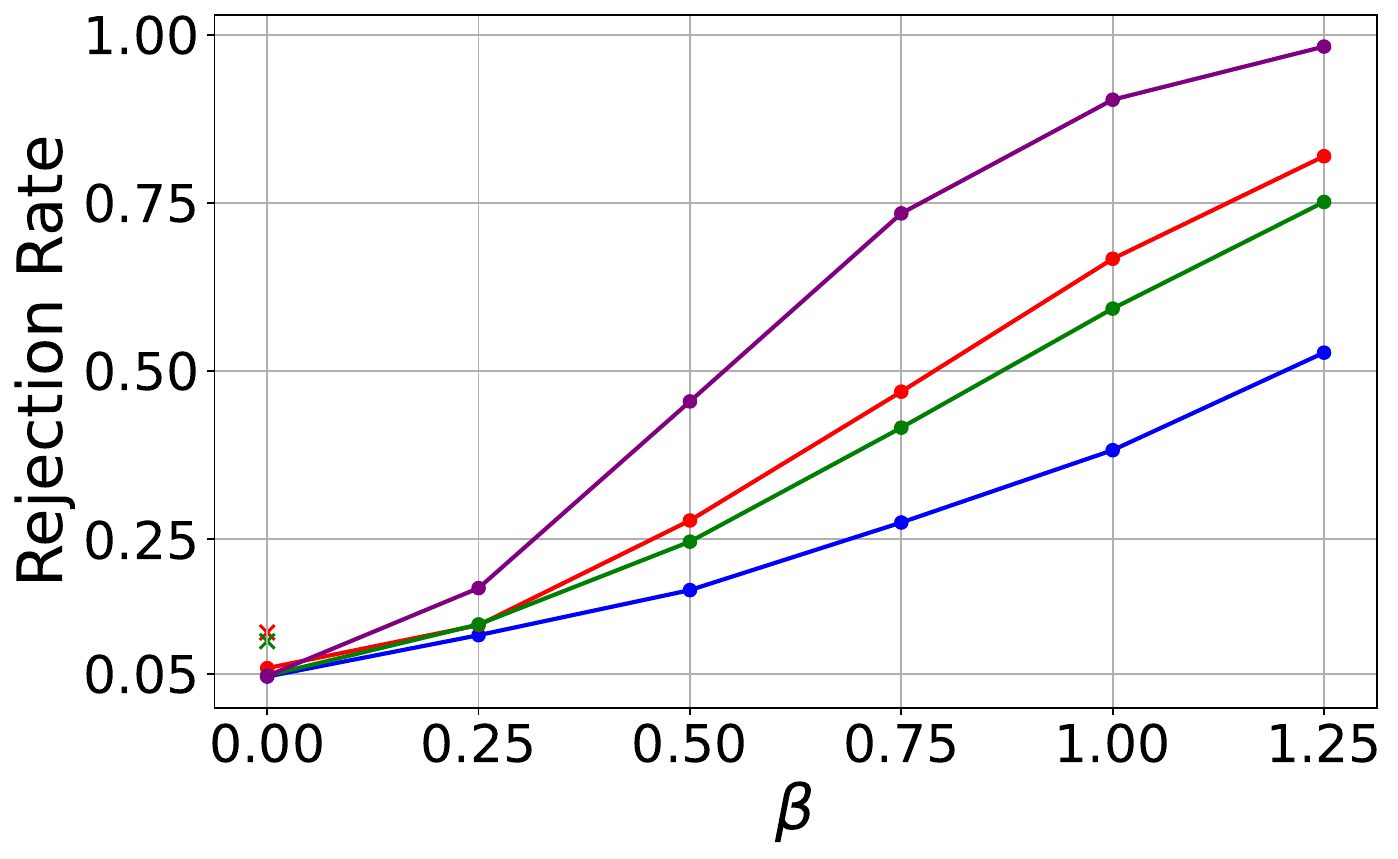}
        \caption{Model 1 ($N=50$)}
      \end{subfigure}
      \hspace{0.1cm}
      \begin{subfigure}[b]{0.4\textwidth}
        \includegraphics[width=\textwidth]{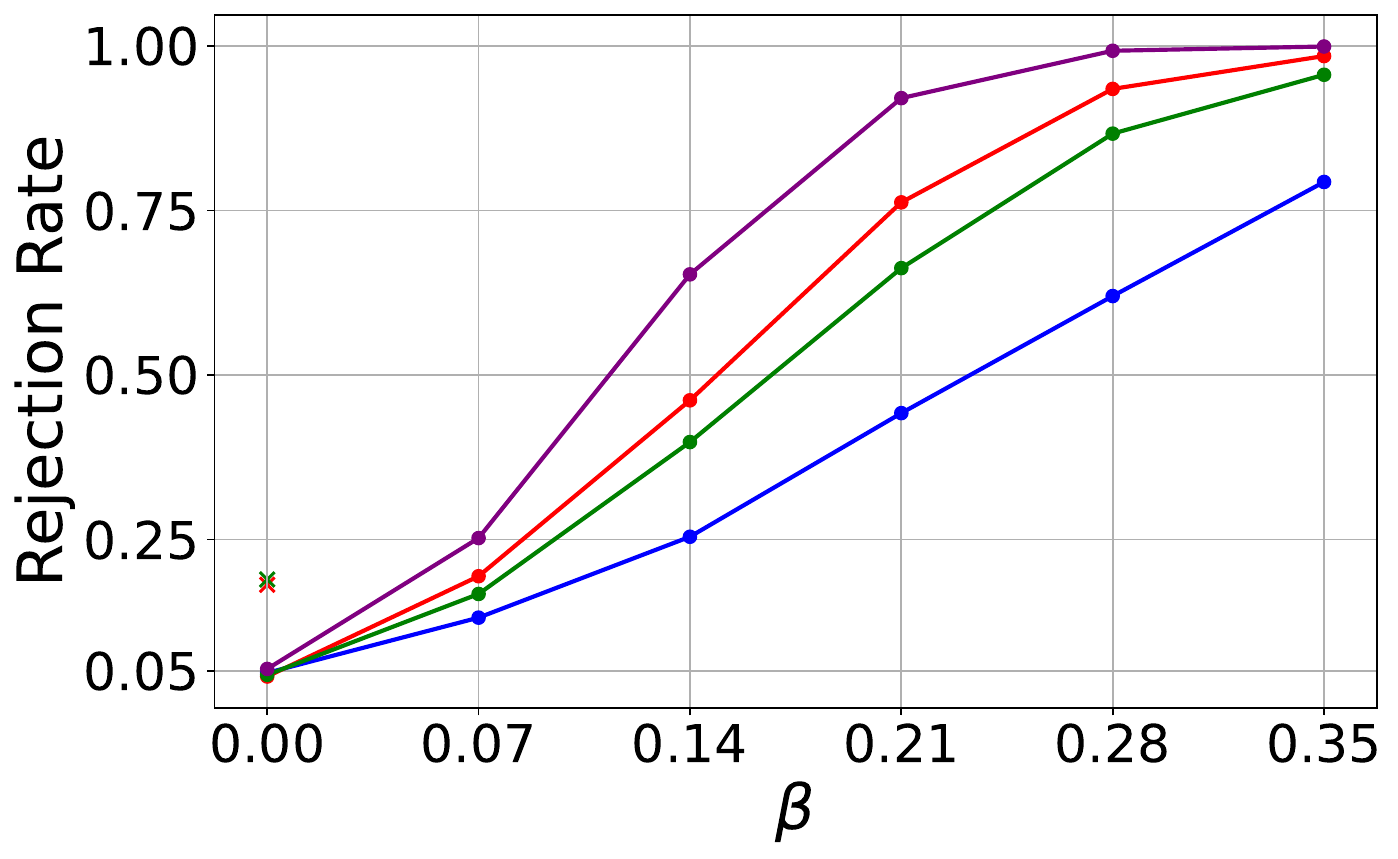}
        \caption{Model 1 ($N=1000$)}
      \end{subfigure}

    % Constant treatment effect; non-linear model for the true outcome; non-linear selection model for the missingness status, without interference

      \begin{subfigure}[b]{0.4\textwidth}
        
        \includegraphics[width=\textwidth]{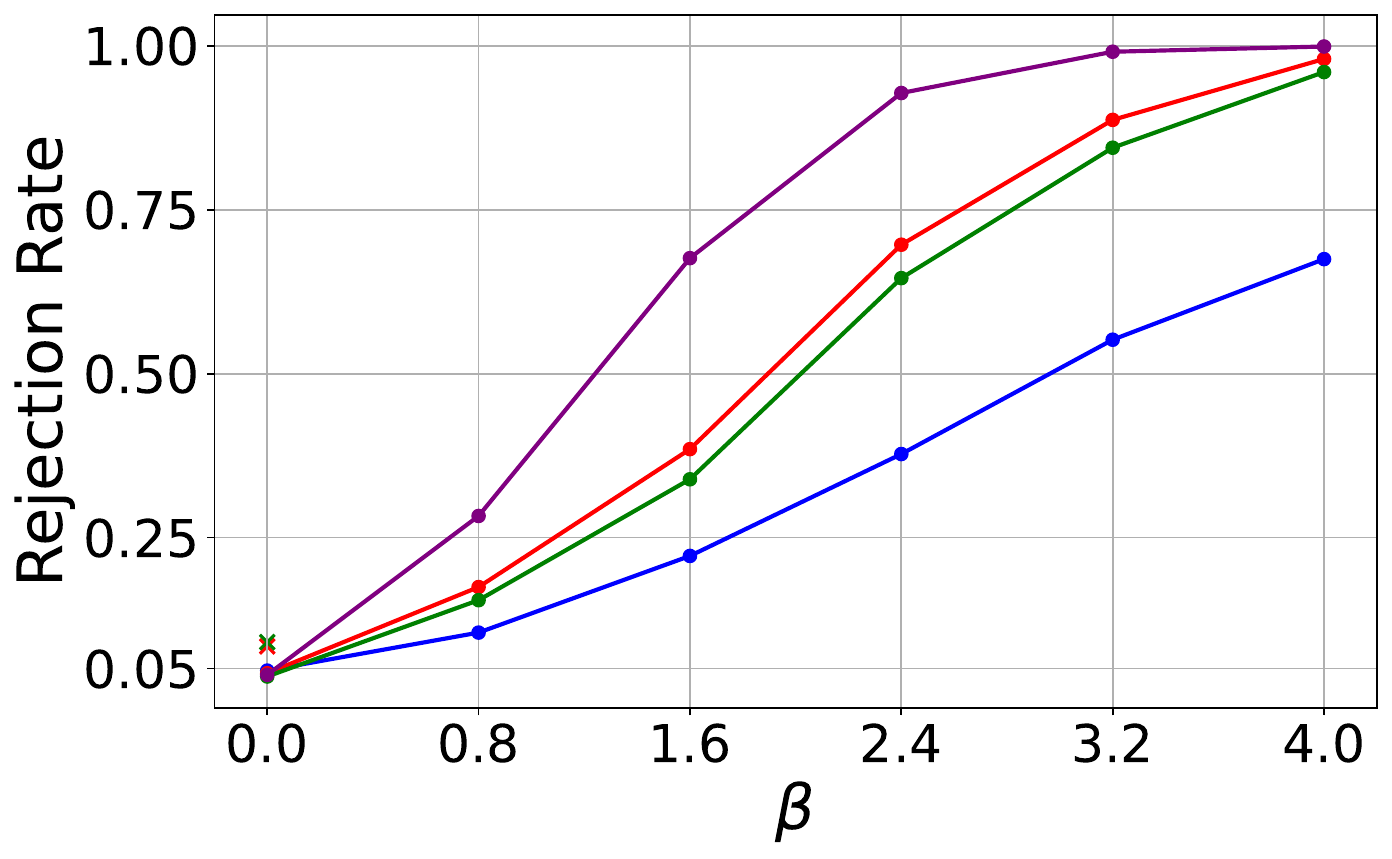}
        \caption{Model 2 ($N=50$)}
      \end{subfigure}
      \hspace{0.1cm}
      \begin{subfigure}[b]{0.4\textwidth}
        
        \includegraphics[width=\textwidth]{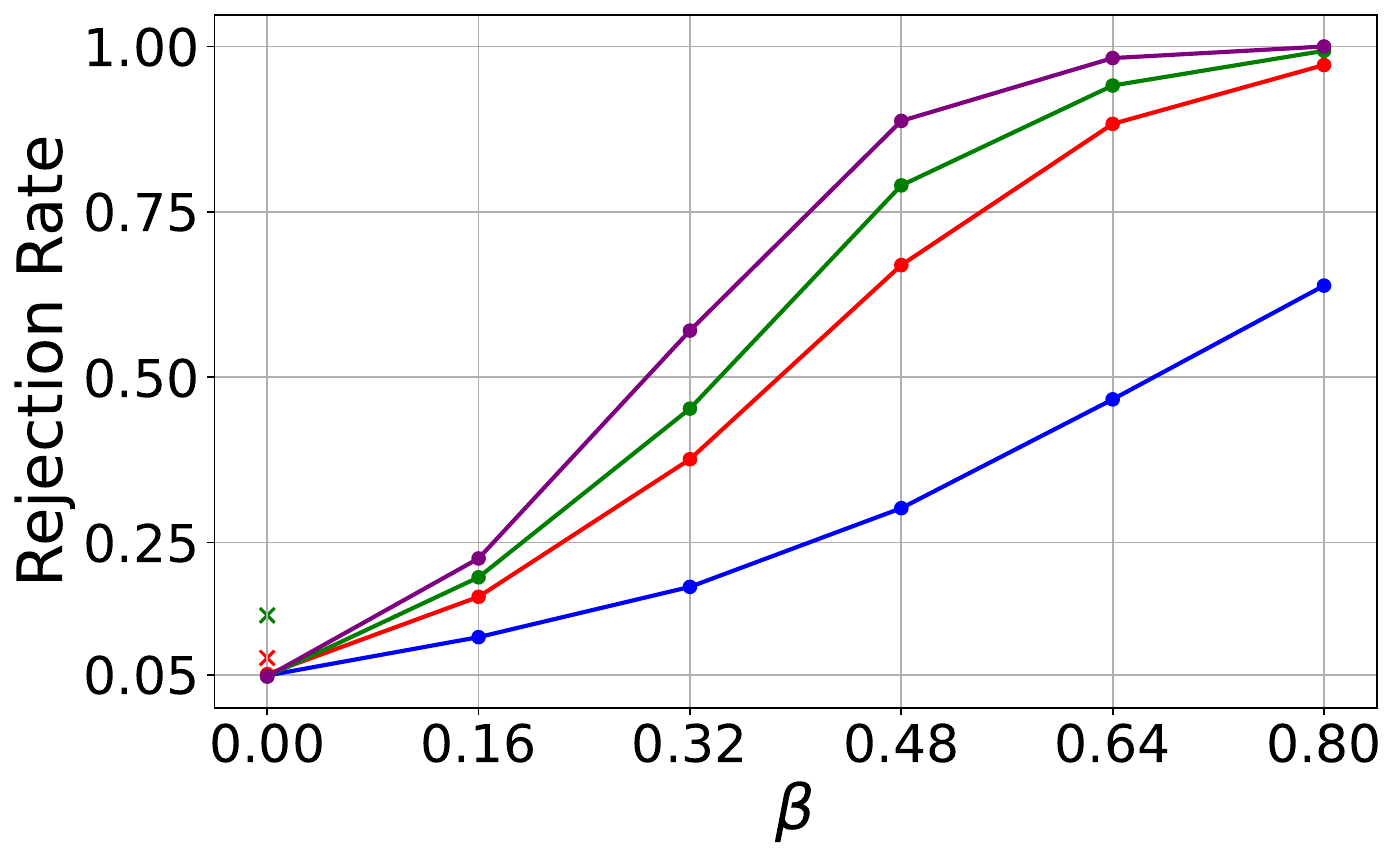}
        \caption{Model 2 ($N=1000$)}
      \end{subfigure}
    
    % Heterogeneous treatment effect; non-linear model for the true outcome; non-linear selection model for the missingness status, without interference
    
      \begin{subfigure}[b]{0.4\textwidth}
        
        \includegraphics[width=\textwidth]{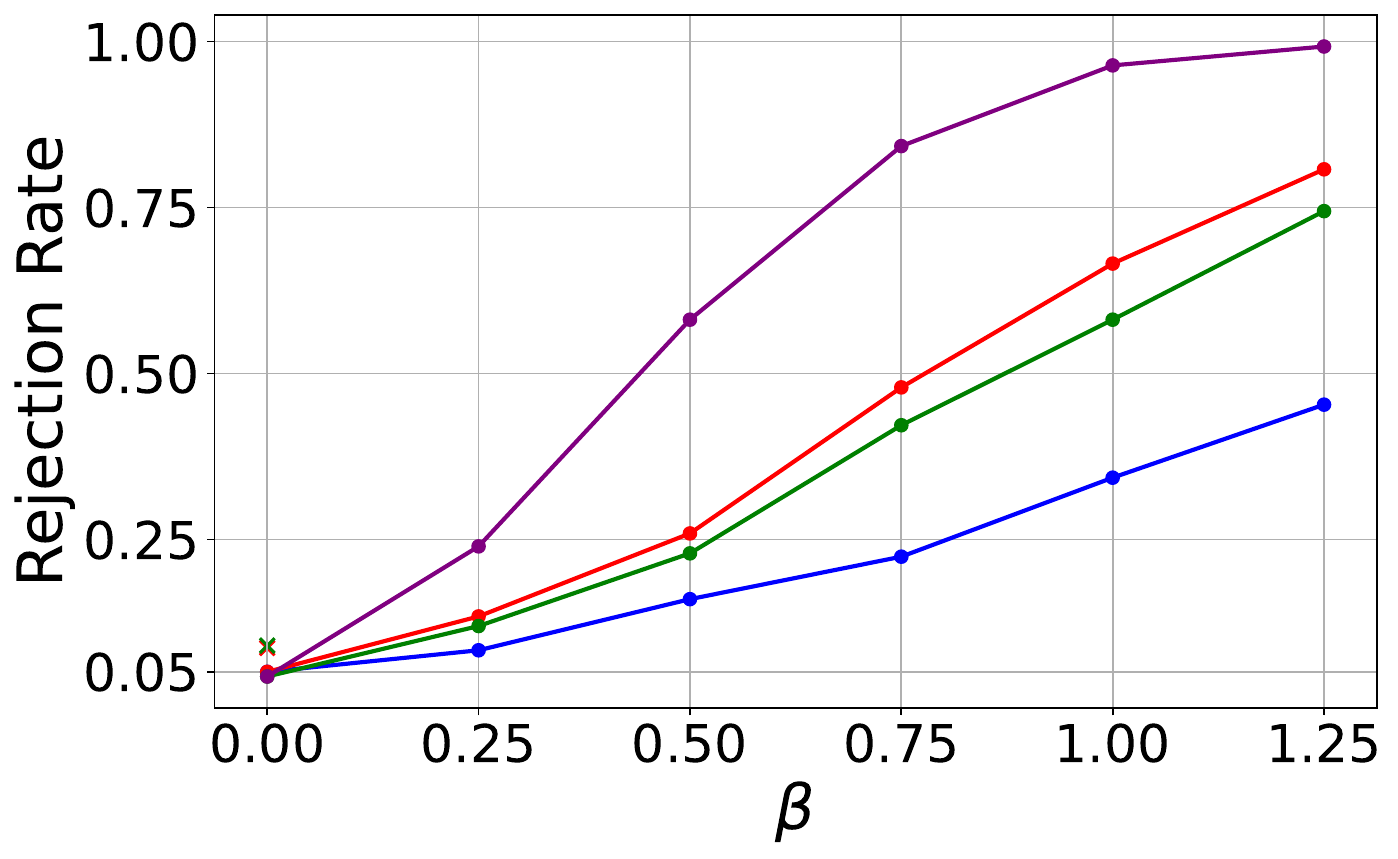}
        \caption{Model 3 ($N=50$)}
      \end{subfigure}
      \hspace{0.1cm}
      \begin{subfigure}[b]{0.4\textwidth}
        
        \includegraphics[width=\textwidth]{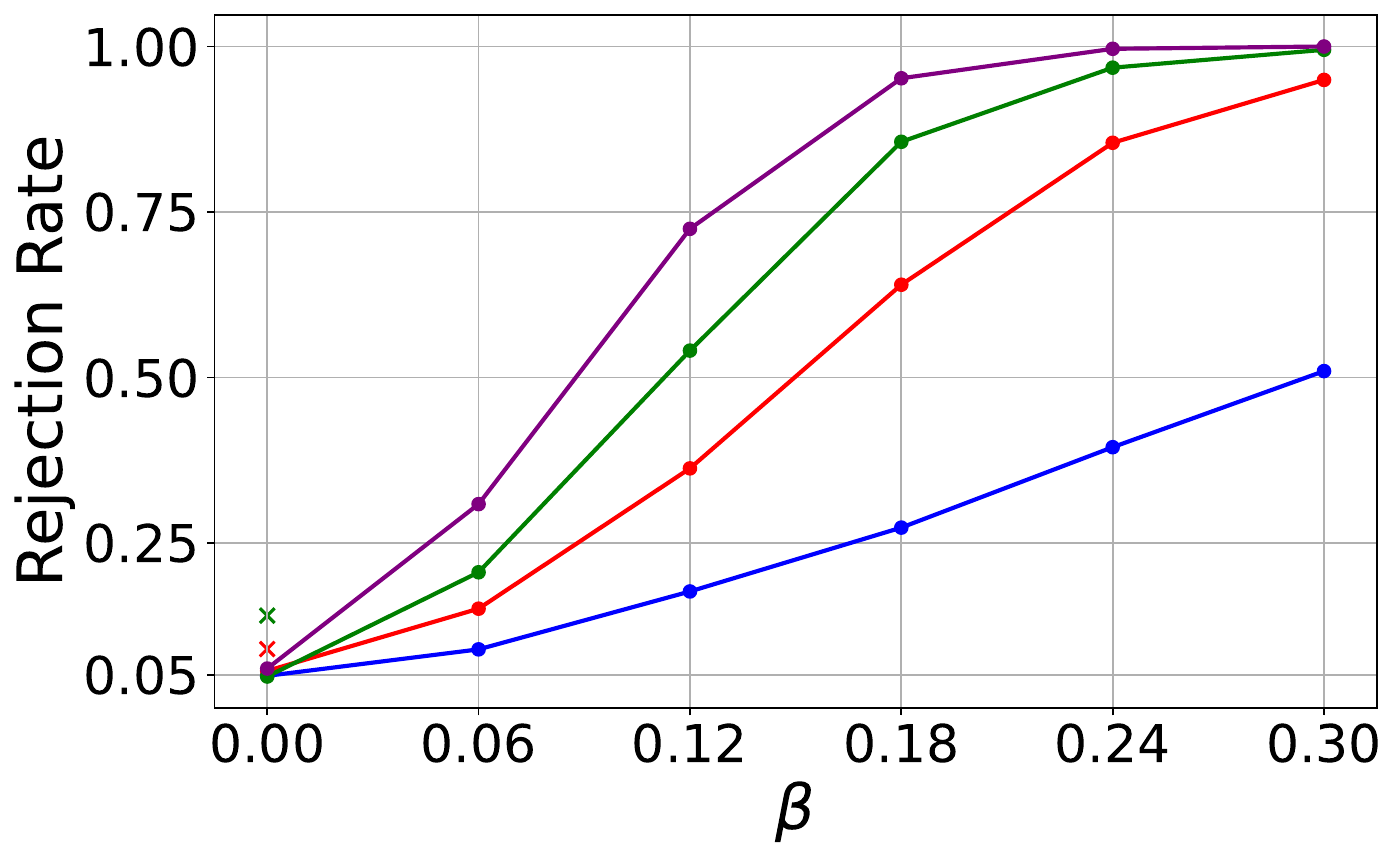}
        \caption{Model 3 ($N=1000$)}
      \end{subfigure}

    % Heterogeneous treatment effect; non-linear model for the true outcome; non-linear selection model for the missingness status, with interference
      
      \begin{subfigure}[b]{0.4\textwidth}
        
        \includegraphics[width=\textwidth]{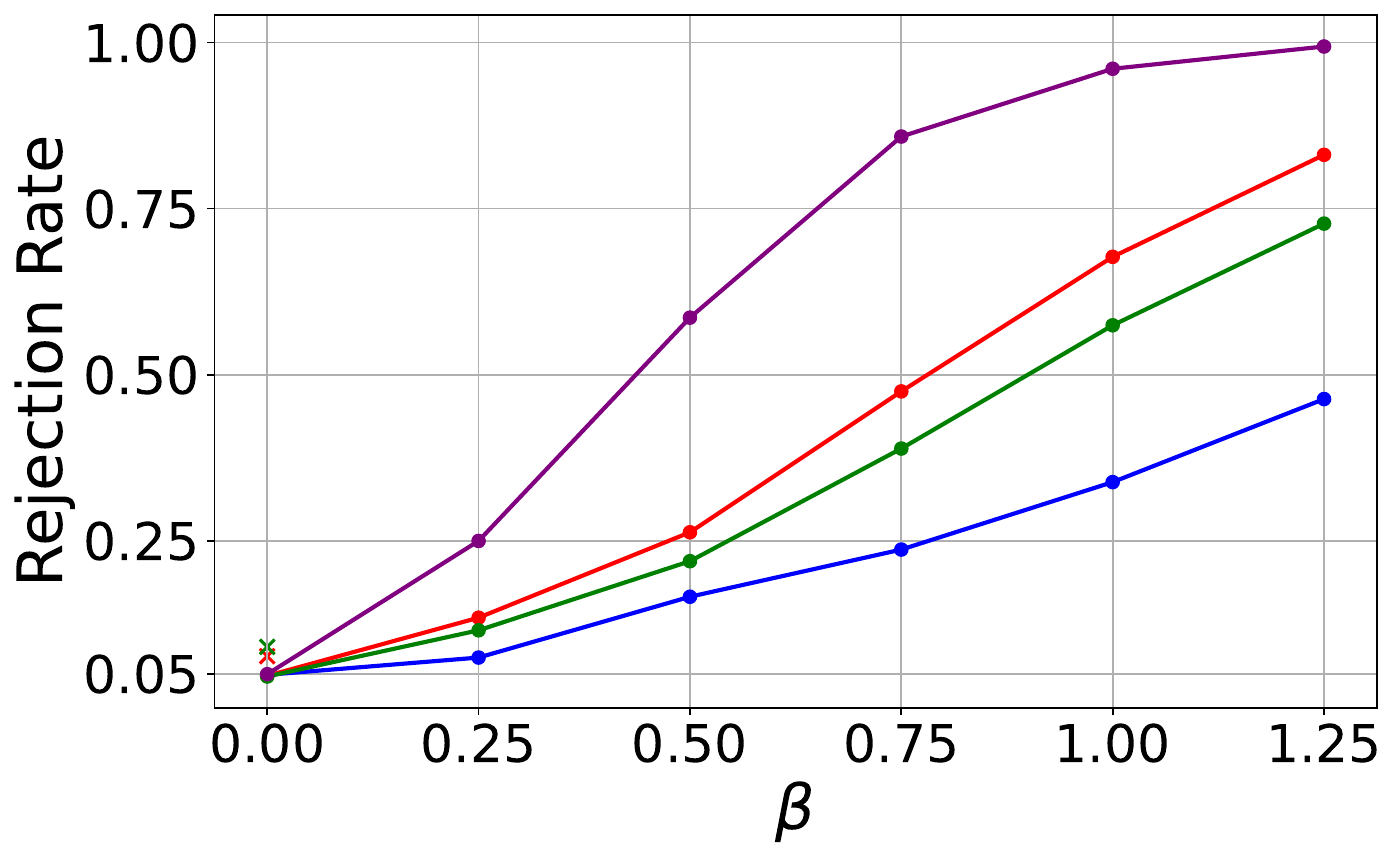}
        \caption{Model 4 ($N=50$)}
      \end{subfigure}
      \hspace{0.1cm}
      \begin{subfigure}[b]{0.4\textwidth}
        
        \includegraphics[width=\textwidth]{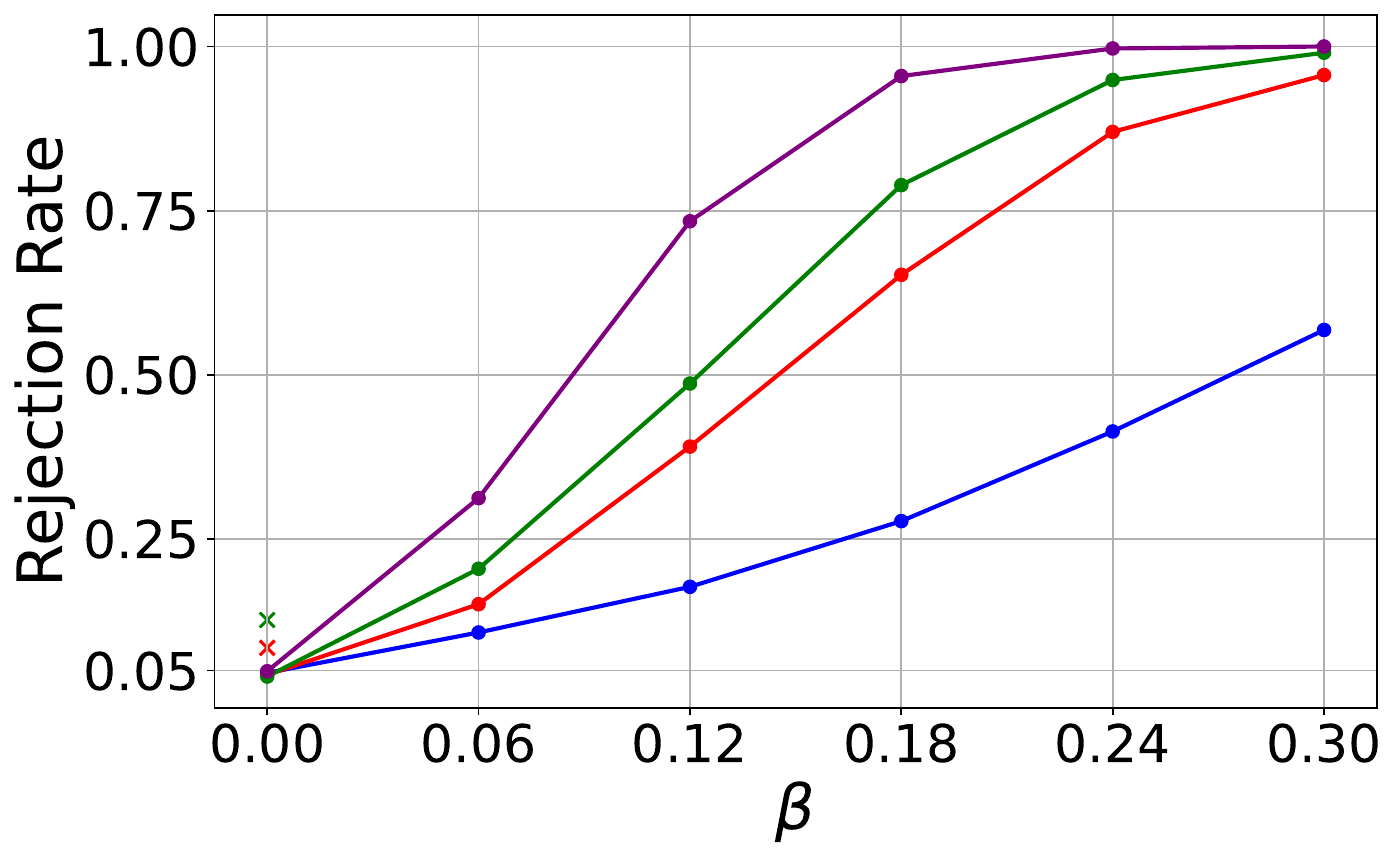}
        \caption{Model 4 ($N=1000$)}
      \end{subfigure}    
      \caption{Type-I error rate (when effect size $\beta=0$) and power (when effect size $\beta>0$) of Methods 1--4 under Models 1--4 with sample size $N=50$ and $N=1000$ (level $\alpha=0.05$).  }
      \label{fig: single outcome simulations}
\end{figure}

In Appendix B.1 in the supplementary materials, we repeat the above simulation studies but change the outcome missingness rate from $50\%$ to $25\%$, and find that the general patterns and insights still hold. In Appendix C.2, we report the computation time for the proposed imputation and re-imputation framework described in Algorithm~\ref{alg: part} under Models 1--4. Additionally, we provide practical guidance for researchers on how to weigh computational costs when choosing the imputation algorithm $\mathcal{G}$ in Algorithm~\ref{alg: part}. In Appendix B.3, we replicate simulation studies under Models 1-4 while allowing for missingness in the covariates $\mathbf{x}_{ij}$. The simulation results confirm the finite-population-exact validity of Algorithm~\ref{alg: part} in the presence of covariate missingness and demonstrate its improved power compared to non-informative imputation approaches in these settings.

\section{Covariate Adjustment in Randomization Tests with Missing Outcomes via the Imputation and Re-Imputation Framework}\label{sec: covariate adjustment}

In design-based causal inference with complete outcome data, researchers often seek to further improve statistical power via incorporating a covariate adjustment procedure in a randomization test \citep{rosenbaum2002covariance, lin2013agnostic, cohen2024no, zhao2024adjust}. For example, in Fisher's randomization tests of no treatment effect with complete outcome data, a common strategy for covariate adjustment is to conduct a randomization test based on residuals obtained from some working outcome model (e.g., some regression model or flexible machine learning algorithm) based only on the observed covariates, without looking at the treatment information \citep{rosenbaum2002covariance}. In Algorithm~\ref{alg: part with covariate adjustment} below (on page 26), we show how to conduct covariate adjustment under the proposed imputation and re-imputation framework. 
\begin{algorithm}  %\SetAlgoNoLine
\SetAlgoLined
\caption{``Imputation and re-imputation" with covariate adjustment.} \label{alg: part with covariate adjustment}
\vspace*{0.12 cm}
\KwIn{A prespecified number of re-imputation runs $L$ (e.g., $L=10{,}000$), some chosen imputation algorithm $\mathcal{G}$, some chosen test statistics $T$ ($T$ can be any test statistics), the randomization design $\mathcal{P}$ (see Assumption~\ref{assump: randomization design}), the observed treatment indicators $\mathbf{Z}$, the observed covariates $\mathbf{X}^{*}$ with possible missingness, and the observed realized outcomes with missingness $\mathbf{Y}^{*}=(\mathbf{Y}_{1}^{*}, \dots, \mathbf{Y}_{K}^{*})$ for the $K$ outcomes ($K\geq 1$). Note that $\mathbf{Y}^{*}$ contains all the observed missingness status information $\mathbf{M}=(\mathbf{M}_{1},\dots, \mathbf{M}_{K})$. }
\vspace*{0.12 cm}
\begin{enumerate}
\item 
    \begin{enumerate}
        \item Use $(\mathbf{Z}, \mathbf{X}^{*}, \mathbf{Y}^{*})$ and the chosen imputation algorithm $\mathcal{G}$ ($\mathcal{G}$ can be any imputation algorithm; recall Remark~\ref{rem: G algorithm}) to obtain the full imputed outcomes $\widehat{\mathbf{Y}}$: 
    \begin{equation*}
 \mathcal{G}: (\mathbf{Z}, \mathbf{X}^{*}, \mathbf{Y}^{*})\mapsto \widehat{\mathbf{Y}}=(\widehat{\mathbf{Y}}_{1}, \dots, \widehat{\mathbf{Y}}_{K}). \quad \text{(The Imputation Step)}
    \end{equation*}
\item Use $(\mathbf{X}^{*}, \widehat{\mathbf{Y}})$ and some chosen outcome prediction model $\mathcal{H}$ (e.g., $\mathcal{H}$ can be any regression model or machine learning algorithm; see Remark~\ref{rem: model H}) to get the fitted outcomes $\widetilde{\mathbf{Y}}$ (without using the treatment information $\mathbf{Z}$):
\begin{equation*}
 \mathcal{H}: (\mathbf{X}^{*}, \widehat{\mathbf{Y}})\mapsto \widetilde{\mathbf{Y}}=(\widetilde{\mathbf{Y}}_{1}, \dots, \widetilde{\mathbf{Y}}_{K}), \quad \text{(Covariate Adjustment after Imputation)}
    \end{equation*}
where each $\widetilde{\mathbf{Y}}_{k}=(\widetilde{Y}_{11k}, \dots, \widetilde{Y}_{In_{I}k})$ denote the predicted $k$-th outcomes. 
    
    \item Calculate $t=T(\mathbf{Z}, \bm{\epsilon})$, where $\bm{\epsilon}=\widehat{\mathbf{Y}}-\widetilde{\mathbf{Y}}$ is the residuals vector and $T$ is some chosen test statistic based on $(\mathbf{Z}, \bm{\epsilon})$.
    \end{enumerate}
    \item For each $l=1, \dots, L$, do the following steps:
\begin{enumerate}
        \item Randomly generate $\mathbf{Z}^{(l)}$ according to the randomization design $\mathcal{P}$.
    \item Obtain the full imputed outcomes $\widehat{\mathbf{Y}}^{(l)}$ based on $(\mathbf{Z}^{(l)}, \mathbf{X}^{*}, \mathbf{Y}^{*})$ and algorithm $\mathcal{G}$:
    \begin{equation*}
  \mathcal{G}: (\mathbf{Z}^{(l)}, \mathbf{X}^{*}, \mathbf{Y}^{*})\mapsto \widehat{\mathbf{Y}}^{(l)}=(\widehat{\mathbf{Y}}_{1}^{(l)}, \dots, \widehat{\mathbf{Y}}_{K}^{(l)}). \quad \text{(The Re-Imputation Step)}
    \end{equation*}
    \item  Use $(\mathbf{X}^{*}, \widehat{\mathbf{Y}}^{(l)})$ and working outcome prediction model $\mathcal{H}$ to get the fitted outcomes $\widetilde{\mathbf{Y}}$ (without using the treatment information $\mathbf{Z}$):
\begin{equation*}
 \mathcal{H}: (\mathbf{X}^{*}, \widehat{\mathbf{Y}}^{(l)})\mapsto \widetilde{\mathbf{Y}}^{(l)}=(\widetilde{\mathbf{Y}}_{1}^{(l)}, \dots, \widetilde{\mathbf{Y}}_{K}^{(l)}), \quad \text{(Covariate Adjustment after Re-Imputation)}
    \end{equation*}
where each $\widetilde{\mathbf{Y}}_{k}^{(l)}=(\widetilde{Y}_{11k}^{(l)}, \dots, \widetilde{Y}_{In_{I}k}^{(l)})$ denote the predicted $k$-th outcomes. 
    \item Calculate $T^{(l)}=T(\mathbf{Z}^{(l)}, 
  \bm{\epsilon}^{(l)})$ and store the value, where $\bm{\epsilon}^{(l)}=\widehat{\mathbf{Y}}^{(l)}-\widetilde{\mathbf{Y}}^{(l)}$ is the residuals vector for the $l$-th permutation.
    \end{enumerate} 
\end{enumerate}
\textbf{Output:} The approximate exact $p$-value $\widehat{p}_{\text{adj}}=\frac{1}{L}\sum_{l=1}^{L}\mathbbm{1}\{T^{(l)}\geq t\}$.
\end{algorithm}

Algorithm~\ref{alg: part with covariate adjustment} can be viewed as embedding the classic Rosenbaum-type covariate adjustment (\citealp{rosenbaum2002covariance}) for Fisher's randomization tests into our proposed imputation and re-imputation framework. We here give some remarks on Algorithm~\ref{alg: part with covariate adjustment}.

\begin{remark}\label{rem: model H}\emph{
    As will be shown in Theorem~\ref{alg: part with covariate adjustment}, the finite-sample validity of the $p$-value reported by Algorithm~\ref{alg: part with covariate adjustment} does \textit{not} need the working covariate adjustment model $\mathcal{H}$ in Algorithm~\ref{alg: part with covariate adjustment} to be correctly specified. In other words, the choice of the working model $\mathcal{H}$ may affect statistical power but not finite-sample validity of Algorithm~\ref{alg: part with covariate adjustment}. } 
\end{remark}
\begin{remark}\emph{In Algorithm ~\ref{alg: part with covariate adjustment}, we have $\bm{\epsilon}=(\bm{\epsilon}_{1}, \dots, \bm{\epsilon}_{K})$ with each $\bm{\epsilon}_{k}=(\widehat{\epsilon}_{11k}, \dots, \widehat{\epsilon}_{In_{I}k})$, where
 $\widehat{\epsilon}_{ijk}=\widehat{Y}_{ijk}-\widetilde{Y}_{ijk}$ are the residuals obtained via conducting covariate adjustment with imputed outcomes, i.e., the difference in the imputed outcome $\widehat{Y}_{ijk}$ using the fully observed data $(\mathbf{Z}, \mathbf{X}^{*}, \mathbf{Y}^{*})$ and the predicted outcome $\widetilde{Y}_{ijk}$ only using $(\mathbf{X}^{*}, \widehat{\mathbf{Y}})$ (without using the information of $\mathbf{Z}$), of the $k$-th outcome of subject $ij$. }
\end{remark}

\begin{remark}\emph{The imputation and re-imputation framework with covariate adjustment (i.e., Algorithm~\ref{alg: part with covariate adjustment}) also allows the covariates to have missing values because its finite-population-exact type-I error rate control (i.e., Theorem~\ref{thm: hypothesis testing with covariate adjustment}) still holds with missing covariates. }
\end{remark}

The following Theorem~\ref{thm: hypothesis testing with covariate adjustment} shows that, under Assumptions \ref{assump: randomization design} and \ref{assump: conditional indep}, as the number of re-imputation runs $L$ increases, the $p$-value $\widehat{p}_{\text{adj}}$ reported by Algorithm~\ref{alg: part with covariate adjustment} converges to the true finite-population-exact $p$-value under $H_{0}$ with the rate $O_{p}(1/\sqrt{L})$. This result holds for any sample size $N$ and does \textit{not} require $\mathcal{G}$ or $\mathcal{H}$ to be correctly specified.
\begin{theorem}\label{thm: hypothesis testing with covariate adjustment}
    Let $\mathcal{G}$ be any outcome imputation algorithm used in Algorithm~\ref{alg: part with covariate adjustment} and $\mathcal{H}$ any prediction model for covariate adjustment based on imputed outcomes used in Algorithm~\ref{alg: part with covariate adjustment}. Let $T=T(\mathbf{Z}, \bm{\epsilon})$ be any test statistic based on $\mathbf{Z}$ and residuals $\bm{\epsilon}$ after covariate adjustment,  and let $t$ denote the value of $T$ based on the observed data $(\mathbf{Z}, \mathbf{X}^{*}, \mathbf{Y}^{*})$. Let $p_{\text{adj}}=P(T\geq t\mid H_{0})$ denote the finite-population-exact $p$-value (with covariate adjustment) under $H_{0}$. For the approximate $p$-value $\widehat{p}_{\text{adj}}$ reported by the imputation and re-imputation approach with covariate adjustment described in Algorithm~\ref{alg: part with covariate adjustment}, under Assumptions \ref{assump: randomization design} and \ref{assump: conditional indep}, we have $\widehat{p}_{\text{adj}}\xrightarrow{a.s.} p_{\text{adj}}$ as the number of re-imputation runs $L\rightarrow \infty$. Moreover, for any $\epsilon>0$ and for all $L$, we have $P(| \widehat{p}_{\text{adj}}-p_{\text{adj}}|\geq \epsilon)\leq 2\exp(-2L\epsilon^{2})$.
\end{theorem}

The key takeaway from Theorem~\ref{thm: hypothesis testing with covariate adjustment} (of which Theorem~\ref{thm: hypothesis testing} is a special case) is a central message of this paper: as long as the treatment assignments do not affect outcome missingness under Fisher's sharp null $H_{0}$ (as ensured by Assumption~\ref{assump: conditional indep}), any test statistics based on treatment assignments $\mathbf{Z}$, outcomes $\mathbf{Y}^{*}$ (which include outcome missingness indicators $\mathbf{M}$), and pre-treatment covariates $\mathbf{X}$ is permutational-invariant under $H_{0}$. This invariance holds because both the observed outcomes and missingness remain fixed under $H_{0}$ and Assumption~\ref{assump: conditional indep}, making the test finite-population-exact for testing $H_{0}$, regardless of the chosen imputation or covariate adjustment model in the proposed framework. However, if the treatment assignments $\mathbf{Z}$ can affect outcome missingness $\mathbf{M}$, the arguments in Theorem~\ref{thm: hypothesis testing with covariate adjustment} (and Theorem~\ref{thm: hypothesis testing}) no longer hold, as the test statistic may incorrectly reject Fisher's sharp null $H_{0}$ due to the induced association between $\mathbf{Z}$ and $\mathbf{M}$.

In Appendix B.2 of the online supplementary materials, we conduct simulation studies to investigate the type-I error rate and power of the imputation and re-imputation approach with covariate adjustment described in Algorithm~\ref{alg: part with covariate adjustment}. The simulation results show that the imputation and re-imputation approach with covariate adjustment (i) can achieve finite-population-exact type-I error rate control for any imputation algorithm $\mathcal{G}$ and covariate adjustment model $\mathcal{H}$, (ii) may further improve the statistical power of the imputation and re-imputation approach in many settings, and (iii) outperforms non-informative imputation (e.g., median imputation) with covariate adjustment. See Appendix B.2 for details.

There are some other recent works on covariate adjustment in randomized experiments with outcome missingness (e.g., \citealp{chang2023covariate, zhao2023covariate, wang2024handling}), all of which have made significant contributions in the super-population setting. In contrast, our approach is tailored for the finite-population setting. Because of this, the statistical settings, the required assumptions, and the proposed approaches of these recent works and those of our work are very different. For example, our framework can still guarantee finite-population-exact type-I error rate control even when both the imputation algorithm/model and covariate adjustment model were misspecified or when unobserved covariates or interference exists in the missingness mechanism. See Remark~\ref{rem: differences with other cov adj} in Appendix D for details.

\section{Finite-Population-Exact Confidence Regions With Missing Outcomes via Inverting An Imputation-Assisted Randomization Test}\label{sec: CI with missing outcomes}

In addition to testing the null hypothesis of no treatment effect, researchers are often also concerned about confidence regions for causal parameters. As reviewed in Section~\ref{assump: conditional indep}, in design-based causal inference with complete outcome data, these two tasks (i.e., hypothesis testing and construction of a confidence region) are equivalent: a finite-population-exact confidence region can be obtained via inverting finite-population-exact randomization tests \citep{rosenbaum2002observational}. However, such equivalence no longer holds in the presence of missing outcomes. As we will show in this section, constructing a finite-population-exact confidence region requires a stronger assumption than those needed for finite-population-exact hypothesis testing (i.e., Assumption~\ref{assump: conditional indep}). We state such an assumption below.  
\begin{assumption}[A General Outcome Missingness Mechanism for Design-Based Confidence Regions]\label{assump: conditional indep with more restrictions}
    Conditional on the finite-population dataset in hand, there exists an unknown map $\phi: \mathbb{R}^{N \times (p_{x}+p_{u})} \rightarrow \{0,1\}^{N \times K}$ such that $(\mathbf{M}_{1},\dots, \mathbf{M}_{K})=\phi(\mathbf{X}, \mathbf{U})$. 
\end{assumption}

Assumption~\ref{assump: conditional indep with more restrictions} can imply an outcome missingness mechanism considered in \citet{heussen2023randomization} (specifically, Condition 1 in \citet{heussen2023randomization}, which states that outcome missingness is stochastically independent of treatment assignments) under the potential outcomes framework. It claims that missingness status can be fully captured by the observed and unobserved covariates (including unobserved pre-treatment error terms in the outcome missingness mechanism), which can be tested based on the observed data. Assumption~\ref{assump: conditional indep with more restrictions} is a much stronger assumption than Assumption~\ref{assump: conditional indep} as the missingness status no longer depends on the post-treatment outcomes and, therefore, will not be affected by treatment assignments (even when Fisher's sharp null $H_{0}$ does not hold).

%However, as we will show in Remark~\ref{rem: Assumption~\ref{assump: conditional indep with more restrictions}}, Assumption~\ref{assump: conditional indep with more restrictions} is nearly the minimal assumption required for constructing an exact confidence region in design-based inference with missing outcomes. 

Theorem~\ref{thm: confidence region} shows that, under Assumptions \ref{assump: randomization design} and \ref{assump: conditional indep with more restrictions}, we can construct a finite-population-exact confidence region with missing outcomes, even when both the imputation and covariate adjustment models were misspecified (the proof is in Appendix A.3).
\begin{theorem}\label{thm: confidence region}
    Consider the same setting as that in Theorem~\ref{assump: conditional indep} and further assume that Assumptions \ref{assump: randomization design} and \ref{assump: conditional indep with more restrictions} hold. Consider the following null hypothesis with prespecified functions $f_{1},\dots, f_{K}$ and prespecified values of the causal parameters $\Vec{\beta}_{0}=(\Vec{\beta}_{01}, \dots, \Vec{\beta}_{0 K})$:
    \begin{equation}\label{eqn: effect model with specified beta_0}
  H_{\Vec{\beta}_{0}}: Y_{ijk}(1)=f_{k}(Y_{ijk}(0), \Vec{\beta}_{0k}, \mathbf{x}_{ij}), \text{ for } i \in [I], j \in [n_{i}], k \in [K]. 
\end{equation}
 Let $\mathbf{Y}^{*}_{\Vec{\beta}_{0}}=(\mathbf{Y}^{*}_{1, \Vec{\beta}_{01}}, \dots, \mathbf{Y}^{*}_{K, \Vec{\beta}_{0K}})$ be the transformed realized outcomes (with missingness), in which we have $\mathbf{Y}^{*}_{k, \Vec{\beta}_{0k}}=(Y^{*}_{11k, \Vec{\beta}_{0k}}, \dots, Y^{*}_{In_{I}k, \Vec{\beta}_{0k}})$ where $Y^{*}_{ijk, \Vec{\beta}_{0k}}=\text{``Missing"}$ if $M_{ijk}=1$ and $Y^{*}_{ijk, \Vec{\beta}_{0k}}=Z_{ij}Y_{ijk}^{*}+(1-Z_{ij})f_{k}(Y_{ijk}^{*}, \Vec{\beta}_{0k}, \mathbf{x}_{ij})=Z_{ij}Y_{ijk}+(1-Z_{ij})f_{k}(Y_{ijk}, \Vec{\beta}_{0k}, \mathbf{x}_{ij})$ if $M_{ijk}=0$. Let $\widehat{p}(\Vec{\beta}_{0})$ denote the approximate $p$-value reported by Algorithm~\ref{alg: part} (without covariate adjustment after imputation) or Algorithm~\ref{alg: part with covariate adjustment} (with covariate adjustment after imputation) based on the inputs $(\mathbf{Z}, \mathbf{X}, \mathbf{Y}^{*}_{\Vec{\beta}_{0}})$. For the prespecified level $\alpha\in (0, 1)$, we define the confidence region $CR=\{\Vec{\beta}_{0}: \widehat{p}(\Vec{\beta}_{0})\geq \alpha\}$. As the number of re-imputation runs $L$ increases, for the unknown true causal parameters $\Vec{\beta}=(\Vec{\beta}_{1},\dots, \Vec{\beta}_{K})$ that satisfy $Y_{ijk}(1)=f_{k}(Y_{ijk}(0), \Vec{\beta}_{k}, \mathbf{x}_{ij})$ for all $i, j, k,$ we have $\lim_{L\rightarrow\infty }  P(\Vec{\beta} \in CR)\geq 1-\alpha$.
    \end{theorem}

\begin{remark}\label{rem: Assumption 3 is strong}\emph{
As shown in Appendix A.3, a key component of the proof of Theorem~\ref{thm: confidence region} is that the missingness indicators $\mathbf{M}$ are invariant under any $\Vec{\beta}_{0}$ in $H_{\Vec{\beta}_{0}}$, which holds under Assumption~\ref{assump: conditional indep with more restrictions}, but may not hold under Assumption~\ref{assump: conditional indep}. Specifically, under Assumption~\ref{assump: conditional indep} (i.e., missingness indicators are allowed to depend on the true outcomes $\mathbf{Y}$ and therefore can change with the treatment assignments), the observed missingness indicators under the observed treatment assignments $\mathbf{Z}$ are $(\mathbf{M}_{1}(\mathbf{Z}), \dots, \mathbf{M}_{K}(\mathbf{Z}))=\eta(\mathbf{X}, \mathbf{U}, \mathbf{Y}(\mathbf{Z}))$ and those under another set of treatment assignments $\mathbf{Z}^{\prime}$ will be $(\mathbf{M}_{1}(\mathbf{Z}^{\prime}), \dots, \mathbf{M}_{K}(\mathbf{Z}^{\prime}))=\eta(\mathbf{X}, \mathbf{U}, \mathbf{Y}(\mathbf{Z}^{\prime}))$. Note that under $H_{\Vec{\beta}_{0}}$, the $\mathbf{Y}(\mathbf{Z}^{\prime})$ may differ from $\mathbf{Y}(\mathbf{Z})$ when $\mathbf{Z}^{\prime}\neq \mathbf{Z}$. Therefore, without additional assumptions (i.e., Assumption~\ref{assump: conditional indep with more restrictions}), missingness status $\mathbf{M}$ may change with different treatment assignments $\mathbf{Z}$, and a finite-population-exact randomization test for $H_{\Vec{\beta}_{0}}$ with an arbitrary $\Vec{\beta}_{0}$ (or equivalently, finite-population-exact confidence region for $\Vec{\beta}$) may not be feasible. However, for testing the null $H_{0}$, Assumption~\ref{assump: conditional indep} is sufficient. }
    \end{remark}

\section{Real Data Application}\label{sec: real data}

We apply our new framework to the Work, Family, and Health Study (\citealp{WFHS2018}), which is a large-scale cluster randomized experiment (CRE) for improving the well-being and work-family balance of employees. In this CRE, half of the 56 working groups (clusters) in an anonymized company were randomly assigned a workplace intervention (treatment), and the remaining half were controls. The average number of study individuals in a cluster is about 18.6, with 1044 study individuals in total. The outcome variable is control over work hours at the 6-month follow-up, a nearly continuous score ranging from 1 to 5 (19.1\% missingness). We consider the following seven covariates (recall that our framework also allows covariate missingness): the baseline value of the outcome measure (10.7\% missingness), employee or manager indicator (no missingness), group job functions (0.8\% missingness), time adequacy (4.9\% missingness), psychological distress (0.1\% missingness), perceived stress (21.5\% missingness), and work-family conflict (0.1\% missingness). In Table~\ref{tab: real data}, we report the one-sided $p$-values under the null hypothesis of no treatment effect $H_{0}$ and one-sided 95\% confidence intervals (CIs) for an additive treatment effect, given by the Wilcoxon rank sum test based on the following six approaches: median imputation without or with covariate adjustment (using a linear covariate adjustment model), the imputation and re-imputation framework without covariate adjustment (i.e., Algorithm~\ref{alg: part}) with linear imputation model or boosting imputation algorithm, and imputation and re-imputation framework with covariate adjustment (i.e., Algorithm~\ref{alg: part with covariate adjustment}) with linear imputation and covariate adjustment models or boosting imputation and covariate adjustment models. See Appendix E for more details on the dataset and procedure. 

%Also, for the same reason as that mentioned in Remark~\ref{rem: we do not compare with CCA and MI}, we here do not compare our approach with previous model-based multiple imputation or weighting approaches in design-based inference with missing outcomes due to their lack of statistical validity (i.e., type-I error rate control) guarantee. 

\setlength{\tabcolsep}{2pt} % default is 6pt
\begin{table}[H]
\small
\centering
\begin{tabular}{lc c c c c c c}
\toprule
\multirow{2}{*}{} &\multicolumn{3}{c}{Without Covariate Adjustment}&\multicolumn{3}{c}{With Covariate Adjustment} \\
\cmidrule(rl){2-4} \cmidrule(rl){5-7} 
 & Median \  & Algo 1 (Linear) \ & Algo 1 (Boosting) \ & Median \ & Algo 2 (Linear) \ & Algo 2 (Boosting) \ \\
\midrule
$p$-value   &   $0.120$ & $0.103$ & $0.104$  & $0.067$ &  $0.030$ & $0.025$  \\
95\% CI & $[-0.124, 4]$ & $[-0.093,4]$ & $[-0.121,4]$  & $[-0.031,4]$ & $[0.036, 4]$ & $[0.024, 4]$ \\
\bottomrule
\end{tabular}
\caption{One-sided $p$-values and CIs reported by median imputation, Algorithm~\ref{alg: part}, and Algorithm~\ref{alg: part with covariate adjustment}. The outcome missingness rate is 19.1\%.}
\label{tab: real data}
\end{table}

Table~\ref{tab: real data} shows that the proposed imputation and re-imputation framework can typically make causal inference for the WFHS more informative (e.g., smaller $p$-values under $H_{0}$ and/or narrower CIs), especially when incorporated with covariate adjustment. For example, under covariate adjustment and one-sided significance level $0.05$, we fail to reject $H_{0}$ when using a classic randomization test with non-informative (e.g., median) imputation. Still, we are able to reject $H_{0}$ if applying the imputation and re-imputation framework incorporated with either linear models or boosting. Recall that these gains in efficiency are based on a 19.1\% outcome missingness rate. We expect these gains to be even more substantial with a higher outcome missingness rate.

\section{Discussion of Future Research Directions}

Our work also suggests future research directions. For example, for the hypothesis testing part, our work focuses on Fisher's sharp null and its extensions in model (\ref{eqn: effect model with specified beta_0}). These null hypotheses are widely used in both methodology research and applied research \citep{rosenbaum2002observational, rosenbaum2020design, imbens2015causal}. However, our work has not explicitly discussed other common null hypotheses in design-based causal inference, such as Neyman's weak null \citep{neyman1923application, imbens2015causal, zhao2024adjust} and attributable effects \citep{rosenbaum2001effects, rosenbaum2002observational, hansen2009attributing, choi2017estimation}. When the outcome variable is discrete, previous studies \citep{rigdon2015randomization, li2016exact, rosenbaum2001effects, hansen2009attributing} showed that testing Neyman's weak null or attributable effects is equivalent to testing a sequence of carefully formulated Fisher's sharp null hypotheses. This sheds light on using our framework to study other null hypotheses (e.g., Neyman's weak null with discrete outcomes and attributable effects) in design-based causal inference with missing outcomes. 

Another meaningful future research direction is to extend our framework to handle missingness in survival data. In Appendix F, we derive some preliminary results on how to apply our framework to right-censored survival data with missing censoring indicators.

%\section*{Acknowledgements}
%Siyu Heng was supported in part by a grant from the New York University Research Catalyst Prize and a New York University School of Global Public Health Research Support Grant. Yang Feng was supported in part by a grant from the National Science Foundation, a grant from the National Institutes of Health, and a New York University School of Global Public Health Research Support Grant. 

\section*{Supplementary Material}

Online supplementary material includes technical proofs, additional theoretical results, additional details, additional simulation studies, and related discussions and remarks.

\section*{Acknowledgements}

The authors thank the joint editor, the associate editor, and two anonymous referees for their valuable comments and helpful advice. The authors thank Rebecca Betensky, Matthew Blackwell, Kan Chen, Iván Díaz, Xinran Li, Naijia Liu, Samuel Pimentel, James Robins, William Rosenberger, Peizan Sheng, Dylan Small, Linbo Wang, Yuhao Wang, and Zeyang Yu for the helpful discussions. The work of Siyu Heng was in part supported by the NIH Grant R21DA060433, a NYU GPH Research Support Grant, and an NYU Research Catalyst Prize. The work of Jiawei Zhang was in part supported by a NYU GPH Research Support Grant and an NYU Research Catalyst Prize when he was working as an undergraduate research assistant at New York University, under the supervision of Siyu Heng. The work of Yang Feng was in part supported by the NIH Grant R21AG074205, the NSF Grant DMS-2324489, and a NYU GPH Research Support Grant.

\putbib[references]  % Appendix references
\end{bibunit}

\clearpage

%% Here are the title, author names and addresses

\begin{bibunit}[apalike]  % Use apalike or agsm style for main text

\begin{center}
    \large \bf Online Supplementary Materials for ``Design-Based Causal Inference with Missing Outcomes: Missingness Mechanisms, Imputation-Assisted Randomization Tests, and Covariate Adjustment" by Heng, Zhang, and Feng
\end{center}

\smallskip
\smallskip
\smallskip

\begin{abstract}
    Online supplementary materials (Appendices A--G) include technical proofs, additional theoretical results, additional details, additional simulation studies, and related discussions and remarks. Specifically, Appendix A includes detailed proofs of all the theorems stated in the main text. Appendix B presents additional simulation studies, including those on the covariate adjustment in design-based causal inference with missing outcomes based on the proposed framework. Appendix C gives more details on the simulation studies. Appendix D includes additional discussions and remarks. Appendix E contains more details about the real data application. Appendices F and G discuss how to extend our proposed approaches to causal inference with survival data (censored data).
\end{abstract}
\clearpage

\section*{Appendix A: Proofs of the Theorems}

\subsection*{Appendix A.1: Proof of Theorem~\ref{thm: hypothesis testing}}

To prove Theorem~\ref{thm: hypothesis testing}, we first prove the following lemma, which gives an explicit form of the finite-population-exact distribution of the test statistic $T(\mathbf{Z}, \widehat{\mathbf{Y}})$ (based on imputed outcomes) under $H_{0}$.

\begin{lemma}\label{lem: p-value}
  Let $\mathcal{G}: (\mathbf{Z}, \mathbf{X}^{*}, \mathbf{Y}^{*})\mapsto \widehat{\mathbf{Y}}$ be any outcome imputation algorithm for obtaining the imputed outcomes $\widehat{\mathbf{Y}}$, and $T(\mathbf{Z}, \widehat{\mathbf{Y}})$ be any test statistic based on $\mathbf{Z}$ and $\widehat{\mathbf{Y}}$. Let $P(T(\mathbf{Z}, \widehat{\mathbf{Y}})\geq t \mid H_{0})$ be the finite-population-exact $p$-value under $H_{0}$ given the observed value $t$ of $T(\mathbf{Z}, \widehat{\mathbf{Y}})$. Under Assumptions \ref{assump: randomization design} and \ref{assump: conditional indep}, if the missing outcome imputation algorithm $\mathcal{G}$ is deterministic (i.e., $\widehat{\mathbf{Y}}=\mathcal{G}(\mathbf{Z}, \mathbf{X}^{*}, \mathbf{Y}^{*})$ is a deterministic vector given $(\mathbf{Z}, \mathbf{X}^{*}, \mathbf{Y}^{*})$), we have
    \begin{equation}\label{eqn: one-sided exact p-value deterministic}
    P(T(\mathbf{Z}, \widehat{\mathbf{Y}})\geq t \mid H_{0})=\sum_{\mathbf{z}\in \Omega}\big[\mathbbm{1}\{T(\mathbf{Z}=\mathbf{z}, \widehat{\mathbf{Y}}=\mathcal{G}(\mathbf{Z}=\mathbf{z}, \mathbf{X}^{*}, \mathbf{Y}^{*}))\geq t\}\times P(\mathbf{Z}=\mathbf{z})\big],
    \end{equation}
    and if the missing outcome imputation algorithm $\mathcal{G}$ is stochastic (i.e., $\widehat{\mathbf{Y}}=\mathcal{G}(\mathbf{Z}, \mathbf{X}^{*}, \mathbf{Y}^{*})$ is a random vector conditional on $(\mathbf{Z}, \mathbf{X}^{*}, \mathbf{Y}^{*})$), we have
    \begin{equation}\label{eqn: one-sided exact p-value stochastic}
    P(T(\mathbf{Z}, \widehat{\mathbf{Y}})\geq t \mid H_{0})=\sum_{\mathbf{z}\in \Omega}\big[P\{T(\mathbf{Z}=\mathbf{z}, \widehat{\mathbf{Y}}=\mathcal{G}(\mathbf{Z}=\mathbf{z}, \mathbf{X}^{*}, \mathbf{Y}^{*}))\geq t\}\times P(\mathbf{Z}=\mathbf{z})\big].
    \end{equation}
\end{lemma}

\begin{proof}
    We prove Lemma~\ref{lem: p-value} by discussing the following two cases separately.

    \textbf{Case 1: The missing outcome imputation algorithm $\mathcal{G}$ is deterministic.}

    To prove that equation (\ref{eqn: one-sided exact p-value deterministic}) holds, the key point is to show that the realized outcomes $\mathbf{Y}^{*}=(Y_{111}^{*}, \dots, Y_{In_{I}K}^{*})$ are invariant under any treatment indicator vector $\mathbf{Z}$. Note that under $H_{0}$, the true outcomes $\mathbf{Y}=(Y_{111}, \dots, Y_{In_{I}K})$ are invariant under any $\mathbf{Z}$. Moreover, based on Assumption~\ref{assump: conditional indep}, such invariance of $\mathbf{Y}$ further implies that the observed missingness indicators $\mathbf{M}=(M_{111}, \dots, M_{In_{I}K})$ are also invariant under any $\mathbf{Z}$. Recall that for each $Y_{ijk}^{*}$ (i.e., the $k$-th realized outcome of subject $ij$), we have $Y_{ijk}^{*}=Y_{ijk}$ if $M_{ijk}=0$ and $Y_{ijk}^{*}=\text{``Missing"}$ if $M_{ijk}=1$. Therefore, the realized outcomes $\mathbf{Y}^{*}=(Y_{111}^{*}, \dots, Y_{In_{I}K}^{*})$ are invariant under any treatment indicator vector $\mathbf{Z}$. That is, if we let $\mathbf{Y}^{*}(\mathbf{z})=(Y_{111}^{*}(\mathbf{z}), \dots, Y_{In_{I}K}^{*}(\mathbf{z}))$ denote the potential realized outcomes under the treatment indicator vector $\mathbf{Z}=\mathbf{z}$, we have the observed realized outcomes $\mathbf{Y}^{*}=\mathbf{Y}^{*}(\mathbf{z})$ for any $\mathbf{z}$. Then, we have 
    \begin{align*}
        P(T(\mathbf{Z}, \widehat{\mathbf{Y}})\geq t \mid H_{0})&=\sum_{\mathbf{z}\in \Omega}\big[\mathbbm{1}\{T(\mathbf{Z}=\mathbf{z}, \widehat{\mathbf{Y}}=\mathcal{G}(\mathbf{Z}=\mathbf{z}, \mathbf{X}^{*}, \mathbf{Y}^{*}(\mathbf{z})))\geq t\}\times P(\mathbf{Z}=\mathbf{z})\big]\\
        &=\sum_{\mathbf{z}\in \Omega}\big[\mathbbm{1}\{T(\mathbf{Z}=\mathbf{z}, \widehat{\mathbf{Y}}=\mathcal{G}(\mathbf{Z}=\mathbf{z}, \mathbf{X}^{*}, \mathbf{Y}^{*}))\geq t\}\times P(\mathbf{Z}=\mathbf{z})\big].
    \end{align*}
    Therefore, equation (\ref{eqn: one-sided exact p-value deterministic}) holds. 
    
 \textbf{Case 2: The missing outcome imputation algorithm $\mathcal{G}$ is stochastic.}

    In Case 1, we have shown that the observed realized outcomes $\mathbf{Y}^{*}=\mathbf{Y}^{*}(\mathbf{z})$ for any $\mathbf{z}$ when \textit{both} the null $H_{0}$ and Assumption~\ref{assump: conditional indep} hold. Then, we have 
    \begin{align*}
        P(T(\mathbf{Z}, \widehat{\mathbf{Y}})\geq t \mid H_{0})&=\sum_{\mathbf{z}\in \Omega}\big[P\{T(\mathbf{Z}=\mathbf{z}, \widehat{\mathbf{Y}}=\mathcal{G}(\mathbf{Z}=\mathbf{z}, \mathbf{X}^{*}, \mathbf{Y}^{*}(\mathbf{z})))\geq t\}\times P(\mathbf{Z}=\mathbf{z})\big]\\
        &=\sum_{\mathbf{z}\in \Omega}\big[P\{T(\mathbf{Z}=\mathbf{z}, \widehat{\mathbf{Y}}=\mathcal{G}(\mathbf{Z}=\mathbf{z}, \mathbf{X}^{*}, \mathbf{Y}^{*}))\geq t\}\times P(\mathbf{Z}=\mathbf{z})\big].
    \end{align*}
    Therefore, equation (\ref{eqn: one-sided exact p-value stochastic}) also holds. 
\end{proof}

We are now ready to prove Theorem~\ref{thm: hypothesis testing}.

\begin{proof}[(Proof of Theorem~\ref{thm: hypothesis testing})]
    Recall that in Algorithm~\ref{alg: part}, we consider 
   \begin{align*}
        \widehat{p}=\frac{1}{L}\sum_{l=1}^{L}\mathbbm{1}\{T^{(l)}\geq t\}=\frac{1}{L}\sum_{l=1}^{L}\mathbbm{1}\{T(\mathbf{Z}^{(l)}, \widehat{\mathbf{Y}}^{(l)})\geq t\}.
   \end{align*}
  For $l=1,\dots, L$, random variables $\mathbbm{1}\{T^{(l)}\geq t\}$ independently and identically follow Bernoulli distribution. Moreover, invoking the arguments in Lemma~\ref{lem: p-value}, if the missing outcome imputation algorithm $\mathcal{G}$ is deterministic, we have 
  \begin{align*}
      &\quad \ P(T^{(l)}\geq t \mid H_{0}) \\
      &=P(T(\mathbf{Z}^{(l)}, \widehat{\mathbf{Y}}^{(l)})\geq t \mid H_{0}) \\
      &= P(T(\mathbf{Z}^{(l)}, \mathcal{G}(\mathbf{Z}^{(l)}, \mathbf{X}^{*}, \mathbf{Y}^{*}))\geq t \mid H_{0}) \quad (\text{recall that $\mathbf{Y}^{*}$ is the observed realized outcomes}) \\
      &=\sum_{\mathbf{z}\in \Omega}\big[\mathbbm{1}\{T(\mathbf{Z}=\mathbf{z}, \widehat{\mathbf{Y}}=\mathcal{G}(\mathbf{Z}=\mathbf{z}, \mathbf{X}^{*}, \mathbf{Y}^{*}))\geq t\}\times P(\mathbf{Z}=\mathbf{z})\big]\\
      &=\sum_{\mathbf{z}\in \Omega}\big[\mathbbm{1}\{T(\mathbf{Z}=\mathbf{z}, \widehat{\mathbf{Y}}=\mathcal{G}(\mathbf{Z}=\mathbf{z}, \mathbf{X}^{*}, \mathbf{Y}^{*}(\mathbf{z})))\geq t\}\times P(\mathbf{Z}=\mathbf{z})\big]\\
      &=P(T(\mathbf{Z}, \widehat{\mathbf{Y}})\geq t \mid H_{0}),
  \end{align*}
  and if the missing outcome imputation algorithm $\mathcal{G}$ is stochastic, we also have 
  \begin{align*}
     P(T^{(l)}\geq t \mid H_{0}) &=P(T(\mathbf{Z}^{(l)}, \widehat{\mathbf{Y}}^{(l)})\geq t \mid H_{0}) \\
      &= P(T(\mathbf{Z}^{(l)}, \mathcal{G}(\mathbf{Z}^{(l)}, \mathbf{X}^{*}, \mathbf{Y}^{*}))\geq t \mid H_{0})\\
      &=\sum_{\mathbf{z}\in \Omega}\big[P\{T(\mathbf{Z}=\mathbf{z}, \widehat{\mathbf{Y}}=\mathcal{G}(\mathbf{Z}=\mathbf{z}, \mathbf{X}^{*}, \mathbf{Y}^{*}))\geq t\}\times P(\mathbf{Z}=\mathbf{z})\big]\\
      &=\sum_{\mathbf{z}\in \Omega}\big[P\{T(\mathbf{Z}=\mathbf{z}, \widehat{\mathbf{Y}}=\mathcal{G}(\mathbf{Z}=\mathbf{z}, \mathbf{X}^{*}, \mathbf{Y}^{*}(\mathbf{z})))\geq t\}\times P(\mathbf{Z}=\mathbf{z})\big]\\
      &=P(T(\mathbf{Z}, \widehat{\mathbf{Y}})\geq t \mid H_{0}).
  \end{align*}
  Therefore, by the law of large numbers, we have 
  \begin{equation*}
      \text{$\widehat{p} \xrightarrow{a.s.} p$ as the number of re-imputation runs $L\rightarrow \infty$.}
  \end{equation*}
  Moreover, by Hoeffding's inequality (\citealp{hoeffding1994probability}), for any $\epsilon>0$ and for all $L$, we have 
  \begin{equation*}
      P(| \widehat{p}-p|\geq \epsilon)\leq 2\exp(-2L\epsilon^{2}).
  \end{equation*}
  
\end{proof}

\begin{remark}
   \emph{The rationale of Theorem~\ref{assump: conditional indep} is clear: under the outcome missingness mechanism in Assumption~\ref{assump: conditional indep} and the null hypothesis $H_{0}$, the realized outcomes $\mathbf{Y}^{*}$ (including the information of missingness status $\mathbf{M}$) are invariant under different treatment assignments $\mathbf{Z}$. Therefore, we can derive an explicit form of the finite-population-exact $p$-value in terms of the observed data and the output of the outcome imputation algorithm $\mathcal{G}$ (see Lemma~\ref{lem: p-value}). Based on this explicit form, we can then use Monte Carlo simulations equipped with a missing data imputation algorithm $\mathcal{G}$ (i.e., the re-imputation step) to approximate the $p$-value $P(T(\mathbf{Z}, \widehat{\mathbf{Y}})\geq t \mid H_{0})$ as described in Algorithm~\ref{alg: part}, which gives us the result in Theorem~\ref{alg: part}.} 
\end{remark}

\subsection*{Appendix A.2: Proof of Theorem~\ref{thm: hypothesis testing with covariate adjustment} }

Similar to the proof of Theorem~\ref{thm: hypothesis testing}, the first step of proving Theorem~\ref{thm: hypothesis testing with covariate adjustment} is to give an explicit form of the finite-population-exact distribution of the test statistic $T(\mathbf{Z}, \bm{\epsilon})$ based on the residuals $\bm{\epsilon}=\widehat{\mathbf{Y}}-\widetilde{\mathbf{Y}}$ under the null hypothesis $H_{0}$.

\begin{lemma}\label{lem: p-value with covariate adjustment}
 Let $\mathcal{G}: (\mathbf{Z}, \mathbf{X}^{*}, \mathbf{Y}^{*})\mapsto \widehat{\mathbf{Y}}$ be any outcome imputation algorithm used in Algorithm~\ref{alg: part with covariate adjustment} and $\mathcal{H}: (\mathbf{X}^{*}, \widehat{\mathbf{Y}})\mapsto \widetilde{\mathbf{Y}}$ any prediction model for covariate adjustment used in Algorithm~\ref{alg: part with covariate adjustment} based on imputed outcomes $\widehat{\mathbf{Y}}$. Let $T=T(\mathbf{Z}, \bm{\epsilon})$ be any test statistic based on $\mathbf{Z}$ and residuals $\bm{\epsilon}=\widehat{\mathbf{Y}}-\widetilde{\mathbf{Y}}$ after covariate adjustment, and let $t$ denote the observed value of $T$ based on the observed data $(\mathbf{Z}, \mathbf{X}^{*}, \mathbf{Y}^{*})$. Let $p_{\text{adj}}=P(T(\mathbf{Z}, \bm{\epsilon})\geq t\mid H_{0})$ denote the finite-population-exact $p$-value (with covariate adjustment) under $H_{0}$. Under Assumptions \ref{assump: randomization design} and \ref{assump: conditional indep}, if both the missing outcome imputation algorithm $\mathcal{G}$ and the prediction model $\mathcal{H}$ for the covariate adjustment are deterministic (i.e., $\widehat{\mathbf{Y}}=\mathcal{G}(\mathbf{Z}, \mathbf{X}^{*}, \mathbf{Y}^{*})$ is a deterministic vector given $(\mathbf{Z}, \mathbf{X}^{*}, \mathbf{Y}^{*})$ and $\bm{\epsilon}$ is a deterministic vector given $(\mathbf{X}^{*}, \widehat{\mathbf{Y}})$), we have
    \begin{align}\label{eqn: one-sided exact p-value with adjustment deterministic}
    &P(T(\mathbf{Z}, \bm{\epsilon})\geq t \mid H_{0}) \nonumber \\ 
    &\quad \quad =\sum_{\mathbf{z}\in \Omega}\big[\mathbbm{1}\{T(\mathbf{Z}=\mathbf{z}, \bm{\epsilon}=\mathcal{G}(\mathbf{Z}=\mathbf{z}, \mathbf{X}^{*}, \mathbf{Y}^{*})-\mathcal{H}(\mathbf{X}^{*}, \mathcal{G}(\mathbf{Z}=\mathbf{z}, \mathbf{X}^{*}, \mathbf{Y}^{*})))\geq t\}\times P(\mathbf{Z}=\mathbf{z})\big],  
    \end{align}
    and if either the missing outcome imputation algorithm $\mathcal{G}$ or the prediction model $\mathcal{H}$ for the covariate adjustment is stochastic (i.e., either $\widehat{\mathbf{Y}}=\mathcal{G}(\mathbf{Z}, \mathbf{X}^{*}, \mathbf{Y}^{*})$ is a random vector given $(\mathbf{Z}, \mathbf{X}^{*}, \mathbf{Y}^{*})$ or $\bm{\epsilon}$ is a random vector given $(\mathbf{X}^{*}, \widehat{\mathbf{Y}})$), we have
    \begin{align}\label{eqn: one-sided exact p-value with adjustment stochastic}
    &P(T(\mathbf{Z}, \bm{\epsilon})\geq t \mid H_{0}) \nonumber \\ 
    &\quad \quad =\sum_{\mathbf{z}\in \Omega}\big[P\{T(\mathbf{Z}=\mathbf{z}, \bm{\epsilon}=\mathcal{G}(\mathbf{Z}=\mathbf{z}, \mathbf{X}^{*}, \mathbf{Y}^{*})-\mathcal{H}(\mathbf{X}^{*}, \mathcal{G}(\mathbf{Z}=\mathbf{z}, \mathbf{X}^{*}, \mathbf{Y}^{*})))\geq t\}\times P(\mathbf{Z}=\mathbf{z})\big].
    \end{align}
\end{lemma}

\begin{proof}
    We prove Lemma~\ref{lem: p-value with covariate adjustment} by discussing the following two cases separately.

    \textbf{Case 1: Both the missing outcome imputation algorithm $\mathcal{G}$ and the prediction model $\mathcal{H}$ for the covariate adjustment are deterministic.}

    In the proof of Lemma~\ref{lem: p-value}, we have shown that the realized outcomes $\mathbf{Y}^{*}=(Y_{111}^{*}, \dots, Y_{In_{I}K}^{*})$ are invariant under any treatment indicator vector $\mathbf{Z}$. Moreover, the residuals $\bm{\epsilon}=\widehat{\mathbf{Y}}-\widetilde{\mathbf{Y}}$ only depend on the observed covariates $\mathbf{X}^{*}$ and the imputed outcomes $\widehat{\mathbf{Y}}$, i.e., the treatment vector $\mathbf{Z}$ can only affect $\bm{\epsilon}$ through its effect on $\widehat{\mathbf{Y}}$. Then, we have 
    \begin{align*}
        &\quad \ P(T(\mathbf{Z}, \bm{\epsilon})\geq t \mid H_{0})\\
        &=P(T(\mathbf{Z}, \widehat{\mathbf{Y}}-\widetilde{\mathbf{Y}})\geq t \mid H_{0})\\
        &=\sum_{\mathbf{z}\in \Omega}\big[\mathbbm{1}\{T(\mathbf{Z}=\mathbf{z}, \bm{\epsilon}=\mathcal{G}(\mathbf{Z}=\mathbf{z}, \mathbf{X}^{*}, \mathbf{Y}^{*}(\mathbf{z}))-\mathcal{H}(\mathbf{X}^{*}, \mathcal{G}(\mathbf{Z}=\mathbf{z}, \mathbf{X}^{*}, \mathbf{Y}^{*}(\mathbf{z}))))\geq t\}\times P(\mathbf{Z}=\mathbf{z})\big]\\
        &=\sum_{\mathbf{z}\in \Omega}\big[\mathbbm{1}\{T(\mathbf{Z}=\mathbf{z}, \bm{\epsilon}=\mathcal{G}(\mathbf{Z}=\mathbf{z}, \mathbf{X}^{*}, \mathbf{Y}^{*})-\mathcal{H}(\mathbf{X}^{*}, \mathcal{G}(\mathbf{Z}=\mathbf{z}, \mathbf{X}^{*}, \mathbf{Y}^{*})))\geq t\}\times P(\mathbf{Z}=\mathbf{z})\big].
    \end{align*}
    Therefore, equation (\ref{eqn: one-sided exact p-value with adjustment deterministic}) holds. 
    
 \textbf{Case 2: If either the missing outcome imputation algorithm $\mathcal{G}$ or the prediction model $\mathcal{H}$ for the covariate adjustment is stochastic.}

    Similarly, since we have shown that the observed realized outcomes $\mathbf{Y}^{*}=\mathbf{Y}^{*}(\mathbf{z})$ for any $\mathbf{z}$ when \textit{both} the null $H_{0}$ and Assumption~\ref{assump: conditional indep} hold, we have 
    \begin{align*}
        &\quad \ P(T(\mathbf{Z}, \bm{\epsilon})\geq t \mid H_{0})\\
        &=P(T(\mathbf{Z}, \widehat{\mathbf{Y}}-\widetilde{\mathbf{Y}})\geq t \mid H_{0})\\
        &=\sum_{\mathbf{z}\in \Omega}\big[P\{T(\mathbf{Z}=\mathbf{z}, \bm{\epsilon}=\mathcal{G}(\mathbf{Z}=\mathbf{z}, \mathbf{X}^{*}, \mathbf{Y}^{*}(\mathbf{z}))-\mathcal{H}(\mathbf{X}^{*}, \mathcal{G}(\mathbf{Z}=\mathbf{z}, \mathbf{X}^{*}, \mathbf{Y}^{*}(\mathbf{z}))))\geq t\}\times P(\mathbf{Z}=\mathbf{z})\big]\\
        &=\sum_{\mathbf{z}\in \Omega}\big[P\{T(\mathbf{Z}=\mathbf{z}, \bm{\epsilon}=\mathcal{G}(\mathbf{Z}=\mathbf{z}, \mathbf{X}^{*}, \mathbf{Y}^{*})-\mathcal{H}(\mathbf{X}^{*}, \mathcal{G}(\mathbf{Z}=\mathbf{z}, \mathbf{X}^{*}, \mathbf{Y}^{*})))\geq t\}\times P(\mathbf{Z}=\mathbf{z})\big].
    \end{align*}
    Therefore, equation (\ref{eqn: one-sided exact p-value with adjustment stochastic}) also holds. 
\end{proof}

We are now ready to prove Theorem~\ref{thm: hypothesis testing with covariate adjustment}.

\begin{proof}[(Proof of Theorem~\ref{thm: hypothesis testing with covariate adjustment})]
    Recall that in Algorithm~\ref{alg: part with covariate adjustment}, we consider 
   \begin{align*}
        \widehat{p}_{\text{adj}}=\frac{1}{L}\sum_{l=1}^{L}\mathbbm{1}\{T^{(l)}\geq t\}=\frac{1}{L}\sum_{l=1}^{L}\mathbbm{1}\{T(\mathbf{Z}^{(l)}, \bm{\epsilon}^{(l)})\geq t\}.
   \end{align*}
  Random variables $\mathbbm{1}\{T^{(l)}\geq t\}$ independently and identically follow Bernoulli distribution. Invoking the arguments in Lemma~\ref{lem: p-value with covariate adjustment}, if both the missing outcome imputation algorithm $\mathcal{G}$ and the prediction model $\mathcal{H}$ for the covariate adjustment are deterministic, we have 
  \begin{align*}
      &\quad \ P(T^{(l)}\geq t \mid H_{0}) \\
      &=P(T(\mathbf{Z}^{(l)}, \bm{\epsilon}^{(l)})\geq t \mid H_{0}) \\
      &= P(T(\mathbf{Z}^{(l)}, \mathcal{G}(\mathbf{Z}^{(l)}, \mathbf{X}^{*}, \mathbf{Y}^{*})-\mathcal{H}(\mathbf{X}^{*}, \mathcal{G}(\mathbf{Z}^{(l)}, \mathbf{X}^{*}, \mathbf{Y}^{*})))\geq t \mid H_{0})\\
      &=\sum_{\mathbf{z}\in \Omega}\big[\mathbbm{1}\{T(\mathbf{Z}=\mathbf{z}, \bm{\epsilon}=\mathcal{G}(\mathbf{Z}^{(l)}, \mathbf{X}^{*}, \mathbf{Y}^{*})-\mathcal{H}(\mathbf{X}^{*}, \mathcal{G}(\mathbf{Z}^{(l)}, \mathbf{X}^{*}, \mathbf{Y}^{*})))\geq t\}\times P(\mathbf{Z}=\mathbf{z})\big]\\
      &=\sum_{\mathbf{z}\in \Omega}\big[\mathbbm{1}\{T(\mathbf{Z}=\mathbf{z}, \bm{\epsilon}=\mathcal{G}(\mathbf{Z}^{(l)}, \mathbf{X}^{*}, \mathbf{Y}^{*}(\mathbf{z}))-\mathcal{H}(\mathbf{X}^{*}, \mathcal{G}(\mathbf{Z}^{(l)}, \mathbf{X}^{*}, \mathbf{Y}^{*}(\mathbf{z}))))\geq t\}\times P(\mathbf{Z}=\mathbf{z})\big]\\
      &=P(T(\mathbf{Z}, \bm{\epsilon})\geq t \mid H_{0}),
  \end{align*}
  and if either the missing outcome imputation algorithm $\mathcal{G}$ or the prediction model $\mathcal{H}$ for the covariate adjustment is stochastic, we also have 
  \begin{align*}
      &\quad \ P(T^{(l)}\geq t \mid H_{0}) \\
      &= P(T(\mathbf{Z}^{(l)}, \mathcal{G}(\mathbf{Z}^{(l)}, \mathbf{X}^{*}, \mathbf{Y}^{*})-\mathcal{H}(\mathbf{X}^{*}, \mathcal{G}(\mathbf{Z}^{(l)}, \mathbf{X}^{*}, \mathbf{Y}^{*})))\geq t \mid H_{0})\\
      &=\sum_{\mathbf{z}\in \Omega}\big[P\{T(\mathbf{Z}=\mathbf{z}, \bm{\epsilon}=\mathcal{G}(\mathbf{Z}^{(l)}, \mathbf{X}^{*}, \mathbf{Y}^{*})-\mathcal{H}(\mathbf{X}^{*}, \mathcal{G}(\mathbf{Z}^{(l)}, \mathbf{X}^{*}, \mathbf{Y}^{*})))\geq t\}\times P(\mathbf{Z}=\mathbf{z})\big]\\
      &=\sum_{\mathbf{z}\in \Omega}\big[P\{T(\mathbf{Z}=\mathbf{z}, \bm{\epsilon}=\mathcal{G}(\mathbf{Z}^{(l)}, \mathbf{X}^{*}, \mathbf{Y}^{*}(\mathbf{z}))-\mathcal{H}(\mathbf{X}^{*}, \mathcal{G}(\mathbf{Z}^{(l)}, \mathbf{X}^{*}, \mathbf{Y}^{*}(\mathbf{z}))))\geq t\}\times P(\mathbf{Z}=\mathbf{z})\big]\\
      &=P(T(\mathbf{Z}, \bm{\epsilon})\geq t \mid H_{0}).
  \end{align*}
  Therefore, by the law of large numbers, we have 
  \begin{equation*}
      \text{$\widehat{p}_{\text{adj}}\xrightarrow{a.s.} p_{\text{adj}}$ as the number of re-imputation runs $L\rightarrow \infty$.}
  \end{equation*}
  Moreover, by Hoeffding's inequality (\citealp{hoeffding1994probability}), for any $\epsilon>0$ and for all $L$, we have 
  \begin{equation*}
      P(|\widehat{p}_{\text{adj}}-p_{\text{adj}}|\geq \epsilon)\leq 2\exp(-2L\epsilon^{2}).
  \end{equation*}
  
\end{proof}

\subsection*{Appendix A.3: Proof of Theorem~\ref{thm: confidence region}}

In Theorem~\ref{thm: confidence region}, the $p$-value $\widehat{p}(\Vec{\beta}_{0})$ can either be reported by Algorithm~\ref{alg: part} (without covariate adjustment after imputation) or Algorithm~\ref{alg: part with covariate adjustment} (with covariate adjustment after imputation) based on the inputs $(\mathbf{Z}, \mathbf{X}, \mathbf{Y}^{*}_{\Vec{\beta}_{0}})$. If we set the prediction model $\mathcal{H}$ for the covariate adjustment as a zero function, Algorithm~\ref{alg: part with covariate adjustment} will reduce to Algorithm~\ref{alg: part}. Therefore, it suffices to prove Theorem~\ref{thm: confidence region} with the $p$-value $\widehat{p}(\Vec{\beta}_{0})$ reported by Algorithm~\ref{alg: part with covariate adjustment}, of which the key is to prove the following lemma. 

\begin{lemma}\label{lem: p-value with adjusted outcomes}
Let $\mathcal{G}$ be any outcome imputation algorithm used in Algorithm~\ref{alg: part with covariate adjustment} and $\mathcal{H}$ any prediction model for covariate adjustment based on imputed outcomes used in Algorithm~\ref{alg: part with covariate adjustment}. Consider the unknown true causal parameters $\Vec{\beta}=(\Vec{\beta}_{1},\dots, \Vec{\beta}_{K})$ that satisfy $Y_{ijk}(1)=f_{k}(Y_{ijk}(0), \Vec{\beta}_{k}, \mathbf{x}_{ij}), \text{ for } i \in [I], j \in [n_{i}], k \in [K],$ with the prespecified maps $f_{k}$ (i.e., the treatment effect models (\ref{eqn: effect model})). Define the transformed realized outcomes (with missingness) under the true causal parameters $\Vec{\beta}=(\Vec{\beta}_{1},\dots, \Vec{\beta}_{K})$ as $\mathbf{Y}^{*}_{\Vec{\beta}}=(\mathbf{Y}^{*}_{1, \Vec{\beta}_{1}}, \dots, \mathbf{Y}^{*}_{K, \Vec{\beta}_{K}})$, in which we have $\mathbf{Y}^{*}_{k, \Vec{\beta}_{k}}=(Y^{*}_{11k, \Vec{\beta}_{k}}, \dots, Y^{*}_{In_{I}k, \Vec{\beta}_{k}})$ where $Y^{*}_{ijk, \Vec{\beta}_{k}}=\text{``Missing"}$ if $M_{ijk}=1$ and $Y^{*}_{ijk, \Vec{\beta}_{k}}=Z_{ij}Y_{ijk}^{*}+(1-Z_{ij})f_{k}(Y_{ijk}^{*}, \Vec{\beta}_{k}, \mathbf{x}_{ij})=Z_{ij}Y_{ijk}+(1-Z_{ij})f_{k}(Y_{ijk}, \Vec{\beta}_{k}, \mathbf{x}_{ij})$ if $M_{ijk}=0$. Let $T_{\Vec{\beta}}=T(\mathbf{Z}, \bm{\epsilon})$ be any test statistic based on $\mathbf{Z}$ and residuals $\bm{\epsilon}$ after covariate adjustment (based on $\mathbf{Y}^{*}_{\Vec{\beta}})$, and let $t_{\Vec{\beta}}$ denote the observed value of $T_{\Vec{\beta}}$ based on the observed $(\mathbf{Z}, \mathbf{X}, \mathbf{Y}^{*}_{\Vec{\beta}})$ (under the oracle ${\Vec{\beta}}$). Let $p(\Vec{\beta})=P(T_{\Vec{\beta}}\geq t_{\Vec{\beta}})$ and consider the approximate $p$-value $\widehat{p}(\Vec{\beta})$ reported by the imputation and re-imputation approach with covariate adjustment described in Algorithm~\ref{alg: part with covariate adjustment} based on the $(\mathbf{Z}, \mathbf{X}, \mathbf{Y}^{*}_{\Vec{\beta}})$. Under Assumptions \ref{assump: randomization design} and \ref{assump: conditional indep with more restrictions}, we have $\widehat{p}(\Vec{\beta})\xrightarrow{a.s.} p(\Vec{\beta})$ as the number of re-imputation runs $L\rightarrow \infty$. 
    \end{lemma}

\begin{proof}
    We let $\mathbf{Y}^{*}_{\Vec{\beta}}(\mathbf{z})$ denote the values of the transformed realized outcomes under $\mathbf{Z}=\mathbf{z}\in \Omega$. Note that under the true causal parameters $\Vec{\beta}=(\Vec{\beta}_{1},\dots, \Vec{\beta}_{K})$ that satisfy $Y_{ijk}(1)=f_{k}(Y_{ijk}(0), \Vec{\beta}_{k}, \mathbf{x}_{ij}) \text{ for all } i, j, k,$, the transformed true outcomes $Y_{ijk, \Vec{\beta}_{k}}=Z_{ij}Y_{ijk}+(1-Z_{ij})f_{k}(Y_{ijk}, \Vec{\beta}_{k}, \mathbf{x}_{ij})$ are invariant under different treatment assignments $\mathbf{Z}$, which follows from $Y_{ijk, \Vec{\beta}_{k}}(1)=Y_{ijk}(1)=f_{k}(Y_{ijk}(0), \Vec{\beta}_{k}, \mathbf{x}_{ij})=Y_{ijk, \Vec{\beta}_{k}}(0)$. Moreover, under Assumption~\ref{assump: conditional indep with more restrictions}, the missingness indicators $M_{ijk}$ are also invariant under different treatment assignments $\mathbf{Z}$. Therefore, the transformed realized outcomes $\mathbf{Y}^{*}_{\Vec{\beta}}(\mathbf{z})$ are invariant under different $\mathbf{z}\in \Omega$. In other words, under the true causal parameters $\Vec{\beta}=(\Vec{\beta}_{1},\dots, \Vec{\beta}_{K})$ and Assumption~\ref{assump: conditional indep with more restrictions}, the observed values of the transformed realized outcomes $\mathbf{Y}^{*}_{\Vec{\beta}}$ equal $\mathbf{Y}^{*}_{\Vec{\beta}}(\mathbf{z})$ for all $\mathbf{z}\in \Omega$.

    Then, we can follow a similar argument to that in the proof of Lemma~\ref{lem: p-value with covariate adjustment} to show the desired conclusion. Specifically, in Algorithm~\ref{alg: part with covariate adjustment} with the transformed realized outcomes $\mathbf{Y}^{*}_{\Vec{\beta}}(\mathbf{z})$ under the true causal parameters $\Vec{\beta}$ in the treatment effects model (\ref{eqn: effect model}), we consider 
   \begin{align*}
        \widehat{p}(\Vec{\beta})=\frac{1}{L}\sum_{l=1}^{L}\mathbbm{1}\{T^{(l)}_{\Vec{\beta}}\geq t_{\Vec{\beta}}\}=\frac{1}{L}\sum_{l=1}^{L}\mathbbm{1}\{T(\mathbf{Z}^{(l)}, \bm{\epsilon}^{(l)}_{\Vec{\beta}})\geq t_{\Vec{\beta}}\},
   \end{align*}
   in which $\bm{\epsilon}^{(l)}_{\Vec{\beta}}=\mathcal{G}(\mathbf{Z}^{(l)}, \mathbf{X}, \mathbf{Y}^{*}_{\Vec{\beta}})-\mathcal{H}(\mathbf{X}, \mathcal{G}(\mathbf{Z}^{(l)}, \mathbf{X}, \mathbf{Y}^{*}_{\Vec{\beta}}))$. For $l=1,\dots, L$, random variables $\mathbbm{1}\{T^{(l)}_{\Vec{\beta}}\geq t_{\Vec{\beta}}\}$ independently and identically follow Bernoulli distribution. Moreover, since we have shown that the observed $\mathbf{Y}^{*}_{\Vec{\beta}}$ equal $\mathbf{Y}^{*}_{\Vec{\beta}}(\mathbf{z})$ for all $\mathbf{z}\in \Omega$, if both the imputation algorithm $\mathcal{G}$ and the working model $\mathcal{H}$ for the covariate adjustment are deterministic, we have 
  \begin{align*}
      &\quad \ P(T^{(l)}_{\Vec{\beta}}\geq t_{\Vec{\beta}} \mid H_{0}) \\
      &=P(T_{\Vec{\beta}}(\mathbf{Z}^{(l)}, \bm{\epsilon}^{(l)}_{\Vec{\beta}})\geq t_{\Vec{\beta}} \mid H_{0}) \\
      &= P(T_{\Vec{\beta}}(\mathbf{Z}^{(l)}, \mathcal{G}(\mathbf{Z}^{(l)}, \mathbf{X}, \mathbf{Y}^{*}_{\Vec{\beta}})-\mathcal{H}(\mathbf{X}, \mathcal{G}(\mathbf{Z}^{(l)}, \mathbf{X}, \mathbf{Y}^{*}_{\Vec{\beta}})))\geq t_{\Vec{\beta}} \mid H_{0})\\
      &=\sum_{\mathbf{z}\in \Omega}\big[\mathbbm{1}\{T_{\Vec{\beta}}(\mathbf{Z}=\mathbf{z}, \bm{\epsilon}_{\Vec{\beta}}=\mathcal{G}(\mathbf{Z}^{(l)}, \mathbf{X}, \mathbf{Y}^{*}_{\Vec{\beta}})-\mathcal{H}(\mathbf{X}, \mathcal{G}(\mathbf{Z}^{(l)}, \mathbf{X}, \mathbf{Y}^{*}_{\Vec{\beta}})))\geq t_{\Vec{\beta}}\}\times P(\mathbf{Z}=\mathbf{z})\big]\\
      &=\sum_{\mathbf{z}\in \Omega}\big[\mathbbm{1}\{T_{\Vec{\beta}}(\mathbf{Z}=\mathbf{z}, \bm{\epsilon}_{\Vec{\beta}}=\mathcal{G}(\mathbf{Z}^{(l)}, \mathbf{X}, \mathbf{Y}^{*}_{\Vec{\beta}}(\mathbf{z}))-\mathcal{H}(\mathbf{X}, \mathcal{G}(\mathbf{Z}^{(l)}, \mathbf{X}, \mathbf{Y}^{*}_{\Vec{\beta}}(\mathbf{z}))))\geq t_{\Vec{\beta}}\}\times P(\mathbf{Z}=\mathbf{z})\big]\\
      &=P(T_{\Vec{\beta}}(\mathbf{Z}, \bm{\epsilon}_{\Vec{\beta}})\geq t_{\Vec{\beta}} \mid H_{0}),
  \end{align*}
  and if either the missing outcome imputation algorithm $\mathcal{G}$ or the prediction model $\mathcal{H}$ for the covariate adjustment is stochastic, we also have 
  \begin{align*}
      &\quad \ P(T^{(l)}_{\Vec{\beta}}\geq t_{\Vec{\beta}} \mid H_{0}) \\
      &= P(T_{\Vec{\beta}}(\mathbf{Z}^{(l)}, \mathcal{G}(\mathbf{Z}^{(l)}, \mathbf{X}, \mathbf{Y}^{*}_{\Vec{\beta}})-\mathcal{H}(\mathbf{X}, \mathcal{G}(\mathbf{Z}^{(l)}, \mathbf{X}, \mathbf{Y}^{*}_{\Vec{\beta}})))\geq t_{\Vec{\beta}} \mid H_{0})\\
      &=\sum_{\mathbf{z}\in \Omega}\big[P\{T_{\Vec{\beta}}(\mathbf{Z}=\mathbf{z}, \bm{\epsilon}_{\Vec{\beta}}=\mathcal{G}(\mathbf{Z}^{(l)}, \mathbf{X}, \mathbf{Y}^{*}_{\Vec{\beta}})-\mathcal{H}(\mathbf{X}, \mathcal{G}(\mathbf{Z}^{(l)}, \mathbf{X}, \mathbf{Y}^{*}_{\Vec{\beta}})))\geq t_{\Vec{\beta}}\}\times P(\mathbf{Z}=\mathbf{z})\big]\\
      &=\sum_{\mathbf{z}\in \Omega}\big[P\{T_{\Vec{\beta}}(\mathbf{Z}=\mathbf{z}, \bm{\epsilon}_{\Vec{\beta}}=\mathcal{G}(\mathbf{Z}^{(l)}, \mathbf{X}, \mathbf{Y}^{*}_{\Vec{\beta}}(\mathbf{z}))-\mathcal{H}(\mathbf{X}, \mathcal{G}(\mathbf{Z}^{(l)}, \mathbf{X}, \mathbf{Y}^{*}_{\Vec{\beta}}(\mathbf{z}))))\geq t_{\Vec{\beta}}\}\times P(\mathbf{Z}=\mathbf{z})\big]\\
      &=P(T_{\Vec{\beta}}(\mathbf{Z}, \bm{\epsilon}_{\Vec{\beta}})\geq t_{\Vec{\beta}} \mid H_{0}).
  \end{align*}
  Therefore, by the law of large numbers, we have 
  \begin{equation*}
      \text{$\widehat{p}(\Vec{\beta})\xrightarrow{a.s.} p(\Vec{\beta})$ as the number of re-imputation runs $L\rightarrow \infty$.}
  \end{equation*}
  Moreover, by Hoeffding's inequality (\citealp{hoeffding1994probability}), for any $\epsilon>0$ and for all $L$, we have 
  \begin{equation*}
      P(|\widehat{p}(\Vec{\beta})-p(\Vec{\beta})|\geq \epsilon)\leq 2\exp(-2L\epsilon^{2}).
  \end{equation*}
\end{proof}

Then Theorem~\ref{thm: confidence region} can be proved based on Lemma~\ref{lem: p-value with adjusted outcomes}.

\begin{proof}[(Proof of Theorem~\ref{thm: confidence region})]

We first follow a similar argument to the proof of Proposition 1.1 in \citet{luo2021leveraging} to show that $P(p(\Vec{\beta})\leq \alpha)\leq \alpha$. Specifically, we let $T_{\Vec{\beta}, (1)}< T_{\Vec{\beta}, (2)}< \dots< T_{\Vec{\beta}, (M)}$ denote the ordered $M$ unique values of the test statistic $T_{\Vec{\beta}}$ among all possible treatment assignments $\mathbf{z}\in \Omega$. For $m=1,\dots, M$, we define $\xi_{m}=P(T_{\Vec{\beta}}=T_{\Vec{\beta}, (m)})>0$, with $\sum_{m=1}^{M}\xi_{m}=1$. Then, the $p(\Vec{\beta})=P(T_{\Vec{\beta}}\geq t_{\Vec{\beta}})$ takes values $\xi_{M}, \xi_{M}+\xi_{M-1}, \dots, \xi_{M}+\xi_{M-1}+\dots + \xi_{1}$ with probabilities $\xi_{M}, \xi_{M-1}, \dots, \xi_{1}$ respectively. Moreover, for any $\alpha\in (\xi_{M},1)$, there exists a unique $1 \leq m_{0} \leq M-1$ such that $\alpha\in (\xi_{M}+\xi_{M-1}+\dots + \xi_{M-m_{0}+1}, \xi_{M}+\xi_{M-1}+\dots + \xi_{M-m_{0}}]/\{1\}$. This implies that
\begin{align*}
    P(p(\Vec{\beta})\leq \alpha)&=\sum_{j=1}^{m_{0} }P(p(\Vec{\beta})= \xi_{M}+\xi_{M-1}+\dots + \xi_{M-j+1})\\
    &=\xi_{M}+\xi_{M-1}+\dots + \xi_{M-m_{0}+1}\\
    &\leq \alpha.
\end{align*}
If $\alpha\in (0, \xi_{M})$, we have $P(p(\Vec{\beta})\leq \alpha)=0\leq \alpha$. If $\alpha= \xi_{M}$, we have $P(p(\Vec{\beta})\leq \alpha)=P(p(\Vec{\beta})= \xi_{M})=\xi_{M}=\alpha$. Therefore, we have shown that $P(p(\Vec{\beta})\leq \alpha)\leq \alpha$ for any $\alpha\in (0, 1)$.

We then use the above result to prove the desired conclusion. If $\alpha<\xi_{M}+\xi_{M-1}+\dots + \xi_{M-m_{0}}$, the $P(p(\Vec{\beta})\leq x)$ is a continuous function of $x$ around the point $x=\alpha$. Therefore, invoking Lemma~\ref{lem: p-value with adjusted outcomes} (i.e., $\widehat{p}(\Vec{\beta})\xrightarrow{a.s.} p(\Vec{\beta})$ as $L\rightarrow\infty$), we have 
\begin{align*}
    \lim_{L\rightarrow\infty }  P(\Vec{\beta} \in CR)&=\lim_{L\rightarrow\infty } P(\widehat{p}(\Vec{\beta})\geq \alpha )\\
    &\geq 1-\lim_{L\rightarrow\infty } P(\widehat{p}(\Vec{\beta})\leq \alpha )\quad \text{(because $\widehat{p}(\Vec{\beta})\xrightarrow{a.s.} p(\Vec{\beta})$ implies $\widehat{p}(\Vec{\beta})\xrightarrow{\mathcal{L}} p(\Vec{\beta})$)} \\
    &=1-P(p(\Vec{\beta})\leq \alpha ) \\
    &\geq 1-\alpha.
\end{align*}
If $\alpha=\xi_{M}+\xi_{M-1}+\dots + \xi_{M-m_{0}}$ with $m_{0}<M-1$, we consider some $\varsigma\in (\xi_{M}+\xi_{M-1}+\dots + \xi_{M-m_{0} }, \xi_{M}+\xi_{M-1}+\dots + \xi_{M-m_{0}-1})=(\alpha, \xi_{M}+\xi_{M-1}+\dots + \xi_{M-m_{0}-1})$. Then, the $P(p(\Vec{\beta})\leq x)$ is a continuous function of $x$ around the point $x=\varsigma$. Therefore, invoking Lemma~\ref{lem: p-value with adjusted outcomes} (i.e., $\widehat{p}(\Vec{\beta})\xrightarrow{a.s.} p(\Vec{\beta})$ as $L\rightarrow\infty$), we have 
\begin{align*}
    \lim_{L\rightarrow\infty }  P(\Vec{\beta} \in CR)&=\lim_{L\rightarrow\infty } P(\widehat{p}(\Vec{\beta})\geq \alpha )\\
    &= 1-\lim_{L\rightarrow\infty } P(\widehat{p}(\Vec{\beta})< \alpha )\\
    &\geq 1-\lim_{L\rightarrow\infty } P(\widehat{p}(\Vec{\beta})\leq \varsigma )\\
    &=1-P(p(\Vec{\beta})\leq \varsigma ) \quad \text{(because $\widehat{p}(\Vec{\beta})\xrightarrow{a.s.} p(\Vec{\beta})$ implies $\widehat{p}(\Vec{\beta})\xrightarrow{\mathcal{L}} p(\Vec{\beta})$)} \\
     &=1-P(p(\Vec{\beta})\leq \alpha ) \\
    &\geq 1-\alpha.
\end{align*}
Therefore, we have shown that $\lim_{L\rightarrow\infty }  P(\Vec{\beta} \in CR)\geq 1-\alpha$.
    
\end{proof}

We briefly summarize the principle of the proof of Theorem~\ref{thm: confidence region}. Note that, under Assumption~\ref{assump: conditional indep with more restrictions}, the missingness indicators $\mathbf{M}$ are invariant under any treatment assignments $\mathbf{Z}$ and any $H_{\Vec{\beta}_{0}}$. Then, under $H_{\Vec{\beta}_{0}}$, the transformed realized outcomes $\mathbf{Y}^{*}_{\Vec{\beta}_{0}}$ are invariant under different $\mathbf{Z}$. In other words, testing $H_{\Vec{\beta}_{0}}$ is equivalent to testing Fisher's sharp null of no effect with the transformed outcomes $\mathbf{Y}^{*}_{\Vec{\beta}_{0}}$. Therefore, applying Algorithm~\ref{alg: part} or Algorithm~\ref{alg: part with covariate adjustment} to $\mathbf{Y}^{*}_{\Vec{\beta}_{0}}$ can produce (approximate) finite-population-exact $p$-values $\widehat{p}(\Vec{\beta}_{0})$ for testing each $H_{\Vec{\beta}_{0}}$, and an approximate finite-population-exact confidence region $CR$ can be obtained via inverting the $p$-values $\widehat{p}(\Vec{\beta}_{0})$ for various $\Vec{\beta}_{0}$. 

\section*{Appendix B: Additional Simulation Studies}

\subsection*{Appendix B.1: Simulation Studies with A Different Outcome Missingness Rate}

In Appendix B.1, we follow the setting and procedure in Section~\ref{subsec: simulation studies} to conduct similar simulation studies, but setting the outcome missingness rate as $25\%$. The corresponding simulation results in Figure~\ref{fig: single outcome simulations with 25 missingness rate} follow a similar pattern to that in Figure~\ref{fig: single outcome simulations} in the main text and confirms finite-population-exact type-I error rate control and gains in power compared with non-informative imputation (e.g., median imputation). Meanwhile, Figure~\ref{fig: single outcome simulations with 25 missingness rate} confirms that as the outcome missingness rate decreases, the gap in power between the randomization test based on oracle outcomes and the imputation and re-imputation framework becomes smaller, especially when the sample size is large. 

\begin{figure}[bp!]
      \centering
        \captionsetup[subfigure]{skip=2pt} % Adjust this value as needed
        \includegraphics[width=\textwidth]{imgs/legendcustomlinestypes.pdf}
      \begin{subfigure}[b]{0.4\textwidth}
        
        \includegraphics[width=\textwidth]{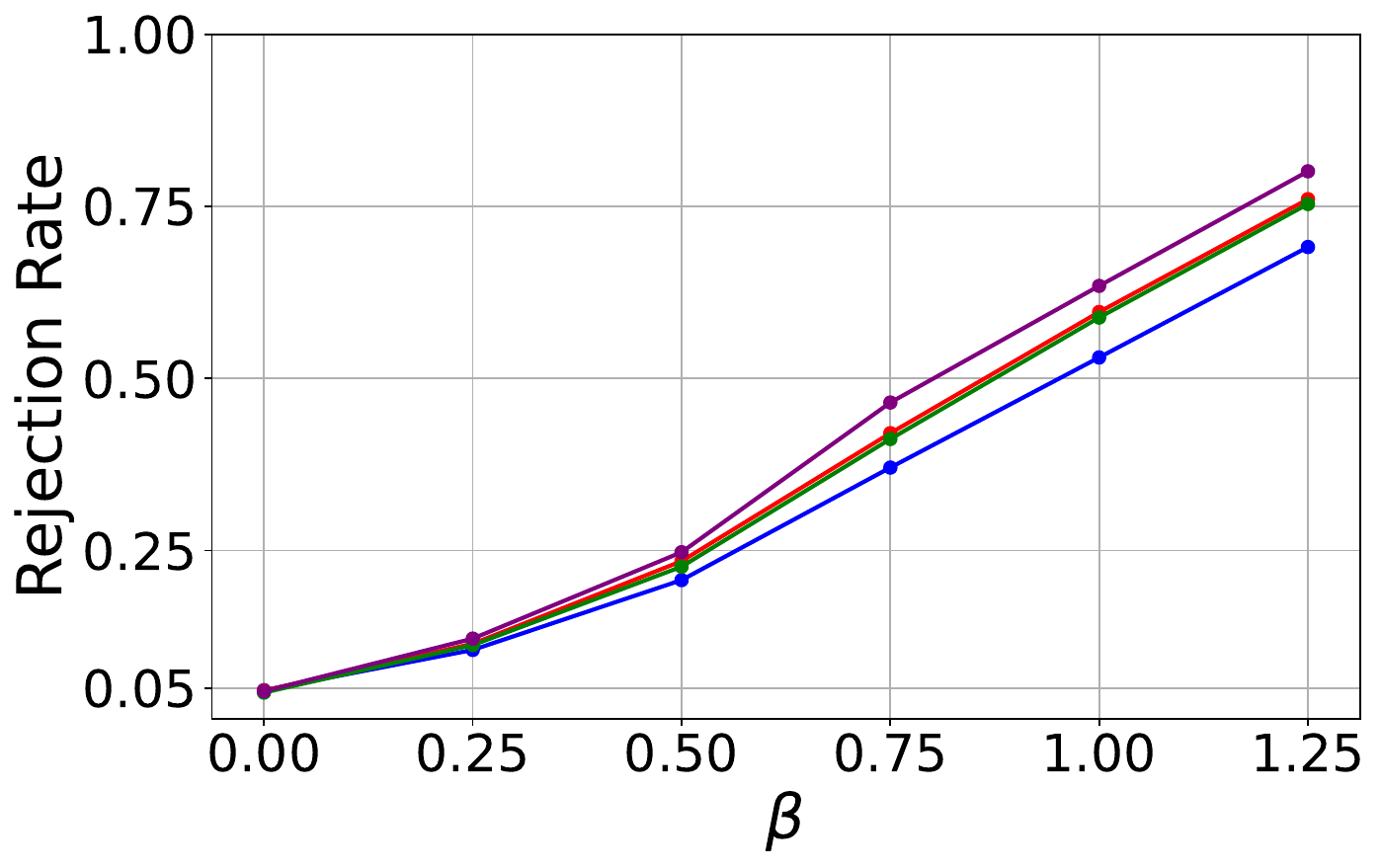}
        \caption{Model 1 ($N=50$)}
      \end{subfigure}
      \hspace{0.1cm}
      \begin{subfigure}[b]{0.4\textwidth}
        \includegraphics[width=\textwidth]{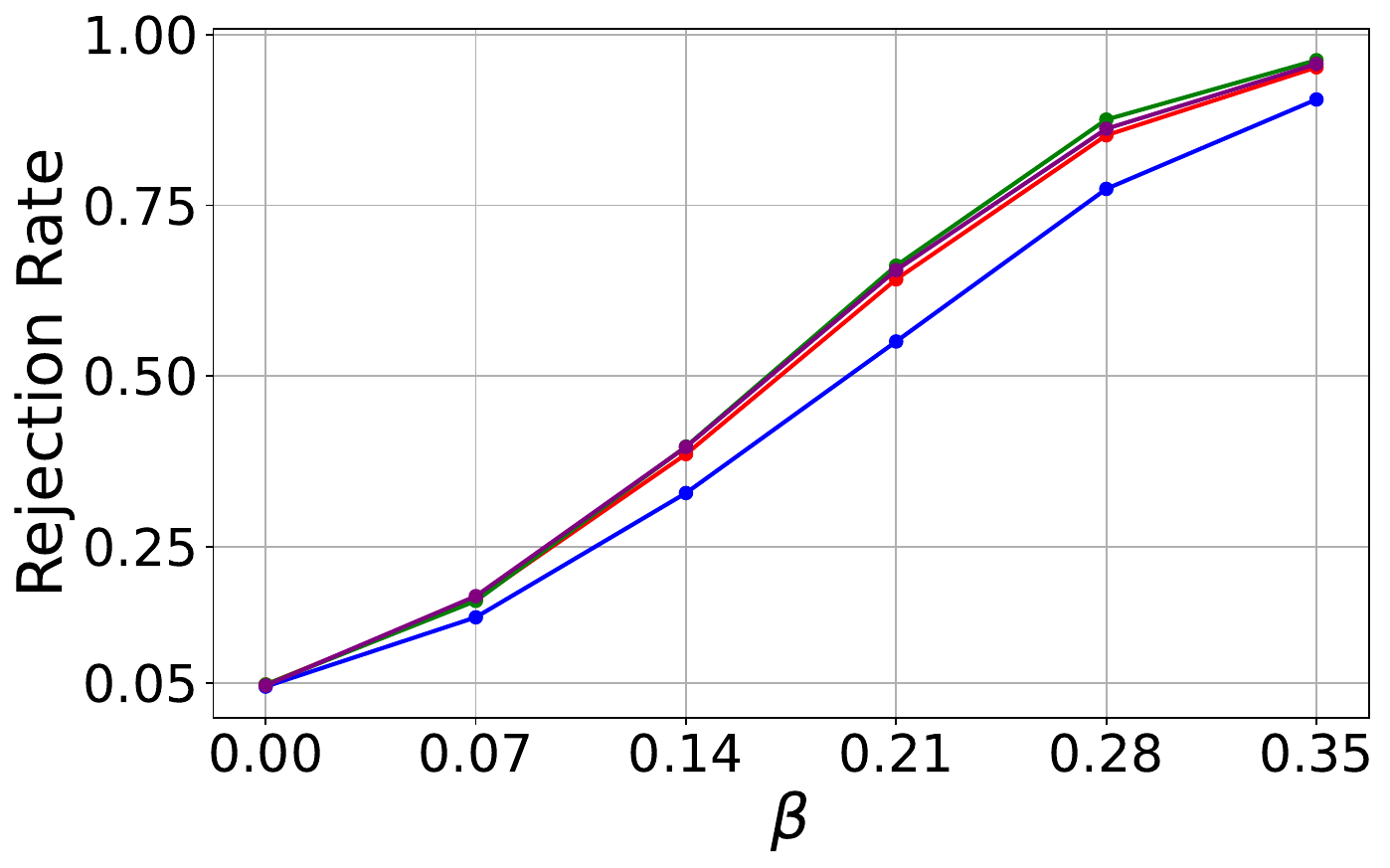}
        \caption{Model 1 ($N=1000$)}
      \end{subfigure}

    % Constant treatment effect; non-linear model for the true outcome; non-linear selection model for the missingness status, without interference

      \begin{subfigure}[b]{0.4\textwidth}
        
        \includegraphics[width=\textwidth]{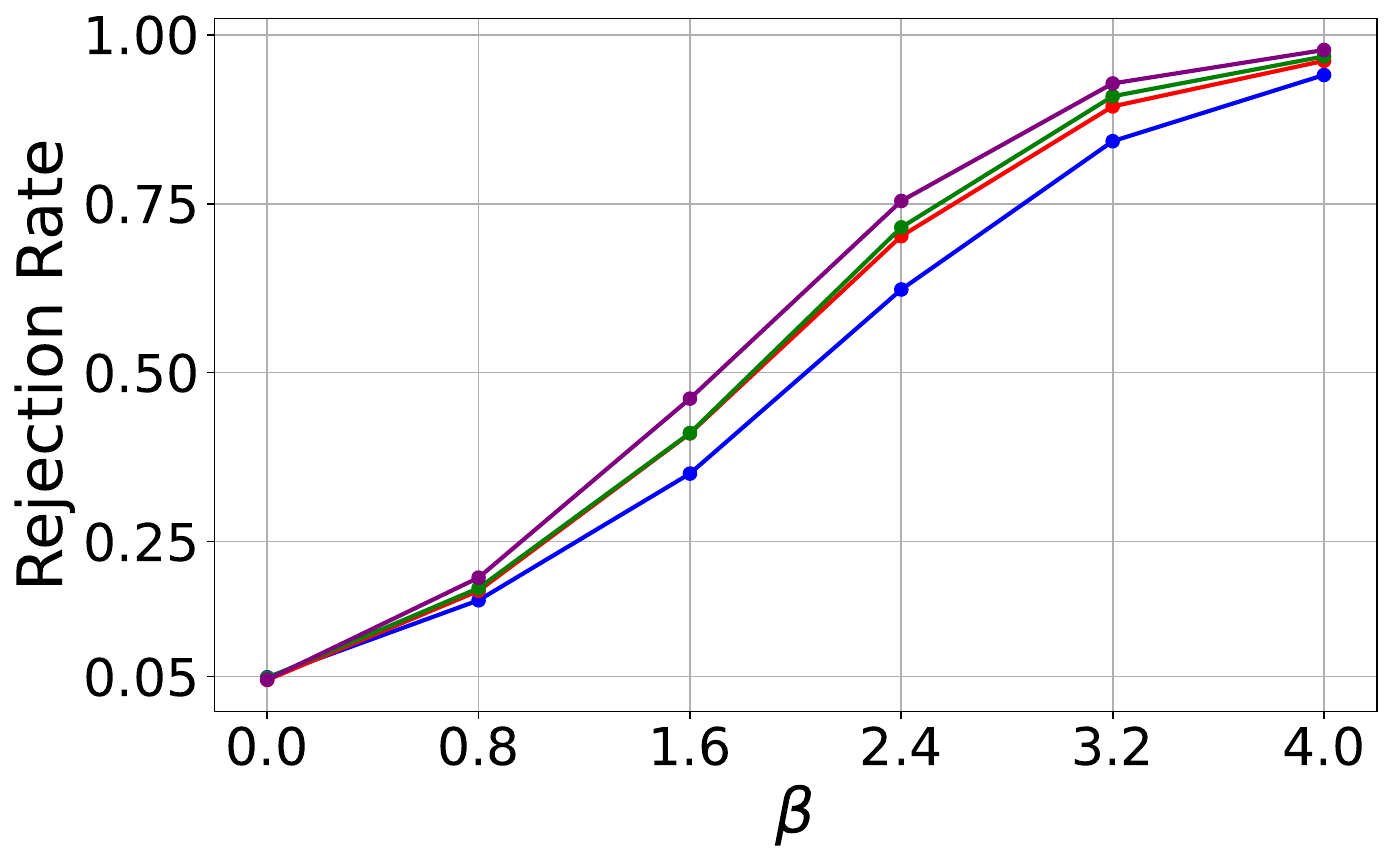}
        \caption{Model 2 ($N=50$)}
      \end{subfigure}
      \hspace{0.1cm}
      \begin{subfigure}[b]{0.4\textwidth}
        
        \includegraphics[width=\textwidth]{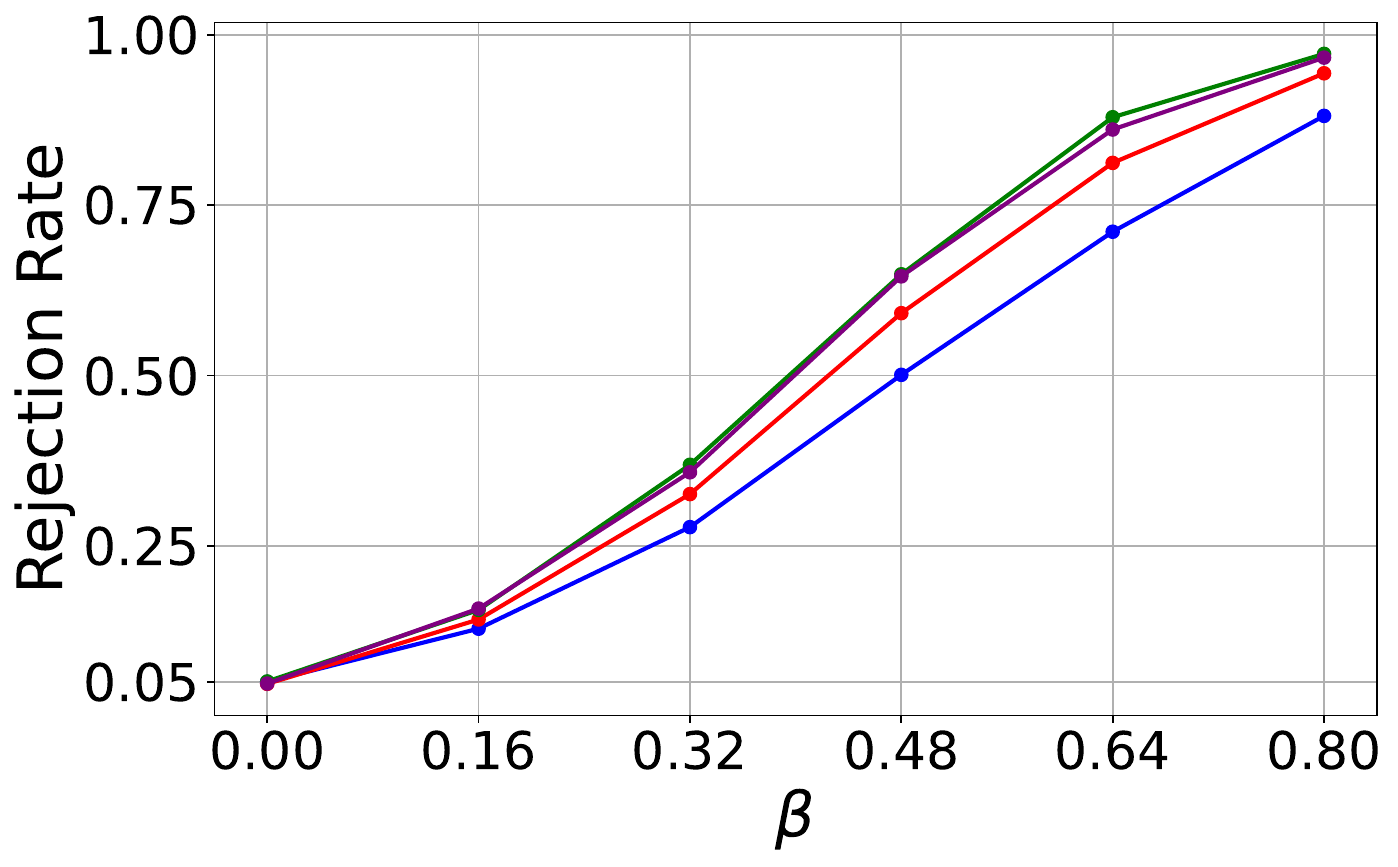}
        \caption{Model 2 ($N=1000$)}
      \end{subfigure}
    
    % Heterogeneous treatment effect; non-linear model for the true outcome; non-linear selection model for the missingness status, without interference
    
      \begin{subfigure}[b]{0.4\textwidth}
        
        \includegraphics[width=\textwidth]{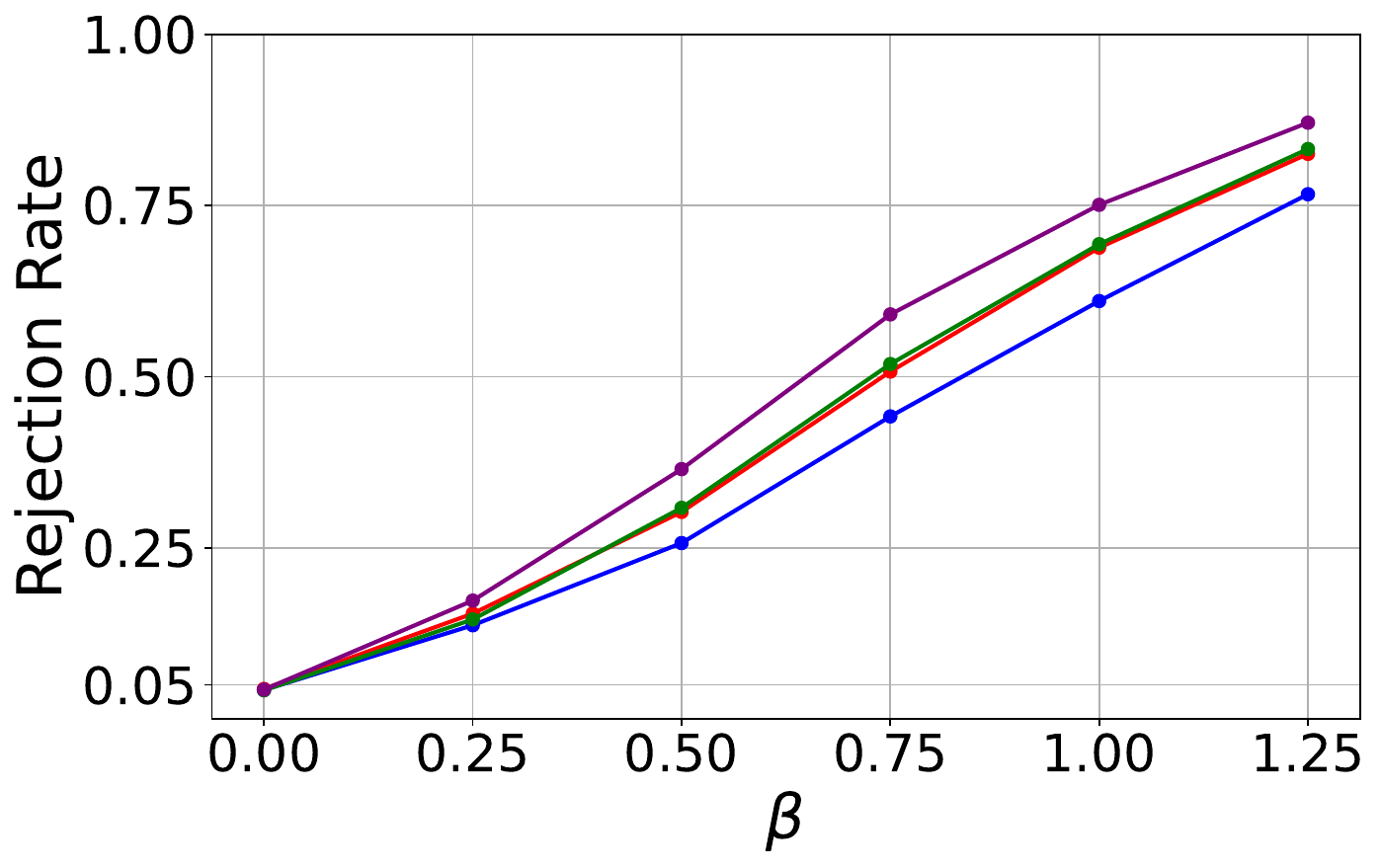}
        \caption{Model 3 ($N=50$)}
      \end{subfigure}
      \hspace{0.1cm}
      \begin{subfigure}[b]{0.4\textwidth}
        
        \includegraphics[width=\textwidth]{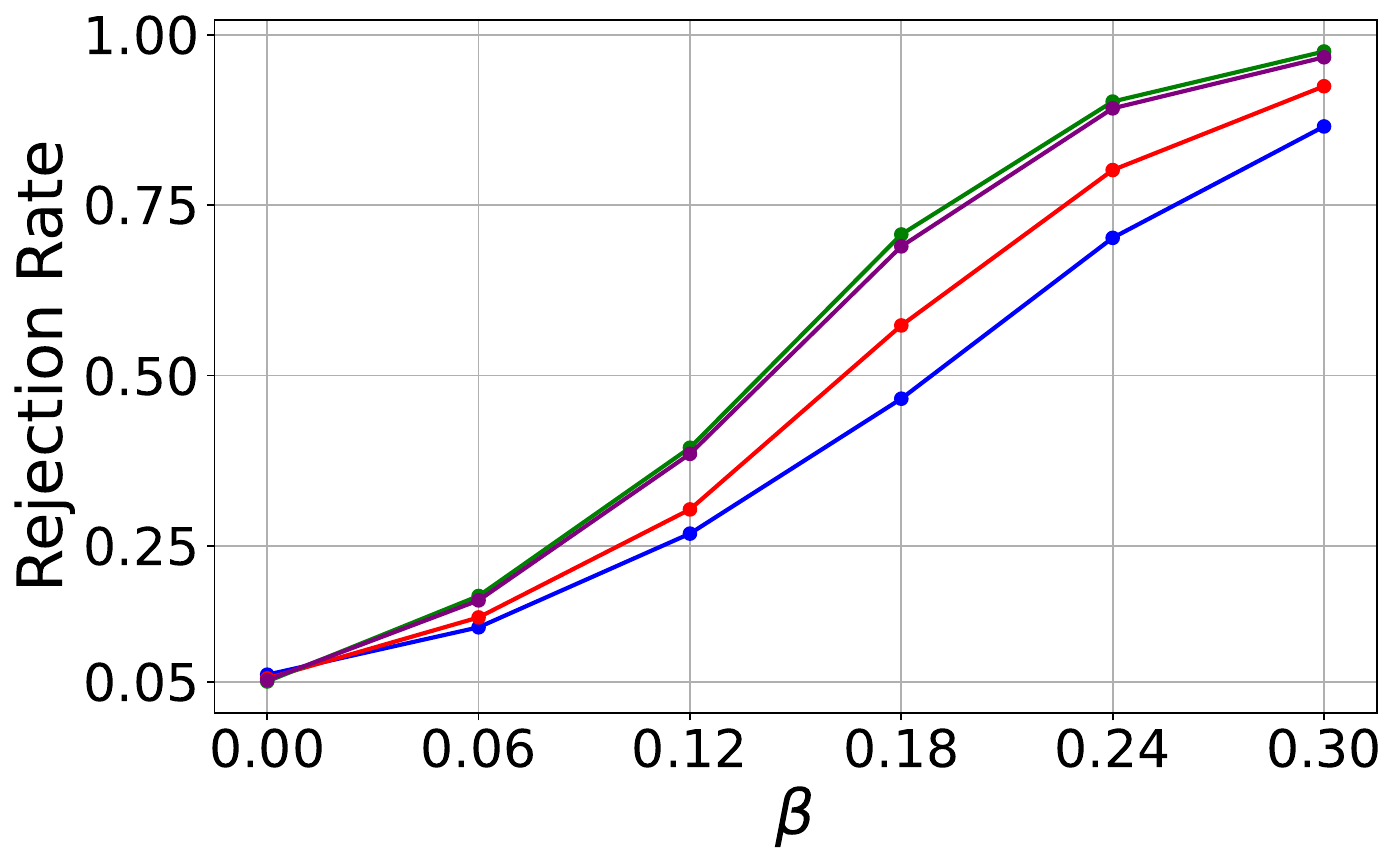}
        \caption{Model 3 ($N=1000$)}
      \end{subfigure}

    % Heterogeneous treatment effect; non-linear model for the true outcome; non-linear selection model for the missingness status, with interference
      
      \begin{subfigure}[b]{0.4\textwidth}
        
        \includegraphics[width=\textwidth]{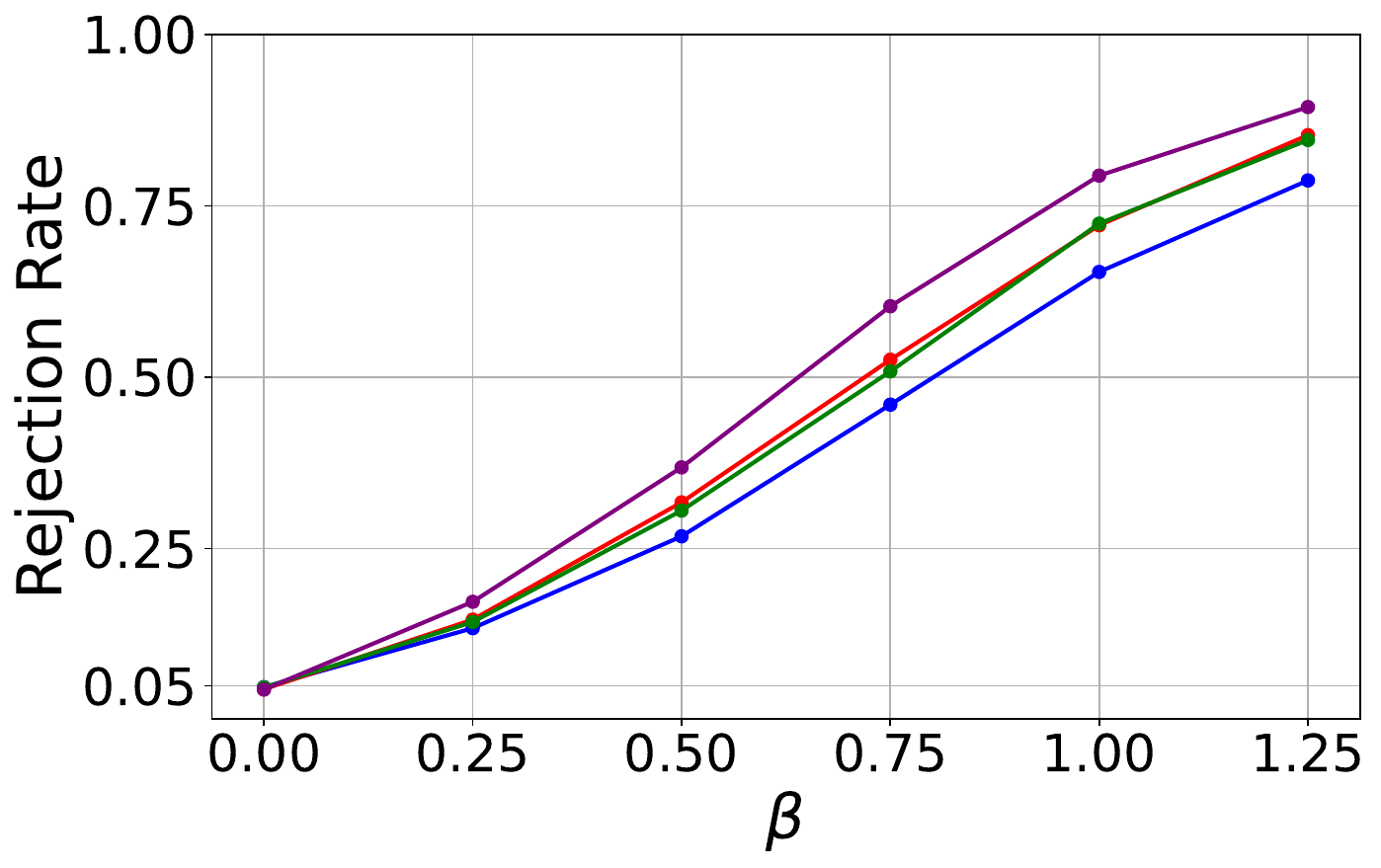}
        \caption{Model 4 ($N=50$)}
      \end{subfigure}
      \hspace{0.1cm}
      \begin{subfigure}[b]{0.4\textwidth}
        
        \includegraphics[width=\textwidth]{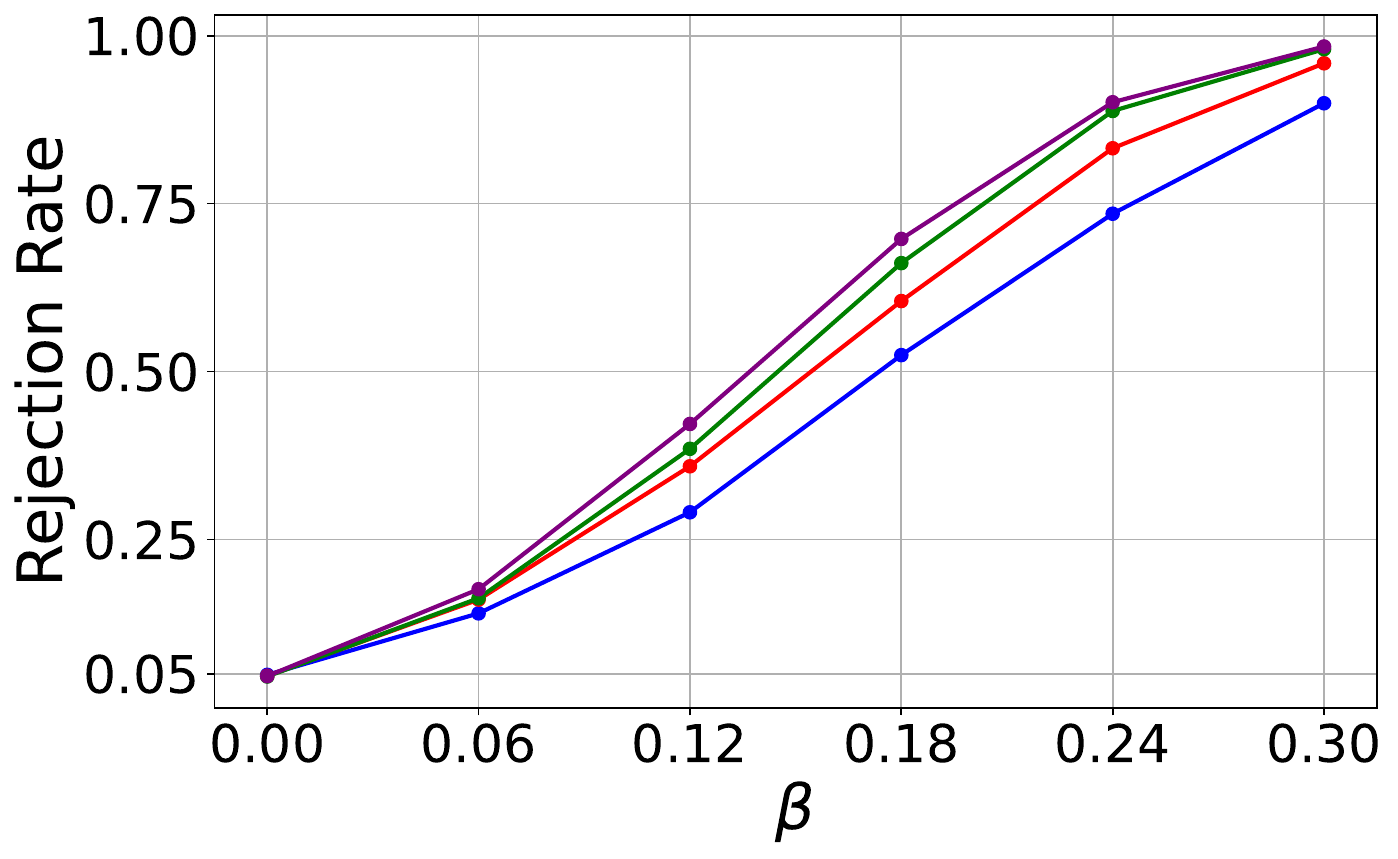}
        \caption{Model 4 ($N=1000$)}
      \end{subfigure}    
      \caption{Type-I error rate (when effect size $\beta=0$) and power (when effect size $\beta>0$) of Methods 1--4 under Models 1--4 with sample size $N=50$ and $N=1000$ (level $\alpha=0.05$). The outcome missingness rate is $25\%$. }
      \label{fig: single outcome simulations with 25 missingness rate}
\end{figure}

\subsection*{Appendix B.2: Simulation Studies on Covariate Adjustment in Randomization Tests with Missing Outcomes}

In Appendix B.2, we present simulation studies on covariate adjustment in imputation-assisted randomization tests with missing outcomes, using the imputation and re-imputation framework with covariate adjustment (i.e., Algorithm~\ref{alg: part with covariate adjustment} in Section~\ref{sec: covariate adjustment} of the main text). Specifically, we consider the same simulation setting and data-generating processes as those in Section~\ref{subsec: simulation studies} (i.e., Models 1--4 in Section~\ref{subsec: simulation studies}), plus an additional Model 6 as follows:
\begin{itemize}
    \item Model 6 (Heterogeneous treatment effects; non-linear model for the true outcome; non-linear selection model for the missingness status; with interference in the missingness mechanism): $Y_{ij} = \beta Z_{ij} x_{ij1}^2 + \frac{1}{\sqrt{5}} \sum_{p=1}^{5} x_{ijp}^2 + u_{ij} + \alpha_{i}+\epsilon_{ij}$ and $M_{ij} = \mathbbm{1} \left\{ \frac{1}{\sqrt{5}} \sum_{p=1}^{5} p x_{ijp} + \frac{1}{\sqrt{5}} \sum_{p=1}^{5} p \cos(x_{ijp}) + 10 \sigma(Y_{ij}) + u_{ij} + \frac{1}{10}\sum_{j=1}^{10}x_{ij1} +  \frac{1}{10}\sum_{j=1}^{10}Y_{ij} > \lambda \right\}$.
 \end{itemize}   
As emphasized in Section~\ref{subsec: simulation studies}, the data-generating processes considered in Section~\ref{subsec: simulation studies} and Appendix B.2 are only intended for automatically and conveniently generating finite-population datasets for our simulation studies, and our framework does not require assuming any outcome-generating distributions or super-population models.

We consider both a small sample size scenario in which we set the total sample size $N=50$ (corresponding to the number of strata $I=5$) and a large sample size scenario in which we set $N=1000$ (corresponding to $I=100$). In each model and each simulation scenario, we implement the following six methods (three without covariate adjustment after imputation and three with covariate adjustment after imputation): 
\begin{itemize}
    \item Method 1 (Algo 1 -- Linear): Imputation-assisted randomization test based on the imputation and re-imputation framework \textit{without covariate adjustment after imputation} (i.e., Algorithm~\ref{alg: part}). We set the embedded outcome imputation algorithm $\mathcal{G}$ to be chained equations imputation based on linear regression (more specifically, Bayesian ridge regression).
    
    \item Method 2 (Algo 2 -- Linear): Imputation-assisted randomization test based on the imputation and re-imputation framework \textit{with covariate adjustment after imputation} (i.e., Algorithm~\ref{alg: part with covariate adjustment}). As in Method 1, we still set the embedded outcome imputation algorithm $\mathcal{G}$ to be chained equations imputation based on linear regression (more specifically, Bayesian ridge regression). Additionally, we also set the embedded working model $\mathcal{H}$ for covariate adjustment after imputation as linear regression (more specifically, Bayesian ridge regression).
    
    \item Method 3 (Algo 1 -- Boosting): Imputation-assisted randomization test based on the imputation and re-imputation framework \textit{without covariate adjustment after imputation} (i.e., Algorithm~\ref{alg: part}), in which we set the embedded outcome imputation algorithm $\mathcal{G}$ to be chained equations imputation based on a boosting algorithm. Following the same strategy used in Section~\ref{subsec: simulation studies}, when $N=50$, the boosting algorithm we choose is the popular XGBoost algorithm (\citealp{chen2016xgboost}), and when $N=1000$, the boosting algorithm we choose is the popular LightGBM algorithm (\citealp{ke2017lightgbm}) (see Section~\ref{subsec: simulation studies} in the main text for the detailed considerations).

    \item Method 4 (Algo 2 -- Boosting): Imputation-assisted randomization test based on the imputation and re-imputation framework \textit{with covariate adjustment after imputation} (i.e., Algorithm~\ref{alg: part with covariate adjustment}), in which we still set the embedded outcome imputation algorithm $\mathcal{G}$ to be chained equations imputation based on a boosting algorithm. Additionally, we also set the embedded working model $\mathcal{H}$ for covariate adjustment after imputation as a boosting algorithm. Following the same strategy used in Section~\ref{subsec: simulation studies}, the boosting algorithm we choose (for both the imputation algorithm $\mathcal{G}$ and the covariate adjustment model $\mathcal{H}$) is the popular XGBoost algorithm (\citealp{chen2016xgboost}) when $N=50$ and the popular LightGBM algorithm (\citealp{ke2017lightgbm}) when $N=1000$ (see Section~\ref{subsec: simulation studies} for the detailed considerations).

   \item Method 5 (Median Imputation): Classic randomization test based on median imputation for missing outcomes (without covariate adjustment).

   \item Method 6 (Median Imputation with Covariate Adjustment): Classic randomization test based on median imputation for missing outcomes, with covariate adjustment based on linear covariate adjustment model (same as that used in Method 2). 
    
\end{itemize}
The simulated type-I error rate (corresponding to the rejection rate when effect size $\beta=0$) and power (corresponding to the rejection rate when effect size $\beta>0$) are reported in Figure~\ref{fig: single outcome simulations with covariate adjustment}. Each number in Figure~\ref{fig: single outcome simulations with covariate adjustment} is based on 2000 simulated datasets and 10,000 re-imputation runs in Algorithm~\ref{alg: part with covariate adjustment}. 
 
%% Here are the title, author names and addresses
\begin{figure}[htp!]
      \centering
        \captionsetup[subfigure]{skip=2pt} % Adjust this value as needed
        \includegraphics[width=\textwidth]{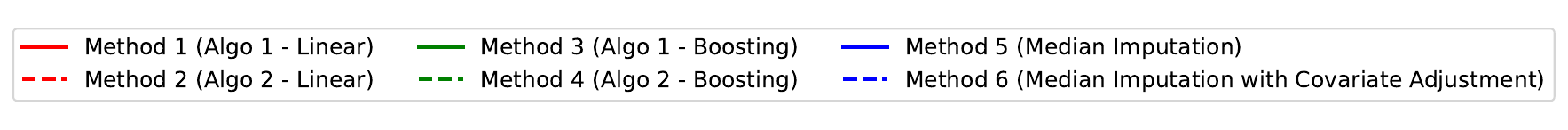}

      % Model 1 (Constant treatment effect; linear model for the true outcome; linear selection model for the missingness status; without interference in the missingness mechanism) 
      \begin{subfigure}[b]{0.4\textwidth}     
        \includegraphics[width=\textwidth]{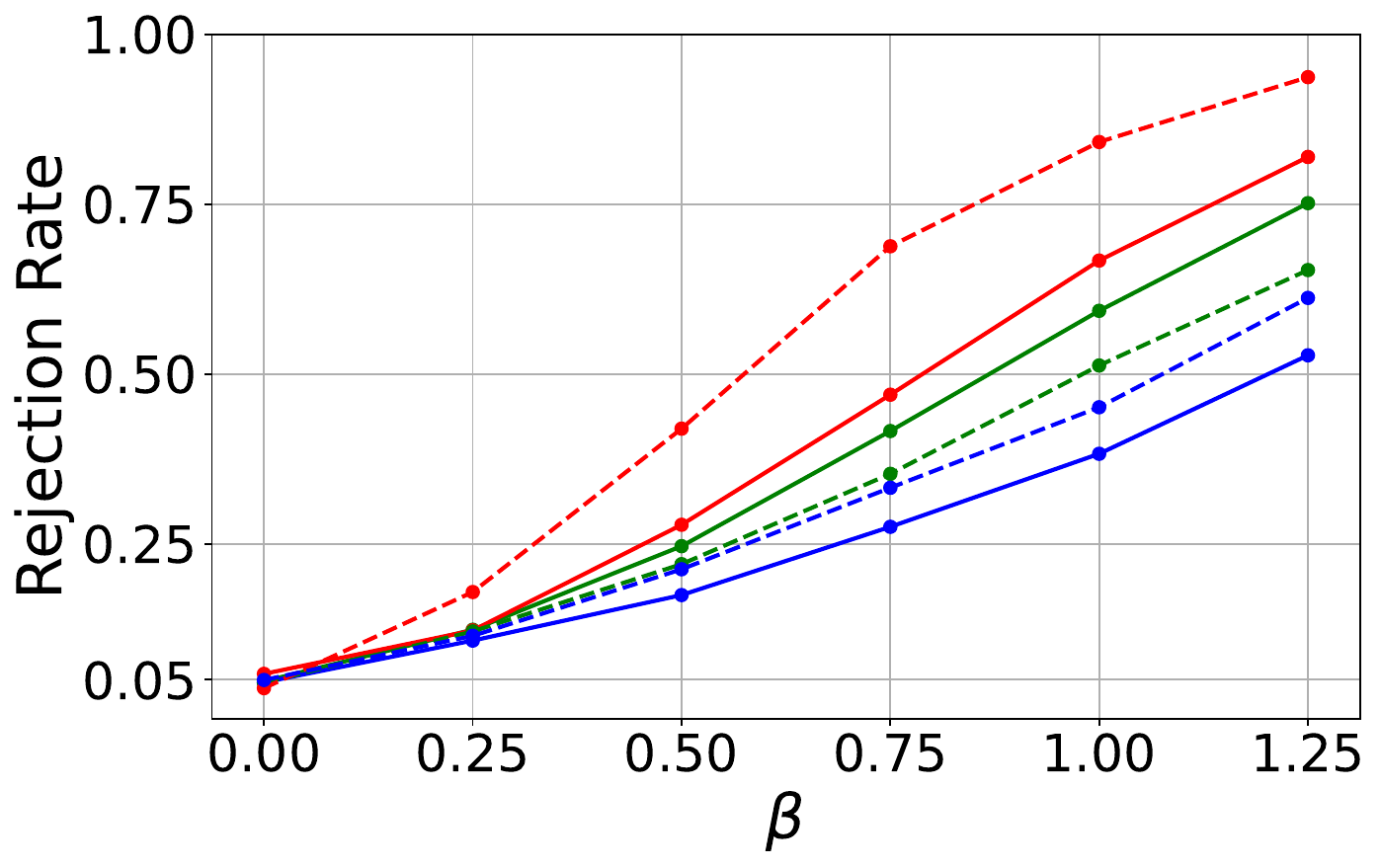}
        \caption{Model 1 ($N=50$)}
      \end{subfigure}
      \hspace{0.1cm}
      \begin{subfigure}[b]{0.4\textwidth}
        \includegraphics[width=\textwidth]{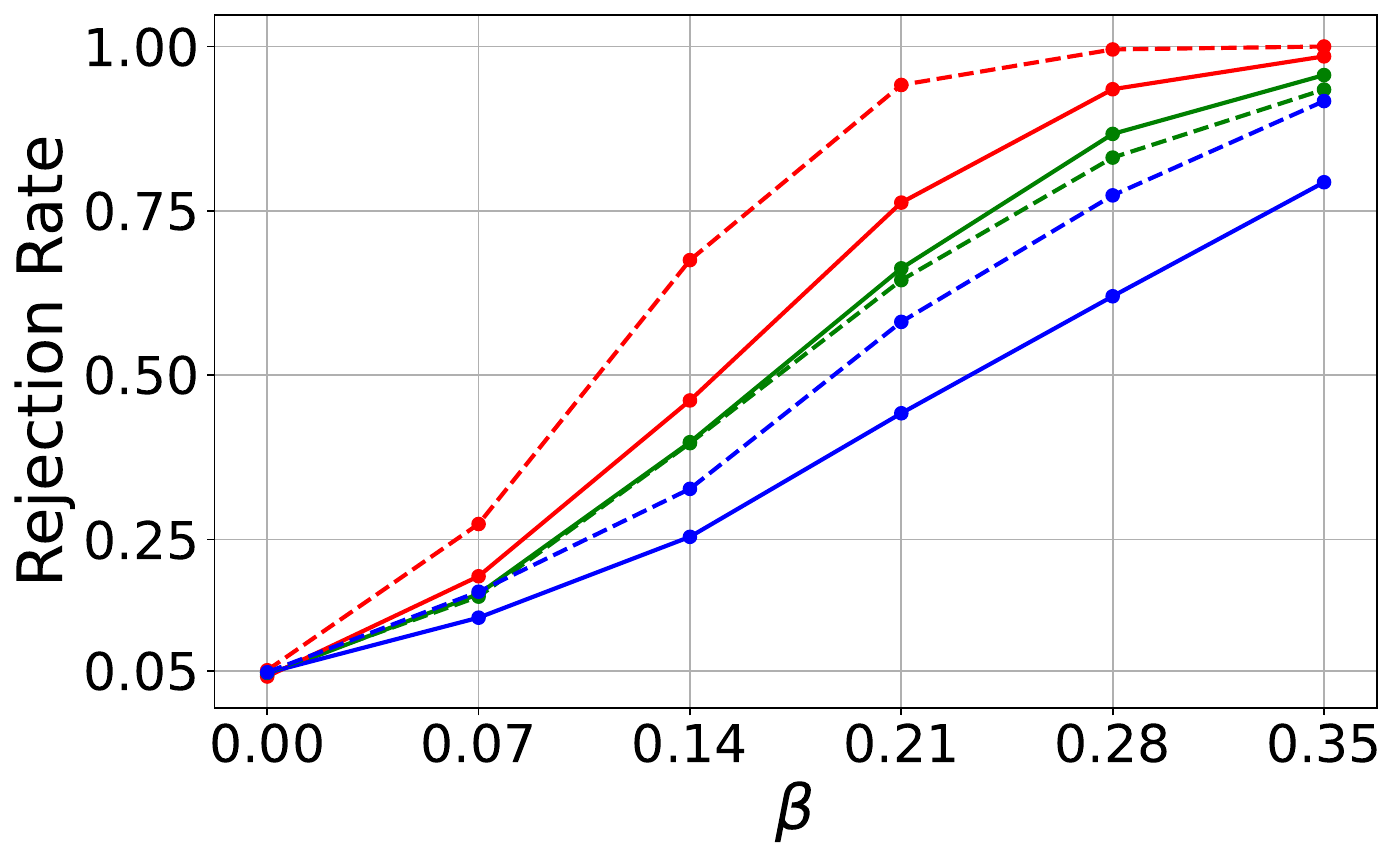}
        \caption{Model 1 ($N=1000$)}
      \end{subfigure}

      % Model 2(Covariate Adjustment): Heterogeneous treatment effects; linear model for the true outcome;non-linear selection model for the missingness status; without interference in the missingness mechanism
      \begin{subfigure}[b]{0.4\textwidth}     
        \includegraphics[width=\textwidth]{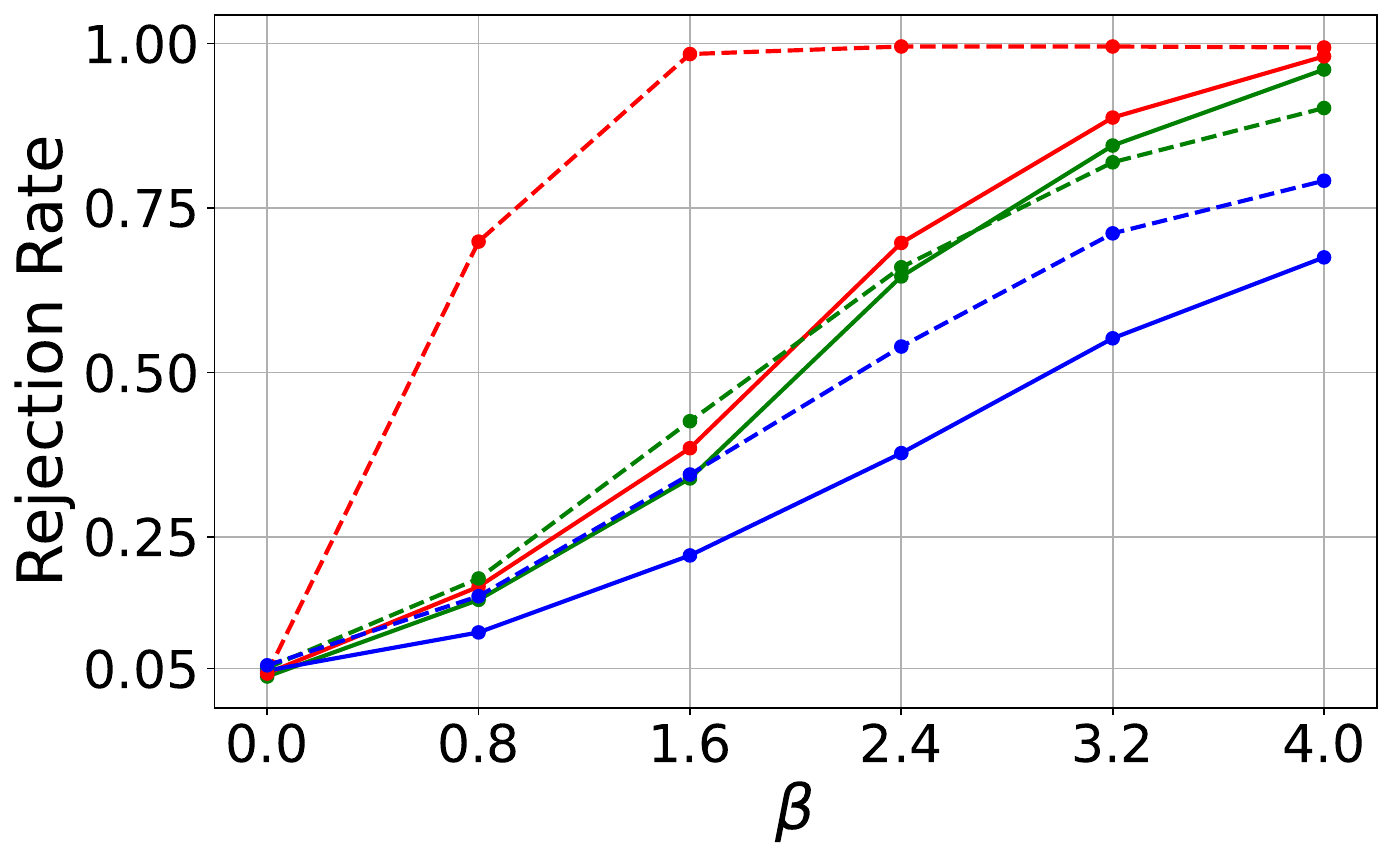}
        \caption{Model 2 ($N=50$)}
      \end{subfigure}
      \hspace{0.1cm}
      \begin{subfigure}[b]{0.4\textwidth}
        \includegraphics[width=\textwidth]{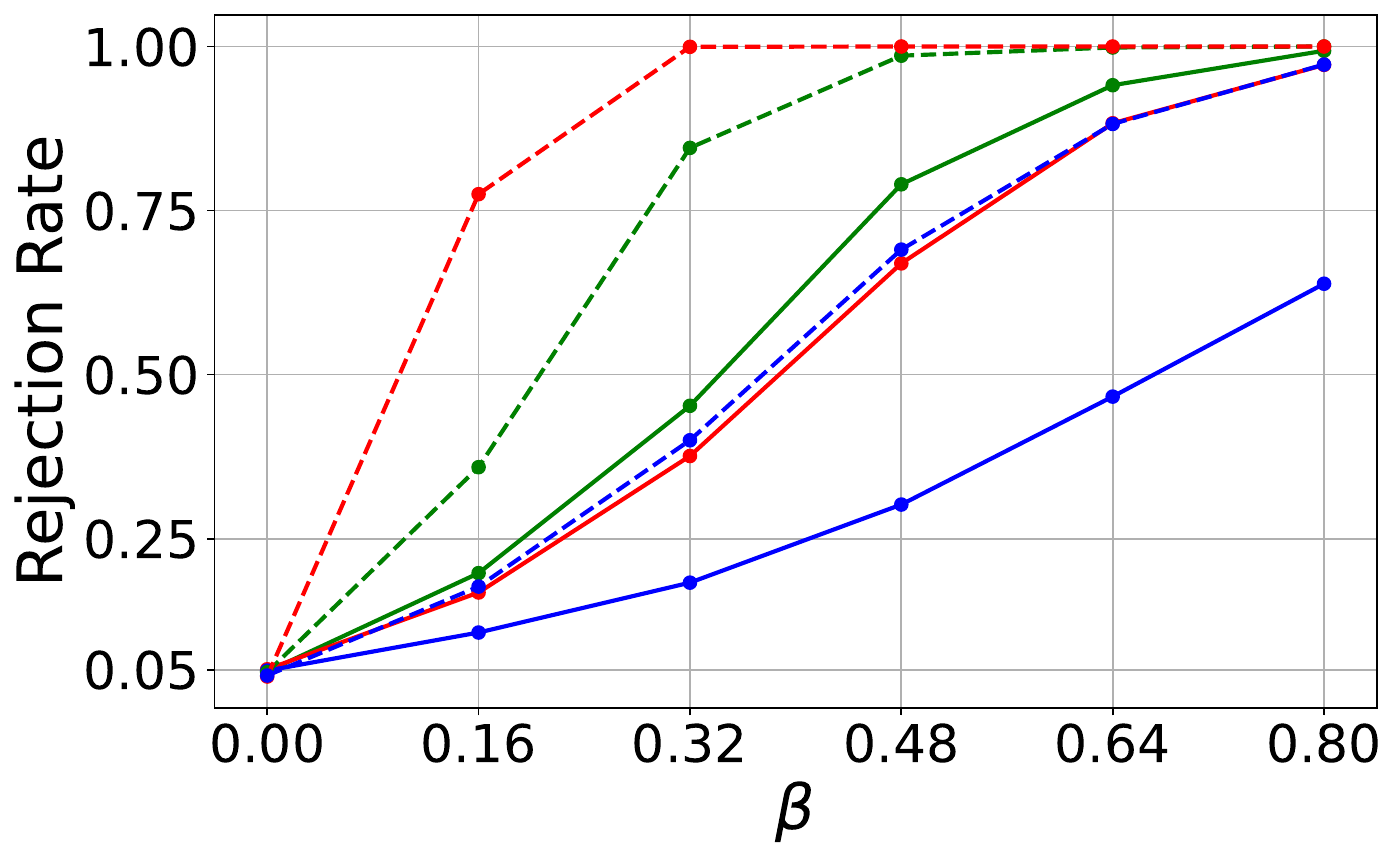}
        \caption{Model 2 ($N=1000$)}
      \end{subfigure}
      
      % Model 3(Covariate Adjustment): Heterogeneous treatment effects; non-linear model for the true outcome;non-linear selection model for the missingness status; without interference in the missingness mechanism
      \begin{subfigure}[b]{0.4\textwidth}     
        \includegraphics[width=\textwidth]{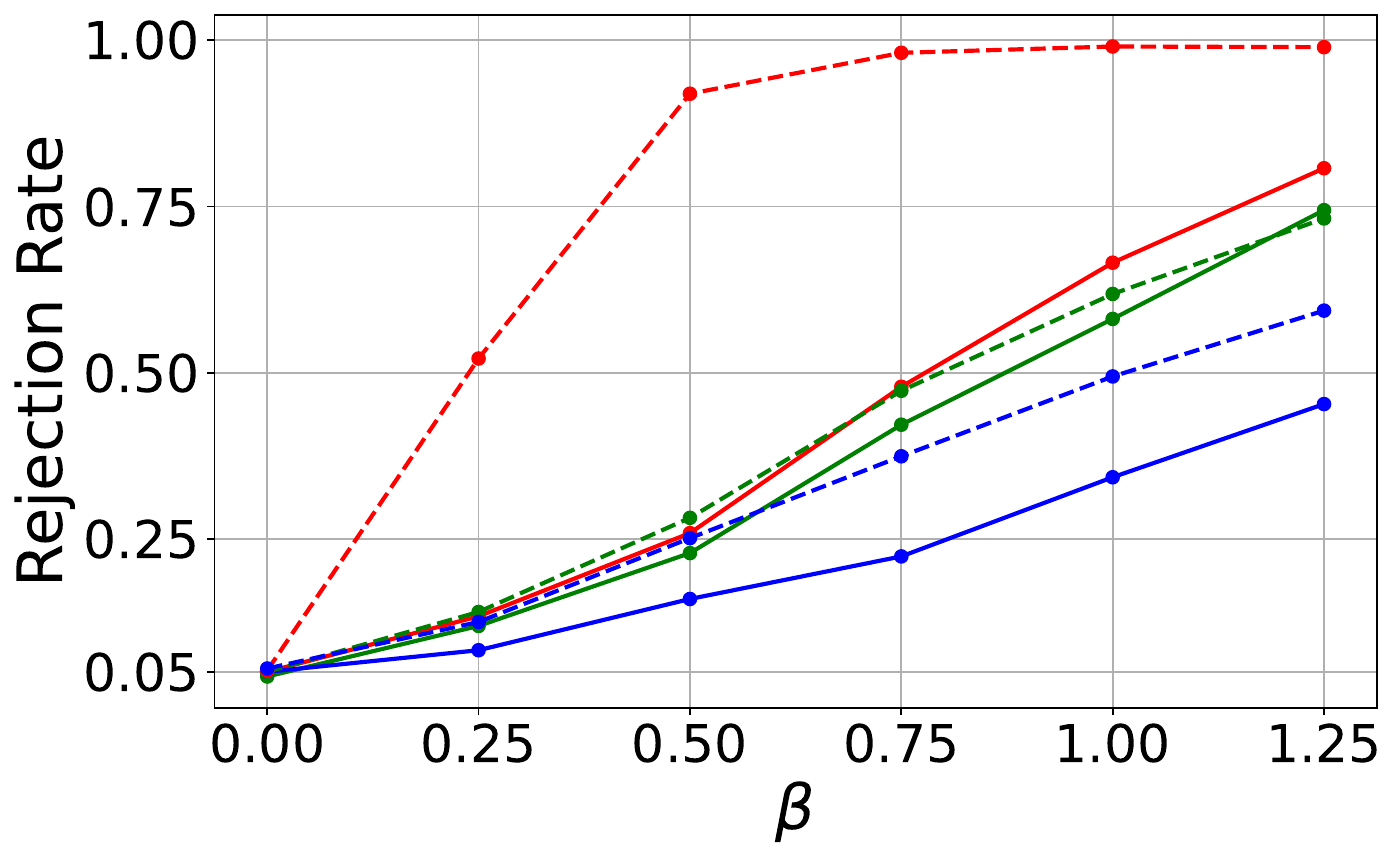}
        \caption{Model 3 ($N=50$)}
      \end{subfigure}
      \hspace{0.1cm}
      \begin{subfigure}[b]{0.4\textwidth}
        \includegraphics[width=\textwidth]{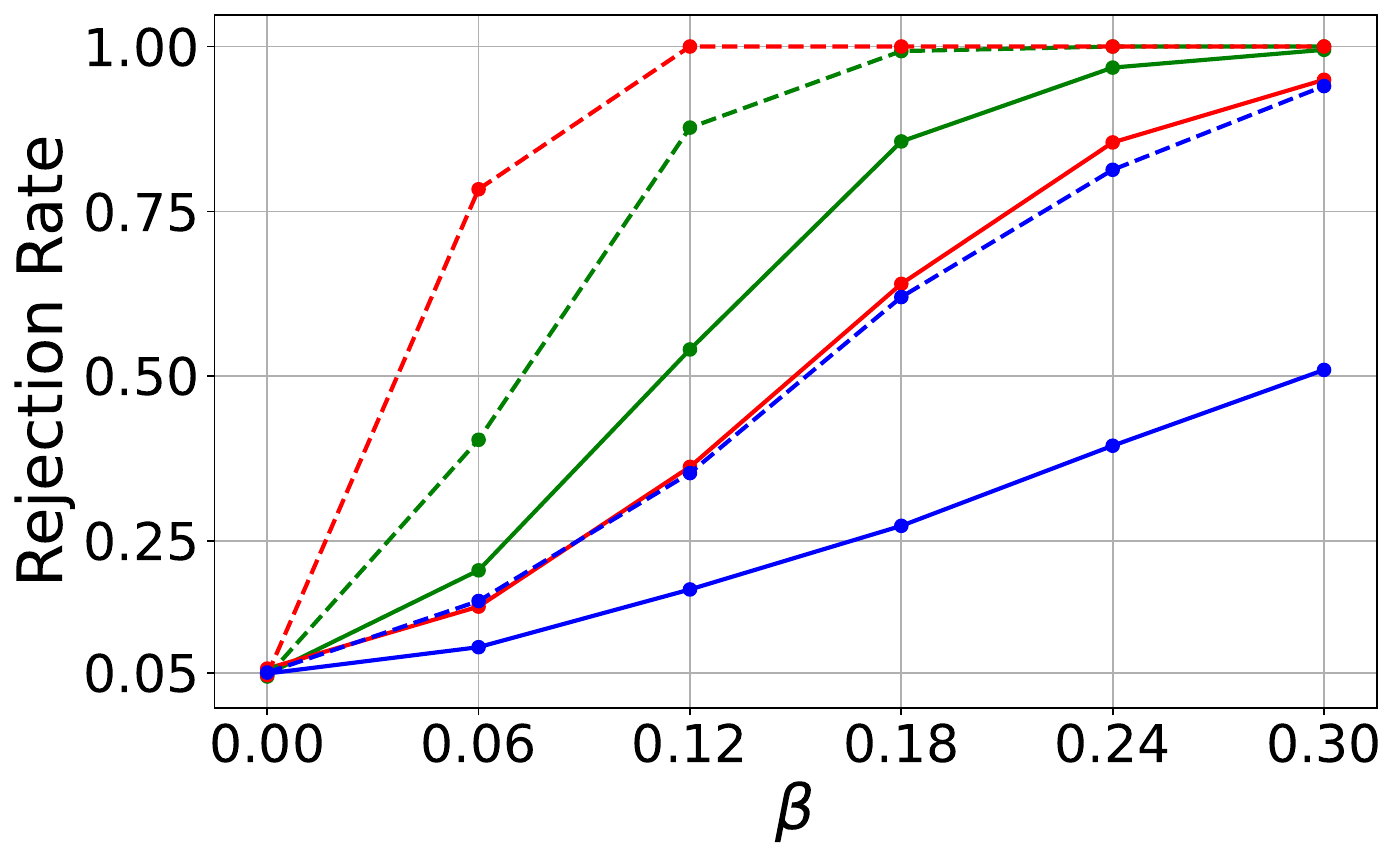}
        \caption{Model 3 ($N=1000$)}
      \end{subfigure}

      % Model 4(Covariate Adjustment): Heterogeneous treatment effects; non-linear model for the true outcome;non-linear selection model for the missingness status; with interference in the missingness mechanism
      \begin{subfigure}[b]{0.4\textwidth}     
        \includegraphics[width=\textwidth]{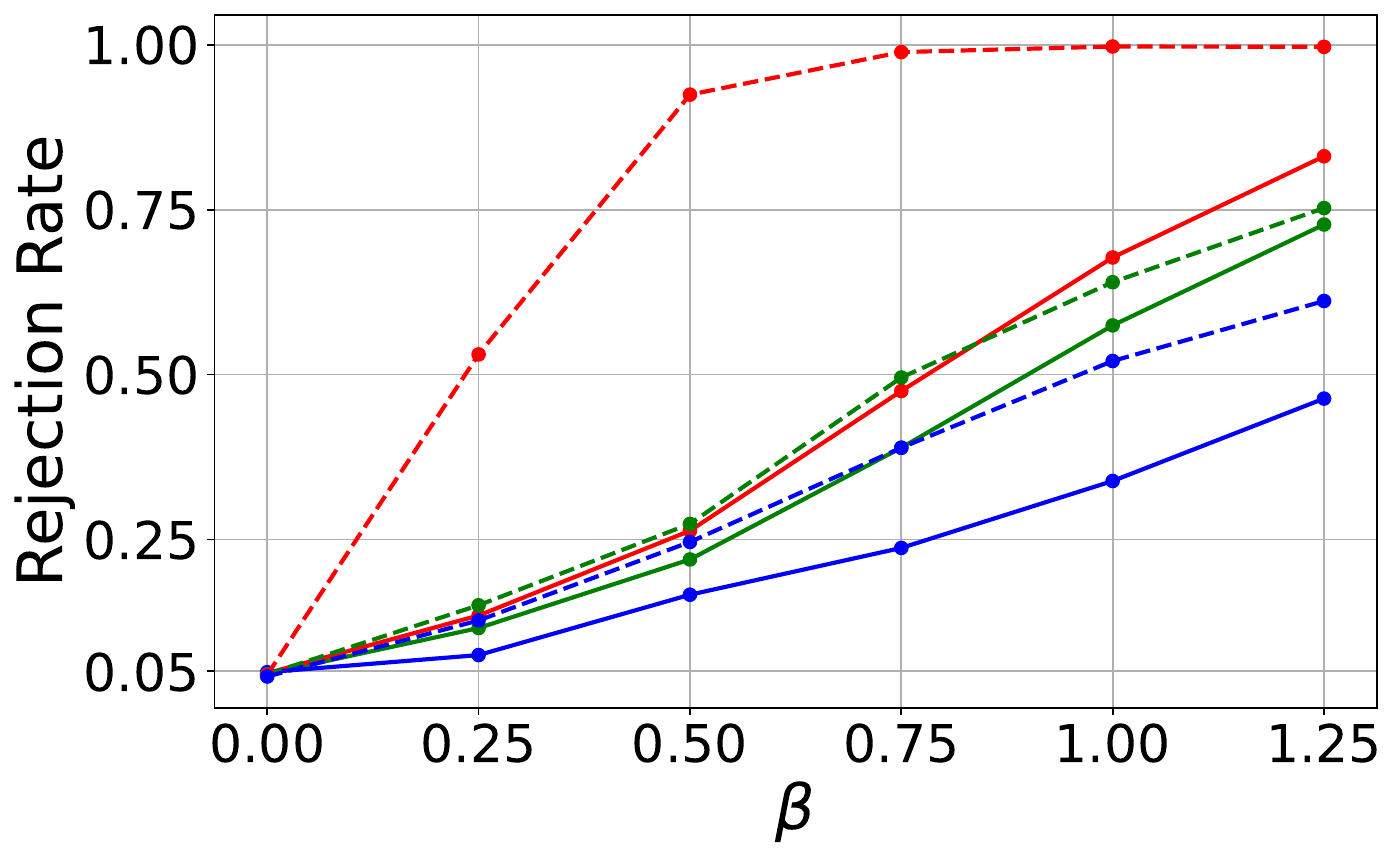}
        \caption{Model 4 ($N=50$)}
      \end{subfigure}
      \hspace{0.1cm}
      \begin{subfigure}[b]{0.4\textwidth}
        \includegraphics[width=\textwidth]{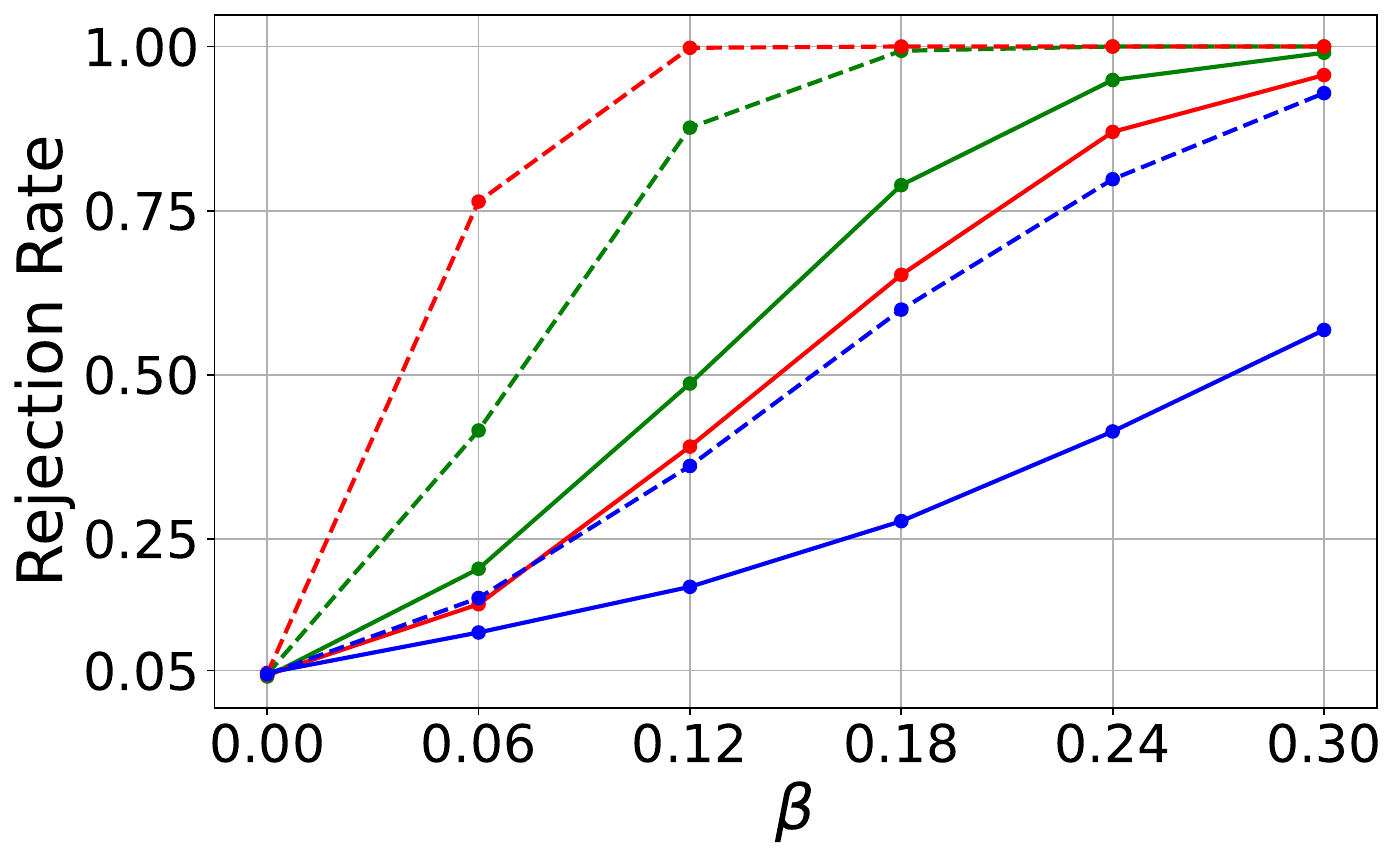}
        \caption{Model 4 ($N=1000$)}
      \end{subfigure}

    % Model 6: Heterogeneous treatment effect; non-linear model for the true outcome; non-linear selection model for the missingness status, with interference

      \begin{subfigure}[b]{0.4\textwidth}
        
        \includegraphics[width=\textwidth]{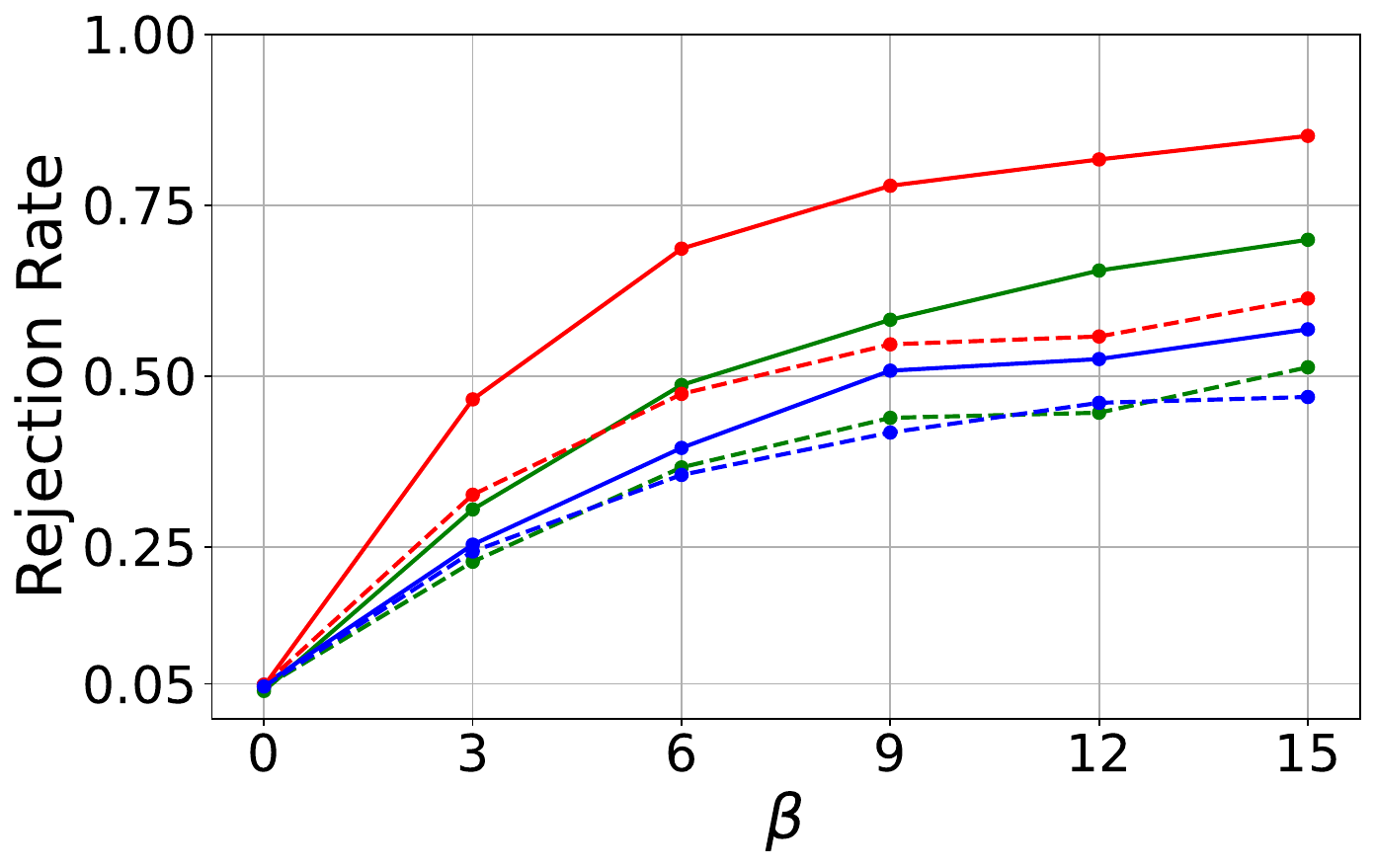}
        \caption{Model 6 ($N=50$)}
      \end{subfigure}
      \hspace{0.1cm}
      \begin{subfigure}[b]{0.4\textwidth}
        
        \includegraphics[width=\textwidth]{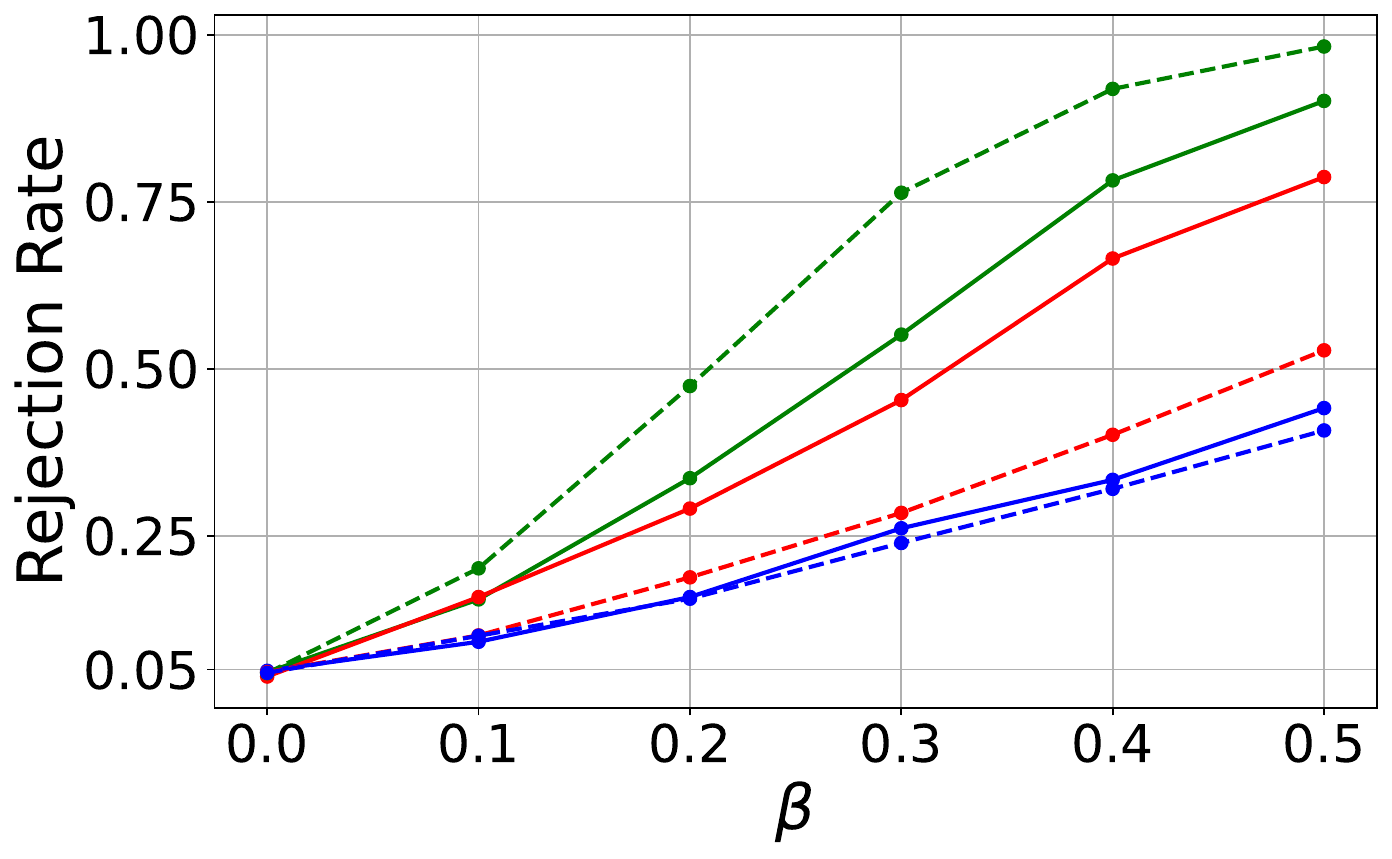}
        \caption{Model 6 ($N=1000$)}
      \end{subfigure}  
      \caption{Type-I error rate (when effect size $\beta=0$) and power (when effect size $\beta>0$) of Methods 1--4 under Models 1--4 and 6 with sample size $N=50$ and $N=1000$ (level $\alpha=0.05$). The outcome missingness rate is $50\%$. }
      \label{fig: single outcome simulations with covariate adjustment}
\end{figure}

From the simulation results in Figure~\ref{fig: single outcome simulations with covariate adjustment}, we can obtain three key insights. First, the imputation and re-imputation framework with covariate adjustment (i.e., Algorithm~\ref{alg: part with covariate adjustment}) can guarantee finite-population-exact type-I error rate control (corresponding to the $\beta=0$ case) with either linear imputation and covariate adjustment models (Method 2) or flexible machine learning imputation and covariate adjustment models (Method 4) under all the considered models (ranging from linear to non-linear outcome models and missingness models, with or without interference in the outcome missingness mechanism, in the presence of unobserved covariates). This confirms the theoretical guarantee of the finite-population-exact type-I error rate control of the imputation and re-imputation framework with covariate adjustment, even when the outcome imputation algorithm/model or the covariate adjustment model was misspecified or when unobserved covariates or interference exists in the missingness mechanism (as stated in Theorem~\ref{alg: part with covariate adjustment}). 

Second, for both the small sample size case ($N=50$) and large sample size case ($N=1000$), in many settings (e.g., Models 2--4), incorporating a covariate adjustment step can further improve the power of an imputation-assisted randomization test constructed by the imputation and re-imputation framework (i.e., comparing Algorithm~\ref{alg: part with covariate adjustment} with Algorithm~\ref{alg: part}). Also, such improvements in power can be substantial, and whether incorporating a linear covariate adjustment model or a flexible machine learning covariate adjustment model results in even more substantial gains in power depends on the specific data-generating process and sample size. However, there is no guarantee that incorporating a covariate adjustment step will always improve the power of the imputation and re-imputation framework; see the simulated results under Model 6 as a counter-example. This pattern also agrees with that of the classic Rosenbaum-type covariate adjustment method for Fisher's randomization tests with complete outcome data -- although implementing a covariate adjustment step in Fisher's randomization tests is promising to improve power in many settings, there is no guarantee that covariate adjustment will always improve power for Fisher's randomization tests (\citealp{rosenbaum2002covariance, small2008randomization}).  

Third, the imputation and re-imputation framework with covariate adjustment typically outperforms non-informative imputation (e.g., median imputation) with covariate adjustment in terms of finite-sample power, especially when the sample size is large (e.g., $N=1000$). 

\subsection*{Appendix B.3: Simulation Studies with Missingness in Covariates}

We replicate the setup (i.e., Models 1-4) and procedures outlined in Section~\ref{subsec: simulation studies} to conduct similar simulation studies while allowing for missingness in covariates. Specifically, we introduce missingness in the covariate $x_{ij5}$ for some subjects, where the corresponding missingness indicator $M^{x}_{ij}$ is generated by the following model:
\begin{equation*}
    M^{x}_{ij} = \mathbbm{1} \left\{ \frac{1}{\sqrt{5}} \sum_{p=1}^{5} x_{ijp} + \frac{1}{\sqrt{5}} \sum_{p=1}^{5} x_{ijp}^2  + \sigma(u_{ij}) > \lambda_{x} \right\}.
\end{equation*}
For each of the 1000 simulated datasets under each setting, we tune the parameter $\lambda_{x}$ to ensure a 25\% missingness rate for $x_{ij5}$. As discussed in Section~\ref{subsec: simulation studies}, the data-generating processes presented in both Section~\ref{subsec: simulation studies} and Appendix B.3 are intended solely to automatically and efficiently generate finite-sample datasets for simulation studies. Our framework does not require assumptions about outcome-generating distributions or super-population models for missingness in outcomes or covariates.

\begin{figure}[H]
      \centering
        \captionsetup[subfigure]{skip=2pt} % Adjust this value as needed

        \begin{minipage}{\textwidth}
          \centering
          \includegraphics[width=\textwidth, trim=0 10 0 0, clip]{imgs/legendcustomlinestypes.pdf}
        \end{minipage}

        \vspace*{-8pt}
        
        \begin{minipage}{\textwidth}
          \centering
          \makebox[\textwidth]{
            \includegraphics[width=1.2\textwidth]{imgs/legend_typeone_error.pdf}
          }
        \end{minipage}
        \vspace*{-26pt}
      
      \begin{subfigure}[b]{0.4\textwidth}
        \includegraphics[width=\textwidth]{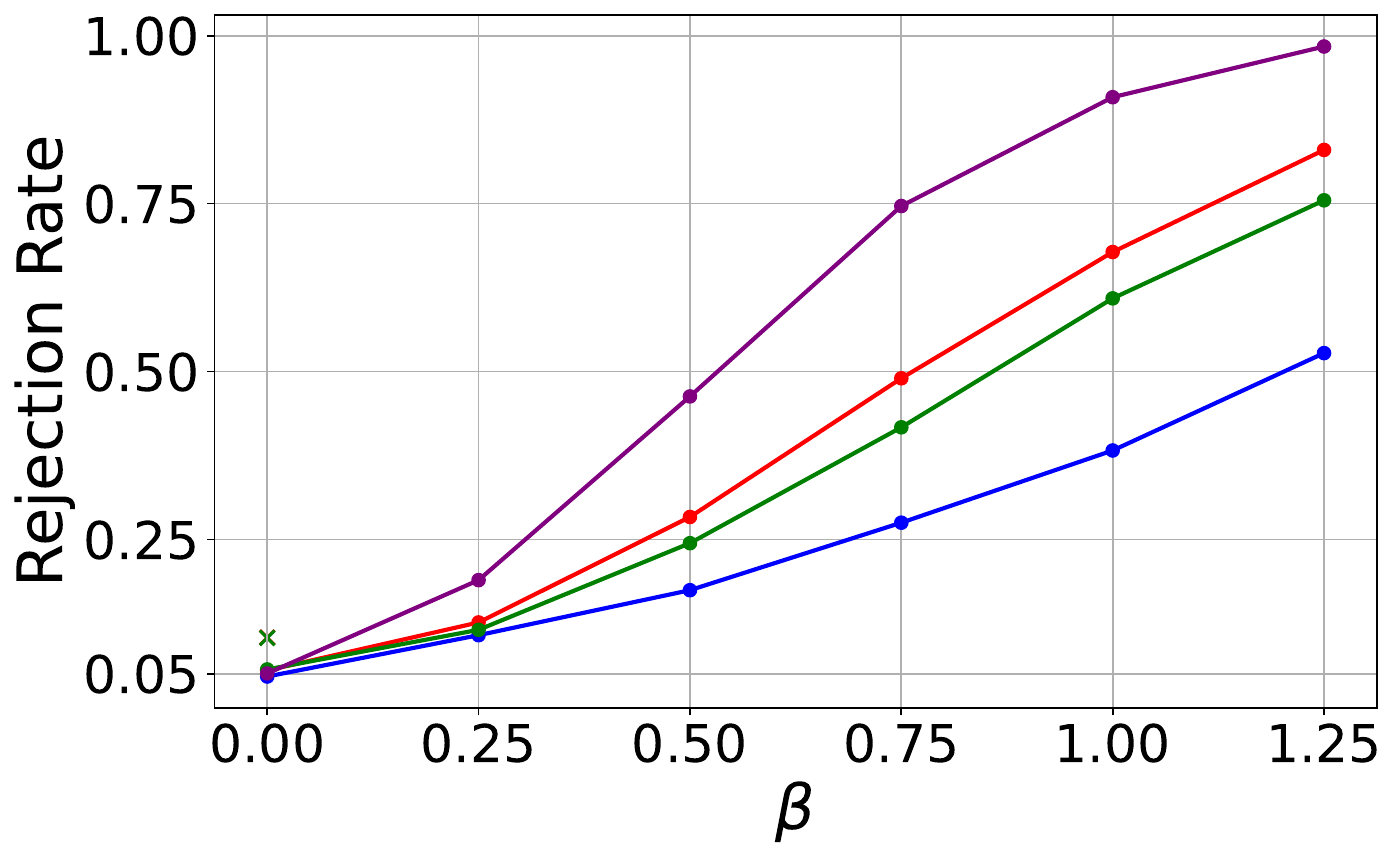}
        \caption{Model 1 ($N=50$)}
      \end{subfigure}
      \hspace{0.1cm}
      \begin{subfigure}[b]{0.4\textwidth}
        \includegraphics[width=\textwidth]{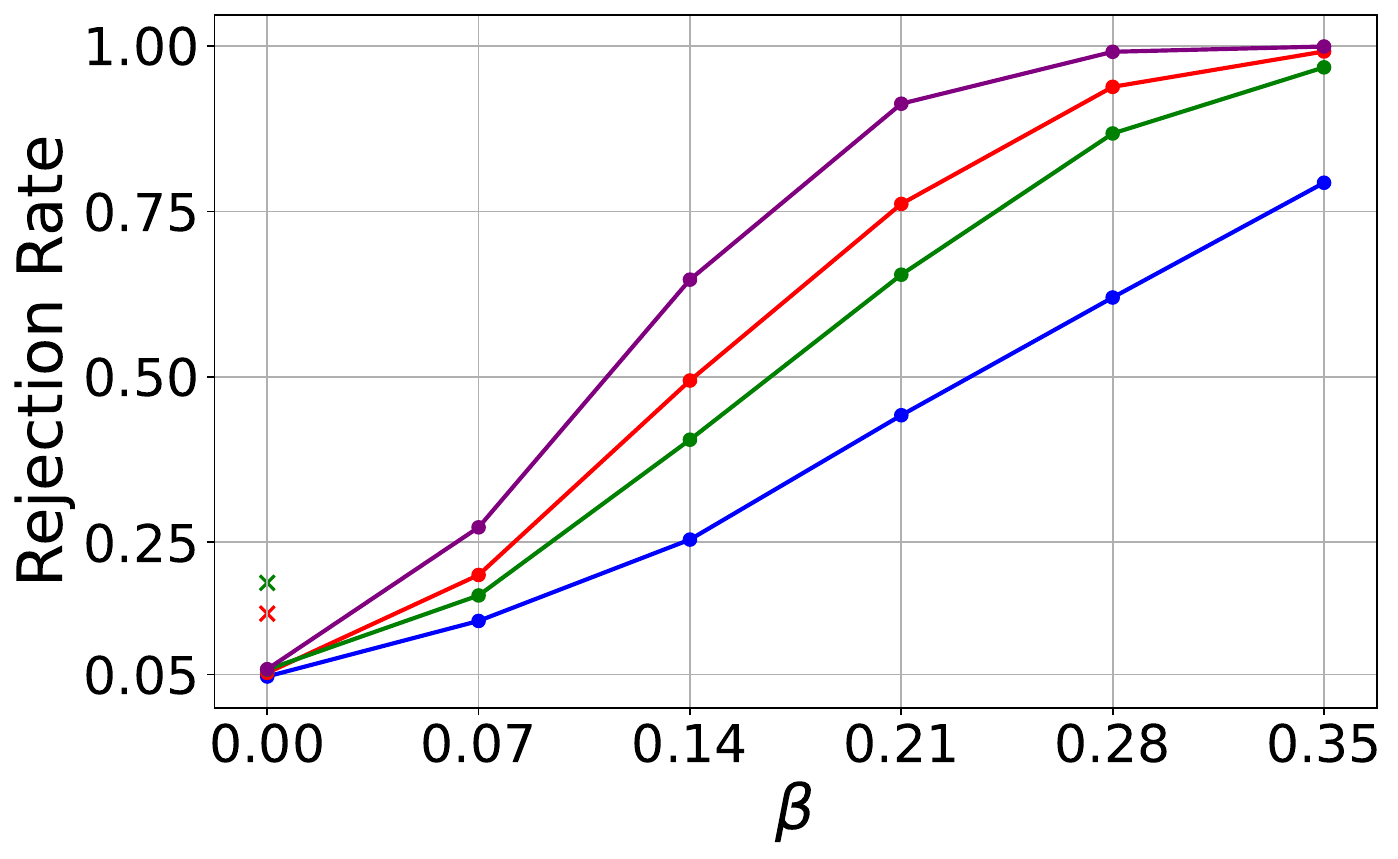}
        \caption{Model 1 ($N=1000$)}
      \end{subfigure}

    % Constant treatment effect; non-linear model for the true outcome; non-linear selection model for the missingness status, without interference

      \begin{subfigure}[b]{0.4\textwidth}
        
        \includegraphics[width=\textwidth]{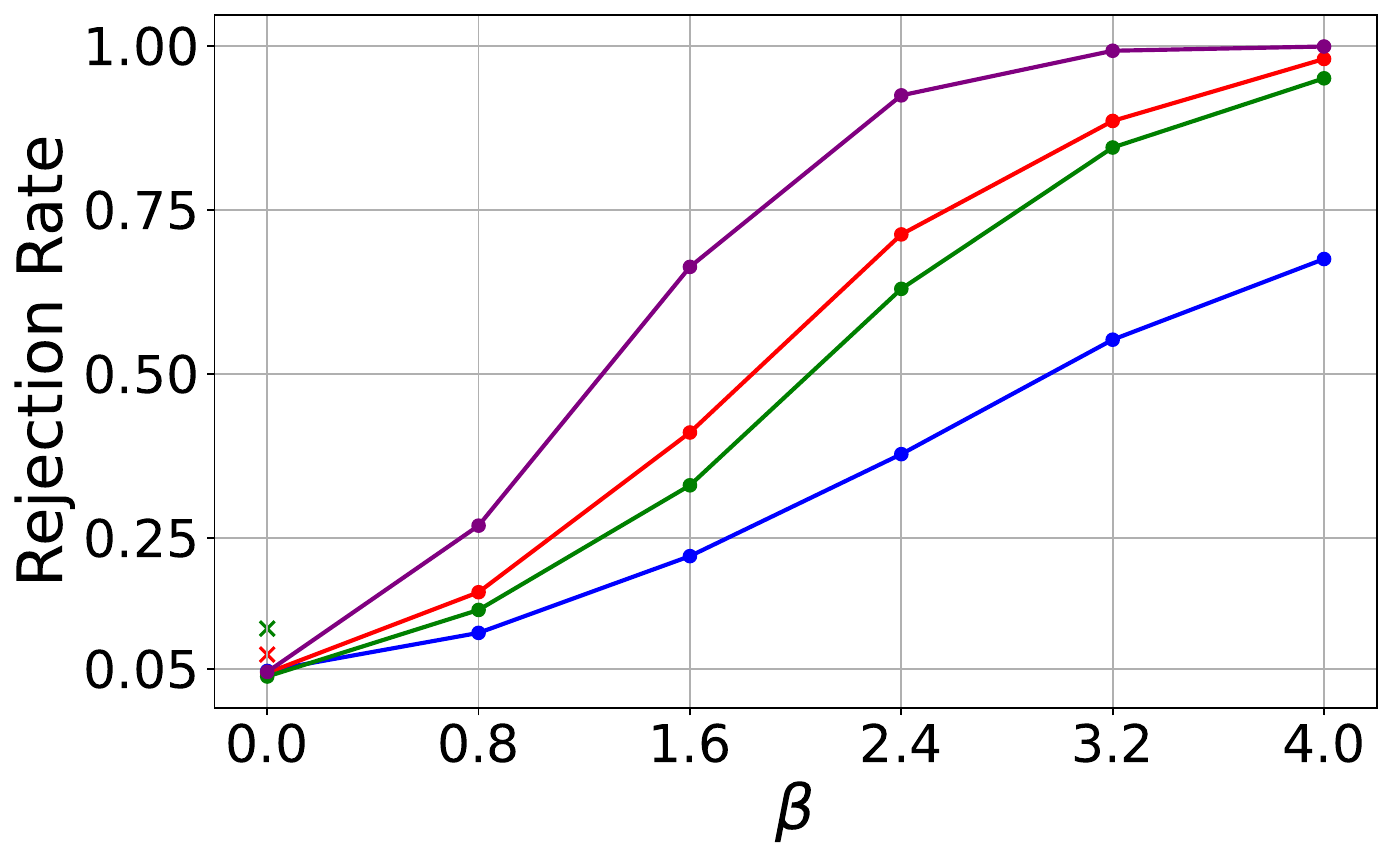}
        \caption{Model 2 ($N=50$)}
      \end{subfigure}
      \hspace{0.1cm}
      \begin{subfigure}[b]{0.4\textwidth}
        
        \includegraphics[width=\textwidth]{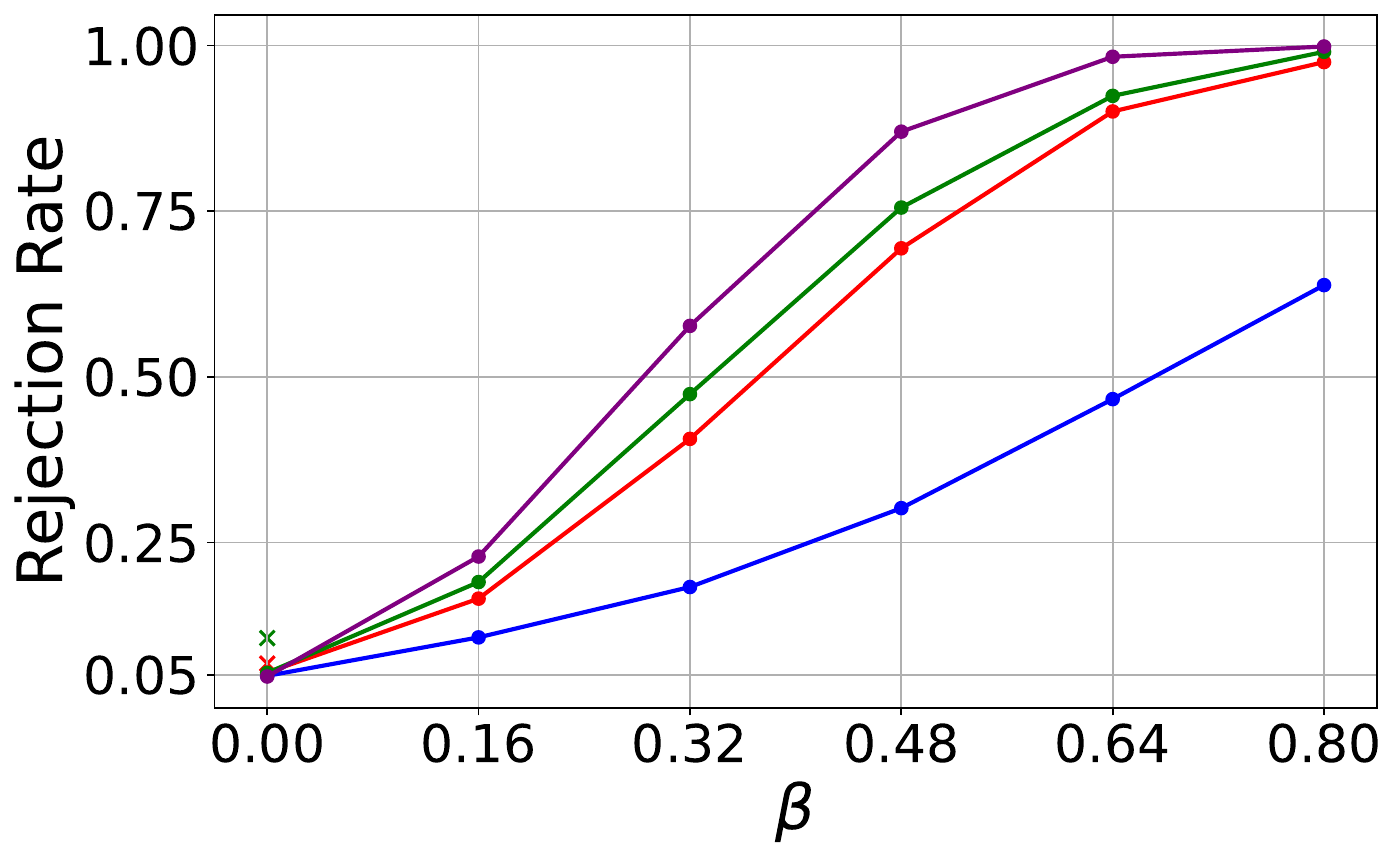}
        \caption{Model 2 ($N=1000$)}
      \end{subfigure}
    
    % Heterogeneous treatment effect; non-linear model for the true outcome; non-linear selection model for the missingness status, without interference
    
      \begin{subfigure}[b]{0.4\textwidth}
        
        \includegraphics[width=\textwidth]{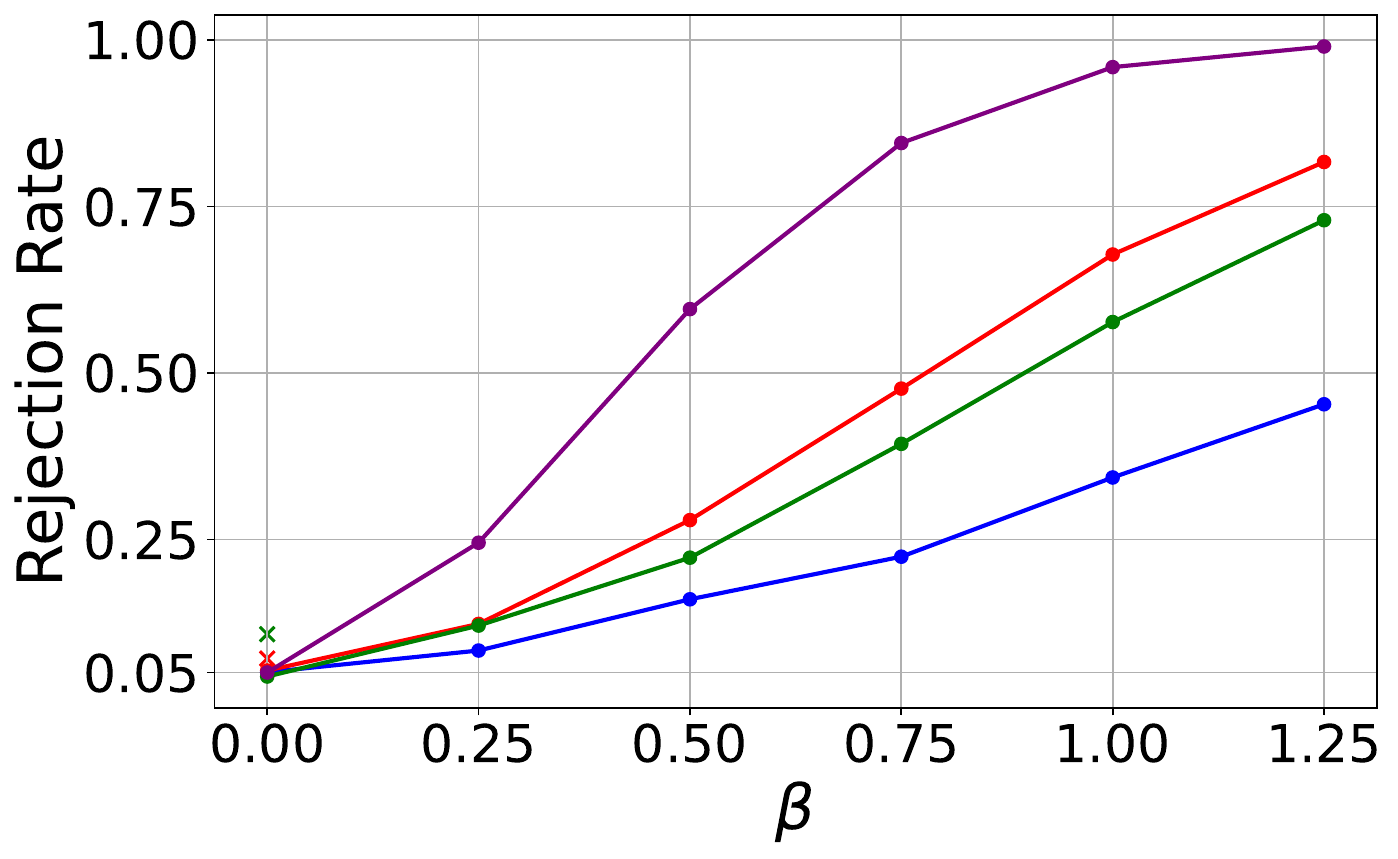}
        \caption{Model 3 ($N=50$)}
      \end{subfigure}
      \hspace{0.1cm}
      \begin{subfigure}[b]{0.4\textwidth}
        
        \includegraphics[width=\textwidth]{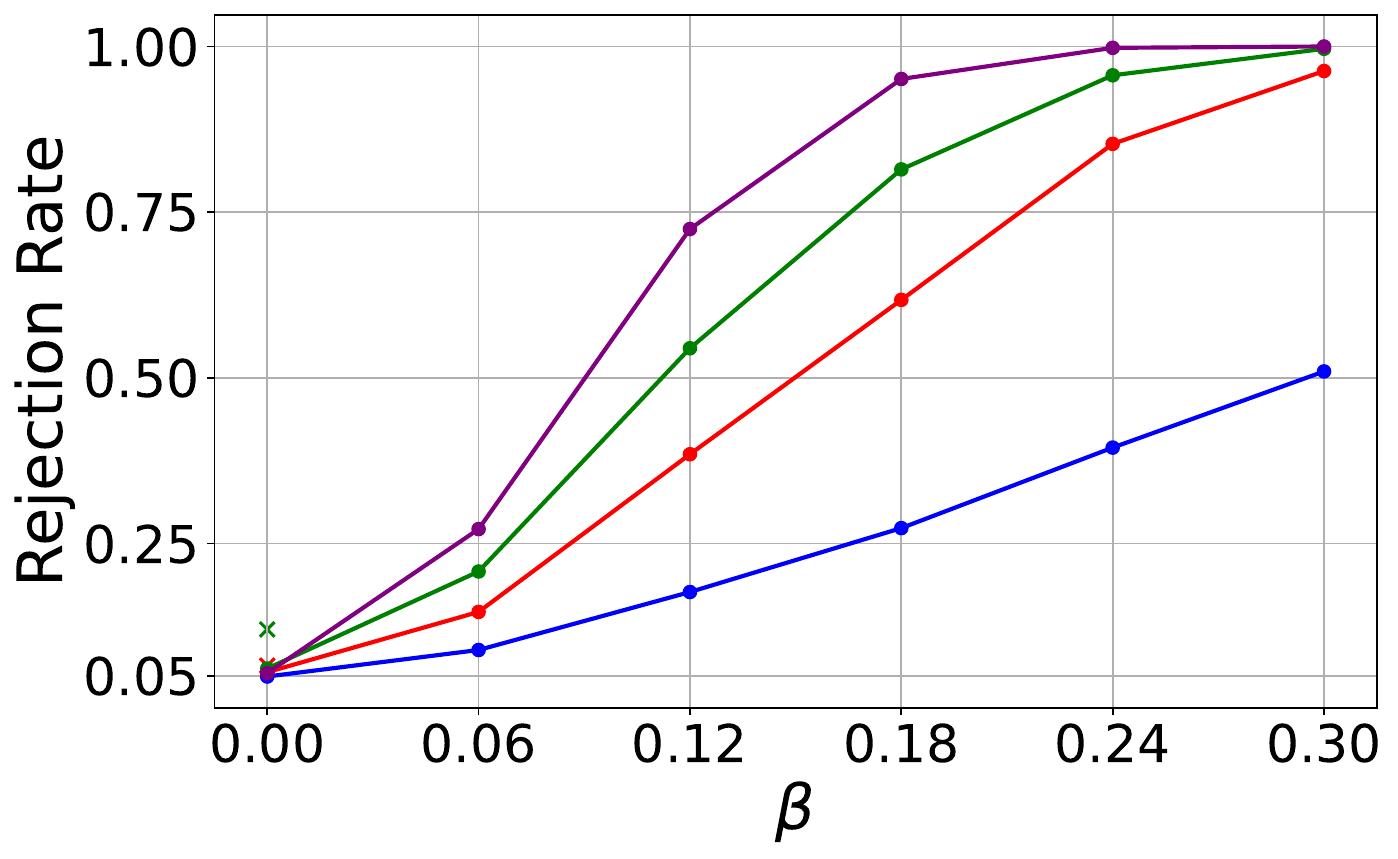}
        \caption{Model 3 ($N=1000$)}
      \end{subfigure}

    % Heterogeneous treatment effect; non-linear model for the true outcome; non-linear selection model for the missingness status, with interference
      
      \begin{subfigure}[b]{0.4\textwidth}
        
        \includegraphics[width=\textwidth]{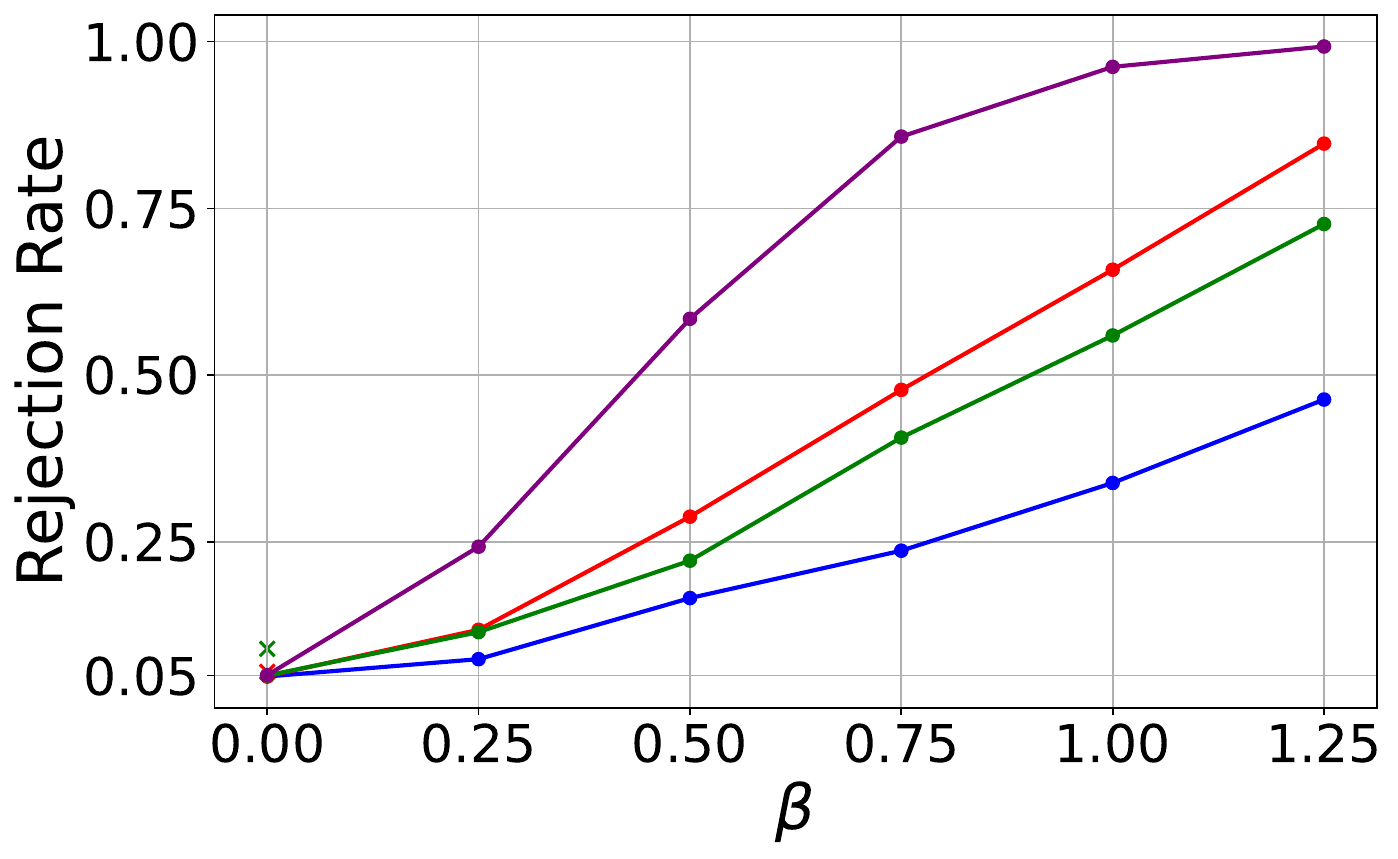}
        \caption{Model 4 ($N=50$)}
      \end{subfigure}
      \hspace{0.1cm}
      \begin{subfigure}[b]{0.4\textwidth}
        
        \includegraphics[width=\textwidth]{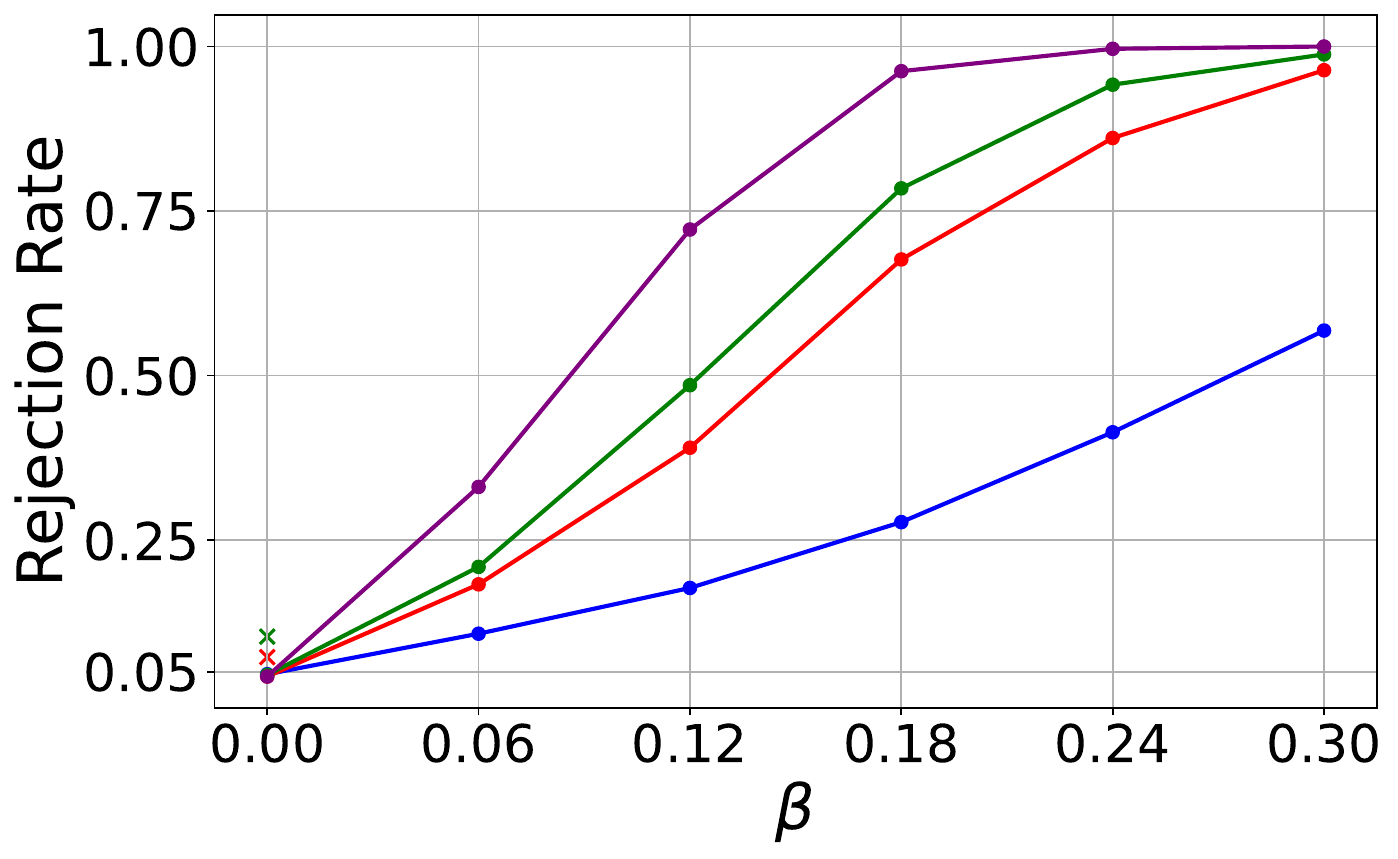}
        \caption{Model 4 ($N=1000$)}
      \end{subfigure}    
      \caption{Type-I error rate (when effect size $\beta=0$) and power (when effect size $\beta>0$) of Methods 1--4 under Models 1--4 with missing covariates. The sample size is set to $N=50$ and $N=1000$ (level $\alpha=0.05$). The covariate missingness rate is $25\%$ (for covariate $x_{ij5}$).} 
      
      \label{fig: covariate missing simulations}
\end{figure}

The simulation results under both covariate and outcome missingness, presented in Figure~\ref{fig: covariate missing simulations}, follow a similar pattern to those in Figure~\ref{fig: single outcome simulations} of the main text. Specifically, under the considered simulation settings, our proposed framework maintains the finite-population-exact type-I error rate (unlike model-based methods, which may violate type-I error rates) while achieving substantially higher power than randomization tests based on non-informative imputation, such as median imputation.

\subsection*{Appendix B.4: Simulation Studies with $T_{M}(\mathbf{Z}, \mathbf{M})$}

Note that the test statistics $T_{M}(\mathbf{Z}, \mathbf{M})$ (the number of missing outcomes among the treated subjects) defined in Remark~\ref{rem: missingness as test} is \textit{not} a valid falsification test for Assumption~\ref{assump: conditional indep}. Specifically, recall that Assumption~\ref{assump: conditional indep} states that outcome missingness $\mathbf{M}$ only depends on observed covariates $\mathbf{X}$, unobserved covariates $\mathbf{U}$, and the true outcomes $\mathbf{Y}$. Under this assumption, the treatment assignments $\mathbf{Z}$ can affect missingness only through their effects on outcomes $\mathbf{Y}$. In other words, Assumption~\ref{assump: conditional indep} only implies that $\mathbf{Z}$ has no effect on $\mathbf{M}$ \textit{if Fisher's sharp null $H_{0}$ holds (i.e., if $\mathbf{Z}$ has no effect on $\mathbf{Y}$)}, but does \textit{not} imply that $\mathbf{Z}$ has no effect on $\mathbf{M}$ regardless of whether $H_{0}$ holds. Instead, Assumption~\ref{assump: conditional indep} allows $\mathbf{Z}$ to affect $\mathbf{M}$ if it has a treatment effect on $\mathbf{Y}$ (through the path $\mathbf{Z}\rightarrow \mathbf{Y} \rightarrow \mathbf{M}$).
    
Therefore, if the test statistic $T_{M}(\mathbf{Z}, \mathbf{M})$ detects an association between $\mathbf{Z}$ and $\mathbf{M}$, there are two possibilities: 
\begin{itemize}
    \item Assumption~\ref{assump: conditional indep} holds, but Fisher's sharp null $H_{0}$ does not hold;
    \item Assumption~\ref{assump: conditional indep} does not hold.
\end{itemize}
Since $T_{M}(\mathbf{Z}, \mathbf{M})$ cannot distinguish between these two possibilities, it is not a valid test for Assumption~\ref{assump: conditional indep}.

To illustrate this point, we conduct a simulation study to present cases where \textit{Assumption~\ref{assump: conditional indep} holds}, yet the rejection rate based on $T_{M}(\mathbf{Z}, \mathbf{M})$ (i.e., the proportion of times the $p$-value from $T_{M}(\mathbf{Z}, \mathbf{M})$ falls below the significance level 0.05) far exceeds the nominal 0.05 level. This simulation highlights that the test statistic $T_{M}(\mathbf{Z}, \mathbf{M})$ is not a valid test for Assumption~\ref{assump: conditional indep}. Specifically, we adopt the same settings and Models 1--4 as described in Section~\ref{subsec: simulation studies}. Note that all these models satisfy Assumption~\ref{assump: conditional indep}. In Table~\ref{tab: T_M outcomes simulation}, we report the rejection rates of $T_{M}(\mathbf{Z}, \mathbf{M})$ under Models 1--4 for various sample sizes $N$ and effect sizes $\beta$. The results show that the rejection rates of $T_{M}(\mathbf{Z}, \mathbf{M})$ can be much higher than the nominal 0.05 level \textit{even when Assumption~\ref{assump: conditional indep} holds}, particularly for large effect size $\beta$. This suggests that $T_{M}(\mathbf{Z}, \mathbf{M})$ cannot guarantee valid type-I error rate control for testing Assumption~\ref{assump: conditional indep} because it would misinterpret an actual treatment effect as evidence against Assumption~\ref{assump: conditional indep}.

\setlength{\tabcolsep}{3pt} % default is 6pt
\begin{table}[ht]
\centering
\begin{tabular}{lcccc}
\toprule
\multirow{2}{*}{}&\multicolumn{2}{c}{$N = 50$}&\multicolumn{2}{c}{$N = 1000$} \\
\cmidrule(rl){2-3} \cmidrule(rl){4-5} 
 &  {$\beta=0.5$} &{$\beta = 1$}  & {$\beta=0.5$}& {$\beta=1$} \\
\midrule
Model 1   & 0.08 & 0.14 & 0.44 &  0.92  \\
Model 2     & 0.08 & 0.14  & 0.42 & 0.87  \\
Model 3  & 0.16 &  0.35   & 0.95 & 1.00   \\
Model 4   & 0.14 & 0.23   & 0.92 & 1.00 \\
\bottomrule
\end{tabular}
\caption{The rejection rates of $T_{M}(\mathbf{Z}, \mathbf{M})$ under Models 1--4 with total sample size $N=50$ and $N=1000$ (based on level $\alpha=0.05$ and 2000 simulated datasets). }
\label{tab: T_M outcomes simulation}
\end{table}

As discussed in Remark~\ref{rem: missingness as test}, under Assumption~\ref{assump: conditional indep}, the $T_{M}(\mathbf{Z}, \mathbf{M})$ is in principle a valid test for Fisher's sharp null $H_{0}$. However, we do not recommend using it for testing $H_{0}$ because it has no statistical power when the treatment has an effect on the true outcomes $\mathbf{Y}$ but no effect on outcome missingness $\mathbf{M}$ (i.e., when Assumption~\ref{assump: conditional indep with more restrictions} in Section~\ref{sec: CI with missing outcomes} holds, which is an important special case of Assumption~\ref{assump: conditional indep}). Also, even when Assumption~\ref{assump: conditional indep} holds but Assumption~\ref{assump: conditional indep with more restrictions} does not, by leveraging the covariate information during outcome imputation, the proposed imputation-assisted randomization tests described in Algorithm~\ref{alg: part} is typically more powerful than $T_{M}(\mathbf{Z}, \mathbf{M})$ when testing $H_{0}$. Specifically, we report the power of $T_{M}(\mathbf{Z}, \mathbf{M})$ and that of Algorithm~\ref{alg: part} under Models 1--4 considered in Section~\ref{subsec: simulation studies} in Figure~\ref{fig: missingness as outcome tests}. The simulation results in Figure~\ref{fig: missingness as outcome tests} suggest that the proposed imputation-assisted randomization tests described in Algorithm~\ref{alg: part} are more powerful than $T_{M}(\mathbf{Z}, \mathbf{M})$ when testing $H_{0}$ under these settings.

\begin{figure}[H]
      \centering
        \captionsetup[subfigure]{skip=2pt} % Adjust this value as needed

        \includegraphics[width=\textwidth]{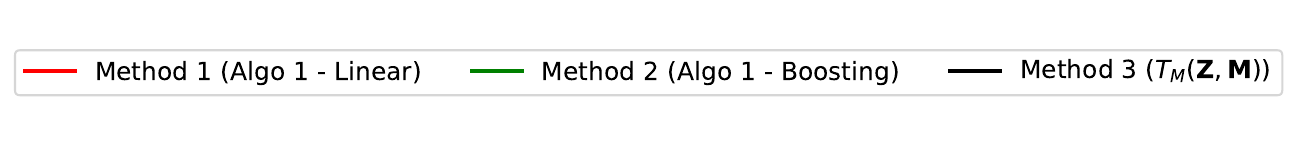}\vspace{-14pt} % Adjust the gap
      
      \begin{subfigure}[b]{0.4\textwidth}
        \includegraphics[width=\textwidth]{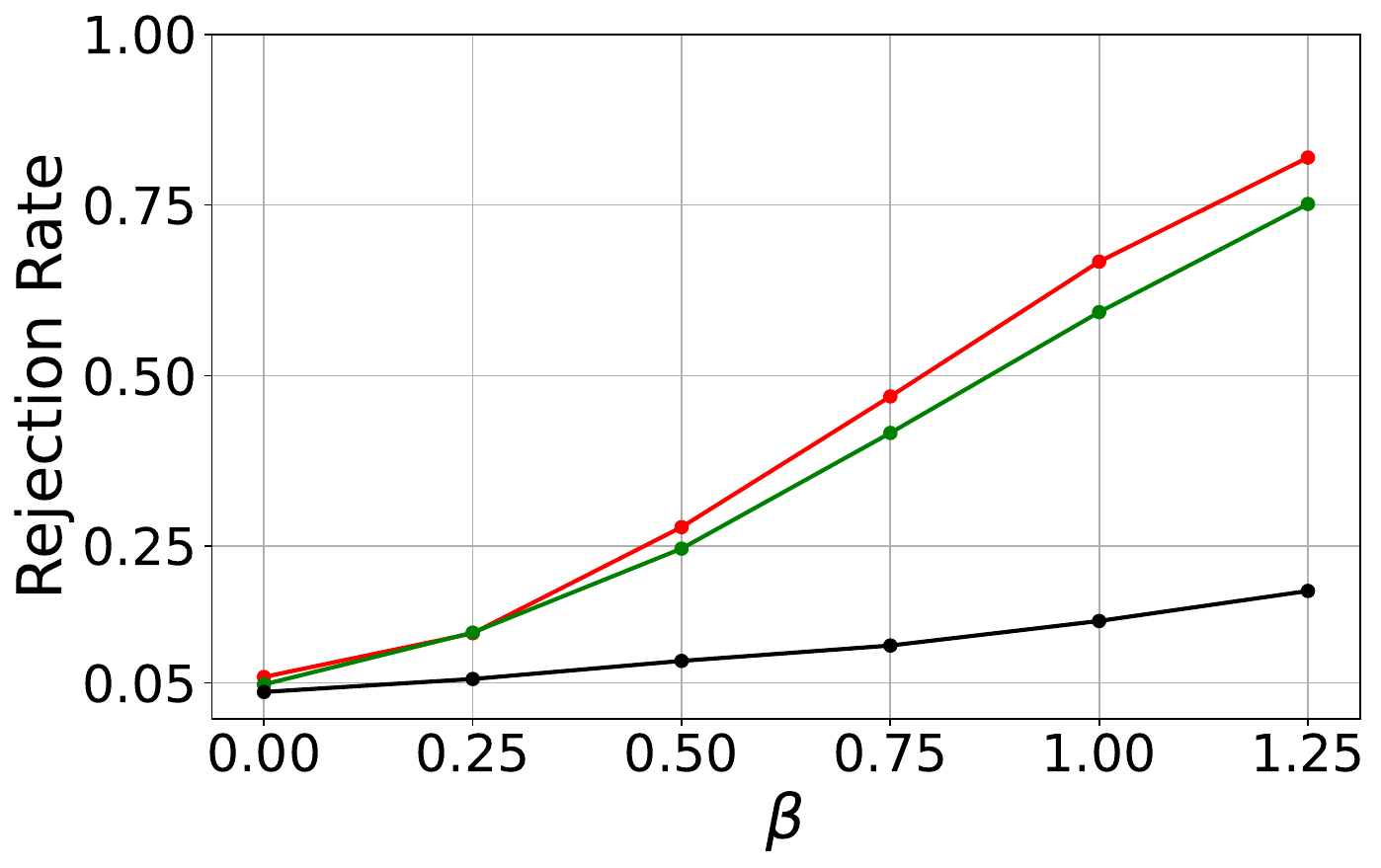}
        \caption{Model 1 ($N=50$)}
      \end{subfigure}
      \hspace{0.1cm}
      \begin{subfigure}[b]{0.4\textwidth}
        \includegraphics[width=\textwidth]{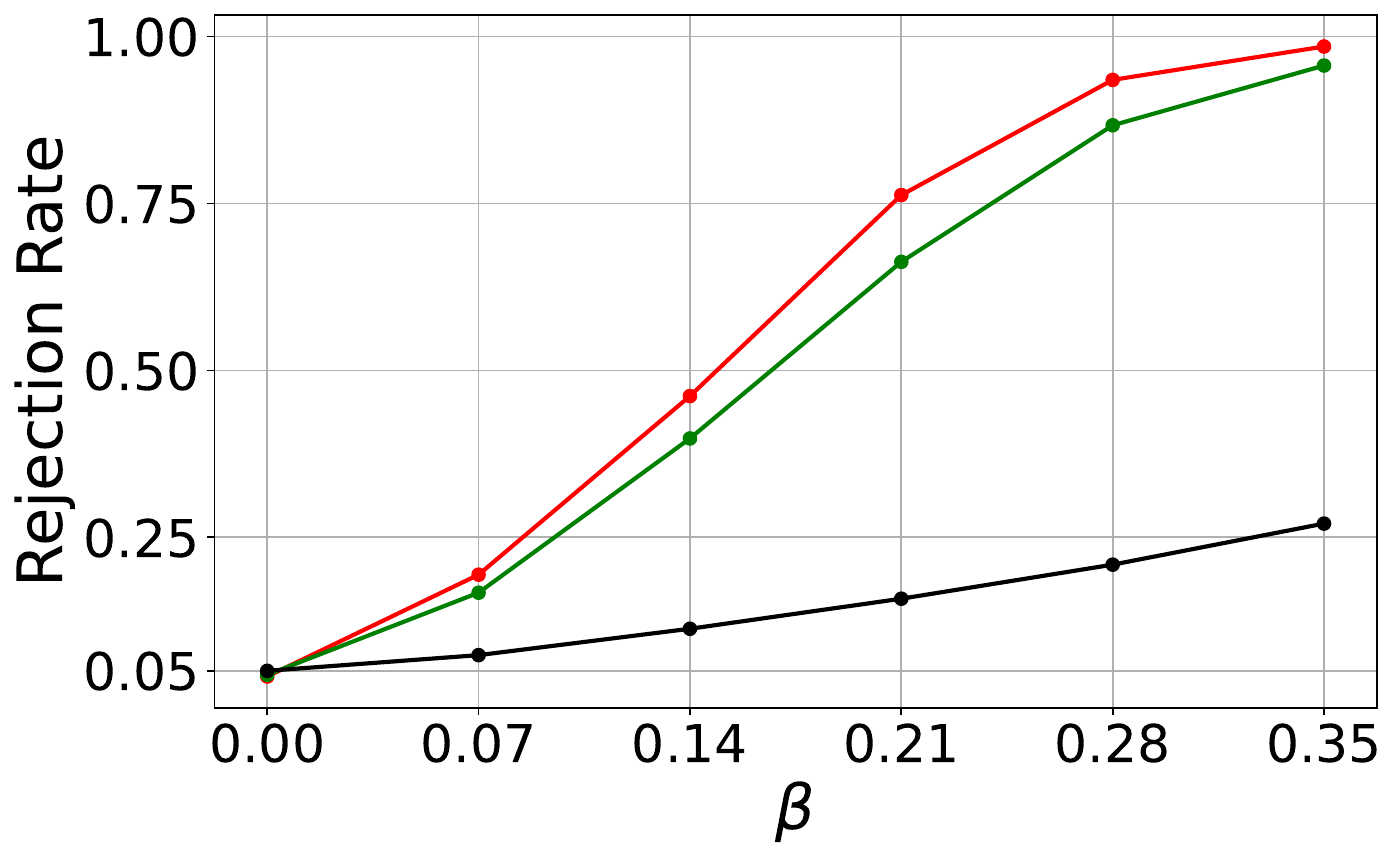}
        \caption{Model 1 ($N=1000$)}
      \end{subfigure}

    % Constant treatment effect; non-linear model for the true outcome; non-linear selection model for the missingness status, without interference

      \begin{subfigure}[b]{0.4\textwidth}
        
        \includegraphics[width=\textwidth]{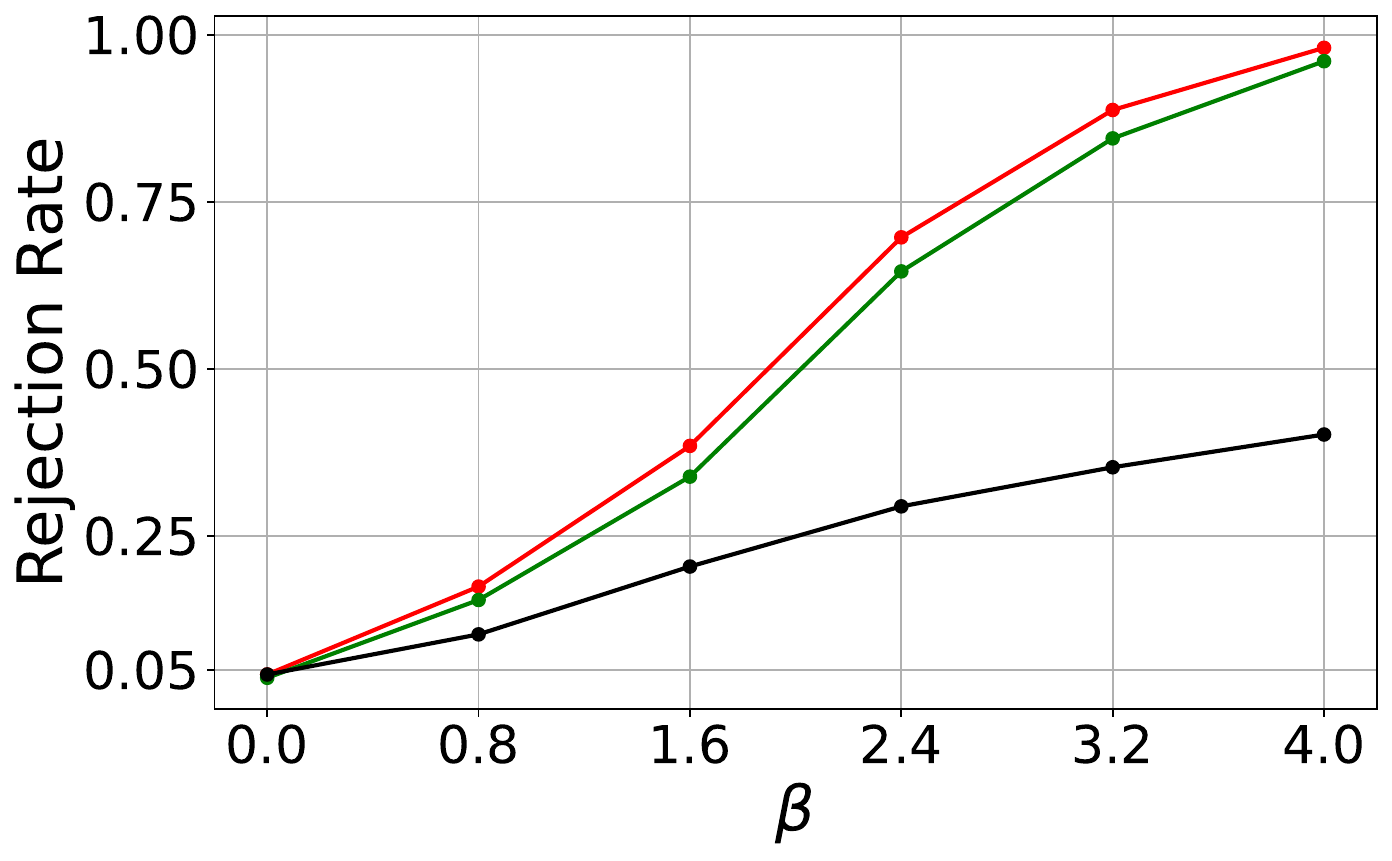}
        \caption{Model 2 ($N=50$)}
      \end{subfigure}
      \hspace{0.1cm}
      \begin{subfigure}[b]{0.4\textwidth}
        
        \includegraphics[width=\textwidth]{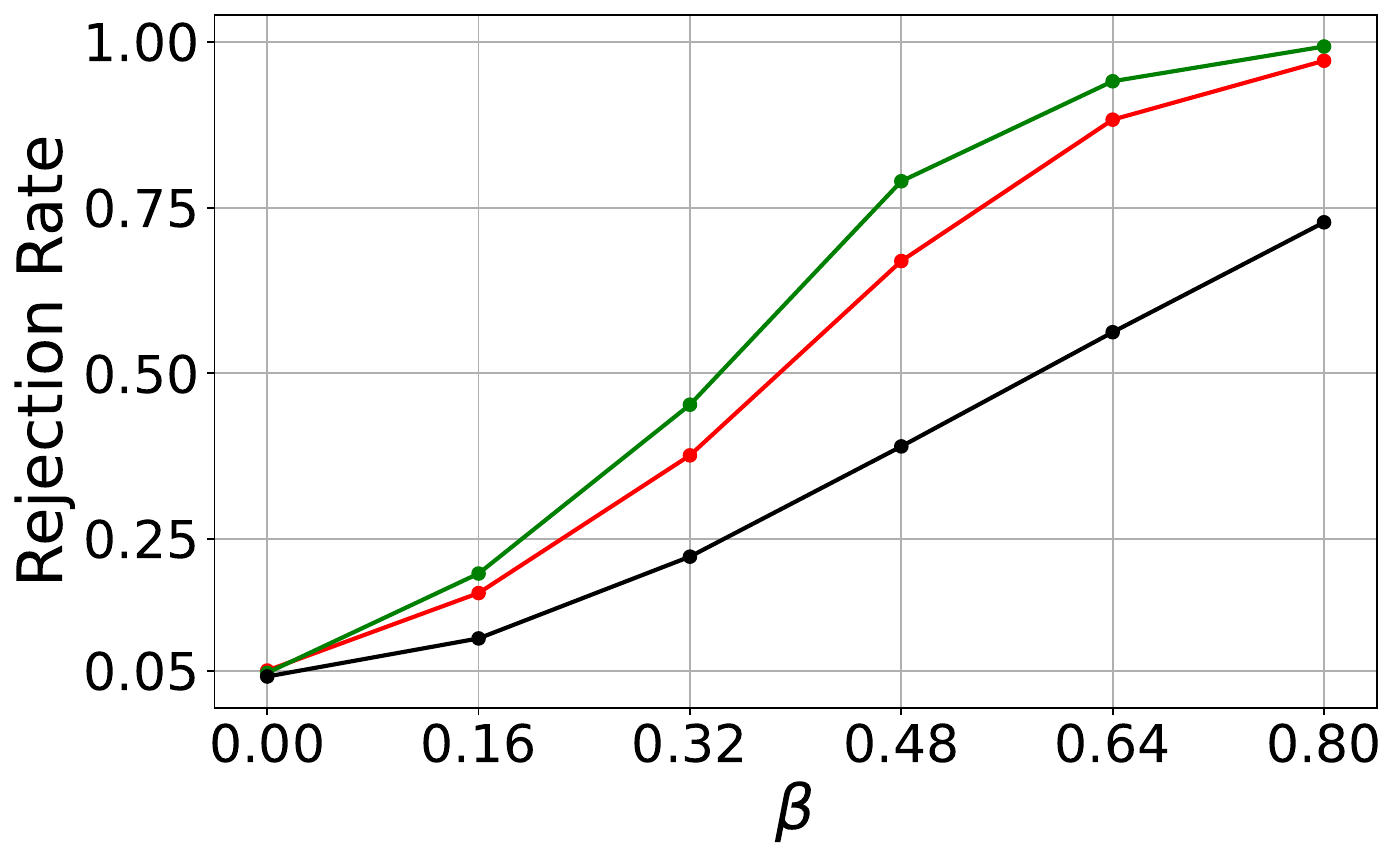}
        \caption{Model 2 ($N=1000$)}
      \end{subfigure}
    
    % Heterogeneous treatment effect; non-linear model for the true outcome; non-linear selection model for the missingness status, without interference
    
      \begin{subfigure}[b]{0.4\textwidth}
        
        \includegraphics[width=\textwidth]{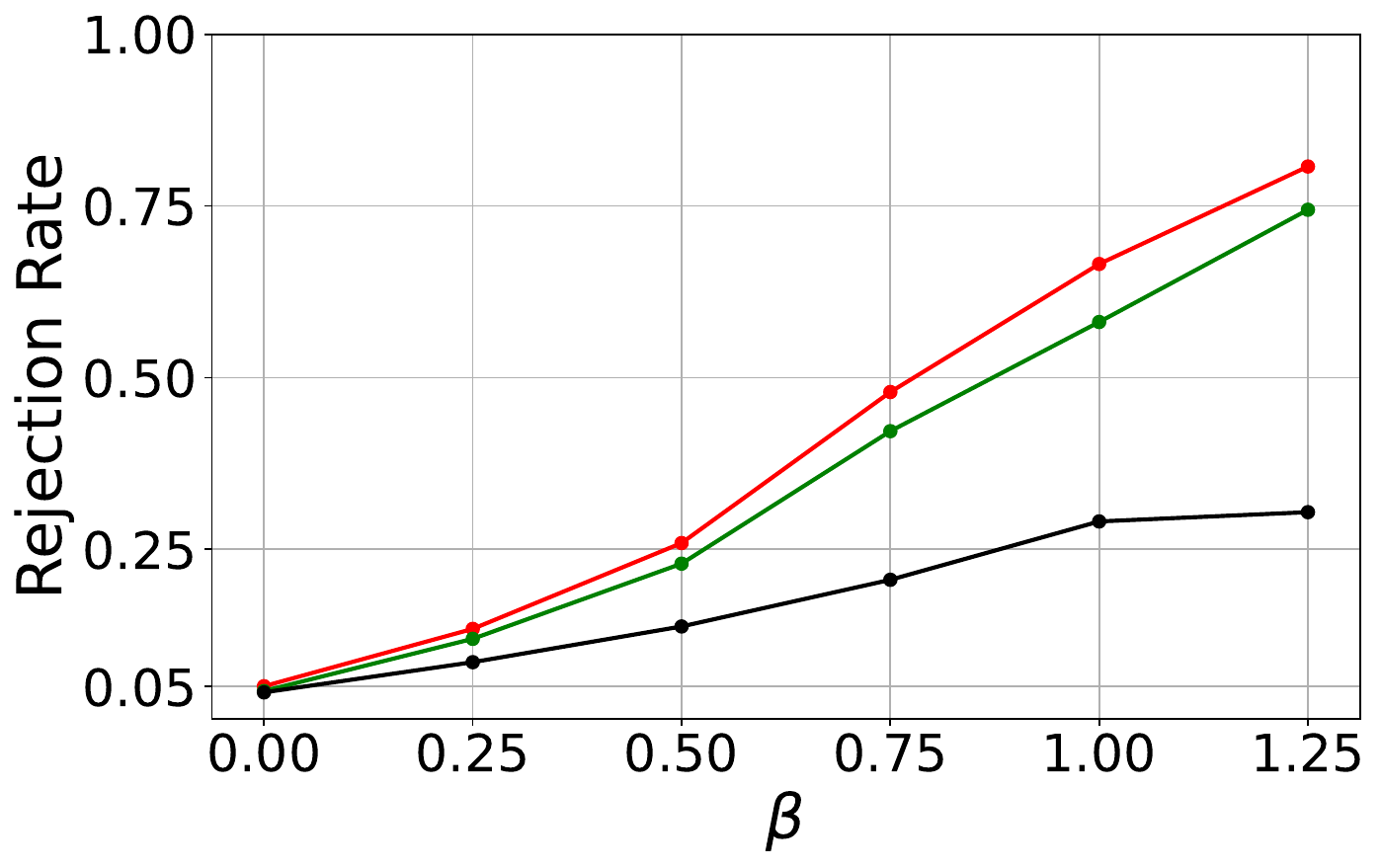}
        \caption{Model 3 ($N=50$)}
      \end{subfigure}
      \hspace{0.1cm}
      \begin{subfigure}[b]{0.4\textwidth}
        
        \includegraphics[width=\textwidth]{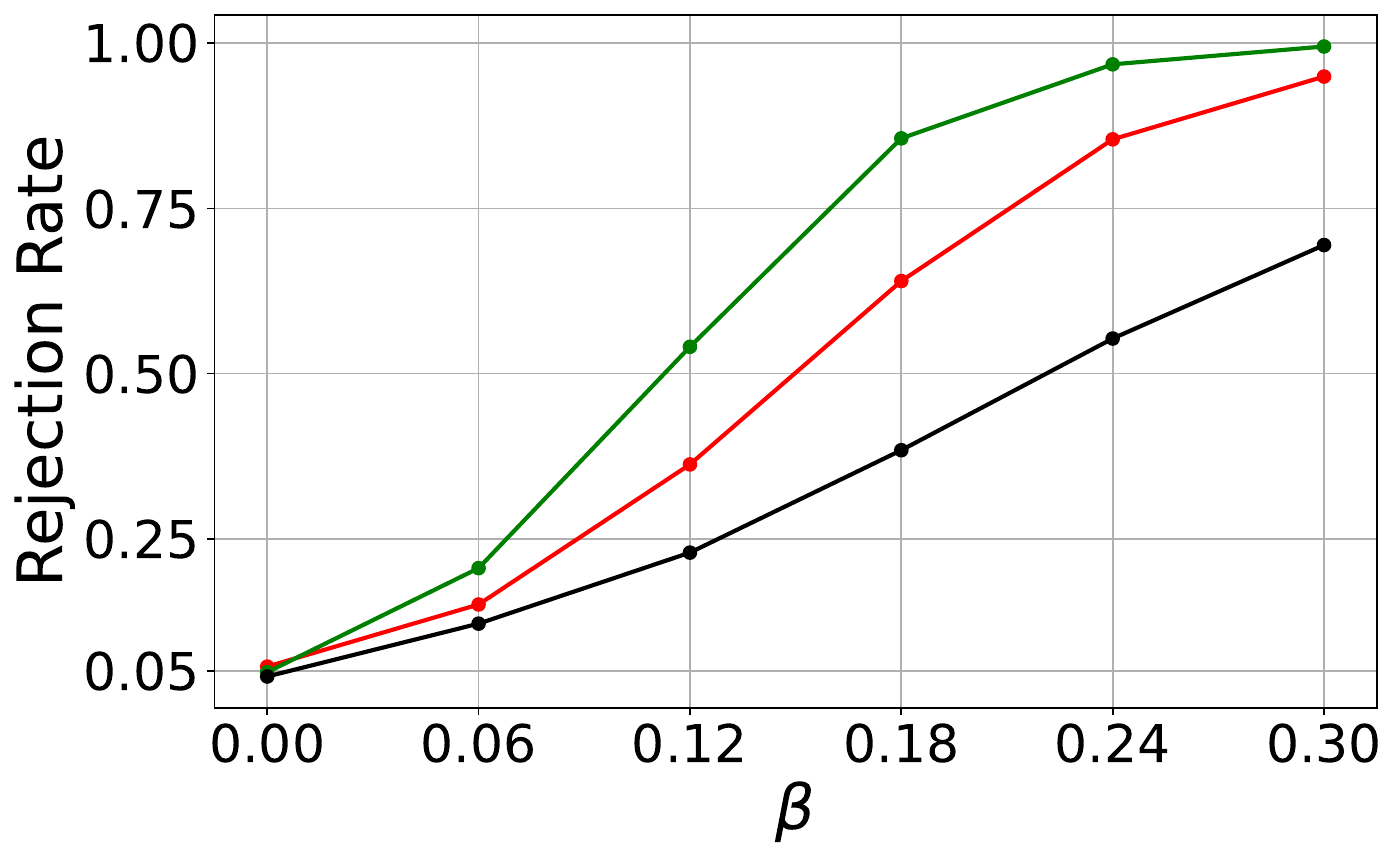}
        \caption{Model 3 ($N=1000$)}
      \end{subfigure}

    % Heterogeneous treatment effect; non-linear model for the true outcome; non-linear selection model for the missingness status, with interference
      
      \begin{subfigure}[b]{0.4\textwidth}
        
        \includegraphics[width=\textwidth]{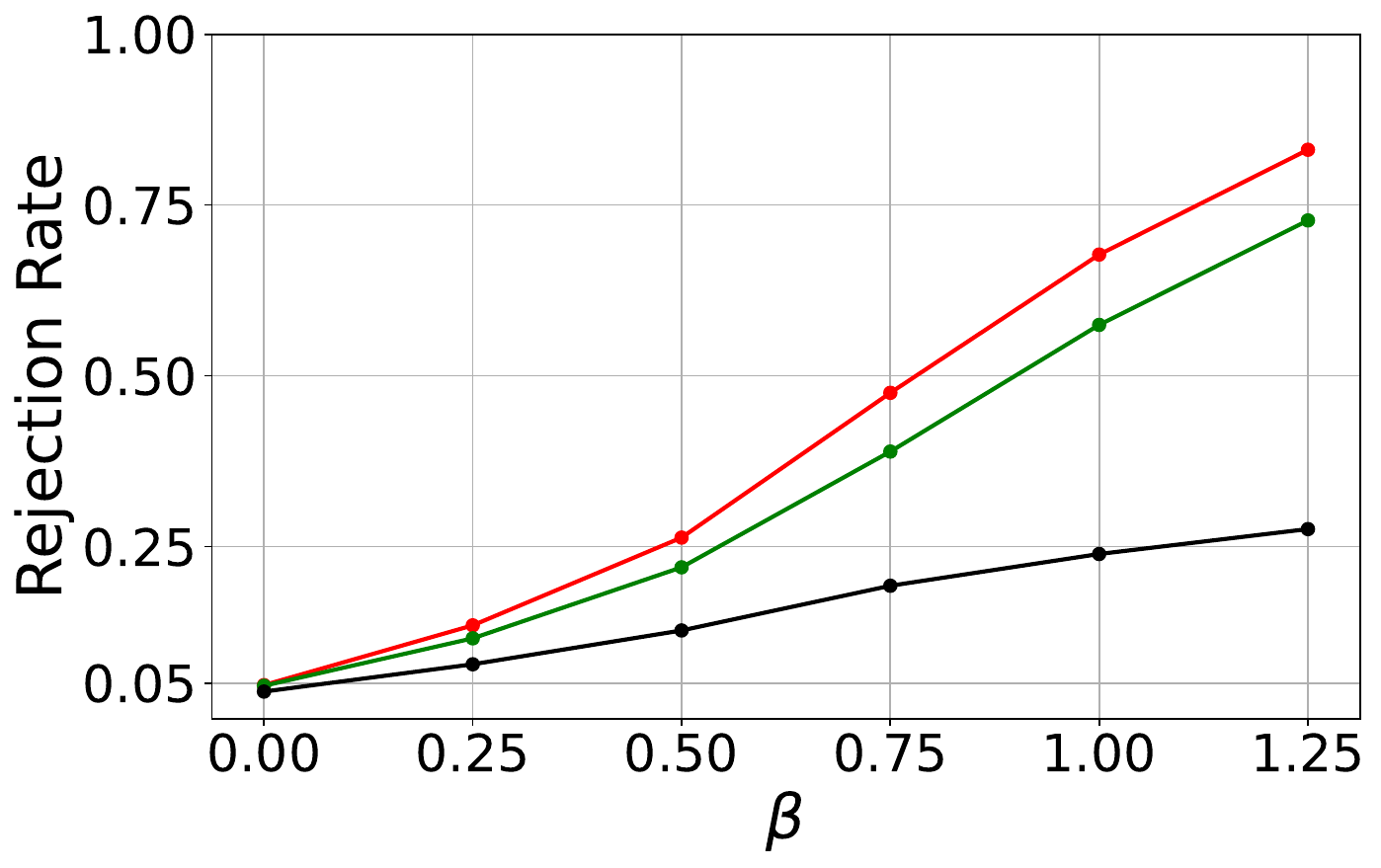}
        \caption{Model 4 ($N=50$)}
      \end{subfigure}
      \hspace{0.1cm}
      \begin{subfigure}[b]{0.4\textwidth}
        
        \includegraphics[width=\textwidth]{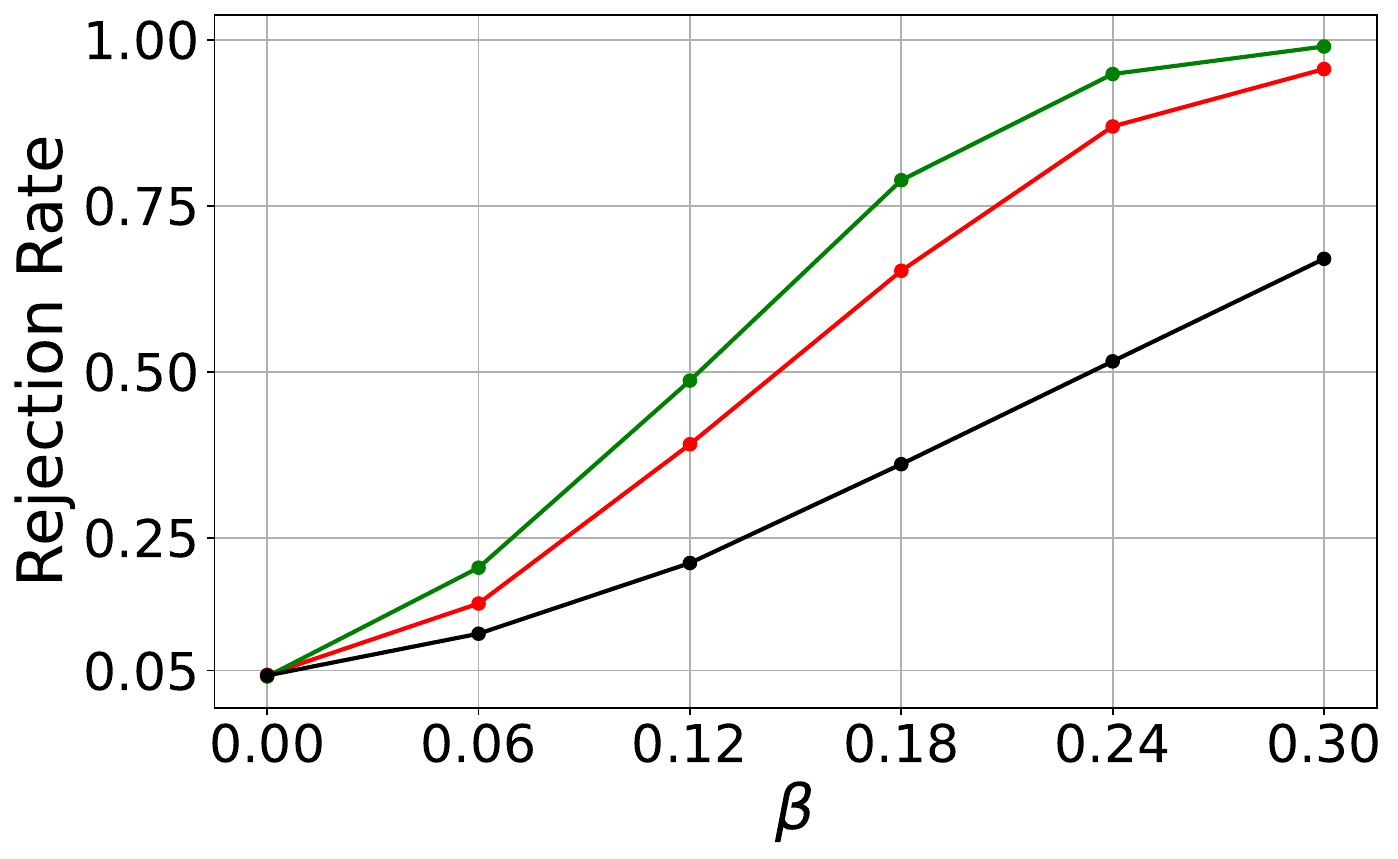}
        \caption{Model 4 ($N=1000$)}
      \end{subfigure}    
      \caption{Type-I error rate (when effect size $\beta=0$) and power (when effect size $\beta>0$) of our method and those of $T_{M}(\mathbf{Z}, \mathbf{M})$ with sample size $N=50$ and $N=1000$ (level $\alpha=0.05$). }
      \label{fig: missingness as outcome tests}
\end{figure}

\subsection*{Appendix B.5: Simulation Studies with Multiple Outcomes}

We then consider the multiple missing outcomes case. Specifically, we consider the following data-generating process for the three true outcomes $Y_{ij1}$, $Y_{ij2}$, and $Y_{ij3}$ and the corresponding missingness indicators $M_{ij1}$, $M_{ij2}$, and $M_{ij3}$:
\begin{itemize}
    \item Model 5 (Multiple outcomes; Heterogeneous treatment effects; non-linear model for the true outcomes; non-linear selection model for the missingness status; with interference in the missingness mechanism): $Y_{ij1} =\frac{1}{4}\beta Z_{ij}  +  \sum \limits_{p=1}^{5} x_{ijp} + \frac{1}{\sqrt{5}} \sum \limits_{p, p'} x_{ijp} x_{ijp^{\prime}} + \sin(u_{ij}) + \alpha_{i1} + \epsilon_{ij1}$, $Y_{ij2} = \beta Z_{ij}\left( 1 + x_{ij1} + u_{ij}\right) - \frac{1}{\sqrt{5}} \sum \limits_{p, p'} x_{ijp} \sigma(1 - x_{ijp^{\prime}}) + \alpha_{i2} + \epsilon_{ij2}$, $Y_{ij3} = \beta Z_{ij}  \sum \limits_{p=1}^{5} |x_{ijp}|  + \frac{1}{5}\sum\limits_{p, p^{\prime}, p^{\prime \prime}} x_{ijp} x_{ijp^{\prime}} \cos(x_{ijp''}) + u_{ij} + \alpha_{i3} + \epsilon_{ij3}$,  \newline $M_{ij1} = \mathbbm{1}\left\{ \frac{1}{\sqrt{5}} \sum \limits_{p=1}^{5} \sigma(x_{ijp}) + \frac{1}{5}  \sum \limits_{p, p'} x_{ijp} x_{ijp^{\prime}} + 5\sigma(Y_{ij1}) + u_{ij} > \lambda_1 \right\}$, 
    \newline $M_{ij2} = \mathbbm{1}\left\{ \frac{1}{\sqrt{5}} \sum \limits_{p=1}^{5} x_{ijp}^3 + \frac{1}{5} \sum \limits_{p, p'} x_{ijp} x_{ijp^{\prime}} + \frac{5}{2} \sum \limits_{k=1}^{3} \sigma(Y_{ijk}) + \sigma(1 - u_{ij}) > \lambda_2 \right\}$, \newline and $M_{ij3} = \mathbbm{1}\left\{ \frac{1}{\sqrt{5}} \sum \limits_{p=1}^{5} p x_{ijp} + \frac{1}{5\sqrt{5}}\sum\limits_{p, p^{\prime}, p^{\prime \prime}} x_{ijp} x_{ijp'} x_{ijp''} + \frac{5}{3} \sum \limits_{k=1}^{3} \sigma(Y_{ijk}) + \sin(u_{ij}^2) > \lambda_3 \right\}$.

\end{itemize}
In Model 5, the stratum-level random effect $\alpha_{ik}$ are generated through: $\alpha_{ik}\overset{\text{i.i.d.}}{\sim} N(0, 0.1)$ for each stratum $i$ and for each $k = 1, 2, 3$, and the individual-level random effect $\epsilon_{ijk}$ are generated through the following process: $\epsilon_{ijk} \overset{\text{i.i.d.}}{\sim} N(0, 0.2)$, for each subject $ij$ and each outcome $k = 1, 2, 3$. Therefore, the total variance of the three terms that observed data cannot capture is $\text{var}(u_{ij})+\text{var}(\alpha_{ik})+\text{var}(\epsilon_{ijk})=0.5$. We set $(\lambda_{1}, \lambda_{2}, \lambda_{3})$ such that the missingness rate is $50\%$ for each outcome and for each simulated dataset. In Table~\ref{tab: multiple outcomes simulation}, we report the simulated power of Methods 1--4 (using the Holm-Bonferroni method to adjust for multiple outcomes) under Model 5, setting level $\alpha=0.05$. Table~\ref{tab: multiple outcomes simulation} shows, similar to the single outcome case considered in Models 1--4, the imputation and re-imputation approach (Methods 2 and 3) can also increase the statistical power of a randomization test with multiple missing outcomes.

\setlength{\tabcolsep}{3pt} % default is 6pt
\begin{table}[ht]
\centering
\begin{tabular}{lcccc}
\toprule
\multirow{2}{*}{}&\multicolumn{2}{c}{$N = 50$}&\multicolumn{2}{c}{$N = 1000$} \\
\cmidrule(rl){2-3} \cmidrule(rl){4-5} 
 &  {$\beta=1.50$} &{$\beta = 2.00$}  & {$\beta=0.20$}& {$\beta=0.25$} \\
\midrule
Method 1 (Median Imputation)  & 0.40 & 0.53 & 0.36 &  0.45  \\
Method 2 (Algo 1 -- Linear)    & 0.80 & 0.92  & 0.58&0.74  \\
Method 3 (Algo 1 -- Boosting) & 0.61 &  0.74  & 0.67&0.85   \\
Method 4 (Oracle)  & 0.97 & 0.99   &0.87 & 1.00 \\
\bottomrule
\end{tabular}
\caption{Power of Methods 1--4 under Model 5 (the multiple outcomes case), with total sample size $N=50$ and $N=1000$ (based on level $\alpha=0.05$ and 2000 simulated datasets). }
\label{tab: multiple outcomes simulation}
\end{table}
\vspace{-0.5 cm}

\section*{Appendix C: Additional Details on the Simulation Studies in Section 3.3 of the Main Text}

\subsection*{Appendix C.1: Implementation of Imputation-Assisted Randomization Tests Using Chained Equations with Bayesian Ridge and Gradient Boosting Predictive Models}
For implementing the imputation-assisted randomization tests under the imputation and re-imputation framework, we need to specify an outcome imputation algorithm $\mathcal{G}$. For implementing the outcome imputation algorithm $\mathcal{G}$ in the simulation studies, we utilize Chained Equations Imputation, which imputes missing data through an iterative series of predictive models (for which we use a linear predictive model and a non-linear boosting predictive model, respectively). Specifically, we implement Chained Equations Imputation through the \texttt{IterativeImputer} module in the \textsf{Python} library \texttt{scikit-learn} (\citealp{pedregosa2011scikit}), in which almost all the parameters are set to their default values except the following three parameters: (i) we explicitly specify some predictive model embedded in Chained Equations Imputation through the command \texttt{estimator=Regressor()}, which will be described in detail in later paragraphs; (ii) we set \texttt{max\_iter} to 3 to reduce simulation time (the default value is 10); and (iii) we set \texttt{random\_state} to 0 to ensure reproducibility. The corresponding \textsf{Python} command is presented below:
\begin{verbatim}
IterativeImputer(estimator=Regressor(), max_iter=3)
\end{verbatim}

Under the Chained Equations Imputation framework, we consider two predictive models (corresponding to specifying the \texttt{Regressor()} in the above command), one linear and one non-linear (boosting). 
\begin{itemize}

    \item (Linear Predictive Model): The linear predictive model embedded in Chained Equations Imputation is Bayesian ridge regression (MacKay, 1992; Tipping, 2001), which is a widely used form of regularized linear regression that can offer enhanced robustness in the face of ill-posed problems compared with a traditional linear regression. All parameters were retained at their default settings during the implementation.
    
    \item (Boosting Predictive Model): The boosting predictive models embedded in Chained Equations Imputation are XGBoost regressor (\citealp{chen2016xgboost}) (when the sample size $N=50$) and LightGBM regressor (\citealp{ke2017lightgbm}) (when sample size $N=1000$), two commonly used robust gradient boosting machines; see Section~\ref{subsec: simulation studies} in the main text for details. All parameters are set to default values. 

\end{itemize}

\subsection*{Appendix C.2: Computation Time of the Proposed Framework}
We report the average computation time of the proposed imputation and re-imputation framework described in Algorithm~\ref{alg: part}, calculated under Models 1--4 in Table~\ref{tab: computation time}. Each value represents the average over $1000$ simulated datasets. All models were run on an Intel Xeon Platinum 8268 CPU @ 2.90GHz with a single core. 

\begin{table}[ht]
    \centering
    \caption{Average computation time (in seconds) of Algorithm~\ref{alg: part} under Models 1--4 based on linear imputation model and boosting, respectively. Each value is obtained by averaging over the 1000 simulated datasets. }
    \label{tab: computation time}
    \setlength{\tabcolsep}{10pt} % Adjusts column spacing
    \begin{tabular}{lcccc}
        \toprule
        \multirow{2}{*}{} & \multicolumn{2}{c}{$N = 50$} & \multicolumn{2}{c}{$N = 1000$} \\
        \cmidrule(lr){2-3} \cmidrule(lr){4-5}
        & Algo 1 - Linear & Algo 1 - Boosting & Algo 1 - Linear & Algo 1 - Boosting \\
        \midrule
        Model 1 & 196.74 & 1837.24 & 217.22 & 7763.28 \\
        Model 2 & 192.74 & 1834.37 & 217.24 & 7744.24 \\
        Model 3 & 192.80 & 1832.37 & 216.31 & 7734.58 \\
        Model 4 & 192.57 & 1831.69 & 215.69 & 7736.12 \\
        \bottomrule
    \end{tabular}
    \vspace{0.3cm} % Adds vertical space below the table
\end{table}

As expected, Table~\ref{tab: computation time} shows that the computational cost of applying boosting methods (XGBoost when $N=50$ and LightGBM when $N=1000$) in Algorithm~\ref{alg: part} is much higher than that of applying linear models. Given the fact that Algorithm~\ref{alg: part} based on boosting methods typically provides greater statistical power than linear models when the underlying data-generating processes of outcome and missingness mechanisms are complex, a meaningful future research direction would be to better understand the trade-off between statistical power and computational cost under the proposed imputation and re-imputation framework.

Here, we provide some practical guidance for researchers implementing the imputation and re-imputation framework: 
\begin{itemize}
\item Small sample sizes: Computational cost is not a major concern when choosing between a simple linear model and a more complex machine learning method (e.g., boosting), as both should yield results within an acceptable time frame.
 \item Large sample sizes with limited computational resources: Linear models are a more computationally efficient choice.  
\item Large sample sizes with ample computational resources: Researchers can use parallel computing and/or GPUs to significantly reduce computation time when applying boosting methods (or other machine learning techniques). Specifically, the re-imputation step in Algorithm~\ref{alg: part} can be executed in parallel across permuted datasets. If multiple CPU or GPU cores are available, parallelization can substantially decrease computation time. Furthermore, since boosting and many other machine learning methods are highly optimized for GPU architectures, using high-performance GPUs can further accelerate computation.
\end{itemize}

\subsection*{Appendix C.3: Simulation Studies on the Convergence Rate of $p$-value Reported by the Proposed Framework}

One important property of Theorems~\ref{thm: hypothesis testing} and \ref{thm: hypothesis testing with covariate adjustment} is that the convergence rate (as the number of re-imputation runs $L$ goes to infinity) of $p$-value under Fisher's sharp null $H_{0}$ reported by Algorithms~\ref{alg: part} and \ref{alg: part with covariate adjustment} does not depend on the specific imputation algorithm (e.g., either a finite-dimensional parametric imputation model or non-parametric methods such as XGBoost) adopted in Algorithms~\ref{alg: part} and \ref{alg: part with covariate adjustment}. This is because, \textit{under Fisher's sharp null $H_{0}$} and Assumption~\ref{assump: conditional indep}, the treatment information \textit{has no predictive power} (i.e., no better than a random guess) for missing outcomes during the re-imputation runs, regardless of the imputation algorithm adopted in Algorithms~\ref{alg: part} and \ref{alg: part with covariate adjustment}; see the proofs of Theorems~\ref{thm: hypothesis testing} and \ref{thm: hypothesis testing with covariate adjustment} in Appendix A for details. 

To illustrate the above point under Models 1--4 considered in Section~\ref{subsec: simulation studies} and Fisher's sharp null $H_{0}$ (i.e., when setting $\beta=0$ in Models 1--4), in Table~\ref{tab: convergence rate}, we report the number of re-imputation runs $L$ required for the $p$-value reported by Algorithm~\ref{alg: part} to achieve some specific convergence level. Specifically, we report the minimal number of re-imputation runs $L^{*}$ (averaged over 2000 simulation datasets) such that as the number of re-imputation runs $L$ surpasses $L^{*}$ and continues to increase, the changes of $p$-values across different $L$ will be less than some prespecified $\gamma$ (e.g., set $\gamma=0.01$ or $0.005$). From Table~\ref{tab: convergence rate}, we can find that the minimal number of re-imputation runs required by some level of convergence of $p$-value reported under the proposed framework does not vary much across different imputation algorithms (e.g., linear imputation model or XGBoost). This confirms that the convergence rate of  $p$-value reported by the proposed framework does not change across different imputation algorithms adopted in the proposed framework.

%This table shows the convergence behavior in simulation studies of Models 1 to 4. We set a threshold for the p-values, and if any p-value reaches a number of simulations, \(L\), where the error between the p-value and the final value is within the threshold, any \(L\) greater than this value ensures the error remains within the threshold. This \(L\) is considered the necessary number of iterations for stable convergence. The table above presents the average convergence \(L\) needed over 2000 simulations for both Algo 1 - Linear and Algo 1 - Boosting methods with thresholds of 0.01 and 0.005 across different values of \(N\) (1000 and 50).

% Table for Threshold 0.01
\begin{table}[H]
    \centering
    \caption{The minimal number of re-imputation runs $L^{*}$ needed for the $p$-value to achieve certain level of convergence (measured by $\gamma$). Each $L^{*}$ reported in the table has been averaged over the $2000$ simulation datasets. }
    \setlength{\tabcolsep}{10pt} % Adjusts column spacing
    \begin{tabular}{lcccc}
        \toprule
        \multirow{2}{*}{$\gamma=0.01$} & \multicolumn{2}{c}{$N = 50$} & \multicolumn{2}{c}{$N = 1000$} \\
        \cmidrule(lr){2-3} \cmidrule(lr){4-5}
        & Algo 1 - Linear & Algo 1 - Boosting & Algo 1 - Linear & Algo 1 - Boosting \\
        \midrule
        Model 1 & 2002.09 & 1945.93 & 2003.97 & 2019.65 \\
        Model 2 & 2038.79 & 2027.33 & 2016.74 & 2003.43 \\
        Model 3 & 1975.20 & 2025.84 & 1963.17 & 1948.30 \\
        Model 4 & 1951.29 & 1950.46 & 1965.92 & 1979.09 \\
        \toprule
        \multirow{2}{*}{$\gamma=0.005$} & \multicolumn{2}{c}{$N = 50$} & \multicolumn{2}{c}{$N = 1000$} \\
        \cmidrule(lr){2-3} \cmidrule(lr){4-5}
        & Algo 1 - Linear & Algo 1 - Boosting & Algo 1 - Linear & Algo 1 - Boosting \\
        \midrule
        Model 1 & 4418.34 & 4339.03 & 4491.54 & 4461.29 \\
        Model 2 & 4519.73 & 4527.74 & 4489.19 & 4433.36 \\
        Model 3 & 4405.48 & 4430.10 & 4400.44 & 4449.95 \\
        Model 4 & 4361.75 & 4327.80 & 4448.62 & 4403.44 \\
        \bottomrule
    \end{tabular}
    \label{tab: convergence rate}
\end{table}
% Table for Threshold 0.005

\section*{Appendix D: Additional Discussions and Remarks}

\begin{remark}\label{rem: counterexample to Assumption 3}
\emph{Assumption~\ref{assump: conditional indep} can be violated if the treatment assignments $\mathbf{Z}$ have a direct effect on outcome missingness $\mathbf{M}$ (i.e., there is a direct path from $\mathbf{Z}$ to $\mathbf{M}$) in addition to the effect through the path $\mathbf{Z}\rightarrow \mathbf{Y}\rightarrow \mathbf{M}$; see Figure~\ref{fig: counterexample} for a detailed illustration. } 
\begin{figure}[h]
        \centering
        \begin{tikzpicture}
            \node (XU) at (0,0) {$X, U$};
            \node (Y) [below = 1.2cm of XU] {$Y$};
            \node (M) [right = of Y] {$M$};
            \node (Z) [left = of Y] {$Z$};
    
            \path[->] (XU) edge (Y);
            \path[->] (XU) edge (M);
            \path[->] (Z) edge (Y);
            \path[->] (Y) edge (M);
            \path[->] (Z) edge[bend right=45] (M);

            \draw[->] (XU) -- (Z);
    
            % Adding cross in red for XU to Z
            \path (XU) -- (Z) coordinate[pos=0.5] (midpoint1);
            \draw[ thick] ($(midpoint1) + (-0.08,0) +(-0.15,-0.15)$) -- ($(midpoint1)+ (-0.08,0)+ (0.15,0.15)$);
            \draw[ thick] ($(midpoint1)+ (-0.08,0)+(0.15,-0.15)$) -- ($(midpoint1)+ (-0.08,0)+(-0.15,0.15)$);
            \node (annotation1) [left = 0.8cm of midpoint1, align=center] {\scriptsize Removed by \\ \scriptsize Randomization};
            \draw[->] (annotation1) -- ($(midpoint1) + (-0.08,0)$);
        \end{tikzpicture}
        \caption{A causal diagram that violates the outcome missingness mechanism stated in Assumption~\ref{assump: conditional indep}.}
           \label{fig: counterexample}
\end{figure}
\end{remark}

\begin{remark}
   \emph{In addition to randomized experiments, the randomization assumption (\ref{eqn: randomization assumption}) is also widely adopted in matched or stratified observational studies when assuming matching or stratifying on observed covariates (i.e., $\mathbf{x}_{ij}=\mathbf{x}_{ij^{\prime}}$ for all $i, j, j^{\prime}$) and no unobserved covariates (\citealp{rosenbaum2002observational}). Therefore, the proposed framework in this work can also be applied in these settings. } 
\end{remark}

\begin{remark}\emph{
    When the sample size is small enough and the outcome imputation algorithm $\mathcal{G}$ is deterministic, instead of randomly drawing a treatment assignment $\mathbf{Z}^{(l)}$ from the randomization distribution $\mathcal{P}$ in each simulation run, we can directly list all possible treatment assignments $\Omega$ and calculate the finite-population-exact $p$-value based on the imputation and re-imputation procedure (see equation (\ref{eqn: one-sided exact p-value deterministic}) in Lemma~\ref{lem: p-value}). }
\end{remark}

\begin{remark}
    \emph{Setting $\widehat{p}=\frac{1}{L}\sum_{l=1}^{L}\mathbbm{1}\{T^{(l)}\geq t\}$ in Algorithm~\ref{alg: part} or $\widehat{p}_{\text{adj}}=\frac{1}{L}\sum_{l=1}^{L}\mathbbm{1}\{T^{(l)}\geq t\}$ in Algorithm~\ref{alg: part with covariate adjustment} can give us an one-sided (greater than) $p$-value. To obtain an one-sided (less than) $p$-value, we just need to let $\widehat{p}=\frac{1}{L}\sum_{l=1}^{L}\mathbbm{1}\{T^{(l)}\leq t\}$ in Algorithm~\ref{alg: part} or $\widehat{p}_{\text{adj}}=\frac{1}{L}\sum_{l=1}^{L}\mathbbm{1}\{T^{(l)}\leq t\}$ in Algorithm~\ref{alg: part with covariate adjustment}. To obtain a two-sided $p$-value, following the general strategy proposed in \citet{luo2021leveraging}, we just need to let $\widehat{p}=2\min\{\frac{1}{L}\sum_{l=1}^{L}\mathbbm{1}\{T^{(l)}\geq t\}, \frac{1}{L}\sum_{l=1}^{L}\mathbbm{1}\{T^{(l)}\leq t\}\}$ in Algorithm~\ref{alg: part} or $\widehat{p}_{\text{adj}}=2\min\{\frac{1}{L}\sum_{l=1}^{L}\mathbbm{1}\{T^{(l)}\geq t\}, \frac{1}{L}\sum_{l=1}^{L}\mathbbm{1}\{T^{(l)}\leq t\}\}$ in Algorithm~\ref{alg: part with covariate adjustment}.}
\end{remark}

\begin{remark}
    \emph{In the simulation studies, the detailed procedure for median imputation is as follows: we first obtain the sample median and sample variance of the available outcome data (denoted as $\widehat{\sigma}^{2}$), and then impute each missing outcome with independent random draws from a normal distribution that centers at the sample median and has variance $\widehat{\sigma}^{2}$. Such a stochastic version of median imputation can avoid underestimating the standard error of imputed outcome data, a common drawback of the deterministic version of median imputation (\citealp{zhang2016missing, jadhav2019comparison}). }
\end{remark}

\begin{remark}
   \emph{As emphasized in the main text, the imputation and re-imputation framework (i.e., Algorithm~\ref{alg: part} and Algorithm~\ref{alg: part with covariate adjustment}) allow the observed covariates $\mathbf{X}$ to have missing values, as their finite-population-exact type-I error rate control (i.e., Theorem~\ref{thm: hypothesis testing} and Theorem~\ref{thm: hypothesis testing with covariate adjustment}) does not require any conditions on whether and how any values of $\mathbf{X}$ were missing. In Theorem~\ref{thm: confidence region} in Section~\ref{sec: CI with missing outcomes} of the main text, for convenience and clarity of the notations and for illustrating the main idea, we temporarily assume that the observed covariates $\mathbf{X}$ in the input data are complete and do not have any missingness. This is because if the treatment effect model in Theorem~\ref{thm: confidence region} involves any observed covariates (i.e., effect modifications or subgroup effects), for any individuals with missing values of any of these effect-modifier covariates, their outcomes would also be considered as missing outcomes because we may not be able to compute the counterfactual outcomes for these individuals due to missingness in their effect-modifier covariates. In this case, we would need to treat the outcomes of these individuals as missing outcomes (even if their original outcome records are available) and then apply the imputation and re-imputation framework to construct confidence regions for the causal parameters involved in the treatment effect model. Then, the theoretical conclusions in Theorem~\ref{thm: confidence region} still hold in this case. }
\end{remark}

\begin{remark}\label{rem: test statistics with additional details}\emph{
   As mentioned in the main text, the test statistic $T=T(\mathbf{Z}, \widehat{\mathbf{Y}})$ embedded in Algorithm~\ref{alg: part} can be any test statistic based on $\mathbf{Z}$ and $\widehat{\mathbf{Y}}$. For example, in the single outcome case (i.e., $K=1$), if we choose the permutational $t$-test statistic, then $T(\mathbf{Z}, \widehat{\mathbf{Y}})=\sum_{i=1}^{I}\sum_{j=1}^{n_{i}}Z_{ij}\widehat{Y}_{ij}$. If we choose the Wilcoxon rank sum test statistic, then $T(\mathbf{Z}, \widehat{\mathbf{Y}})=\sum_{i=1}^{I}\sum_{j=1}^{n_{i}}Z_{ij}\text{rank}(\widehat{Y}_{ij})$ where $\text{rank}(\widehat{Y}_{ij})=\sum_{i^{\prime}=1}^{I}\sum_{j^{\prime}=1}^{n_{i^{\prime}}}\mathbbm{1}\{\widehat{Y}_{ij}\geq \widehat{Y}_{i^{\prime}j^{\prime}}\}$. In the multiple outcomes case, we can choose $T(\mathbf{Z}, \widehat{\mathbf{Y}})$ to be some linear combination of the $K$ test statistics $T_{1}, \dots, T_{K}$, among which each $T_{k}$ focuses on the $k$-th outcome. For example, we can let $T_{k}(\mathbf{Z}, \widehat{\mathbf{Y}}_{k})=\sum_{i=1}^{I}\sum_{j=1}^{n_{i}}Z_{ij}\text{rank}(\widehat{Y}_{ijk})$ where $\text{rank}(\widehat{Y}_{ijk})=\sum_{i^{\prime}=1}^{I}\sum_{j^{\prime}=1}^{n_{i^{\prime}}}\mathbbm{1}\{\widehat{Y}_{ijk}\geq \widehat{Y}_{i^{\prime}j^{\prime}k}\}$, and then let $T=\sum_{k=1}^{K}\lambda_{k} T_{k}$, where $\lambda_{1},\dots, \lambda_{K}$ are prespecified weights (\citealp{rosenbaum2016using}). In addition to this combined statistic strategy, an alternative way to hypothesis testing in the multiple outcomes case is to combine the $K$ $p$-values $\widehat{p}_{1}, \dots, \widehat{p}_{K}$ through the Bonferroni correction or the Holm-Bonferroni method, where each $\widehat{p}_{k}$ is reported by applying the individual test $T_{k}$ in Algorithm~\ref{alg: part}.}
\end{remark}

\begin{remark}\label{rem: differences with other cov adj}\rm
  As mentioned in Section~\ref{sec: covariate adjustment}, there are a few other recent works on covariate adjustment in randomized experiments with incomplete outcome data (e.g., \citealp{chang2023covariate}; \citealp{zhao2023covariate}; \citealp{wang2024handling}), all of which have made important contributions to this topic. These recent works focus on the average treatment effect among the super-population, while our work focuses on Fisher's sharp null and its extensions (in model (\ref{eqn: effect model with specified beta_0}) in the main text) among the finite population (i.e., design-based causal inference). Because of this, the required assumptions and proposed approaches of these recent works and those of our work are very different. First, these recent works in the super-population setting typically require the observed data to fully capture outcome missingness. In contrast, our framework allows outcome missingness to depend on unobserved data such as unobserved covariates. Second, these recent works in the super-population setting require either the outcome model or the missingness model to be correctly specified. In contrast, our finite-population framework can still ensure finite-sample validity even when both the imputation model and the covariate adjustment model were misspecified. Third, these recent works assume that the outcome missingness indicators are independently and identically distributed realizations (at the individual or cluster level) from some super-population model or distribution to invoke super-population asymptotics. In contrast, focusing on the finite-population setting, our framework allows outcome missingness interference between study subjects and between different outcome variables, and its statistical validity holds for any sample size. 
\end{remark}

\begin{remark}\label{rem: Assumptions 2 and 3}
    \rm Recall that the key difference between Assumption~\ref{assump: conditional indep} and Assumption~\ref{assump: conditional indep with more restrictions} is that Assumption~\ref{assump: conditional indep} allows outcome missingness $\mathbf{M}$ to depend on the true outcomes $\mathbf{Y}$ (in addition to dependence on observed covariates $\mathbf{X}$ and unobserved covariates $\mathbf{U}$) but Assumption~\ref{assump: conditional indep with more restrictions} does \textit{not} allow dependence of $\mathbf{M}$ on $\mathbf{Y}$. This difference is exactly the key reason why Assumption~\ref{assump: conditional indep with more restrictions} (along with the randomization assumption stated in Assumption \ref{assump: randomization design}) suffices to ensure the validity of a confidence region/interval obtained by inverting a randomization test (as shown in Theorem~\ref{thm: confidence region}) but Assumption~\ref{assump: conditional indep} does not. Specifically, note that the feasibility of inverting a hypothesis test to construct a confidence region requires that the potential outcome missingness indicators $\mathcal{M}=\{M_{ijk}(\mathbf{z}): i \in [I], j \in [n_{i}], k \in [K], \mathbf{z}\in\Omega \}$ are known to us under each hypothetical effect size of the treatment, which ensures that the permutational distribution of the randomization test statistic under each hypothetical effect size can be obtained by Monte-Carlo simulation (i.e., the re-imputation step). When the treatment effect size is non-zero, Assumption~\ref{assump: conditional indep with more restrictions} can still ensure that each outcome missingness indicator $M_{ijk}$ does not change with different treatment assignments $\mathbf{Z}$, which implies that the potential missingness indicator $M_{ijk}(\mathbf{z})$ equals the observed value $M_{ijk}$ for all possible treatment assignments $\mathbf{z}\in \Omega$. Instead, if outcome missingness $\mathbf{M}$ can depend on the true outcomes $\mathbf{Y}$ as stated in Assumption~\ref{assump: conditional indep}, the counterfactual $M_{ijk}(\mathbf{z})$ is unknown unless we impose strong modeling assumptions on the dependence of $\mathbf{M}$ on $\mathbf{Y}$. Therefore, Assumption~\ref{assump: conditional indep} does not suffice to obtain the permutational distribution of the randomization test statistic when the hypothetical effect size is non-zero, which makes it infeasible to construct a confidence region by inverting a randomization test.
\end{remark}

\begin{remark}\label{rem: Assumption 3 is strong}
   \em Assumption~\ref{assump: conditional indep with more restrictions} does not hold when outcome missingness $\mathbf{M}$ depends on the true outcomes $\mathbf{Y}$, which is referred to as ``self-censoring" or ``outcome-dependent missingness" in the missing outcomes literature \citep{d2010new, wang2014instrumental}. Self-censoring may exist in several practical settings. For example, suppose the outcome of interest is some survival time censored at 2 years. In this case, self-censoring exists because survival time beyond 2 years is always missing. Another example of self-censoring is the missingness of self-reported outcomes concerning some sensitive topics such as drug use or risky sexual behaviors. In this setting, self-censoring may occur when participants refuse to report their true outcomes on these sensitive questions. For other relevant examples, see \citet{joffe2012g}, \citet{ chen2023causal}, and \citet{li2023self}. Meanwhile, as discussed in Remark~\ref{rem: self-censoring in Assumption 2} in the main text, Assumption~\ref{assump: conditional indep} allows self-censoring (i.e., the dependence of $\mathbf{M}$ on $\mathbf{Y}$) and is, therefore, weaker than Assumption~\ref{assump: conditional indep with more restrictions}. 
\end{remark}

\section*{Appendix E: More Details on the Real Data Application}

First, we clarify the notations for the cluster randomized experiments here. In the main text, we use index $i$ to denote the $i$-th stratum. In this section, we instead use index $i$ to denote the $i$-th cluster, but we still use index $j$ to denote $j$-th subject (in a cluster) and index $k$ to denote the $k$-th outcome. For example, the $Z_{ij}$, $Y_{ijk}$, $M_{ijk}$, and $Y^{*}_{ijk}$ denote the treatment indicator, $k$-th true outcome, missingness indicator for the $k$-th outcome, and $k$-th realized  outcome, of subject $j$ in cluster $i$. Similarly, all the other notations, concepts, null hypotheses, treatment effect models, and algorithms indexed by $ijk$ introduced in the main text can be translated to the cluster randomized experiments setting following the above argument. Similar to the argument in the main text, if there is only one outcome variable of interest, we can omit index $k$ here.

In a cluster (group) randomized experiment, the unit of randomization is each cluster. For example, in a cluster randomized experiment following complete randomization (in which we assume $I_{T}$ out of total $I$ clusters receive the treatment), we have $\mathcal{Z}=\{\mathbf{z}=(z_{11}, \dots, z_{In_{I}})\in \{0,1\}^{N}: z_{ij}=z_{ij^{\prime}} \text{ for all $i, j, j^{\prime}$ and } \sum_{i=1}^{I}z_{i1}=I_{T}\}$ (so we have $|\mathcal{Z}|={I\choose I_{T}}$) and 
\begin{equation}\label{eqn: cluster randomization assumption}
    P(\mathbf{Z}=\mathbf{z})={I\choose I_{T}}^{-1} \quad \text{for all $\mathbf{z}\in \mathcal{Z}$}.
\end{equation}

The original Work, Family, and Health Study (WFHS) adopted a biased-coin adaptive design for cluster-level treatment assignments. For methodology illustration purposes, we assume a more straightforward and basic randomization design (i.e., cluster complete randomization described above). A similar simplification strategy has also been adopted by \citet{wang2024handling}, in which the authors assume the treatment assignments follow independent coin flips. As pointed out by \citet{wang2024handling}, such a simplification strategy should lead to a slight overestimation of the variance (i.e., slightly more conservative statistical inference). Also, it is worth emphasizing that, as discussed in the main text, the current version of the imputation and re-imputation framework (e.g., Algorithm~\ref{alg: part} and Algorithm~\ref{alg: part with covariate adjustment}) works for a wide range of common randomization designs, such as complete randomization, stratified or blocked randomization, paired randomization, Bernoulli randomization, and cluster randomization, among many others. It is a meaningful future research direction to investigate how to extend the imputation and re-imputation framework to other randomization designs, such as biased-coin adaptive design, covariate-adaptive or response-adaptive randomization design, and stepped wedge design. 

The Work, Family, and Health Study \citep{WFHS2018} was conducted in the information technology division of a Fortune 500 company (pseudonym: Tomo). The description of the treatment (i.e., intervention) taken from the codebook of the WFHS (\url{https://www.icpsr.umich.edu/web/DSDR/studies/36158}) is listed as follows:
\begin{longtable}{p{5cm}| p{11cm}}
\hline
\textbf{Treatment Name} & \textbf{Description from the user guide of the WFHS (\url{https://www.icpsr.umich.edu/web/DSDR/studies/36158})}\\
\hline

Treatment condition (coded as ``CONDITION"): 1 (51.4\%); 2 (48.6\%); Missingness rate (0.0\%).
&  ``The intervention included both supervisory training on strategies to facilitate employees’ control over work time and work redesign activities that identified new ways to work that meet business needs, while increasing employee control over work time. 1=Intervention; 2=Control."  
\\
\hline
\end{longtable}

We also list the detailed description taken from the codebook of the WFHS (\url{https://www.icpsr.umich.edu/web/DSDR/studies/36158}) of each of the seven covariates:
\begin{longtable}{p{5cm}| p{11cm}}
\hline
\textbf{Covariate Name} & \textbf{Description from the codebook of the WFHS (\url{https://www.icpsr.umich.edu/web/DSDR/studies/36158})}\\
\hline
The baseline value of control over work hours (coded as ``SCWM\_CWH"): Mean: 3.53; Standard deviation: 0.64; Missingness rate (10.7\%).
& ``Measures the degree to which employees control the arrangement of the hours that they work (\citealp{thomas1995impact}). \citet{thomas1995impact} reasoned that inflexible work hours would create work-family conflict but that certain supportive elements of work (“family-supportive variables”) could help control the effects. We used a modified scale from the original 14 items proposed by \citet{thomas1995impact}." 
\\
\hline
Employee (coded as ``EMPLOYEE"): 0 (21.2\%); 1 (78.8\%); Missingness rate (0.0\%).
& ``Employee Indicator Variable: 0 = Not an employee (i.e., a manager); 1 = Employee (i.e., not a manager)." 
\\
\hline
Group job function (coded as ``RMZFN"):  1 (40.7\%); 2 (58.5\%); Missingness rate (0.8\%).
& ``Randomized Group’s Job Function: 1=Core; 2=Support." 
\\
\hline
Work-family conflict (coded as ``SCWM\_FTWC"): 
Mean: 2.12; Standard deviation: 0.65; Missingness rate (0.2\%)
& ``Work-family conflict is defined as a type of inter-role conflict where work and family roles are incompatible (\citealp{greenhaus1985sources}). This scale measures the extent to which time and effort devoted to work interfere with family responsibilities and the extent to which time and effort devoted to family interfere with work responsibilities." 
\\
\hline
Time adequacy (coded as ``SCWM\_TIMEALL"): Mean: 3.58; Standard deviation: 0.59; Missingness rate (4.9\%).
& ``Time adequacy questions are part of the larger Family Resource Scale-Revised (FRS) (\citealp{van2001family}), which is an assessment of a family member’s perceptions of available resources across a range of areas." 
\\
\hline
Perceived stress scale (PSS) (coded as ``SCEM\_STRS"): Mean: 8.54; Standard deviation: 2.68; Missingness rate (21.2\%).
& ``The Perceived Stress Scale (PSS) is by far the most widely used scale of stress appraisals and has been found to be more predictive of physical and mental health outcomes and use of health services than event-based stress checklists. The PSS has been found to predict many adverse physical and mental health outcomes (\citealp{cohen1983global}). The PSS has shown discriminant validity with regard to life event measures of stress (i.e., each measure of stress predicts different stress-related processes or outcomes). The PSS has good internal consistency (alpha = 0.78)."  
\\
\hline
Psychological distress scale (coded as ``SCEM\_DIST"): Mean: 10.74; Standard deviation: 3.14; Missingness rate (0.2\%).
& ``The K6 is the most widely used mental health screening scale in the United States and has been used in numerous psychiatric and social epidemiology studies, including the National Household Survey on Drug Abuse (e.g., Centers for Disease Control and Prevention, 2004; U.S. Department of Health and Human Services, 2004). Extensive clinical validation of the K-6 has been conducted. When compared to a gold-standard structured diagnostic interview, the K-6 has demonstrated excellent precision, including sensitivity and specificity (area under the curve = 0.86). The K-6 has excellent internal consistency (alpha = 0.89) and has been demonstrated to better predict a diagnosis of serious mental illness than a number of other well-validated scales of psychological distress (\citealp{kessler2003screening})."  
\\
\hline

\end{longtable}

The outcome variable is the 6-month follow-up value of control over work hours (coded as ``SCWM\_CWH" at the follow-up survey): Mean: 3.72; Standard deviation: 0.63; Missingness rate (19.1\%).

\section*{Appendix F: Extension of the Proposed Framework to Handling Missingness in Survival Data}

In survival analysis, there are two major types of missingness in outcomes: missing survival time due to censoring (e.g., the participants leaving the study) and missing censoring indicators (i.e., missing censoring status information) (\citealp{cook2004analysis, brownstein2021bayesian}). While missing survival times due to censoring has been extensively studied in the causal inference literature (\citealp{rosenbaum2002observational, rubin2006causal, li2023randomization}), missing censoring indicators remain an underexplored issue, despite their presence in many settings (see Example~\ref{example: missing censoring} below) (\citealp{cook2004analysis, brownstein2015parameter, brownstein2021bayesian}). 
    
To fill this gap in design-based causal inference, we extend our proposed framework to conduct randomization tests in right-censored survival data with missing censoring indicators. 

\begin{example}[Missing Censoring Indicators in Survival Data]\label{example: missing censoring}
In many survival (time-to-event) datasets, missing censoring indicators can arise due to the gap between the screening and formal diagnosis during the last follow-up (\citealp{brownstein2015parameter, brownstein2021bayesian}). Specifically, at the last follow-up, some studies first conduct screening (e.g., cancer screening or mental health disorder screening) for participants, followed by a formal diagnosis of the event of interest (e.g., cancer or a mental health disorder) for those who screen positive. However, in such studies, not all participants who screened positive during the last follow-up necessarily undergo the final diagnosis for the event of interest. In this case, the censoring indicators of those participants are missing because it remains unclear whether their survival time (time-to-event) surpassed the final study period (i.e., whether their events of interest occurred before they left the study). See Figure~\ref{fig: missing censoring status} for an illustration. 
\end{example}

Specifically, for each subject $j$ in stratum $i$, we let $T_{ij}$ denote the true, possibly unobserved, survival time and let $C_{ij}$ denote the censoring time. Then, we let the time at risk $R_{ij}=\min\{T_{ij}, C_{ij}\}$ and the event indicator $\Delta_{ij}=\mathbbm{1}\{T_{ij}\leq C_{ij}\}$. In some studies (e.g., those described in Example~\ref{example: missing censoring}),  the censoring indicators $\Delta_{ij}$ for some subjects are missing because it is unknown whether $T_{ij}\leq C_{ij}$ or $T_{ij}> C_{ij}$ due to unknown status of the event (e.g., having cancer or a specific disorder) at the last follow-up. We let $M^{\Delta}_{ij}$ denote the missingness indicator for $\Delta_{ij}$, i.e., $M^{\Delta}_{ij}=1$ if $\Delta_{ij}$ is missing and $M^{\Delta}_{ij}=0$ otherwise. Let $\mathbf{T}=(T_{11},\dots, T_{In_{I}})$ denote the vector of the true (but possibly unobserved) survival times, $\mathbf{C}=(C_{11},\dots, C_{In_{I}})$ the vector of the true censoring times, $\mathbf{R}=(R_{11},\dots, R_{In_{I}})$ the vector of the true times at risk, $\mathbf{\Delta}=(\Delta_{11},\dots, \Delta_{In_{I}})$ the vector of true censoring indicators, and $\mathbf{M}^{\Delta}=(M^{\Delta}_{11},\dots, M^{\Delta}_{In_{I}})$ the vector of missingness indicators for the censoring status. Following the potential outcomes framework, for each subject $j$ in stratum $i$, we let $T_{ij}(1)$ (or $T_{ij}(0)$) denote the potential outcome of survival time under the treatment (or under control), $C_{ij}(\mathbf{z})$ denote the potential outcome of censoring time under the treatment assignments $\mathbf{Z}=\mathbf{z}\in \Omega$, $R_{ij}(\mathbf{z})$ denote the potential time at risk under $\mathbf{Z}=\mathbf{z}\in \Omega$, $\Delta_{ij}(\mathbf{z})$ denote the potential censoring indicator under $\mathbf{Z}=\mathbf{z}\in \Omega$, and $M^{\Delta}_{ij}(\mathbf{z})$ denote the potential missingness indicator for censoring status under $\mathbf{Z}=\mathbf{z}\in \Omega$. In design-based causal inference, all the potential post-treatment variables are fixed values (including the potential survival time $\mathcal{T}=\{(T_{ij}(1), T_{ij}(0)): i \in [I], j \in [n_{i}]\}$, the potential censoring time $\mathcal{C}=\{C_{ij}(\mathbf{z}): i \in [I], j \in [n_{i}], \mathbf{z}\in \Omega \}$, the potential time at risk $\mathcal{R}=\{R_{ij}(\mathbf{z}): i \in [I], j \in [n_{i}], \mathbf{z}\in \Omega\}$, the potential censoring indicator $\mathit{\Delta}=\{\Delta_{ij}(\mathbf{z}): i \in [I], j \in [n_{i}], \mathbf{z}\in \Omega\}$, and the potential missingness indicators for censoring status $\mathcal{M}^{\Delta}=\{M^{\Delta}_{ij}(\mathbf{z}): i \in [I], j \in [n_{i}], \mathbf{z}\in \Omega\}$), and the only source of randomness is the randomization of treatment assignments $\mathbf{Z}$ \citep{rosenbaum2002observational, li2017general}. 

\begin{figure}    
    \centering
    \caption{An illustration of missing censoring indicators due to the gap between the screening and formal diagnosis during the last follow-up. This figure follows Figure 1 in \citet{brownstein2021bayesian} but with some notational adjustments. }\label{fig: missing censoring status}
    \begin{tikzpicture}[
    node distance=0.6cm and 1cm, % Even tighter spacing
    every node/.style={rectangle, rounded corners, draw=black, text width=2.8cm, align=center, minimum height=0.5cm, font=\small},
    >=stealth
]
% Nodes
\node (start) {Patient \( ij \) enrolled with baseline covariate \( \mathbf{X}_{ij} \)};
\node (followup1) [below left=of start, text width=2.6cm, yshift=-0.1cm] {Screening \( Q_{ij} = 1 \) at last follow-up };
\node (followup2) [below right=of start, text width=2.6cm, yshift=-0.1cm] {Screening \( Q_{ij} = 0 \) at last follow-up };
\node (examined) [below=of followup1, text width=2.6cm, yshift=-0.1cm] {Patient examined: \( M_{ij}^{\Delta} = 0 \)};
\node (notexamined) [right=of examined, text width=2.6cm, yshift=0.1cm] {Patient not examined: \\ \( M_{ij}^{\Delta} = 1 \Rightarrow \Delta_{ij} \text{ missing} \)};
\node (censored) [below=of followup2, text width=2.6cm, yshift=-0.1cm] {Patient censored: \\ \( \Rightarrow \Delta_{ij} = 0 \)};
\node (diagnosed) [below left=of examined, text width=2.6cm, yshift=-0.1cm] {Diagnosed with disease \\ \( \Rightarrow \Delta_{ij} = 1 \)};
\node (notdiagnosed) [below right=of examined, text width=2.6cm, yshift=-0.1cm] {No evidence of disease \\ \( \Rightarrow \Delta_{ij} = 0 \)};

% Arrows
\draw[->] (start) -- (followup1);
\draw[->] (start) -- (followup2);
\draw[->] (followup1) -- (examined);
\draw[->] (followup1) -- (notexamined);
\draw[->] (followup2) -- (censored);
\draw[->] (examined) -- (diagnosed);
\draw[->] (examined) -- (notdiagnosed);

\end{tikzpicture}
\end{figure}

We consider the following general family of the censoring mechanisms. 

\begin{assumption}[A General Censoring Mechanism]\label{assump: censoring}
     Conditional on the finite-population dataset in hand, there exists an unknown map $\psi: \mathbb{R}^{N \times (p_{x}+p_{u}+1)} \rightarrow \mathbb{R}^{N}$ such that $\mathbf{C}=\psi(\mathbf{X}, \mathbf{U}, \mathbf{T})$. 
\end{assumption}
The censoring mechanisms stated in Assumption~\ref{assump: censoring} are flexible as they allow arbitrary forms of dependence of censoring on observed and unobserved covariates, the true survival time, and interference between study subjects. 

In addition, we consider the following general family of mechanisms of missingness in censoring indicators. 

\begin{assumption}[A General Missingness Mechanism for Censoring Status]\label{assump: missingness of censoring}
     Conditional on the finite-population dataset in hand, there exists an unknown map $\upsilon: \mathbb{R}^{N \times (p_{x}+p_{u}+2)} \rightarrow \{0,1\}^{N}$ such that $\mathbf{M}^{\Delta}=\upsilon(\mathbf{X}, \mathbf{U}, \mathbf{T}, \mathbf{C})$. 
\end{assumption}
Assumption~\ref{assump: missingness of censoring} allows missingness in censoring indicators to depend on observed and unobserved covariates, the true survival time, the true censoring time, and the interference between study subjects. Moreover, there are no restrictions on the forms of such dependence as the map $\upsilon$ in Assumption~\ref{assump: missingness of censoring} can be any map from $\mathbb{R}^{N \times (p_{x}+p_{u}+2)}$ to $\{0,1\}^{N}$.

In design-based causal inference with survival data, Fisher's sharp null of no treatment effect on survival time (\citealp{rosenbaum2002covariance, li2023randomization}) claims that:
\begin{equation*}
    \widetilde{H}_{0}: T_{ij}(1)=T_{ij}(0) \text{ for all $i,j$}.
\end{equation*}
When the censoring status information is complete, i.e., when $M^{\Delta}_{ij}=0$ for all $i, j$, a commonly used family of test statistics for testing $\widetilde{H}_{0}$ is the weighted logrank test statistics (Chapter 7, \citet{fleming2013counting}), which includes the logrank test, the Gehan-Wilcoxon test, and the Wilcoxon-Prentice test as specific examples, among many others. Specifically, for any $t\geq 0$, let 
\begin{align*}
    &\overline{N}_{1}(t)=\sum_{i=1}^{I}\sum_{j=1}^{n_{i}}Z_{ij}\mathbbm{1}\{R_{ij}\geq t \},\\
    &\overline{N}_{0}(t)=\sum_{i=1}^{I}\sum_{j=1}^{n_{i}}(1-Z_{ij})\mathbbm{1}\{R_{ij}\geq t \},\\
    &\overline{N}(t)=\overline{N}_{1}(t)+\overline{N}_{0}(t)=\sum_{i=1}^{I}\sum_{j=1}^{n_{i}}\mathbbm{1}\{R_{ij}\geq t \}
\end{align*}
be the number of subjects at risk at time $t$ in the treated, control, and both groups, respectively. Let 
\begin{align*}
    &\overline{D}_{1}(t)=\sum_{i=1}^{I}\sum_{j=1}^{n_{i}}Z_{ij}\Delta_{ij} \mathbbm{1}\{R_{ij}= t \},\\
    &\overline{D}_{0}(t)=\sum_{i=1}^{I}\sum_{j=1}^{n_{i}}(1-Z_{ij})\Delta_{ij}\mathbbm{1}\{R_{ij}= t \},\\
    &\overline{D}(t)=\overline{D}_{1}(t)+\overline{D}_{0}(t)=\sum_{i=1}^{I}\sum_{j=1}^{n_{i}}\Delta_{ij}\mathbbm{1}\{R_{ij}= t \}
\end{align*}
be the number of subjects having events at time $t$ in the treated, control, and both groups, respectively. Let $t_{1}< t_{2}< \dots < t_{d}$ be the distinct values of the observed event times. A weighted logrank test statistic takes the form 
\begin{equation*}
    A_{W}(\mathbf{Z}, \mathbf{R}, \mathbf{\Delta})=\sum_{l=1}^{d}W(t_{l})\{\overline{D}_{1}(t_{l})-\overline{E}(t_{l})\}, \ \text{where $\overline{E}(t)=\frac{\overline{D}(t)\overline{N}_{1}(t) }{\overline{N}(t)}$},
\end{equation*}
and $W(t)$ is a prespecified weight function. For example, if $W(t)\equiv 1$, we get the logrank test (\citealp{mantel1966evaluation}). If $W(t)=\overline{N}(t)/(N+1)$, we get the Gehan-Wilcoxon test (\citealp{gehan1965generalized}), which reduces to the two-sample Wilcoxon test in uncensored data. If $W(t)=\widehat{S}(t-)=\prod_{l: t_{l}< t}\{1-\frac{\overline{D}(t_l)}{\overline{N}(t_l)}\}$ (i.e., $W(t_{l})=\prod_{l^{\prime}=1}^{l-1}\{1-\frac{\overline{D}(t_{l^{\prime}})}{\overline{N}(t_{l^{\prime}})}\}$ for $l\geq 2$ and $W(t_{l})=1$ for $l=1$), we get the Prentice-Wilcoxon test (\citealp{prentice1978linear}), which uses the Kaplan-Meier estimator as the weight. Under Assumption~\ref{assump: randomization design}, Assumption~\ref{assump: censoring}, and Fisher's sharp null $\widetilde{H}_{0}$, both $\mathbf{R}$ and $\mathbf{\Delta}$ are fixed under different treatment assignments $\mathbf{Z}$. Therefore, when there is no missingness in censoring indicators, a randomization test using the weighted logrank test statistic $A_{W}(\mathbf{Z}, \mathbf{R}, \mathbf{\Delta})$ is finite-population-exact for testing $\widetilde{H}_{0}$ under Assumptions~\ref{assump: randomization design} and \ref{assump: censoring}.

We consider the setting that missingness of censoring indicators may exist at the final screening stage (i.e., at time $t_{d}$ of the study). Specifically, let $M^{\Delta}_{ij}$ denote the missingness indicator of subject $ij$: $M^{\Delta}_{ij}=1$ if the censoring indicator $\Delta_{ij}$ is missing and $M^{\Delta}_{ij}=0$ otherwise. In this case, for a subject $ij$ with $M^{\Delta}_{ij}=1$, we have either $\Delta_{ij}=1$ (i.e., $T_{ij}=C_{ij}=R_{ij}=t_{d}$) or $\Delta_{ij}=0$ (i.e., $T_{ij}>C_{ij}=R_{ij}=t_{d}$). For subject $ij$, we denote the observed realized censoring indicator as $\Delta^{*}_{ij}$, which may be subject to missingness at time $t_{d}$ (i.e., the last follow-up). That is, the realized value $\Delta_{ij}^{*}$ equals the true value $\Delta_{ij}$ if $M_{ij}^{\Delta}=0$, and equals ``Missing" if $M_{ij}^{\Delta}=1$. Let $\mathbf{\Delta}^{*}=(\Delta^{*}_{11},\dots, \Delta^{*}_{In_{I}})$. In Algorithm~\ref{alg: iart for censored data}, we extend the imputation and re-imputation framework proposed in Section~\ref{subsec: PART} to conduct a design-based test for $\widetilde{H}_{0}$ with missing censoring indicators.

\begin{algorithm} %\SetAlgoNoLine
\SetAlgoLined
\caption{The ``imputation and re-imputation" framework for censored data with missingness in censoring indicators.} \label{alg: iart for censored data}
\vspace*{0.12 cm}
\KwIn{A prespecified number of re-imputation runs $L$ (e.g., $L=10{,}000$), some chosen imputation algorithm $\mathcal{G}$, some chosen test statistics $A$ (e.g., $A$ can be a weighted logrank test $A_{W}$), the randomization design $\mathcal{P}$ (see Assumption~\ref{assump: randomization design}), the observed treatment indicators $\mathbf{Z}$, the observed covariates $\mathbf{X}^{*}$ with possible missingness, the observed time at risk $\mathbf{R}$, and the observed censoring indicators with missingness $\mathbf{\Delta}^{*}$. Note that $\mathbf{\Delta}^{*}$ contain all the observed missingness status information $\mathbf{M}^{\Delta}$. }
\vspace*{0.12 cm}
\begin{enumerate}
    \item Use $(\mathbf{Z}, \mathbf{X}^{*}, \mathbf{R}, \mathbf{\Delta}^{*})$ and the chosen imputation algorithm $\mathcal{G}$ for imputing the missing censoring indicators ($\mathcal{G}$ can be any imputation algorithm) to obtain the full imputed censoring indicators $\widehat{\mathbf{\Delta}}$ at the final testing stage (i.e., time $t_{d}$): 
    \begin{equation*}
 \mathcal{G}: (\mathbf{Z}, \mathbf{X}^{*}, \mathbf{R}, \mathbf{\Delta}^{*})\mapsto \widehat{\mathbf{\Delta}}=(\widehat{\Delta}_{11}, \dots, \widehat{\Delta}_{In_{I}}). \quad \text{(The Imputation Step)}
    \end{equation*}
    We then calculate $a=A(\mathbf{Z}, \mathbf{R}, \widehat{\mathbf{\Delta}})$ based on $\widehat{\mathbf{\Delta}}$.
    \item For each $l=1, \dots, L$:
\begin{enumerate}
        \item Randomly generate $\mathbf{Z}^{(l)}$ according to the randomization design $\mathcal{P}$.
    \item Obtain the full imputed censoring indicators $\widehat{\mathbf{\Delta}}^{(l)}$ based on $(\mathbf{Z}^{(l)}, \mathbf{X}^{*}, \mathbf{R}, \mathbf{\Delta}^{*})$ and algorithm $\mathcal{G}$:
    \begin{equation*}
  \mathcal{G}: (\mathbf{Z}^{(l)}, \mathbf{X}^{*}, \mathbf{R}, \mathbf{\Delta}^{*})\mapsto \widehat{\mathbf{\Delta}}^{(l)}=(\widehat{\Delta}_{11}^{(l)}, \dots, \widehat{\Delta}_{In_{I}}^{(l)}). \quad \text{(The Re-Imputation Step)}
    \end{equation*}
    \item Calculate $A^{(l)}=A(\mathbf{Z}^{(l)}, \mathbf{R}, \widehat{\mathbf{\Delta}}^{(l)})$ for the $l$-th re-imputation run.
    \end{enumerate} 
\end{enumerate}
\textbf{Output:} The approximate finite-population-exact $p$-value $\widehat{p}=\frac{1}{L}\sum_{l=1}^{L}\mathbbm{1}\{A^{(l)}\geq a\}$.
\end{algorithm}

The following Theorem~\ref{thm: hypothesis testing for censored data} gives the theoretical guarantee of finite-population-exact type-I error rate control for testing Fisher's sharp null $\widetilde{H}_{0}$ using Algorithm~\ref{alg: iart for censored data}. The detailed proof is in Appendix G.

%https://www.karlin.mff.cuni.cz/~vavraj/cda/exercise_05.html
\begin{theorem}\label{thm: hypothesis testing for censored data}
   Consider the imputation and re-imputation framework in Algorithm~\ref{alg: iart for censored data} paired with any chosen imputation algorithm $\mathcal{G}: (\mathbf{Z}, \mathbf{X}^{*}, \mathbf{R}, \mathbf{\Delta}^{*})\mapsto \widehat{\mathbf{\Delta}}$ and any chosen test statistic $A=A(\mathbf{Z}, \mathbf{R}, \widehat{\mathbf{\Delta}})$. Let $a$ be the observed value of $A$ based on the observed data $(\mathbf{Z}, \mathbf{R}, \mathbf{\Delta}^{*})$ and imputation algorithm $\mathcal{G}$. Let $p=P(A\geq a\mid \widetilde{H}_{0})$ denote the true finite-population-exact $p$-value under $\widetilde{H}_{0}$. For the approximate $p$-value $\widehat{p}$ reported by Algorithm~\ref{alg: iart for censored data}, under Assumptions \ref{assump: randomization design}, \ref{assump: censoring}, and ~\ref{assump: missingness of censoring}, we have $\widehat{p}\xrightarrow{a.s.} p$ as the number of re-imputation runs $L\rightarrow \infty$. Moreover, for any $\epsilon>0$ and for all $L$, we have $P(| \widehat{p}-p|\geq \epsilon)\leq 2\exp(-2L\epsilon^{2})$.
\end{theorem}

We conduct a simulation study to assess the type-I error rate control and finite-sample power of Algorithm~\ref{alg: iart for censored data}. Similar to the simulation studies in Section~\ref{subsec: simulation studies} in the main text, we consider a stratified randomized experiment with $I$ strata. Each stratum has $10$ subjects, among which we randomly assign the treatment to $5$ subjects and control to the remaining $5$ subjects. 

We follow the same procedure as that in Section~\ref{subsec: simulation studies} to generate the five-dimensional observed covariates $(x_{ij1}, \dots, x_{ij5})$ for each subject $j$ in stratum $i$: $(x_{ij1}, x_{ij2}) \overset{\text{i.i.d.}}{\sim} \mathcal{N}\left(\begin{pmatrix} \frac{1}{2} \\ -\frac{1}{3} \end{pmatrix}, \begin{pmatrix} 1 & \frac{1}{2} \\ \frac{1}{2} & 1 \end{pmatrix}\right)$, $(x_{ij3}, x_{ij4}) \overset{\text{i.i.d.}}{\sim} \text{Laplace}\left(\begin{pmatrix} 0 \\ \frac{1}{\sqrt{3}} \end{pmatrix}, \begin{pmatrix} 1 & \frac{1}{\sqrt{2}} \\ \frac{1}{\sqrt{2}} & 1 \end{pmatrix}\right)$, and $x_{ij5} \overset{\text{i.i.d.}}{\sim} \text{Bernoulli}(1/3)$. Similar to Section~\ref{subsec: simulation studies}, we also generate an aggregate unobserved covariate $u_{ij}$ for each subject $ij$: $u_{ij} \overset{\text{i.i.d.}}{\sim} N(0, 0.2)$. In the outcome-generating process, we still include a stratum-level random effect $\alpha_{i}$ and an individual-level random effect $\epsilon_{ij}$: $\alpha_{i} \overset{\text{i.i.d.}}{\sim} N(0, 0.1)$ and $\epsilon_{ij} \overset{\text{i.i.d.}}{\sim} N(0, 0.2)$. Therefore, the total variance of unobserved terms is $\text{var}(u_{ij})+\text{var}(\alpha_{i})+\text{var}(\epsilon_{ij})=0.5$ (we will normalize the coefficients of the observed covariates to make the variance contributed by observed terms and that contributed by unobserved terms comparable). 

We consider the following data-generating process consisting of: 
\begin{itemize}
    \item 1) A model for generating the true time of event $T_{ij}$: $T_{ij}\sim \text{Poisson}(8 \cdot|\lambda_{T, ij}|)$, where $\lambda_{T,ij}=\beta Z_{ij} \left(1 + x_{ij1} + \frac{1}{\sqrt{5}} \sum_{p=1}^{5} |x_{ijp}| \right) + \frac{1}{\sqrt{5}} \sum_{p=1}^{5} x_{ijp} + \frac{1}{5} \sum_{p=1}^{5} \sum_{p'=1}^{5} x_{ijp} \sigma(1 - x_{ijp^{\prime}}) + u_{ij} + \alpha_{i}+\epsilon_{ij}$;
    \item 2) A model for generating the censoring time $C_{ij}$: $C_{ij}=\widetilde{C}_{ij}\cdot \mathbbm{1}\{\widetilde{C}_{ij}<10\}+10\cdot \mathbbm{1}\{\widetilde{C}_{ij}\geq 10\}$, where $\widetilde{C}_{ij}\sim \text{Poisson}(0.6 \cdot |\lambda_{Cij}|)$ with $\lambda_{Cij}=\frac{1}{\sqrt{5}} \sum \limits_{p=1}^{5} p x_{ijp} + \frac{1}{\sqrt{5}} \sum \limits_{p=1}^{5} p \cos(x_{ijp}) + 10 \sigma(T_{ij}) + u_{ij} + \frac{1}{3}\sum \limits_{j'=1}^{10}x_{ij'1} +  \frac{1}{3}\sum \limits_{j'=1}^{10}T_{ij'}$;
    \item 3) A model for generating the missingness indicator $M^{\Delta}_{ij}$ for censoring status at time $t_{d}$: $M_{ij}^{\Delta}=\mathbbm{1}\{\frac{1}{\sqrt{5}} \sum \limits_{p=1}^{5} p \cos(x_{ijp}) + 10 \sigma(T_{ij}) + 10 \sigma(C_{ij}) + u_{ij} + \frac{1}{3}\sum \limits_{j'=1}^{10}x_{ij'1} +  \frac{1}{3}\sum \limits_{j'=1}^{10}T_{ij'}+ \frac{1}{3}\sum \limits_{j'=1}^{10}C_{ij'}>\lambda_{\Delta}\}$, where $\lambda_{\Delta}$ is a tuning parameter for controlling the outcome missingness rate (we set $\lambda_{\Delta}$ such that the missingness rate is $50\%$ for each simulated dataset), and $\sigma(x)=\exp(x)/(1+\exp(x))$. 
\end{itemize} 

Note that the above data-generating process satisfies Assumption~\ref{assump: randomization design} (randomization design), Assumption~\ref{assump: censoring} (censoring mechanism), and Assumption~\ref{assump: missingness of censoring} (mechanism of missingness in censoring indicators). We consider two small sample size scenarios in which we set the total sample size $N=50$ (corresponding to $I=5$) and $N=200$ (corresponding to $I=20$), and two large sample size scenarios in which we set $N=1000$ (corresponding to $I=100$) and $N=2000$ (corresponding to $I=200$). 

In each model and each simulation scenario, we implement the following four methods for design-based hypothesis testing with missing outcomes (for all  methods, we use the Prentice-Wilcoxon test statistic): 
\begin{itemize}
    \item Method S1 (Non-Informative Imputation): The classic randomization-based Prentice-Wilcoxon test by imputing all the missing censoring status $\Delta_{ij}=\mathbbm{1}\{T_{ij}\leq C_{ij}\}$ with value $0$ (i.e., treat all lost to follow-up as censoring). 
    
    \item Method S2 (Algo 3 -- Logistic): The imputation-assisted randomization-based Prentice-Wilcoxon test based on the imputation and re-imputation framework described in Algorithm~\ref{alg: iart for censored data}, in which we set the embedded outcome imputation algorithm $\mathcal{G}$ to be chained equations imputation based on logistic regression. 
    
    \item Method S3 (Algo 3 -- Boosting): The imputation-assisted randomization-based Prentice-Wilcoxon test based on the imputation and re-imputation framework described in Algorithm~\ref{alg: iart for censored data}, in which we set the embedded outcome imputation algorithm $\mathcal{G}$ to be chained equations imputation based on a boosting algorithm for classification. Specifically, in the small sample size scenarios (when $N=50$ and $N=100$), the boosting algorithm we choose is the commonly used XGBoost algorithm \citep{chen2016xgboost}. In the large sample size scenarios (when $N=1000$ and $N=2000$), the boosting algorithm we use is the widely used LightGBM algorithm \citep{ke2017lightgbm}, a newer gradient boosting decision tree algorithm with less computational cost than XGBoost, which is particularly suitable for implementing more re-imputation runs (i.e., larger $L$) in Algorithm~\ref{alg: iart for censored data} with larger datasets. 
    
    \item Method S4 (Oracle): The classic randomization-based Prentice-Wilcoxon test based on the complete (oracle) censoring status information. 
\end{itemize}

\setcounter{figure}{10} % Set the previous figure number to 10 so the next is 11
\begin{figure}[H]
      \centering
        \captionsetup[subfigure]{skip=2pt} % Adjust this value as needed

        \includegraphics[width=\textwidth]{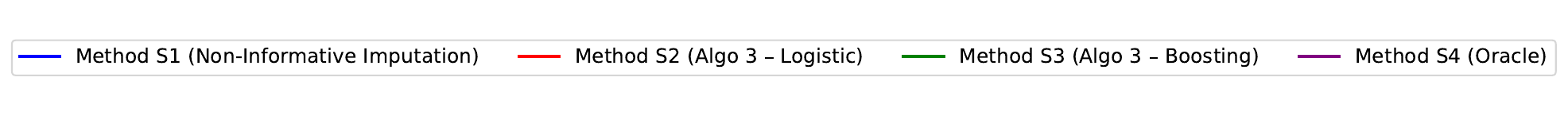}\vspace{-10pt} % Adjust the gap
      
      \begin{subfigure}[b]{0.4\textwidth}
        \includegraphics[width=\textwidth]{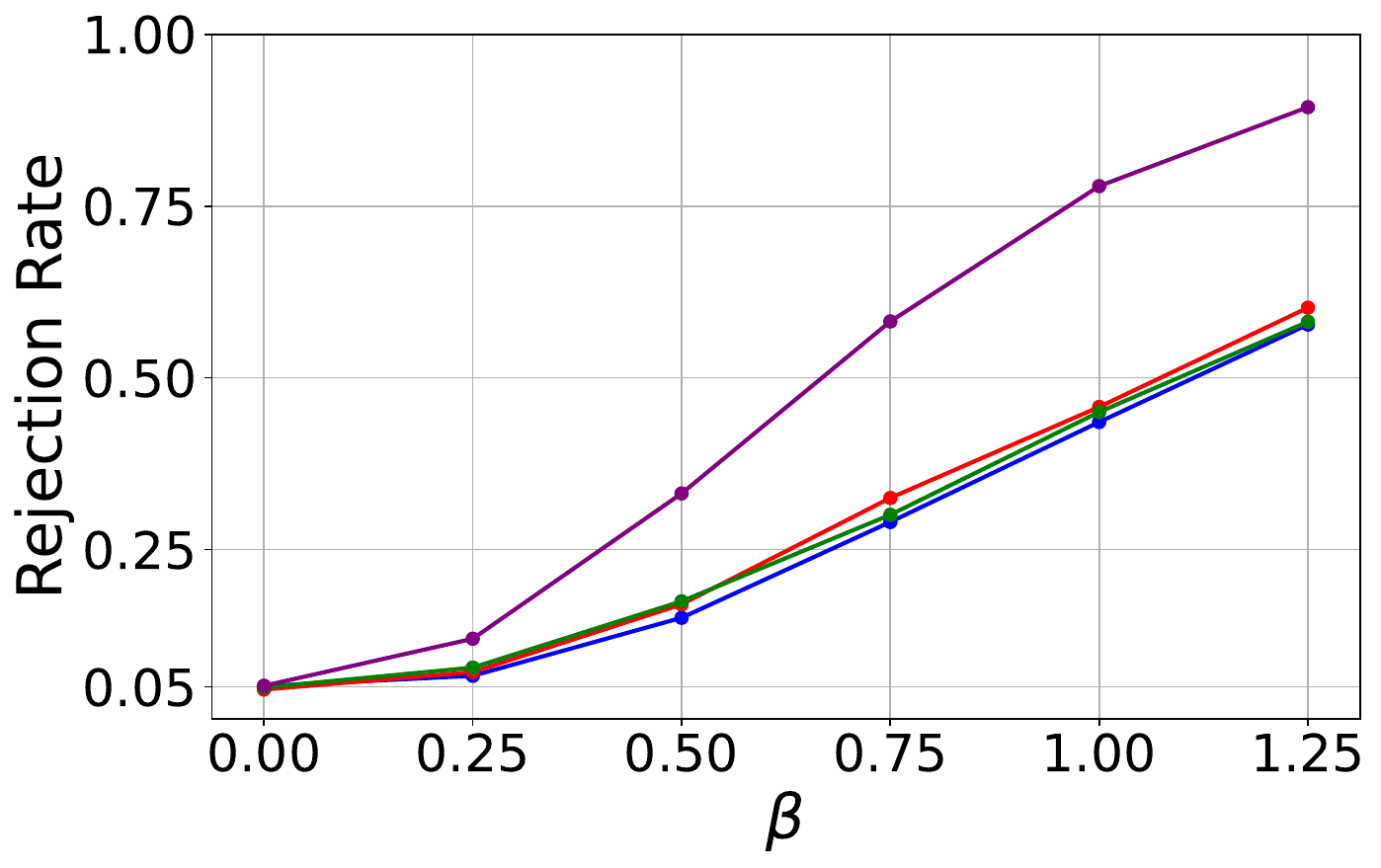}
        \caption{Survival Model ($N=50$)}
      \end{subfigure}
      \hspace{0.1cm}
      \begin{subfigure}[b]{0.4\textwidth}
        \includegraphics[width=\textwidth]{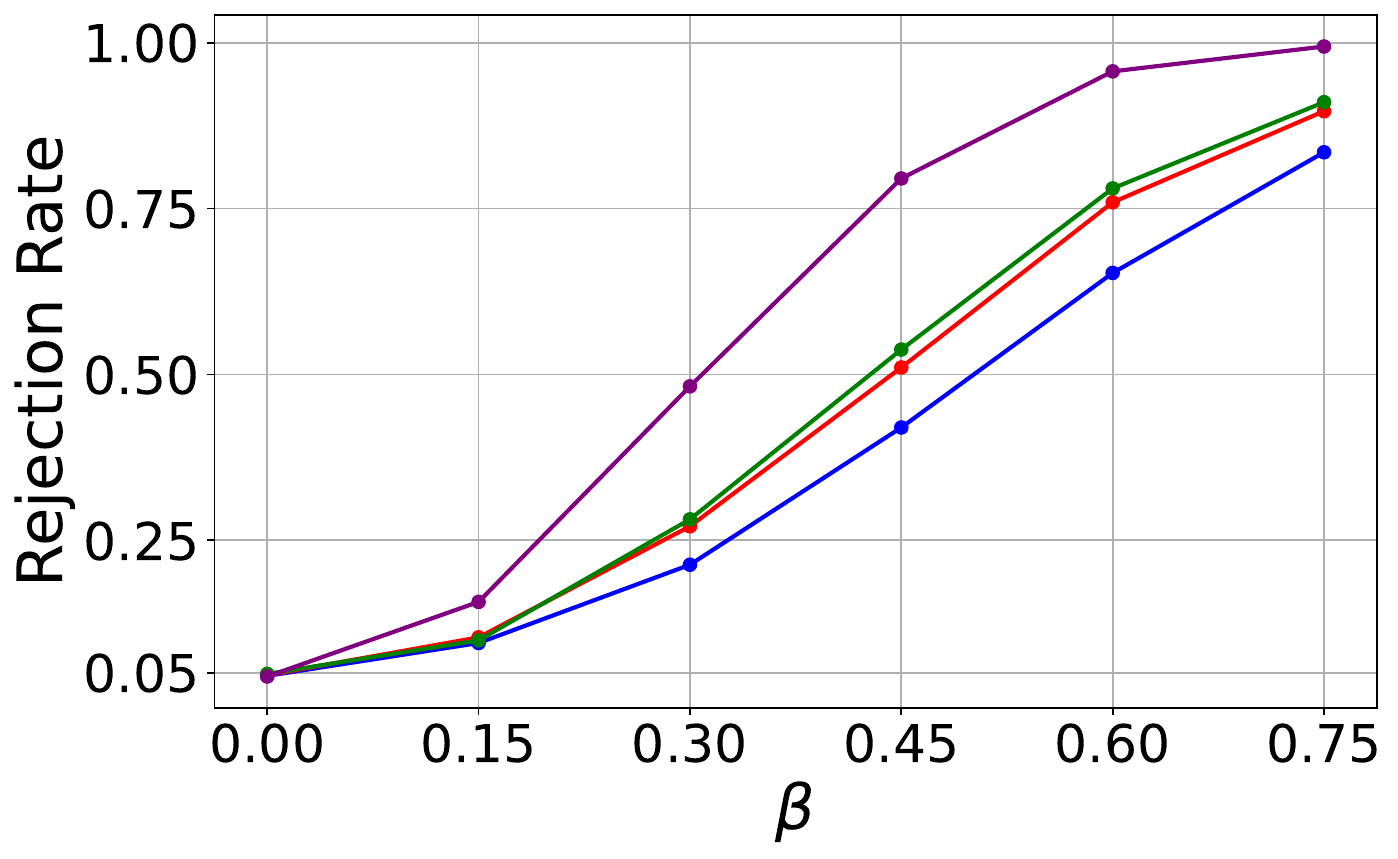}
        \caption{Survival Model ($N=200$)}
      \end{subfigure}

      \begin{subfigure}[b]{0.4\textwidth}
        
        \includegraphics[width=\textwidth]{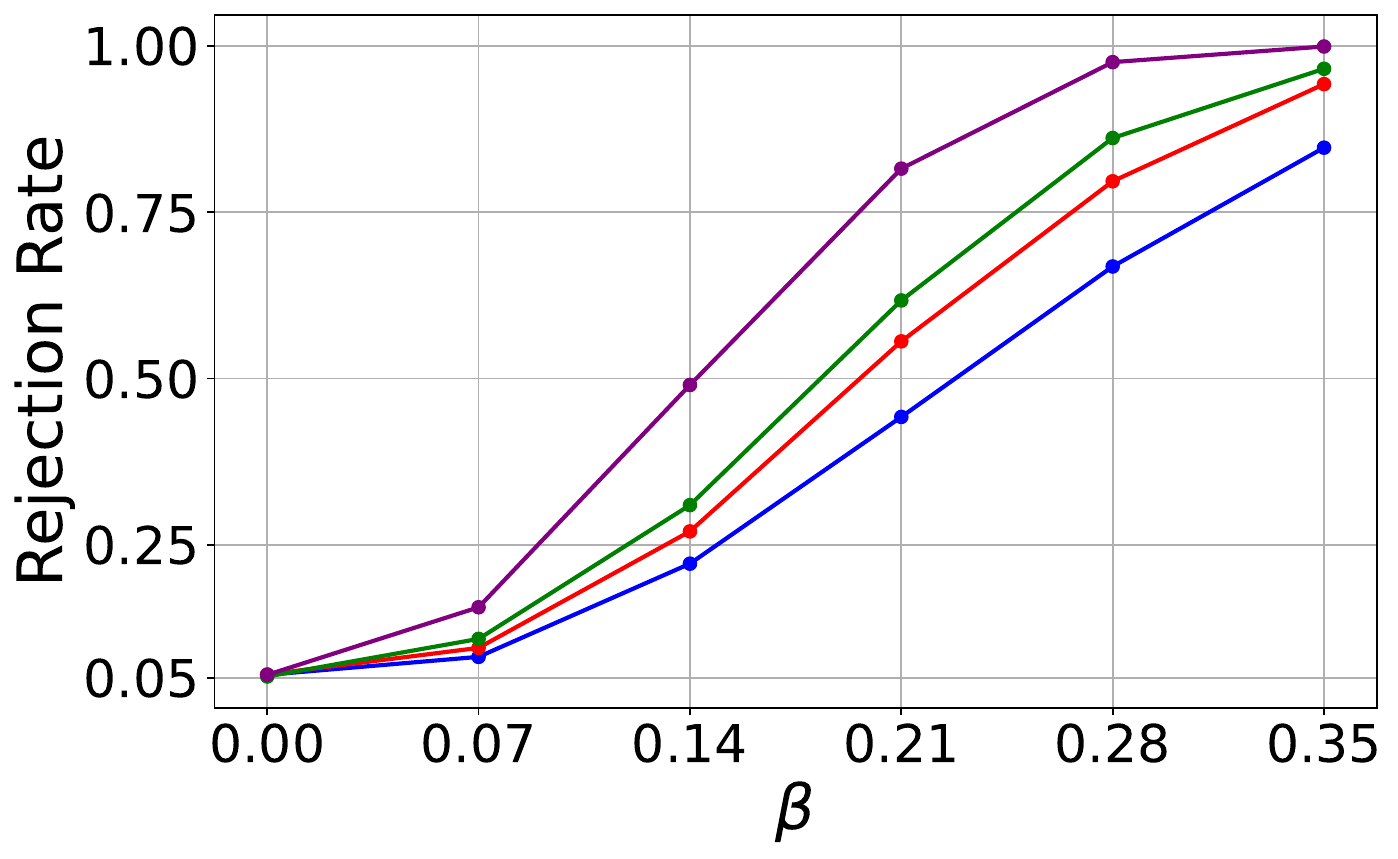}
        \caption{Survival Model ($N=1000$)}
      \end{subfigure}
      \hspace{0.1cm}
      \begin{subfigure}[b]{0.4\textwidth}
        
        \includegraphics[width=\textwidth]{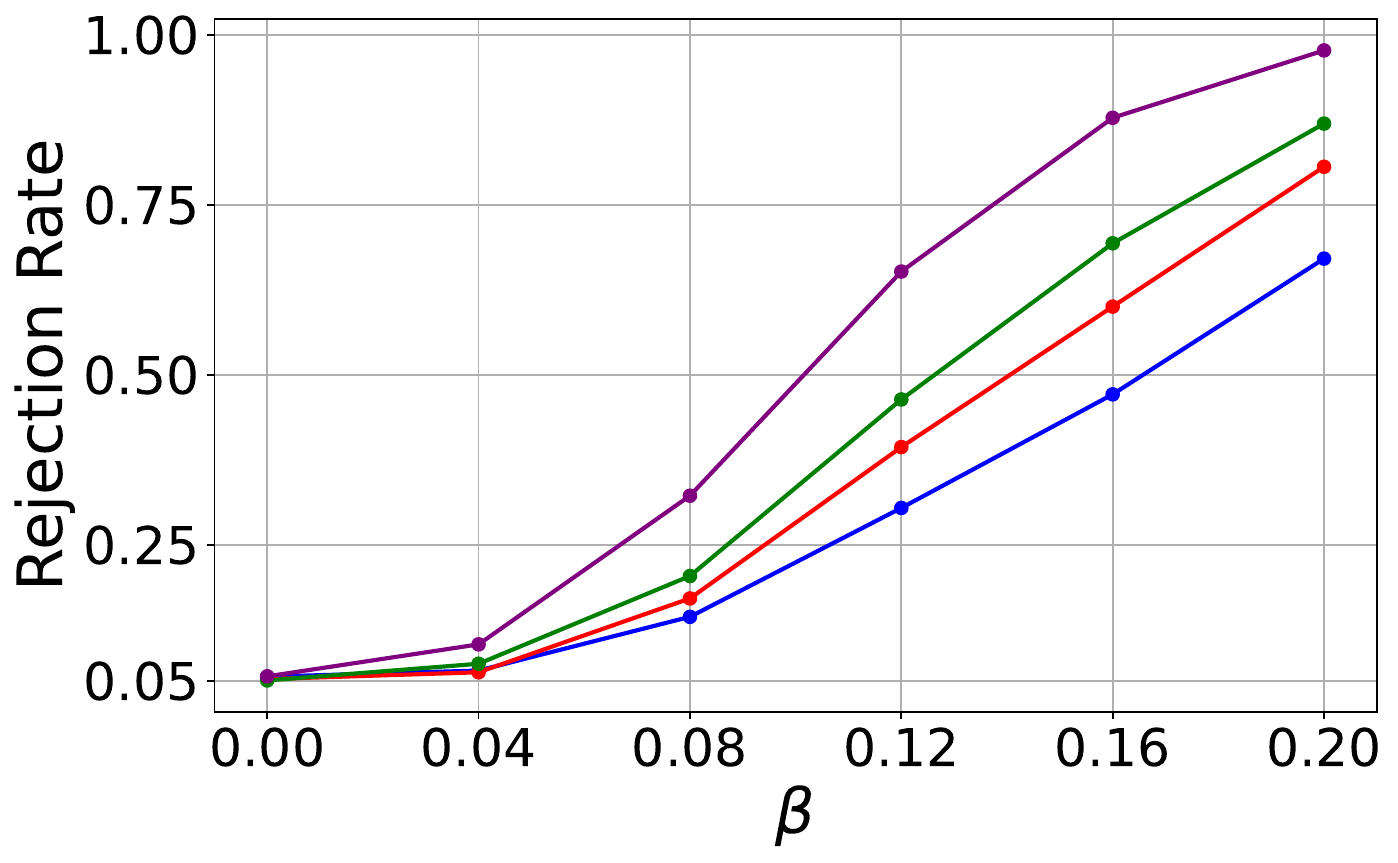}
        \caption{Survival Model ($N=2000$)}
      \end{subfigure}
      \caption{Type-I error rate (when effect size $\beta=0$) and power (when effect size $\beta>0$) of Methods S1--4 under survival model using Algorithm~\ref{alg: iart for censored data} with sample size $N=50$, $N=200$, $N=1000$ and $N=2000$ (level $\alpha=0.05$). }
  \label{fig: censoring simulations}
\end{figure}

The simulation results are presented in Figure~\ref{fig: censoring simulations}, which are based on $2000$ simulated datasets and $L=10{,}000$ re-imputation runs. Figure~\ref{fig: censoring simulations} delivers two important messages: First, the proposed imputation and re-imputation framework described in Algorithm~\ref{alg: iart for censored data}, when applied to the Prentice-Wilcoxon test for survival data with missingness in censoring indicators, can ensure finite-population-exact type-I error rate control with either parametric imputation model (e.g., logistic model) or nonparametric boosting imputation method under the considered data generating process that allows existence of non-linearity, unobserved covariates, outcome-dependent censoring, outcome-dependent missingness in censoring, and interference between subjects. This confirms the theoretical guarantee of type-I error rate control stated in Theorem~\ref{thm: hypothesis testing for censored data}. Second, compared with the commonly used non-informative imputation method, the proposed imputation and re-imputation framework described in Algorithm~\ref{alg: iart for censored data} can improve the statistical power of randomization-based Prentice-Wilcoxon test (without sacrificing the exact type-I error rate control) by incorporating the covariate information into the imputation procedure for missingness in censoring indicators. Such gains in statistical power become more substantial as the sample size increases because more samples are available for training a powerful imputation model used in the imputation and re-imputation framework.

\section*{Appendix G: Proof of Theorem~\ref{thm: hypothesis testing for censored data}}

To prove Theorem~\ref{thm: hypothesis testing for censored data}, we first prove the following lemma that gives an explicit form of the finite-population-exact distribution of the test statistic $A(\mathbf{Z}, \mathbf{R}, \widehat{\mathbf{\Delta}})$ (based on imputed censoring indicators) under $\widetilde{H}_{0}$.

\begin{lemma}\label{lem: p-value under censoring}
  Let $\mathcal{G}: (\mathbf{Z}, \mathbf{X}^{*}, \mathbf{R}, \mathbf{\Delta}^{*})\mapsto \widehat{\mathbf{\Delta}}$ be any imputation algorithm for obtaining the imputed censoring indicators $\widehat{\mathbf{\Delta}}$, and $A(\mathbf{Z}, \mathbf{R}, \widehat{\mathbf{\Delta}})$ be any test statistic based on $\mathbf{Z}$, $\mathbf{R}$, and $\widehat{\mathbf{\Delta}}$. Let $P(A(\mathbf{Z}, \mathbf{R}, \widehat{\mathbf{\Delta}})\geq a \mid \widetilde{H}_{0})$ be the finite-population-exact $p$-value under $\widetilde{H}_{0}$ given the observed value $a$ of $A(\mathbf{Z}, \mathbf{R}, \widehat{\mathbf{\Delta}})$. Under Assumptions \ref{assump: randomization design}, \ref{assump: censoring}, and \ref{assump: missingness of censoring}, if the imputation algorithm $\mathcal{G}$ is deterministic (i.e., $\widehat{\mathbf{\Delta}}=\mathcal{G}(\mathbf{Z}, \mathbf{X}^{*}, \mathbf{R}, \mathbf{\Delta}^{*})$ is a deterministic vector given $(\mathbf{Z}, \mathbf{X}^{*}, \mathbf{R}, \mathbf{\Delta}^{*})$), we have
    \begin{equation}\label{eqn: one-sided exact p-value deterministic under censoring}
    P(A(\mathbf{Z}, \mathbf{R}, \widehat{\mathbf{\Delta}})\geq a \mid \widetilde{H}_{0})=\sum_{\mathbf{z}\in \Omega}\big[\mathbbm{1}\{A(\mathbf{Z}=\mathbf{z}, \mathbf{R}, \widehat{\mathbf{\Delta}}=\mathcal{G}(\mathbf{Z}=\mathbf{z}, \mathbf{X}^{*}, \mathbf{R}, \mathbf{\Delta}^{*}))\geq a\}\times P(\mathbf{Z}=\mathbf{z})\big],
    \end{equation}
    and if the imputation algorithm $\mathcal{G}$ is stochastic (i.e., $\widehat{\mathbf{\Delta}}=\mathcal{G}(\mathbf{Z}, \mathbf{X}^{*}, \mathbf{R}, \mathbf{\Delta}^{*})$ is a random vector conditional on $(\mathbf{Z}, \mathbf{X}^{*}, \mathbf{R}, \mathbf{\Delta}^{*})$), we have
    \begin{equation}\label{eqn: one-sided exact p-value stochastic under censoring}
     P(A(\mathbf{Z}, \mathbf{R}, \widehat{\mathbf{\Delta}})\geq a \mid \widetilde{H}_{0})=\sum_{\mathbf{z}\in \Omega}\big[P\{A(\mathbf{Z}=\mathbf{z}, \mathbf{R}, \widehat{\mathbf{\Delta}}=\mathcal{G}(\mathbf{Z}=\mathbf{z}, \mathbf{X}^{*}, \mathbf{R}, \mathbf{\Delta}^{*}))\geq a\}\times P(\mathbf{Z}=\mathbf{z})\big].
    \end{equation}
\end{lemma}

\begin{proof}
    We prove Lemma~\ref{lem: p-value under censoring} by discussing the following two cases separately.

    \textbf{Case 1: The imputation algorithm $\mathcal{G}$ is deterministic.}

    To prove that equation (\ref{eqn: one-sided exact p-value deterministic under censoring}) holds, the key point is to show that the realized censoring indicators $\mathbf{\Delta}^{*}=(\Delta_{11}^{*}, \dots, \Delta_{In_{I}}^{*})$ are invariant under any treatment assignments $\mathbf{Z}$. Note that under $\widetilde{H}_{0}$, the true survival time $\mathbf{T}=(T_{11}, \dots, T_{In_{I}})$ are invariant under any $\mathbf{Z}$. Then, based on Assumption~\ref{assump: censoring}, such invariance of $\mathbf{T}$ further implies that the true censoring time $\mathbf{C}=(C_{11}, \dots, C_{In_{I}})$ are also invariant under any $\mathbf{Z}$. Finally, under Assumption~\ref{assump: missingness of censoring}, invariance of $\mathbf{T}$ and $\mathbf{C}$ implies that the censoring status missingness indicators $\mathbf{M}^{\Delta}=(M_{11}^{\Delta}, \dots, M_{In_{I}}^{\Delta})$ are also invariant under different $\mathbf{Z}$. Recall that for each realized censoring indicator $\Delta_{ij}^{*}$, we have $\Delta_{ij}^{*}=\Delta_{ij}=\mathbbm{1}\{T_{ij}\leq C_{ij}\}$ if $M_{ij}^{\Delta}=0$ and $\Delta_{ij}^{*}=\text{``Missing"}$ if $M^{\Delta}_{ij}=1$. Therefore, the realized censoring indicators $\mathbf{\Delta}^{*}=(\Delta_{11}^{*}, \dots, \Delta_{In_{I}}^{*})$ are invariant under any treatment assignments $\mathbf{Z}$. That is, if we let $\mathbf{\Delta}^{*}(\mathbf{z})=(\Delta_{11}^{*}(\mathbf{z}), \dots, \Delta_{In_{I}}^{*}(\mathbf{z}))$ denote the potential realized censoring indicators under $\mathbf{Z}=\mathbf{z}$, we have the observed realized censoring indicators $\mathbf{\Delta}^{*}=\mathbf{\Delta}^{*}(\mathbf{z})$ for any $\mathbf{z}$. Then, we have 
    \begin{align*}
        P(A(\mathbf{Z}, \mathbf{R}, \widehat{\mathbf{\Delta}})\geq a \mid \widetilde{H}_{0})&=\sum_{\mathbf{z}\in \Omega}\big[\mathbbm{1}\{A(\mathbf{Z}=\mathbf{z}, \mathbf{R}, \widehat{\mathbf{\Delta}}=\mathcal{G}(\mathbf{Z}=\mathbf{z}, \mathbf{X}^{*}, \mathbf{R}, \mathbf{\Delta}^{*}(\mathbf{z})))\geq a\}\times P(\mathbf{Z}=\mathbf{z})\big]\\
        &=\sum_{\mathbf{z}\in \Omega}\big[\mathbbm{1}\{A(\mathbf{Z}=\mathbf{z}, \mathbf{R}, \widehat{\mathbf{\Delta}}=\mathcal{G}(\mathbf{Z}=\mathbf{z}, \mathbf{X}^{*}, \mathbf{R}, \mathbf{\Delta}^{*}))\geq a\}\times P(\mathbf{Z}=\mathbf{z})\big].
    \end{align*}
    Therefore, equation (\ref{eqn: one-sided exact p-value deterministic under censoring}) holds. 
    
 \textbf{Case 2: The imputation algorithm $\mathcal{G}$ is stochastic.}

    In Case 1, we have shown that the observed realized outcomes $\mathbf{\Delta}^{*}=\mathbf{\Delta}^{*}(\mathbf{z})$ for any $\mathbf{z}$ when \textit{both} Fisher's sharp null $\widetilde{H}_{0}$ and Assumptions~\ref{assump: censoring} and \ref{assump: missingness of censoring} hold. Then, we have 
    \begin{align*}
        P(A(\mathbf{Z}, \mathbf{R}, \widehat{\mathbf{\Delta}})\geq a \mid \widetilde{H}_{0})&=\sum_{\mathbf{z}\in \Omega}\big[P\{A(\mathbf{Z}=\mathbf{z}, \mathbf{R}, \widehat{\mathbf{\Delta}}=\mathcal{G}(\mathbf{Z}=\mathbf{z}, \mathbf{X}^{*}, \mathbf{R}, \mathbf{\Delta}^{*}(\mathbf{z})))\geq a\}\times P(\mathbf{Z}=\mathbf{z})\big]\\
        &=\sum_{\mathbf{z}\in \Omega}\big[P\{A(\mathbf{Z}=\mathbf{z}, \mathbf{R}, \widehat{\mathbf{\Delta}}=\mathcal{G}(\mathbf{Z}=\mathbf{z}, \mathbf{X}^{*}, \mathbf{R}, \mathbf{\Delta}^{*}))\geq a\}\times P(\mathbf{Z}=\mathbf{z})\big].
    \end{align*}
    Therefore, equation (\ref{eqn: one-sided exact p-value stochastic under censoring}) also holds. 
\end{proof}

Now, we are ready to prove Theorem~\ref{thm: hypothesis testing for censored data}.

\begin{proof}[(Proof of Theorem~\ref{thm: hypothesis testing for censored data})]
    Recall that in Algorithm~\ref{alg: iart for censored data}, we consider 
   \begin{align*}
        \widehat{p}=\frac{1}{L}\sum_{l=1}^{L}\mathbbm{1}\{A^{(l)}\geq a\}=\frac{1}{L}\sum_{l=1}^{L}\mathbbm{1}\{A(\mathbf{Z}^{(l)},\mathbf{R}, \widehat{\mathbf{\Delta}}^{(l)})\geq a\}.
   \end{align*}
  For $l=1,\dots, L$, random variables $\mathbbm{1}\{A^{(l)}\geq a\}$ independently and identically follow Bernoulli distribution. Moreover, applying the arguments in Lemma~\ref{lem: p-value under censoring}, if the imputation algorithm $\mathcal{G}$ is deterministic, we have 
  \begin{align*}
      &\quad \ P(A^{(l)}\geq a \mid \widetilde{H}_{0}) \\
      &=P(A(\mathbf{Z}^{(l)}, \mathbf{R}, \widehat{\mathbf{\Delta}}^{(l)})\geq a \mid \widetilde{H}_{0}) \\
      &= P(A(\mathbf{Z}^{(l)}, \mathbf{R}, \mathcal{G}(\mathbf{Z}^{(l)}, \mathbf{X}^{*}, \mathbf{R}, \mathbf{\Delta}^{*}))\geq a \mid \widetilde{H}_{0}) \quad (\text{$\mathbf{\Delta}^{*}$ is the observed realized censoring indicators}) \\
      &=\sum_{\mathbf{z}\in \Omega}\big[\mathbbm{1}\{A(\mathbf{Z}=\mathbf{z}, \mathbf{R}, \widehat{\mathbf{\Delta}}=\mathcal{G}(\mathbf{Z}=\mathbf{z}, \mathbf{X}^{*}, \mathbf{R}, \mathbf{\Delta}^{*}))\geq a\}\times P(\mathbf{Z}=\mathbf{z})\big]\\
      &=\sum_{\mathbf{z}\in \Omega}\big[\mathbbm{1}\{A(\mathbf{Z}=\mathbf{z}, \mathbf{R}, \widehat{\mathbf{\Delta}}=\mathcal{G}(\mathbf{Z}=\mathbf{z}, \mathbf{X}^{*}, \mathbf{R}, \mathbf{\Delta}^{*}(\mathbf{z})))\geq a\}\times P(\mathbf{Z}=\mathbf{z})\big]\\
      &=P(A(\mathbf{Z}, \mathbf{R}, \widehat{\mathbf{\Delta}})\geq a \mid \widetilde{H}_{0}),
  \end{align*}
  and if the imputation algorithm $\mathcal{G}$ is stochastic, we also have 
  \begin{align*}
      &\quad \ P(A^{(l)}\geq a \mid \widetilde{H}_{0}) \\
      &=P(A(\mathbf{Z}^{(l)}, \mathbf{R}, \widehat{\mathbf{\Delta}}^{(l)})\geq a \mid \widetilde{H}_{0}) \\
      &= P(A(\mathbf{Z}^{(l)}, \mathbf{R}, \mathcal{G}(\mathbf{Z}^{(l)}, \mathbf{X}^{*}, \mathbf{R}, \mathbf{\Delta}^{*}))\geq a \mid \widetilde{H}_{0})  \\
      &=\sum_{\mathbf{z}\in \Omega}\big[P\{A(\mathbf{Z}=\mathbf{z}, \mathbf{R}, \widehat{\mathbf{\Delta}}=\mathcal{G}(\mathbf{Z}=\mathbf{z}, \mathbf{X}^{*}, \mathbf{R}, \mathbf{\Delta}^{*}))\geq a\}\times P(\mathbf{Z}=\mathbf{z})\big]\\
      &=\sum_{\mathbf{z}\in \Omega}\big[P\{A(\mathbf{Z}=\mathbf{z}, \mathbf{R}, \widehat{\mathbf{\Delta}}=\mathcal{G}(\mathbf{Z}=\mathbf{z}, \mathbf{X}^{*}, \mathbf{R}, \mathbf{\Delta}^{*}(\mathbf{z})))\geq a\}\times P(\mathbf{Z}=\mathbf{z})\big]\\
      &=P(A(\mathbf{Z}, \mathbf{R}, \widehat{\mathbf{\Delta}})\geq a \mid \widetilde{H}_{0}).
  \end{align*}
  Therefore, applying the law of large numbers, we have 
  \begin{equation*}
      \text{$\widehat{p} \xrightarrow{a.s.} p$ as the number of re-imputation runs $L\rightarrow \infty$.}
  \end{equation*}
  Also, invoking Hoeffding's inequality (\citealp{hoeffding1994probability}), for any $\epsilon>0$ and any $L$, we have 
  \begin{equation*}
      P(| \widehat{p}-p|\geq \epsilon)\leq 2\exp(-2L\epsilon^{2}).
  \end{equation*}
\end{proof}

\putbib[references]
\end{bibunit}

\end{document}